\documentclass[
11pt, 
oneside, 
english, 
singlespacing, 
headsepline, 
]{MastersDoctoralThesis} 

\usepackage[utf8]{inputenc} 
\usepackage[T1]{fontenc} 

\usepackage{import} 
\usepackage{enumerate} 

\usepackage{dirtytalk} 

\usepackage{stmaryrd} 

\usepackage{float} 

\usepackage{graphicx} 
\usepackage{tcolorbox} 

\usepackage{dsfont}

\usepackage{pdfpages} 

\usepackage{xspace}  
\usepackage{centernot}
\usepackage[scaled=0.86]{berasans}  
\usepackage[scaled=1.03]{inconsolata} 
\usepackage[babel]{microtype}  
\usepackage{amsmath,amssymb,amsthm,bm,mathtools,amsfonts,mathrsfs,bbm} 

\usepackage{environ}
\makeatletter
\NewEnviron{specialalign}
  {\def\align@preamble{%
     &\hfil
      \strut@
      \setboxz@h{\@lign$\m@th\displaystyle{####}$}%
      \ifmeasuring@\savefieldlength@\fi
      \set@field
      \hfil
      \tabskip2\tabcolsep
     &\setboxz@h{\@lign$\m@th\displaystyle{{}####}$}%
      \ifmeasuring@\savefieldlength@\fi
      \set@field
      \hfil
      \tabskip\alignsep@
  }%
  \begin{align}\BODY\end{align}}
\makeatother

\usepackage{mathpazo} 

\usepackage [backend=bibtex,style=numeric-comp,sorting=none,natbib=true] {biblatex} 

\newbibmacro{string+url}[1]{%
  \iffieldundef{url}{#1}{\href{\thefield{url}}{#1}}}
\DeclareFieldFormat{title}{\usebibmacro{string+url}{\mkbibemph{#1}}}
\DeclareFieldFormat[article]{title}{\usebibmacro{string+url}{\mkbibquote{#1}}}
\DeclareFieldFormat[inproceedings]{title}{\usebibmacro{string+url}{\mkbibquote{#1}}}
\DeclareFieldFormat[inbook]{title}{\usebibmacro{string+url}{\mkbibquote{#1}}}

\usepackage[autostyle=true]{csquotes} 

\newcommand{\bra}[1]{\left\langle #1 \right|}
\newcommand{\ket}[1]{\left| #1 \right\rangle}
\newcommand{\braket}[2]{\left\langle #1 \middle| #2 \right\rangle}

\DeclareMathOperator{\Tr}{Tr}

\newtheorem{theorem}{Theorem}
\newtheorem{problem}{Problem}
\newtheorem{conjecture}{Conjecture}
\newtheorem{lemma}{Lemma}
\newtheorem{definition}{Definition}
\newtheorem{corollary}{Corollary}
\newtheorem{proposition}{Proposition}
\newtheorem{question}{Question}
\newtheorem{counterexample}{Counterexample}

\newcommand{\ba}{\begin{eqnarray}}
\newcommand{\ea}{\end{eqnarray}}

\newcommand{\rvline}{\hspace*{-\arraycolsep}\vline\hspace*{-\arraycolsep}}

\DeclareMathOperator*{\argmax}{arg\,max} 

\thesistitle{Optimal Manipulation Of Correlations And Temperature In Quantum Thermodynamics} 
\supervisor{Prof. Nicolas \textsc{Brunner} \\ \& Dr. Marcus \textsc{Huber}} 
\examiner{} 
\degree{Doctor of Physics} 
\author{Fabien \textsc{Clivaz}} 
\addresses{} 

\subject{Physics} 
\university{\href{http://www.unige.ch}{University of Geneva}} 
\department{\href{http://www.unige.ch/sciences/physique/physics-section/research/department-applied-physics/}{Department of Applied Physics}} 
\group{\href{https://www.unige.ch/sciences/quantumcorrelations/}{Group of Quantum Correlations}} 
\faculty{\href{/http://www.unige.ch/sciences/index_en.html}{Sciences}} 

\AtBeginDocument{
\hypersetup{pdftitle=\ttitle} 
\hypersetup{pdfauthor=\authorname} 
\hypersetup{pdfkeywords=\keywordnames} 
}

\begin{document}

\frontmatter 

\pagestyle{plain} 


\begin{titlepage}
\begin{center}

\begin{tabular}{ll}
{\scshape Université de Genève} \hfill & {\scshape Faculté des Sciences}\\
{\scshape Département de Physique Appliquée} \hfill &Professeur Nicolas Brunner\\
{\scshape Département de Physique Appliquée} \\
{\scshape Institue for Quantum Optics and Quantum}\\
{\scshape Information Vienna (IQOQI Vienna)} \hfill &Docteur Marcus Huber\\
\end{tabular}

\HRule \\[0.4cm] 

\centering
\includegraphics{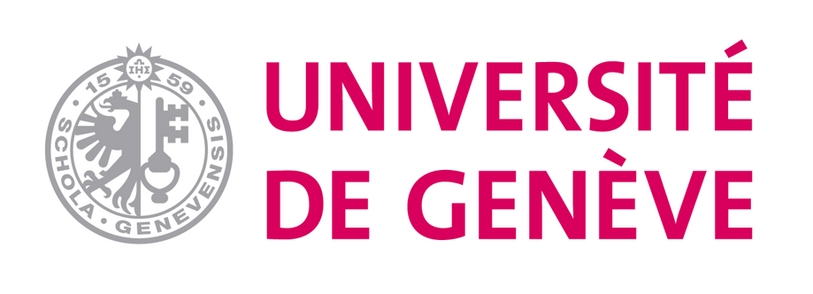} 

\HRule \\[1.4cm] 
{\huge \bfseries \ttitle\par}\vspace{0.4cm} 
\hfill \\[2.5cm] 

{\LARGE THÈSE}\\[0.4cm]
Présentée à la Faculté des sciences de l’Université de Genève

Pour obtenir le grade de Docteur ès sciences, mention Physique\\[1.5cm]

Par\\[1.5cm]
{\LARGE Fabien {\scshape Clivaz}}\\
de\\
Chermignon (VS)\\[1.5cm]

Thèse N°5476\\[0.4cm]

GENÈVE\\
Atelier Repromail, Université de Genève\\
2020


%
%
%
%

 
\vfill
\end{center}
\end{titlepage}

\begin{abstract}

\addchaptertocentry{\abstractname} 


This thesis is devoted to the study of two tasks: refrigeration and the creation of correlations. Both of these tasks are investigated within the realm of quantum thermodynamics. Our approach is influenced by that of quantum information but connections with other approaches are desired and emphasized whenever possible. \\

Cooling is one of the tasks of paramount importance of quantum thermodynamics and thermodynamics itself. Reaching cold temperatures is of undisputed technological relevance and ubiquitous among the many sub-fields of physics, allowing, on the one hand, to reach fascinating states of matter such as superconductivity or Bose-Einstein condensation, while, on the other hand, being a fundamentally imposed background constraint in areas such as astrophysics. In thermodynamics, temperature plays a special role in that it is one of the premises of the theory. While its value is unbounded from above, its statistical physical interpretation sets a hard lower bound that is by now thought of as inherent to the theory. Investigating how this lower bound can be attained within the framework of quantum theory is thus intimately related to the understanding of the fundamental laws that dictate quantum thermodynamics. It is therefore of no surprise that many different cooling schemes emerged from different approaches of quantum thermodynamics in the recent literature on the subject. While each of these approaches work with their own set of assumptions, they all have in common that
\begin{enumerate}
\item they assume a background temperature and with it an available thermal state,
\item the system to be cooled is in some sort open/connected to the environment. \label{enum:opensyst}
\end{enumerate}

The first assumption is actually not obvious at all and understanding how and under which circumstances one can assume a given closed system to be thermal is on its own an exciting and active research field. However, once the existence of a thermal state is accepted, in order for anything to happen to the state, \ref{enum:opensyst}. has to hold. The diversity of the approaches is then formed in how exactly the system is open to its environment. Our take on it consists in dividing the environment in two parts: a \enquote{machine} and the rest of the environment. While the rest of the environment is thought of as being weakly or not controlled in that it is only used to rethermalize the machine, we assume an increased control on the part of the environment we call \enquote{machine}, hence the name. To prevent any hidden energy supply, we furthermore explicitly take the source of energy into account. This gives rise to two variations of our paradigm. In the incoherent version the machine benefits from a fully entropic source of energy and is controlled in an energy conserving manner. In the coherent one, the source of energy is provided in an entropy-less manner and we allow for full unitary control. Both variations of this paradigm are designed to be related to other existing approaches and as such provide a common platform to on the one hand unify these approaches and on the other hand compare their performances on an equal footing.\\

Investigating each of our paradigms we furthermore find a bound that holds for both of our paradigms and that is attainable under minimal assumptions for all finite dimensional systems and machines. The bound is in particular a single letter bound that does not depend on the particular intricacies of the system and machine, and as such identifies the relevant parameter of interest, fulfilling one of the central goal of statistical physics. The bound is furthermore already attainable for minimal machines. Investigating these minimal machines in greater details, we find that the energy spent for achieving this bound is dependent on the level of control and that for non-maximal cooling there is no universally better paradigm that achieves a given temperature in terms of energy expenditure. \\

Our second task of interest, the creation of correlations, is motivated by a fundamental property of our current understanding of nature: correlations lie at the heart of every scientific prediction. It is therefore natural to wonder how much correlations can be created at a finite amount of invested energy. Assuming an initially uncorrelated system in a thermal background, it can be shown that the creation of correlations comes at an energy cost. In other words, no correlation is for free in that context. Viewing correlations as the resource of information theory and energy as that of thermodynamics, the above stipulates that the acquisition of information in a physical system necessarily comes at a thermodynamic cost. While lower bounds on the amount of energy that has to be invested in the process can be formulated, their reachability remains an open question. We here derive a framework that allows to investigate this question more closely for a pair of initially uncorrelated identical systems in a thermal background. This framework is based on decomposing the Hilbert space in a Latin square manner and as such allows us to harness the tools of majorization theory. Doing so, we are able to provide protocols that achieve the lower bound for any 3 dimensional and 4 dimensional systems. We furthermore provide a set of conditions to be fulfilled in order for the bound to be achievable in any dimension.

\end{abstract}

\chapter{R\'esum\'e}

Cette thèse est consacrée à l'étude de deux tâches: la réfrigération et la création de corrélations. Ces deux tâches sont étudiées du point de vue de la thermodynamique quantique. Notre approche est influencée par celle de l'information quantique, mais des connexions avec d'autres approches sont souhaitées et soulignées dans la mesure du possible. \\

Le refroidissement est l'une des tâches de la plus haute importance de la thermodynamique quantique et de la thermodynamique elle-même. Atteindre des températures froides est d'une pertinence technologique incontestée et omniprésente parmi les nombreux sous-domaines de la physique, nous permettant d'une part d'atteindre des états fascinants de la matière tels que la supraconductivité ou la condensation de Bose-Einstein, tout en étant d'autre part tout simplement une contrainte de fond imposée fondamentalement dans des domaines tels que l'astrophysique. En thermodynamique, la température joue un rôle particulier en ce qu'elle est l'une des prémisses de la théorie. Bien que sa valeur ne soit pas bornée en haut, son interprétation physique statistique établit une borne inférieure dure qui est désormais considérée comme inhérente à la théorie. Étudier comment cette limite inférieure peut être atteinte dans le cadre de la théorie quantique est donc intimement lié à la compréhension des lois fondamentales qui dictent la thermodynamique quantique. Il n'est donc pas surprenant que de nombreux schémas de refroidissement différents aient émergé de différentes approches de la thermodynamique quantique dans la littérature récente sur le sujet. Bien que chacune de ces approches fonctionne avec son propre ensemble d'hypothèses, elles ont toutes en commun que
\begin{enumerate}
\item elles supposent une température de fond et avec elle un état thermique disponible,
\item le système à refroidir est d'une manière ou d'une autre ouvert sur l'environnement. \label{enum:opensystfr}
\end{enumerate}

La première hypothèse n'est en fait pas évidente du tout et comprendre comment et dans quelles circonstances on peut supposer qu'un système donné en contact avec un environnement soit thermique est en soi un domaine de recherche passionnant et actif. Cependant, une fois que l'existence d'un état thermique est acceptée, pour que quelque chose arrive à l'état, \ref{enum:opensystfr}. est forcé d'être vrai. La diversité des approches se forme alors dans la façon dont le système est ouvert à son environnement. Notre approche consiste à diviser l'environnement en deux parties: une \enquote{machine} et le reste de l'environnement. Alors que le reste de l'environnement est considéré comme faiblement ou non contrôlé en ce qu'il n'est utilisé que pour rethermaliser la machine, nous supposons un contrôle accru de la part de l'environnement que nous appelons \enquote{machine}, d'où le nom. Pour éviter tout approvisionnement énergétique caché, nous tenons en outre explicitement compte de la source d'énergie. Cela donne lieu à deux variantes de notre paradigme. Dans la version incohérente, la machine bénéficie d'une source d'énergie entièrement entropique et est contrôlée d'une façon conservant l'énergie. Dans la version cohérente, la source d'énergie est fournie sans entropie et nous permettons un contrôle unitaire complet. Les deux variantes de ce paradigme sont conçues pour être liées à d'autres approches existantes et, en tant que telles, fournissent une plate-forme commune qui d'un côté unifie ces approches et, d'autre part, compare leurs performances sur un même pied d'égalité. \\

En étudiant chacun de nos paradigmes en tant que tel, nous trouvons en outre une borne qui s'applique à nos deux paradigmes et qui est réalisable sous des hypothèses minimales pour tous les systèmes et machines de dimension finie. La borne est en particulier de type lettre unique et ne dépend pas des subtilités particulières du système et de la machine, et en tant que telle identifie le paramètre d'intérêt pertinent, remplissant l'un des piliers centraux de la physique statistique. La borne est en outre déjà atteignable pour des machines minimales. En étudiant ces machines minimales de façon plus détaillée, nous constatons que l'énergie dépensée pour atteindre cette borne dépend du niveau de contrôle et que pour un refroidissement non maximal, il n'y a pas de paradigme universellement meilleur qui atteint une température donnée en termes de dépense énergétique. \\

Notre deuxième tâche d'intérêt, la création de corrélations, est motivée par une propriété fondamentale de notre compréhension actuelle de la nature: les corrélations sont au cœur de toute prédiction scientifique. Il est donc naturel de se demander combien de corrélations peuvent être créées pour une quantité finie d'énergie investie. En supposant un système initialement non corrélé dans un contexte thermique, il peut être démontré que la création de corrélations a un coût énergétique. En d'autres termes, aucune corrélation n'est gratuite dans ce contexte. Considérant les corrélations comme la ressource de la théorie de l'information et l'énergie comme celle de la thermodynamique, ce qui précède stipule que l'acquisition d'informations dans un système physique a nécessairement un coût thermodynamique. Alors que des bornes inférieures sur la quantité d'énergie qui doit être investie dans le processus peuvent être formulées, leur atteignabilité reste une question ouverte. Nous dérivons ici un cadre qui permet d'étudier cette question de plus près pour une paire de systèmes identiques initialement non corrélés dans un contexte thermique. Ce cadre est basé sur la décomposition de l'espace de Hilbert de type carré latin et à ce titre nous permet d'exploiter les outils de la théorie de la majorisation. Ce faisant, nous sommes en mesure de fournir des protocoles qui atteignent la borne inférieure pour tous les systèmes de dimension 3 et 4. Nous fournissons en outre un ensemble de conditions à remplir pour que la borne soit atteignable dans n'importe quelle dimension.



\begin{acknowledgements}
\addchaptertocentry{\acknowledgementname} 

I would like to first of all thank both of my supervisors, Nicolas Brunner and Marcus Huber, for their invaluable support throughout my PhD. In particular, thank you to both of them for having given me the freedom and flexibility I needed while at the same time providing guidance, clarity, a whole lot of precious ideas, as well as a fruitful research context.\\

On the Geneva side, my gratitude extends to the whole quantum theory group as well as the GAP. Thank you in particular to Ralph Silva for the tight and fruitful collaboration that developed throughout the years. I equally thank Géraldine Haack and Jonatan Bohr Brask in that regard. Thank you also to Flavien Hirsch who has by now become a precious research companion.\\

On the Vienna side, I am grateful to the entire Huber group. Thank you to Nicolai Friis for his helpful scientific support. Thank you also to Faraj Bakhshinezhad for his incredible perseverance. And thank you to Jessica Bavaresco for her wise and beneficial advice.\\

I would also like to thank Philip Taranto, Falvien Hirsch, Raya Polishchuk, Marcus Huber, Nicolas Brunner, Géraldine Haack, Mart\'{i} Perarnau-Llobet, and Andreas Winter for valuable feedback on and around this manuscript.\\

Thank you also to Kristina Eisfeld, Raya Polishchuk, and Tatjana Boczy for creating an online working space that enabled me to keep being productive in writing this thesis while being confined at home due to the Covid-19 measures.\\

Last but not least, I would like to wholeheartedly thank my family, friends, colleagues, and others who have contributed in what one may describe as a more indirect but at least as important way to the academic path I have taken so far.
A very special thanks in that regard goes to my wife, Raya, who besides having contributed in the specific above mentioned instances, has been an incredible source of continuous love and support throughout the years. I feel very privileged to be lucky enough to have you in my life.
\end{acknowledgements}


\tableofcontents 

%

\chapter{List of Publications}

This thesis is based on the following articles: 
\medbreak

\begin{itemize}
\item F. Bakhshinezhad, \underline{F. Clivaz}, G. Vitagliano, P. Erker, A. Rezakhani, M. Huber, and N. Friis, "Thermodynamically optimal creation of correlations," \href{https://dx.doi.org/10.1088/1751-8121/ab3932}{Journal of Physics A: Mathematical and Theoretical {\bf 52}, 465303 (2019)}, \href{https://arxiv.org/abs/1904.07942}{arXiv:1904.07942}.

\item \underline{F. Clivaz}, R. Silva, G. Haack, J. Bohr Brask, N. Brunner, and M. Huber, "Unifying paradigms of quantum refrigeration: A universal and attainable bound on cooling," \href{https://dx.doi.org/10.1103/PhysRevLett.123.170605}{Phys. Rev. Lett. {\bf 123}, 170605 (2019)}, \href{https://arxiv.org/abs/1903.04970}{arXiv:1903.04970}.

\item \underline{F. Clivaz}, R. Silva, G. Haack, J. Bohr Brask, N. Brunner, and M. Huber, "Unifying paradigms of quantum refrigeration: fundamental limits of cooling and associated work costs," \href{https://dx.doi.org/10.1103/PhysRevE.100.042130}{Phys. Rev. E {\bf 100}, 042130 (2019)}, \href{https://arxiv.org/abs/1710.11624}{arXiv:1710.11624}.

\end{itemize}

\mainmatter 

\pagestyle{thesis} 



\part{Setting the Stage} 

\label{intro} 

\chapter{Research Context}

It is a delicate task to define a research field, especially when it is still very much active. On the one hand, one would like to include all what is done in the direction that the supposed field investigates. On the other hand, delimiting it in too broad terms might inadvertently include research that definitely considers itself as being outside of the said field. That being said, the research that formed the basis of this thesis can be classified as being part of the field of quantum thermodynamics. Quantum thermodynamics is a hybrid field that studies the interplay between two major theories of physics: quantum mechanics and thermodynamics. The motivations for doing so are quite diverse. From a theoretical perspective both theories are strikingly complementary. Quantum mechanics is the widely accepted theory of microscopic objects, while thermodynamics is very effective at describing macroscopic phenomena. It is therefore appealing to:

\begin{enumerate}
	\item Investigate how the thermodynamic behavior of macroscopic systems emerges from their microscopic quantum mechanical description.
	\item Extend useful thermodynamic concepts such as temperature, heat, work and entropy to the quantum mechanical realm.
\end{enumerate}  

From a practical viewpoint, the recent miniaturization of technologies naturally pushes for the development of a theory that allows to best engineer them. This is timely, as it is complemented with our increased ability to control larger and larger quantum systems, giving us the opportunity to experimentally test theoretical ideas.\\

The dialogue between thermodynamics and quantum mechanics started right at the dawn of quantum mechanics: in 1905 Einstein introduced the concept of a photon by, using a thermodynamic argument, studying the photoelectric effect \cite{Einstein-1905}. Despite this intimate relation, quantum mechanics and thermodynamics branched off in the following decades to further develop independently of each other. Only in 1959, when the first solid state lasers were developed, did the two fields meet again as Scovil and Schulz DuBois \cite{Scovil-1959} discovered the equivalence between the three level maser and the Carnot heat engine \cite{Carnot-1872}. The former was developed thanks to our understanding of stimulated emission, a truly quantum mechanical phenomena, while the latter is a pillar of thermodynamics. Their seminal work set the stage for the rich field of quantum thermodynamics that would blossom in the decades to come.

Due to its broad scope, the field has attracted researchers from many backgrounds, ranging from statistical physics, many-body theory, and quantum optics to mesoscopic physics and quantum information theory \cite{Binder-2018}. This diversity is reflected both in the number of different sub fields and the different approaches found to address a given question, where consensus is sometimes yet to be reached \cite{Vinjanampathy-2016}.\\

One of the most fundamental questions that has been asked since the beginning of thermodynamics and statistical physics regards the problem of equilibration and thermalization \cite{Binder-2018,Gogolin-2016}: why do systems tend to equilibrate? Early contributions to the topic date back to the works of Schrödinger in 1927 \cite{Schroedinger-1927} and von Neumann in 1929 \cite{Neumann-1929}. See \cite{Neumann-2010} for an English translation of the work of von Neumann. It is, however, only recently that we were able to gain more understanding onto this challenging topic. Current research in this direction is contributing to the understanding of why many-body systems appear to equilibrate and the time-scales under which they do so. A large body of work also investigates the properties of many-body systems at thermal equilibrium, such as their correlations. Notions such as (dynamical) typicality \cite{Neumann-1929,Schroedinger-1927,Neumann-2010,Goldstein-2010,Popescu-2006,Gemmer-2010} and the eigenstates thermalization hypothesis (ETH) \cite{Srednicki-1994,Deutsch-1991}  are central to this endeavor.

While studying the dynamics of quantum systems to see if, and under which conditions, they tend to equilibrate is central in understanding how the microscopic and macroscopic worlds are related, the dynamics of small quantum systems out of equilibrium is \emph{per se} also of interest \cite{Campbell-2019,Binder-2018,Vinjanampathy-2016}. Investigating the behavior of these small systems in regimes far from that of their environment, fluctuation theorems, which can be seen as generalizations of the second law of thermodynamics, can be formulated. As their name suggests, their most basic insight is that dynamics in this regime are governed by fluctuations. Two major results of this sub-field are the Jarzynski equality \cite{Jarzynski-1997} and the Crooks fluctuation theorem  \cite{Crooks-1999}.

More recently, a quantum information perspective of the field has blossomed \cite{Goold-2016,Binder-2018}. In a broad sense, this approach investigates how information and thermodynamics are related. Doing so, researchers have, on the one hand, looked at how quantum thermodynamics can be reconstructed from an information theoretic perspective and, on the other hand, at how quantum thermodynamics provides a platform to implement quantum information tasks. A significant contribution regarding that approach has been to formulate thermodynamics as a resource theory \cite{Lostaglio-2019,Gour-2015,Brandao-2013}. This enables one to answer questions such as: given a certain state and a limited set of operations, what final state can be achieved?

Last, but not least, a central pillar of quantum thermodynamics is the thermodynamics of open quantum systems \cite{Kosloff-2013,Potts-2019,Vinjanampathy-2016,Binder-2018}. Of major importance here is the so called Gorini, Kossakowski, Lindblad, Sudarshan (GKLS) equation \cite{Gorini-1975,Lindblad-1976}, which is a type of master equation valid in the markovian (i.e., memory-less) regime, where the system of interest is weakly-coupled to its environment, whose temporal correlations decay much faster than those of the system. This approach has enabled many researchers to extensively investigate quantum heat devices \cite{Mitchison-2019}. Going beyond the weak-coupling regime in order to investigate strongly-coupled open systems has recently gained substantial interest within this community \cite{Binder-2018}.

Experimentally, the field is still in its infancy; nevertheless, there has recently been a number of advances regarding various platforms \cite{Pekola-2015,Binder-2018}. Among them nitrogen-vacancy (NV) centers provide a good overall platform \cite{Klatzow-2019, Bar-Gill-2018}. One dimensional atomic particles have been used to test equilibration concepts \cite{Langen-2015, Schmiedmayer-2018}. And trapped ions have been demonstrated to be a good avenue to probe the dynamics of fluctuations \cite{An-2014, Lu-2018}, as well as to realize quantum heat engines \cite{Rossnagel-2016, Dawkins-2018}.\\

Despite the tremendous progress of the past years, there remain many open questions to be addressed in the field. One of the many challenges is to understand how the distinct approaches, fit to address a same problem from different angles, relate to one another. This is what Part \ref{part:refri} of this thesis is addressing for the ubiquitous thermodynamic task of refrigeration. 

Another intriguing and still poorly understood question is that of (quantitative) resource inter-convertibility. Indeed, while the resource for cooling can be identified in terms of available energy, it is often more natural to think in terms of other resources when faced with different tasks. Part \ref{part:corre} of this thesis will be concerned with how, and how well, one can transform a given resource into another. More precisely, we will be interested in how well energy, the paradigmatic thermodynamics resource, can be transformed into correlations, the resource for information processing tasks.

\chapter{Encompassing Idea and Main Results}

The core of this thesis is structured in two distinct Parts. In Part~\ref{part:refri}, the task of cooling a quantum system is investigated. Part~\ref{part:corre} is concerned with how much correlations can be created in a bipartite quantum system for a given amount of energy. While both tasks inscribe themselves within the same research field, namely that of quantum thermodynamics, and while their inner workings are based upon the same mathematical theory, namely that of majorization, their physical relevance is quite independent. We therefore chose to make each Part stand alone and equipped them with their own introduction and conclusion, enabling the reader interested in one Part only to solely concentrate on it without having to worry about the other Part.\\

In Part~\ref{part:refri}, we develop a paradigm for quantum refrigeration. The idea behind the paradigm stems from the observation that a variety of different approaches to refrigeration exist in the literature, but the set of assumptions under which they operate being quite different, comparison among them has proven challenging. Our contribution intends to fill this gap by providing a platform to compare the different existing approaches. In general terms, the different existing approaches have two features in common.

\begin{enumerate}
\item They assume an environment at some background temperature that allows accessing/preparing thermal state at that temperature. \label{assump:bath}
\item The system to be cooled is initially from that environment and remains, in some sense, open to that environment during the cooling process. \label{assump:openstate}
\end{enumerate}

As von Neumann and Schrödinger already realized in the late 1920s~\cite{Neumann-1929, Schroedinger-1927}, while assumption \ref{assump:bath}. might seem obvious from an empirical thermodynamics perspective, it is far from obviously derivable from first quantum mechanical principles. Removing it essentially amounts to diving deep into the field of equilibration and thermalization and, while being a fascinating topic to investigate, effects in diverting the attention from the original intent. Even more so, one could argue that the concept of refrigeration only makes sense once it is possible to speak of a well-defined background temperature. In the end, one always cools a system with respect to its environment, and when we speak of a cold/warm system, what we really mean is that it is colder/warmer than its environment.

The first part of assumption~\ref{assump:openstate} makes sure that no additional resource is hidden in the initial state. To illustrate this idea with one extreme example, if the initial state of the system to be cooled were allowed being a pure state, the task of refrigeration would become trivial. This assumption also goes back to the basic idea of refrigeration. The very meaning of the fact that one cools a system with respect to its environment entails that the initial state of that system is from that environment. Most of the approaches, actually all the approaches we consider as well as our paradigm, even go a step further in assuming that the state of the system is initially uncorrelated with that of the environment. This is a debatable simplification. It is, however, beyond the scope of this thesis to challenge it.

Given that the initial state of the system is thermal, the second part of assumption~\ref{assump:openstate} is necessary for the system to be cooled at all. Indeed, letting the system evolve according to its own Hamiltonian will not induce any cooling since this evolution commutes with the state of the system. What is more, allowing the system to interact with itself in any possible way, meaning allowing turning on any local Hamiltonian, also does not lead to any cooling of the system. This is a consequence of the fact that the state of the system is passive~\cite{Pusz-1978, Lenard-1978}. So, in order for the system to be cooled, one really needs to open it to the environment. The diversity of the approaches results in how one decides to open the system to the environment.\\

Our take on it is to divide the environment into two parts: a \enquote{machine} and the rest of the environment. While the rest of the environment is thought of as being weakly or not controlled in that it is only used to rethermalize the machine, we assume an increased control on the part of the environment we call \enquote{machine}, hence the name. To prevent any hidden energy supply, we furthermore explicitly take the source of energy --- referred to as the resource --- into account. This gives rise to two variations of our paradigm that we dub coherent and incoherent. In the incoherent scenario, the resource is embodied by a thermal bath at a hotter temperature. The hot thermal bath invests energy by rethermalizing (part of) the machine and the system is cooled via an energy conserving unitary applied on the joint system of machine and target system. This ensures that the entire energy comes from the hot thermal bath, making the energy used be extracted at maximal entropy. In the coherent scenario the entropy of the resource is, in contrast, left unchanged upon energy extraction. In order to cool the system, a (possibly energy non-conserving) unitary is in turn applied to the joint system of machine and target system. In both variations, repetitions are enabled by rethermalizing the machine to the respective baths and repeating the above.\\

Both scenarios are per design related to other refrigeration paradigms studied in the literature. In particular, the coherent scenario is a generalization of heat bath algorithmic cooling (HBAC)~\cite{Schulman-1999, Boykin-2002, Raeisi-2015, Rodriguez-Briones-2016, Rodriguez-Briones-2017} without compression qubit. In fact, the coherent scenario fully generalizes HBAC if one extends it as in~\cite{Taranto-2020}. The coherent scenario furthermore includes any quantum Otto engine implementation~\cite{Abah-2012, Rossnagel-2016, Niedenzu-2016, Niedenzu-2018}. In the limit of infinite machine size, a single application of the coherent scenario is able to implement any completely positive and trace preserving (CPTP) map to the system~\cite[Chapter~8]{Nielsen-2010}. In the same limit of infinite machines, a single application of the incoherent scenario investigates whether state transitions of interest are possible within the framework of thermal operations (TO)~\cite{Lostaglio-2019,Gour-2015,Brandao-2013}. Last but not least, the incoherent scenario is intimately related to autonomous cooling~\cite{Linden-2010, Skrzypczyk-2011}.\\

Investigating each of our scenario on its own, we find a bound on cooling valid in both scenarios for arbitrary machines and target systems. The bound is in particular a single letter bound that, for a given target system dimension, only depends on the relevant parameter of the machine: its maximal energy gap. As such, it achieves one of the central pillar of statistical physics: distilling the pertinent parameters from the intricacies of complex systems. The bound is furthermore attainable for any machine in the coherent scenario and achievable under minimal modifications of the machine in the incoherent scenario. In particular, minimal machines already suffice to attain the bound. Investigating these minimal machines in greater details, we identify the operations that spend a minimal amount of resource to reach a given temperature. Doing so, we find that the amount of resource spent to reach the bound depends on the level of control, i.e., on the scenario, but that for non-maximal cooling there is no universally better scenario. Finally, we define a generalized notion of temperature, called sumtemperature, that is based on the notion of majorization and uniquely captures our cooling bound for both scenarios.\\

Part~\ref{part:corre} of this thesis treats of the creation of correlations within a thermodynamic setting. More precisely, we are interested in knowing how much correlations can be created in a bipartite system for a given amount of energy. The general question motivating this research line is that of resource inter-convertibility. Indeed, at the time of the development of thermodynamics, the resources at hand were fairly unanimous and determined in concrete terms. However, in quantum thermodynamics, this is not the case anymore and identifying the relevant resources in that regime is arguably one of the major open endeavor of the field. Whether one resource can be converted into another is therefore of major importance. This is especially true when speaking of the ubiquitous resource of thermodynamics, energy, and that of information theory, correlations. Indeed, according to our current understanding of nature, correlations lie at the heart of every scientific prediction. It is therefore only natural to ask how much correlations can be created for a given amount of energy.

Assuming an initially uncorrelated bipartite system in a thermal background, it can be shown that, as long as both local systems are not persistently interacting, establishing correlations necessarily comes at an energy cost~\cite{Friis-2016}. The acquisition of information in a physical system therefore automatically comes at a thermodynamic cost. Conversely, energy can be extracted from any kind of correlations~\cite{PerarnauLlobet-2015}. This qualitatively settles the question of resource inter-convertibility and sets the stage for the next fundamental question, namely that of how much resource can in principle be inter-converted.\\

Our contribution is to provide a framework to investigate how much correlations can be created in an initially uncorrelated bipartite system for a given amount of energy. The precise question we investigate was first formulated in~\cite{Huber-2015}, see also~\cite{Vitagliano-2018}. There, the question could already be answered for the case of two equally gaped identical systems. Our framework is based on decomposing the Hilbert space in a Latin square manner and as such allows us to harness the tools of majorization theory. In doing so, we are able to fully answer the question for the case of two identical qutrits and ququarts of arbitrary energy gaps. We furthermore conjecture a bound on the amount of correlations for arbitrary identical systems to be reachable in any dimension and provide a set of conditions under which the bound it achievable. For non-identical systems the bound can be violated, and we provide evidence for it, see also~\cite{Vitagliano-2018}. Lastly, we solve the problem in its full generality, i.e., lifting the identical and equally gap constraints, for when the systems are in a vanishing background temperature.

\chapter{General Notation}

We would here like to state some general notation used throughout the thesis. Some part or chapter specific notation will be stated later to avoid confusion. This chapter as well as the subsequent notation chapters or sections, i.e., Chapter~\ref{chap:notationref}, Section~\ref{sec:notationqubit}, and Chapter~\ref{chap:corrnotation}, are intended for reference in order to facilitate a cherry-picked reading. They might therefore simply be skipped upon a linear reading.

We work in units of $k_B=\hbar=1$, where $k_B$ is the Boltzmann constant and $\hbar$ the reduced Planck constant. The background temperature is fixed throughout the analysis and denoted by $T_R$. The subscript $R$ highlights the fact that this is the room temperature. We emphasize, however, that $T_R$ might be of any numerical value, as long as fixed throughout the analysis. We denote the inverse background temperature by $\beta_R$, i.e.
\begin{equation}
\beta_R=\frac{1}{T_R}.
\end{equation}
In Part~\ref{part:refri}, we consider a hot bath and denote its temperature by $T_H$, and inverse temperature by $\beta_H=\frac{1}{T_H}$. In Part~\ref{part:corre}, we encounter temperatures $T' \geq T_R$ and denote their inverse temperature by $\beta'=\frac{1}{T'}$. While $T_H$ is thought of as fixed, $T'$ is thought of as varying, hence the difference in notation.\\

For system $i$, where typically $i=S,M,SM$ in Part~\ref{part:refri} and $i=A,B,AB$ in Part~\ref{part:corre}, we always associate a Hilbert space $\mathcal{H}_i$ as well as a Hamiltonian $H_i$ to it. We denote by $\tau_i^R$, $\tau_i^H$ and $\tau_i(\beta')$ the thermal state of the Hamiltonian $H_i$ at $T_R$, $T_H$ and $T'$ respectively, i.e.
\begin{align}
\tau_i^R&= \frac{e^{-\beta_R H_i}}{\Tr(e^{-\beta_R H_i})},\\
\tau_i^H&= \frac{e^{-\beta_H H_i}}{\Tr(e^{-\beta_H H_i})},\\
\tau_i(\beta')&= \frac{e^{-\beta' H_i}}{\Tr(e^{-\beta' H_i})}.
\end{align}
We will often drop the superscript $R$. i.e.
\begin{equation}
\tau_i=\tau_i^R.
\end{equation}

Furthermore, the initial state of the system $i$ is denoted by $\rho_i$ and is thermal at $T_R$ unless otherwise stated, that is
\begin{align}
\rho_i &= \tau_i.
\end{align}

A general density matrix on system $i$ is denoted by $\sigma_i$ and the set of all density matrices on $i$ by $\mathcal{S}(\mathcal{H}_i)$, i.e.

\begin{equation}
\mathcal{S}(\mathcal{H}_i)=\{\sigma_i:\mathcal{H}_i \rightarrow \mathcal{H}_i \mid \Tr(\sigma_i)=1, \sigma_i \geq 0 \}.
\end{equation}

The dimension of system $i$ is always finite and denoted by $d_i$. We start to count from 0. The computational basis notation, $\left(\ket{k}_i \right)_{k=0}^{d_i-1}$, is used to denote the energy eigenbasis of system $i$. Given two systems $i_1$ and $i_2$, we denote $\ket{k}_{i_1} \otimes \ket{l}_{i_2}$ by

\begin{equation}
\ket{k,l}_{i_1 i_2}.
\end{equation}
We will drop the comma and the subscripts, i.e., use $\ket{kl}$, whenever no confusion arises. Components of the density operator $\sigma$ in the energy eigenbasis are denoted by $[\sigma]_{kl}$, i.e.,
\begin{equation}
[\sigma]_{kl}=\bra{k} \sigma \ket{l}.
\end{equation}
For composed systems,
\begin{equation}
[\sigma]_{ik,jl}=\bra{ik} \sigma \ket{jl}.
\end{equation}
We denote by $\mathfrak{D}(\sigma)$ the vector of diagonal elements of $\sigma$ in the energy eigenbasis, i.e.
\begin{equation}
\mathfrak{D}(\sigma)= (\bra{0} \sigma \ket{0}, \dots, \bra{d-1} \sigma \ket{d-1}).
\end{equation} 

$\lambda(\sigma)$ denotes the vector of eigenvalues of $\sigma$ and $\text{Eig}_A(\lambda)$ the eigenspace of the matrix $A$ with eigenvalue $\lambda$, i.e.
\begin{equation}
\text{Eig}_A(\lambda)=\{\ket{v} \mid A \ket{v} = \lambda \ket{v}\}.
\end{equation}
Given a reference set $X$, sometimes called a universe, and a set $B \subset X$, we denote by $B^c$ the complement of the set $B$ in $X$, i.e.
\begin{equation}
B^c= X \setminus B = \{x \in X \mid x \notin B\}.
\end{equation}
Given two compatible matrices or operators $A$ and $B$, we denote by $[A,B]$ the commutator between $A$ and $B$, i.e.
\begin{equation}
[A,B]=AB-BA.
\end{equation}
Given an integer $d \in \mathds{N}$ and a vector $v=(v_0, \dots, v_{d-1})$, we denote by 
\begin{equation}
v_0^{\downarrow} \geq v_1^{\downarrow} \geq \dots \geq v_{d-1}^{\downarrow}
\end{equation}
 the components of v arranged in decreasing order and 
 \begin{equation}
 v^{\downarrow}=(v_0^{\downarrow},\dots,v_{d-1}^{\downarrow}).
 \end{equation}
 
 Analogously
 
 \begin{equation}
v_0^{\uparrow} \leq v_1^{\uparrow} \leq \dots \leq v_{d-1}^{\uparrow}
\end{equation}
 denote the components of v arranged in increasing order and 
 \begin{equation}
 v^{\uparrow}=(v_0^{\uparrow},\dots,v_{d-1}^{\uparrow}).
 \end{equation}
 
 Note that we use increasing and decreasing in place of the sometimes preferred non-decreasing and non-increasing. If we want to rule out equality from the relation, we will use strictly increasing and strictly decreasing. 
 
 We denote by $M$ a general doubly stochastic matrix, by  $P$ a permutation matrix, by $Q$ a permutation matrix exchanging only 2 indices and by $\Pi$ the permutation matrix that \enquote{pushes down} every element of a vector once, i.e.
\begin{equation}
\Pi = (\Pi_{ij}), \quad \Pi_{ij} = \delta_{i j+1 \, \text{mod}\, d}.
\end{equation}
We denote by $T$ a special kind of doubly stochastic matrix called a T-transform. With the above notation, a T-transform can be written as
\begin{equation}
T= (1-t) \mathds{1} + t Q,
\end{equation}
for some $t\in [0,1]$ and some $Q$. We will write $T(t)$ when emphasizing the parametric dependance is desired.


\part{Refrigeration} 

\label{part:refri} 


\emph{This part is based on the following papers:} 

\begin{itemize}

\item \underline{F. Clivaz}, R. Silva, G. Haack, J. Bohr Brask, N. Brunner, and M. Huber, ``Unifying paradigms of quantum refrigeration: A universal and attainable bound on cooling,'' \href{https://doi.org/10.1103/PhysRevLett.123.170605}{Phys. Rev. Lett. {\bf 123}, 170605 (2019)}, \href{https://arxiv.org/abs/1903.04970}{arXiv:1903.04970}.

\item \underline{F. Clivaz}, R. Silva, G. Haack, J. Bohr Brask, N. Brunner, and M. Huber,``Unifying paradigms of quantum refrigeration: fundamental limits of cooling and associated work costs,'' \href{https://doi.org/10.1103/PhysRevE.100.042130}{Phys. Rev. E {\bf 100}, 042130 (2019)}, \href{https://arxiv.org/abs/1710.11624}{arXiv:1710.11624}.

\end{itemize}

\chapter{Introduction} \label{chap:intro}

Understanding the performance of thermal machines is intimately related to the fundamental laws of thermodynamics. Indeed, the formulation of fundamental laws in the early days of thermodynamics were a tremendous help to our understanding of how to better design machines. More generically, pioneering ideas in the design of relevant machines paves the way for the derivation of fundamental laws governing their behavior. Those laws in turn challenge for the design of better machines as well as for out-of-the-box designs that force us to rethink the applicability of the laws and ultimately for the derivations of new ones. \\

While the resources as well as operations at hand during the development of thermodynamics were clearly dictated by first-hand experiences, it is a main challenge to identify the natural and reasonable resources as well as operations in a quantum thermodynamic setting. This ambiguity gave rise to a plethora of conceptually different approaches.

Taking a stroke type, or timeless, perspective, the resource theory of thermodynamics has been very successful at deriving fundamental laws by characterizing possible state transitions within a single application of a physically motivated completely positive and trace preserving (CPTP) map~\cite[Chapter~8]{Nielsen-2010} called \enquote{thermal operation} \cite{Horodecki-2013, Brandao-2013, Gour-2015, Lostaglio-2015, Lostaglio-2015b, Cwiklinski-2015, Brandao-2015, Skrzypczyk-2014, Guryanova-2016, Masanes-2017, Wilming-2017, Scharlau-2018}. While the interaction of the system with the thermal bath is here explicitly imposed to be energy conserving, the implicit assumptions are 
\begin{enumerate}[i)]
\item\label{enum:perfect_time} a perfect timing device (or clock), 
\item arbitrary spectra in the bath,
\item interaction Hamiltonian of arbitrary complexity.
\end{enumerate}
A similar timeless  perspective with no access to a thermal bath but increased unitary control lead to the well-known concept of passivity~\cite{Pusz-1978, Lenard-1978, Alicki-2013, Hovhannisyan-2013, Skrzypczyk-2015, PerarnauLlobet-2015b}. Here the implicit assumptions are \ref{enum:perfect_time}) and
\begin{enumerate}
\item[iv)] the ability to implement any cyclic change in the Hamiltonian of a quantum system.
\end{enumerate}

Finally, combining both aspects, access to a bath and arbitrary unitary control, as well as allowing for repeated applications of the CPTP map at hand leads to the idea of heat bath algorithmic cooling (HBAC)~\cite{Schulman-1999, Boykin-2002, Raeisi-2015, Rodriguez-Briones-2016, Rodriguez-Briones-2017}. In HBAC, however, the bath one has access to is generally limited to being a string of uncorrelated qubits.

Taking a dynamical perspective on the other hand, one can model the interaction of the systems of interest with their environment as open quantum systems coupled to external baths. Of particular interest in this regime is usually the asymptotic non-equilibrium steady state obtained. Autonomous machines have been designed within this paradigm~\cite{Linden-2010, Skrzypczyk-2011, Levy-2012, Brunner-2012, Venturelli-2013, Mitchison-2016, Hofer-2016, Hofer-2016b, Maslennikov-2019}. Here the interactions with the various thermal baths are modeled as time-independent Hamiltonian. Taking a similar stance to the thermal operations, the time-independent interactions give rise to energy conserving dynamics and make sure that no external source of work is implicitly needed. The interactions being always turned on and time-independent, this paradigm also gets rid of the hidden assumption~\ref{enum:perfect_time}) otherwise ubiquitous in all other paradigms. At the other end of the spectrum, by allowing for the implementation of complex unitary cycles, the design of quantum Otto engines has been explored in an open quantum system setting~\cite{Abah-2012, Rossnagel-2016, Niedenzu-2016, Niedenzu-2018}.\\

While all the above discussed paradigms are perfectly valid within their own set of assumptions, the fact that their approach stems from so drastically different point of views makes it hard to draw parallels between them, to compare them, to carry one result from one approach to the other, as well as to get a unified view of how well a given quantum thermodynamic task can be performed, see~\cite{Mitchison-2015, Brask-2015, Pozas-Kerstjens-2018, Torrontegui-2017, Chubb-2018, Brown-2016, Perry-2018, Sparaciari-2017, Erker-2017} for preliminary established connections.

Focusing on the task of refrigeration, it is with this idea of unification in mind that we have designed the two cooling scenarios called coherent and incoherent, that we will discuss in the rest of this Part. While each of this scenario works with its own set of assumptions, they are both constructed in a way that allows comparing them naturally. Furthermore, we will see that each scenario is closely related to the previously discussed paradigms of thermal operations, HBAC, autonomous cooling, and quantum Otto engines. While comparing both scenarios with one another and with other existing paradigms motivates how we define them, exploring them in their own sake will nevertheless be the focus here. This will allow us to develop a better understanding of the workings of each scenario as well as derive new results that can easily be carried on to other approaches. This being said, we do not claim that our take on refrigeration settles the question once and for all. First of all, as we will see, both scenarios still need to be further studied to be completely understood. The parallel they draw with the other existing paradigms also need some more investigation. But more crucially, both scenarios have their own limitations in the way they are defined. They are for example inherently markovian and can as such never include non-markovian cooling strategies. They can nevertheless be extended to the non-markovian regime as done for example in~\cite{Taranto-2020}.\\

The rest of this Part is organized as follows. After having set some notation specifically used for this Part in Chap.~\ref{chap:notationref}, we will define each scenario in Chap.~\ref{chap:coh} and Chap.~\ref{chap:inc} respectively. We will then motivate our definition of work cost and temperature in Chap.~\ref{chap:worktemp} before seeing how both scenarios relate to other existing paradigms in Chap.~\ref{chap:other}. In Chap.~\ref{chap:remarks} we will make some general remarks about each scenario that will give us some intuition and set the interesting problems to investigate as well as hint at the useful tools to utilize. Following a bottom-up approach, we will then focus in Chap.~\ref{chap:qubitsyst} on a target qubit system and investigate in detail the case of the one and two qubit machines to cool this target. This will hep us build some crucial intuition about each scenario and will motivate the formulation of two open problems in Sec.~\ref{sec:twoopenproblems}. In Chap.~\ref{chap:quditsystem} we will then move our attention to more general instances of our scenarios and will consider the cooling of an arbitrary qudit system by an arbitrary finite dimensional machine. There we will in particular derive a bound on cooling valid for both paradigms and define a new notion of temperature that encompasses other temperature notions within our scenarios.

\chapter{Notation} \label{chap:notationref}

We would here like to state some further notation that will be used throughout this Part. We remind the reader that since this chapter is intended for reference, it might simply be skipped upon a linear reading.

In the following we will be interested in cooling a system that we will denote by $S$. We will refer to this system as \emph{the system of interest}, \emph{the target system}, \emph{the target}, or sometimes as \emph{the system} simply. We will also be considering a machine that we will denote by $M$ and refer to as simply \emph{the machine}. The target and the machine together will be denoted by $SM$ and often referred to as \emph{the joint system}.\\

The Hamiltonian of the $d_S < \infty$ dimensional system S will be denoted by
\begin{equation}
H_S=\sum_{i=0}^{d_S-1} E_i \ket{i} \bra{i}_S, \quad \text{with } E_i \leq E_{i+1} \forall i=0,\dots,d_S-1.
\end{equation}

Throughout this chapter $E_0=0$. In the case of a system being a qubit we will denote its energy gap by $E_S$, i.e.,
\begin{equation}
E_S=E_{d_S-1}, \quad \text{if } d_S=2.
\end{equation}
We will denote the Hamiltonian of the $d_M < \infty$ dimensional machine M by
\begin{equation}
H_M= \sum_{i=0}^{d_M-1} \mathcal{E}_i \ket{i} \bra{i}_M, \quad \text{with } \mathcal{E}_i \leq \mathcal{E}_{i+1} \forall i=0,\dots,d_M-1.
\end{equation}
We will sometimes write $\mathcal{E}_{\text{max}}$ for $\mathcal{E}_{d_M-1}$, i.e.,
\begin{equation}
\mathcal{E}_{\text{max}} = \mathcal{E}_{d_M-1}.
\end{equation}
Furthermore, $\mathcal{E}_0=0$ throughout this chapter. We will encounter 2 specific machines. In the case where the machine consists of a single qubit, we will denote its energy gap by $\mathcal{E}_M$. If the machine consists of 2 qubits, we will denote the first qubit by $M_1$ and its associated energy gap by $\mathcal{E}_{M_1}$. The second qubit will be denoted by $M_2$ and its associated energy gap by $\mathcal{E}_{M_2}$. \\

We will denote the Hamiltonian of the (non-interacting) joint system $SM$ by
\begin{equation}
H_{SM}= \sum_{i=0}^{d_{SM}-1} \epsilon_i \ket{i} \bra{i}_{SM}, \quad \text{with } \epsilon_i \leq \epsilon_{i+1} \forall i=0,\dots,d_{SM}-1.
\end{equation}
Note that as $H_{SM} = H_S \otimes \mathds{1} + \mathds{1} \otimes H_M$, $d_{SM}=d_S+d_M$ and 
\begin{equation}
H_{SM}= \sum_{i=0}^{d_S-1} \sum_{j=0}^{d_M-1} (E_i+\mathcal{E}_j) \ket{ij} \bra{ij}_{SM}.
\end{equation}
So that for all $i=0,\dots,d_{SM}-1$, $\epsilon_i=E_k+\mathcal{E}_l$ and $\ket{i}_{SM}=\ket{kl}_{SM}$ for some $k \in \{0,\dots,d_S-1\}$, $l \in \{0,\dots,d_M-1\}$.\\

It will sometimes be useful to refer to the $k_i$ distinct energy eigenvalues of system $i=S,M,SM$. We will denote them by $\tilde{E}_0,\dots \tilde{E}_{k_S-1}$, $\tilde{\mathcal{E}}_0,\dots \tilde{\mathcal{E}}_{k_M-1}$ and $\tilde{\epsilon}_0,\dots \tilde{\epsilon}_{k_{SM}-1}$ respectively.\\

A single application of the coherent scenario is denoted by $\Lambda_{\text{coh}}$ and $n$ applications by $\Lambda_{\text{coh}}^n$, i.e.,
\begin{align}
\Lambda_{\text{coh}}(\rho_S)&= \Tr_M (U \rho_S \otimes \rho_M U^{\dagger}), \quad U: \text{ arbitrary unitary}\\
\Lambda_{\text{coh}}^n(\rho_S)&= \underbrace{\Lambda_{\text{coh}}(\dots \Lambda_{\text{coh}}(\Lambda_{\text{coh}}(}_{\text{n times}}\rho_S))\dots).
\end{align}
Note that the chosen unitary can vary from step to step. If we want to specifically emphasize that $\Lambda_{\text{coh}}$ at step $i$ might be different from $\Lambda_{\text{coh}}$ at step $j \neq i$, we use the notation
\begin{equation}
\Lambda_{\text{coh}}^n = \Lambda_{\text{coh}}^{(n)} \circ \Lambda_{\text{coh}}^{(n-1)} \circ \dots \circ \Lambda_{\text{coh}}^{(1)}.
\end{equation}

Similarly, for the incoherent scenario we have
\begin{align}
\Lambda_{\text{inc}}(\rho_S)&= \Tr_M (U_{\text{inc}} \rho_S \otimes \rho_M^{R,H} U^{\dagger}_{\text{inc}}), \quad U_{\text{inc}}: [U_{\text{inc}}, H_{SM}]=0,\\
\Lambda_{\text{inc}}^n(\rho_S)&= \underbrace{\Lambda_{\text{inc}}(\dots \Lambda_{\text{inc}}(\Lambda_{\text{inc}}(}_{\text{n times}}\rho_S))\dots)
=\Lambda_{\text{inc}}^{(n)} \circ \Lambda_{\text{inc}}^{(n-1)} \circ \dots \circ \Lambda_{\text{inc}}^{(1)},
\end{align}
where
\begin{equation}
\rho_{M}^{R,H}=\tau_{M_R} \otimes \tau_{M_H}^H,
\end{equation}
and $M_H$ is the part of $M$ thermalized to $T_H$ and $M_R$ and the part of $M$ thermalized to $T_R$. Note that in the incoherent scenario it is the choice of $M_H$ and $M_R$ that may vary at each step.\\

$\Delta F$ denotes the free energy difference of the resource, i.e.,

\begin{equation}
\Delta F = \Delta \mathcal{U} - T_R \Delta S,
\end{equation}

with $\Delta \mathcal{U}$ the internal energy change of the resource and $\Delta S$ the entropy change of the resource. The letter $r$ is reserved to denote the ground state population, i.e.
\begin{equation}
r= \bra{0} \sigma \ket{0}.
\end{equation}
For a qubit system, i.e., $d_{SM}=2 d_M$, and $v \in \mathds{R}^{d_{SM}}$,
\begin{align}
a_v&=(v_0, \dots, v_{d_M-1}),\\\
b_v&=(v_{d_M}, \dots, v_{d_{SM}-1}).
\end{align}
In the case of a qubit target and the one and two qubit machines
\begin{equation}
\mathcal{L}_{SM_k}= i \ket{01} \bra{10}_{SM_k} - i \ket{10} \bra{01}_{SM_k}, \quad k=1,2,
\end{equation}
denotes the Hamiltonian generating the swapping between the energy levels $\ket{0 1}_{SM_k}$ and $\ket{1 0}_{SM_k}$.
In the case of a qubit target and two qubit machine, for $i,j = 0,\dots 7$,
\begin{equation}
U_{ij} = \ket{i_2} \bra{j_2} + \ket{j_2} \bra{i_2} + \mathds{1}_{\text{span}^c\{\ket{i_2}, \ket{j_2}\}},
\end{equation}
where $i_2$ denotes the 3 digit display of $i= 0,\dots,7$ in base 2, e.g., $0_2=000$, and $\text{span}^c\{\ket{i_2}, \ket{j_2}\}$ denotes the complement of the set $\text{span}\{\ket{i_2}, \ket{j_2}\}$.
\chapter{Coherent Scenario} \label{chap:coh}

In this section we would like to define what we will refer to as the coherent scenario in the rest of this thesis. The term \emph{coherent scenario} was coined to highlight the fact that one can create coherence between the different energy levels of the machine and target joint system. This is in contrast with the incoherent scenario, that will be defined in Chapter~\ref{chap:inc}, where coherence can only be built within degenerate energy levels of the joint system.\\

The scenario itself consists of three elements that allow cooling the system of interest, see Figure~\ref{fig:elements}. These elements are comprised of a \emph{machine}, that will draw what it needs to operate from some \emph{resource}. This will in turn allow applying some \emph{operation} on the target system. These three elements: machine, resource, and operation, are the basis of the analysis that will follow. There are also what both the incoherent and the coherent scenario have in common and as such build a platform from which other scenarios can be defined.\\

\begin{figure}
\centering
\includegraphics[width=0.7 \linewidth]{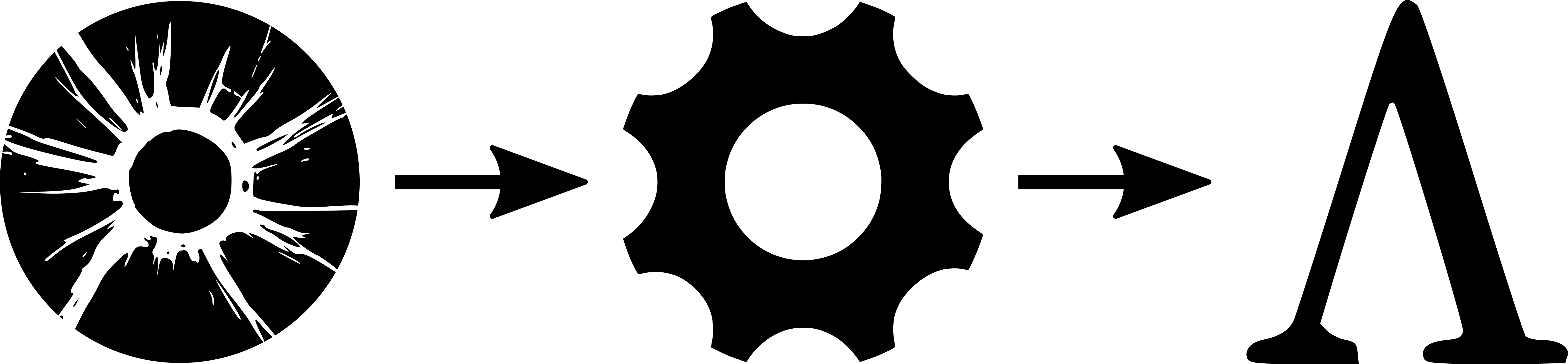}
\caption{\label{fig:elements}The three building blocks of the coherent as well as the incoherent scenario. From left to right, the \emph{resource} fueling the \emph{machine} applying an \emph{operation} on the target system.}
\end{figure}

In the coherent scenario the machine is embodied by a finite dimensional quantum system given in a thermal state at the fixed background temperature $T_R$; or, as we will later in this thesis prefer to state, at the fixed inverse background temperature

\begin{equation}
\beta_{R}= \frac{1}{T_R}.
\end{equation}
Note that we work in units of $k_B=\hbar=1$, where $k_B$ is the Boltzmann constant and $\hbar$ the reduced Planck constant. The machine is given as resource the ability to implement any unitary on the machine and target joint system. Finally, the operation this scenario allows implementing on the target system is given by the following completely positive and trace preserving (CPTP) map

\begin{equation}
\Lambda_{\text{coh}}(\rho_S)= \Tr_{M} \left(U \, \rho_S \otimes \rho_M \, U^{\dagger} \right),
\end{equation}
where $\rho_{S}$ denotes the (target) system state, $\rho_{M}$ the machine state and $U$ the joint unitary operation. We remind the reader that $\rho_{M}$ is a thermal state at inverse temperature $\beta_R$, i.e.,

\begin{equation}
\rho_{\text{M}}= \frac{e^{-\beta_R H_{M}}}{\Tr\left( e^{-\beta_R H_{M}} \right)}.
\end{equation}
See Figure~\ref{fig:cohdiagram} for an illustration of the scenario.
\begin{figure}
\centering
\def\svgwidth{0.7 \columnwidth}
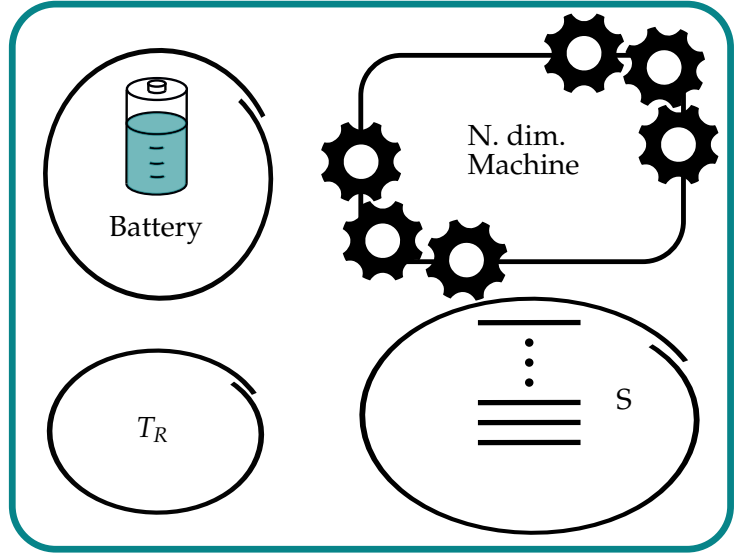
\caption{\label{fig:cohdiagram}A depiction of the coherent scenario. The target system, denoted by S, is cooled by a finite dimensional machine of Hilbert space dimension N by means of a unitary operation acting on the system and the machine jointly. This unitary invests some energy, depicted here as a battery, in the joint system. Upon repetition of the cooling procedure the machine is thermalized back to its original state thanks to a thermal bath at room temperature $T_R$.}
\end{figure}

We would like to stress that the resource of this scenario, i.e., the ability to implement the joint unitary operation, is quite abstract. Indeed, this resource can be concretely implemented in various ways. In the spirit of~\cite{aberg_2014}, one could, for example, imagine it to materialize as an external quantum system upon which a suitable operation, such as a joint energy conserving unitary, is implemented. Alternatively, taking an agent or experimentalist perspective, this resource translates as the ability to engineer any interaction Hamiltonian within the machine and target joint system. The scenario itself is, however, independent on how the resource is implemented in concrete terms.

One should also note that the unitary may be arbitrary and is in particular allowed to be energy non-conserving. It may therefore supply the joint target and machine system with energy. This energy change notably comes at no entropy change of the joint system. In that sense it is a maximally ordered source of energy. This motivates our use of the battery pictogram to symbolize the entropy-less source of energy in this scenario. Following the intuition of classical thermodynamics that work being an ordered source of energy is found to be useful, we intuitively expect this source of energy to perform efficiently.\\

We have up to now described one single application of the scenario. Repeated applications of it are however also of interest. To allow it, we rethermalize the machine to the background temperature $T_R$ and reapply the map $\Lambda_{\text{coh}}$ to the target system. After n repetitions of the scenario we therefore get the state

\begin{equation}
\Lambda_{\text{coh}}^n(\rho_S)= \underbrace{\Lambda_{\text{coh}}(\dots \Lambda_{\text{coh}}(\Lambda_{\text{coh}}(}_{\text{n times}}\rho_S))\dots).
\end{equation}
We will also consider the case $n\rightarrow \infty$ in this thesis.

\chapter{Incoherent Scenario } \label{chap:inc}

This section is devoted to defining the incoherent scenario. The name of this scenario stems from the fact that it does not allow creating coherence between the different energy levels of the machine and target joint system. It is therefore incoherent in the energy eigenbasis. This scheme is built upon the same three elements that the coherent scenario uses, namely a machine, a resource, and an operation, see Figure~\ref{fig:elements}.\\

The machine is given, as in the coherent scenario, by a finite dimensional quantum system in a thermal state at the fixed background temperature $T_R$. The resource is here given by a thermal bath at temperature $T_H > T_R$. This thermal bath allows thermalizing, depending on the structure of the machine, the entire machine or parts of it to the hotter temperature $T_H$. If the machine has no tensor product structure then either the whole machine is thermalized to $T_H$ or the hot thermal bath leaves the machine unaltered. If however the machine exhibits a tensor product structure, that is if $\rho_M=  \otimes_{i=1}^k \rho_{M_i}$ for some $k \in \mathbb{N}$, then the hot bath is able to access the different parts of the machine, i.e., the different $M_i$'s, individually and can either thermalize them to $T_H$ or leave them unchanged at $T_R$ as desired. Note that this tensor product structure effectively assumes that the Hamiltonian of the machine is comprised of k non-interacting parts, that is \begin{equation}
H_M=\sum_{i=1}^k H_{M_i} \otimes \mathds{1}_{M_i^c}.
\end{equation}

Upon (partial) rethermalization of the machine, an energy conserving unitary is applied to the machine and target joint system effecting in implementing the following CPTP map to the target system

\begin{figure}
\centering
\def\svgwidth{0.7 \columnwidth}
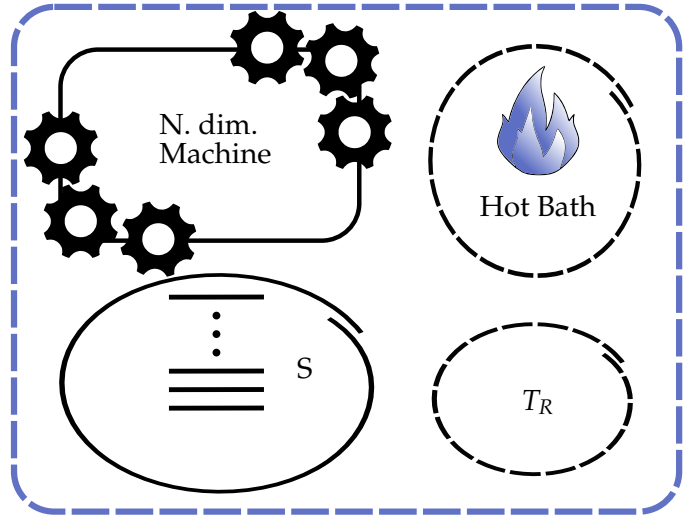
\caption{\label{fig:incodiagram}The incoherent scenario. A hot bath is allowed to get in contact with the finite dimensional machine to thermalize part of it to the hot temperature $T_H$, leaving the rest of the machine thermal at $T_R$. The target system, denoted by S, is then cooled upon applying an energy conserving unitary acting on the system and the machine jointly. The unitary being energy conserving ensures that the invested energy solely originates from the hot thermal bath. Upon repetition of the cooling procedure the room temperature part of the machine, $M_R$, is thermalized back to its original state thanks to a thermal bath at room temperature $T_R$.}
\end{figure}

\begin{equation}
\Lambda_{\text{inc}}(\rho_S)= \Tr_M(U_{\text{inc}} \, \rho_S \otimes \rho_M^{R,H} \, U_{\text{inc}}^{\dagger}),
\end{equation}
where $\rho_M^{R,H}$ denotes the state of the machine in parts thermal at $T_H$ and in parts thermal at $T_R$. Furthermore, $U_{\text{inc}}$ denotes the energy conserving unitary, meaning in this context that

\begin{equation}
[U_{\text{inc}}, H_{SM}]= 0,
\end{equation}
where $H_{SM}=H_S \otimes \mathds{1} + \mathds{1} \otimes H_M$ is the Hamiltonian of the machine and target joint system and $[A,B]$ denotes the commutator between $A$ and $B$, i.e.
\begin{equation}
[A,B]=AB-BA.
\end{equation}
 Figure~\ref{fig:incodiagram} provides a depiction of the incoherent scenario. Note that given a machine M and a choice of which parts of the machine are left unchanged and which are thermalized to $T_H$, one can always label the untouched parts as a whole as $M_R$ and the rethermalized parts as $M_R$. We then have

\begin{equation}
\rho_M^{R,H}= \tau_{M_R} \otimes \tau_{M_H}^H.
\end{equation}

We stress that $M_R$ and $M_H$ are not per se structurally given for a machine and that we might have some choice in selecting them for a given machine.\\

One may wonder why the unitary is restricted to be energy conserving in this scenario. This is simply demanded to ensure that the entire energy that is invested in the joint machine and target system solely comes from the hot thermal bath. In stark contrast to the resource of the coherent scenario, the hot thermal bath is a fully entropic source of energy. It corresponds in this sense to a maximally disordered source of energy and according to the classical thermodynamic intuition that heat being a disordered source of energy is rather useless, we intuitively expect this source of energy to perform less efficiently than an ordered one.

As for the coherent scenario, one may repeatedly apply this scenario by making explicit use of the room temperature bath to rethermalize the $M_R$ part of the machine to $T_R$. One then proceeds to thermalize the $M_H$ part of the machine to $T_H$, which enables to reapply the map $\Lambda_{\text{inc}}$ to the target system. After n repetitions of the scenario, we get the state 

\begin{equation}
\Lambda_{\text{inc}}^n(\rho_S)=\underbrace{\Lambda_{\text{inc}}(\dots \Lambda_{\text{inc}}(\Lambda_{\text{inc}}(}_{\text{n times}}\rho_S))\dots).
\end{equation}
We will also consider the case $n \rightarrow \infty$.
Note that in terms of state attainability, one can equivalently make use of the room temperature bath to rethermalize the entire machine to $T_R$ and then thermalize $M_H$ to $T_H$.  This is however generically less effective in terms of energy expenditure, as will be exemplified in Section~\ref{subsec:rep}.

\chapter{Work Cost and Temperature Quantifiers} \label{chap:worktemp}

\section{Work Cost}

We characterize the work cost in each scenario by the free energy change of the resource
\begin{equation}
\Delta F = \Delta \mathcal{U} - T_R \Delta S,
\end{equation}
with $\Delta \mathcal{U}$ the internal energy change of the resource and $\Delta S$ the entropy change of the resource. We will denote by $\Delta F_{\text{inc}}$ the change of free energy in the incoherent scenario and by $\Delta F_{\text{coh}}$ that in the coherent scenario. This choice of work quantifier is motivated as follows. For one, the free energy is a well-established monotone in thermodynamics, meaning that its value decreases the closer a state is to the thermal state at the background temperature $T_R$ \cite{Brandao-2015, Gallego-2016, Muller-2018}. Moreover, the free energy quantifies the maximum amount of work that can be extracted on average from a resource when access to a background temperature $T_R$ is granted \cite{Gallego-2016}. It therefore quantifies the useful energy that one draws from a given resource and as such makes perfect sense as a work cost quantifier. Last but not least, it is a notion that is well-defined in both of our scenarios and is therefore well suited for comparing their performance on an equal footing. We also point out that we are interested in the free energy change of the resource, not of the state of the joint system. This is because we are interested in knowing how much resource we consume to perform the desired state transformation.\\

In the incoherent scenario the resource is the hot bath and so $\Delta \mathcal{U}$ corresponds to the heat drawn from the hot bath. We denote this heat by $Q^H$. The superscript H stresses the fact that this quantity is dependent on the temperature of the hot bath $T_H$. By the first law, or more precisely by basic energy conservation, this equals the change in energy induced on $M_H$, the part of the machine thermalized to $T_H$, upon bringing it in contact with the hot bath. That is

\begin{equation}
Q^H= \Tr \left( (\tau_{M_H}^H-\tau_{M_H}) H_{M_H} \right).
\end{equation}

The change of entropy $\Delta S$ also takes a simple form in the incoherent scenario. Indeed, the entropy change of the hot bath is given by
\begin{equation}
\Delta S= \frac{Q^H} {T_H}.
\end{equation}

Here it should be noted that by assuming equality, as opposed to the generic inequality that the second law would dictate, we assume the heat bath to be infinite as well as the interaction between our system and the heat bath to be markovian and in the weak coupling regime.

We therefore have for the free energy change in the incoherent scenario
\begin{equation} \label{eq:generaldeltaFinc}
\Delta F_{\text{inc}}( T_H)= Q^H \left(1-\frac{T_R}{T_H} \right).
\end{equation}

In the coherent scenario the entropy is unchanged and $\Delta F_{\text{coh}}$ therefore corresponds to the change of energy induced on SM upon applying the unitary operation $U$ to the joint system, i.e.,

\begin{equation}
\Delta F_{\text{coh}}= \Tr \left( (U \rho_{SM} U^{\dagger} - \rho_{SM}) H_{SM} \right).
\end{equation}

Note that as desired this change of energy is, again by invoking the first law, equal to the change of energy that the unitary operation induces in the resource of the coherent scenario, whatever its concrete implementation might be.
\section{Temperature} \label{sec:temperature}

There are many notions of temperature that make sense to consider. Given a system S in some state $\sigma_S$, one can for example consider its

\begin{itemize}
\item ground state population: $\bra{0} \sigma_S \ket{0}_S$,
\item average energy: $\Tr(\sigma_S H_S)$,
\item von Neumann entropy: $- \Tr(\sigma_S \log \sigma_S)$,
\item purity: $1-\Tr (\sigma_S^2)$.
\end{itemize}

That these notions may be used as a way to determine the temperature of the state relies on the fact that they all are strictly monotonic as a function of $T$ when evaluated on the thermal states, i.e., when
\begin{equation}
\sigma_S(T)=\frac{e^{-\frac{1}{T} H_S}}{\Tr(e^{-\frac{1}{T} H_S})}.
\end{equation}

This means that given a thermal state, one can uniquely extract its temperature $T$ by knowing its ground state population, average energy, entropy, or purity. This uniquely defines functions

\begin{align}
f_{gs}:& [0,1] \rightarrow \mathbb{R} \cup \{-\infty, +\infty\}\\
f_{ae}:& [0,E_{d_S-1}] \rightarrow \mathbb{R} \cup \{-\infty, +\infty\}\\
f_{S}:& [0,\log(d_S)] \rightarrow \mathbb{R} \cup \{-\infty, +\infty\}\\
f_{p}:& [\frac{1}{d_S},1] \rightarrow \mathbb{R} \cup \{-\infty, +\infty\},
\end{align}
such that for all $T \in \mathbb{R} \cup \{-\infty, +\infty\}$,

\begin{align}
f_{gs}(\bra{0} \tau_S(T) \ket{0}_S)&= f_{ae}(\Tr(\tau_S(T) H_S))\\
&=f_S (-\Tr(\tau_S(T) \log(\tau_S(T)))\\ 
&= f_p (\Tr(\tau_S(T)^2))=T.
\end{align}

While all these notions coincide for thermal states, they are incomparable in general. That is, for a given $\sigma_S$, generically

\begin{align}
f_{gs}(\bra{0} \sigma_S \ket{0}_S)&\neq f_{ae}(\Tr(\sigma_S H_S))\\
&\neq f_S (-\Tr(\sigma_S \log(\tau_S(T)))\\ 
&\neq f_p (\Tr(\sigma_S^2)).
\end{align}

One is therefore forced to make a choice that is to some extent a matter of taste. Here we will choose to work with the ground state population. This choice implies that we implicitly have a preferred basis, the energy eigenbasis, and that we will mainly be interested in the change of the diagonal elements of the target system state in that basis. We are, in fact, even a bit more ambitious than that. We are indeed interested in maximizing the ground state population of the target system. However, once this is done, we also aim at maximizing the sum of the ground state population and the population of the first energy excited state of the target.  With that done, we then focus our attention on maximizing the sum of the three lowest excited energy eigenstates. We subsequently move our way up to the most excited state. In short, we are interested in maximizing

\begin{equation}\label{equ:ambitiouscool}
\sum_{k=0}^l \bra{k} \sigma_S \ket{k}_S, \quad \forall l=0,\dots, d_S-1.
\end{equation}

For qubit target systems, this more ambitious view on cooling will not be any different from the traditional ground state population one. Indeed, since the trace of the state is fixed, there is only one partial sum to maximize, namely that consisting of the ground state population only. In fact, all the mentioned notions of cooling actually coincide for (diagonal) qubit target systems. Our results for qubit target systems will therefore not only also apply to the ground state notion of cooling but also to the entropy, average energy and purity notions. We will come back to this in Chapter~\ref{chap:qubitsyst}. For general qudit systems, we will also find that most of our results will carry on to the other notions of cooling discussed above. There are, however, a couple of subtleties that one should be careful about in that regime. We will come back to them in Section~\ref{sec:sumtemperature}.

\chapter{Other Existing Scenarios} \label{chap:other}

Refrigeration being of high interest to the physics community, there naturally exists an array of proposed paradigms for cooling. We would like here to study how both of our scenarios relate to some of them. This will in particular help us to make conceptual links between the different paradigms as well as contextualize and compare the results of the subsequent sections with other existing scenarios.

We will start by studying the limiting case of the coherent scenario. This will result in Lemma~\ref{lemma:unitary} and Proposition~\ref{prop:cohlimit}. We will then see what the limiting case of the incoherent scenario relates to. Finally we will investigate the parallels with other paradigms for each scenario (in their non-limiting case). \\

We would like to emphasize the finite dimensionality of our machine in terms of Hilbert space dimension. Indeed, granting access to an arbitrary infinite dimensional machine allows implementing operations that are impossible with finite dimensional ones. This fact is stringent in the coherent scenario. Having access to infinite dimensional machines indeed allows implementing any CPTP map on the target system. In particular, it enables to implement the ground state cooling map

\begin{equation}
\Lambda_{gc}(\rho_S)= \Tr(\rho_S) \ket{0} \bra{0}_S.
\end{equation}

This map is positive since $\Tr(\rho_S) \geq 0$ for positive $\rho_S$. It is furthermore completely positive. Indeed, for any ancillary system A and some orthonormal basis on A that we will denote by $(\ket{i}_A)$, we have with $\rho_{SA}=\sum_{ijkl} a_{ijkl} \ket{ij} \bra{kl}_{SA}$, that

\begin{align}
\Lambda_{gc} \otimes \mathds{1}_A (\rho_{SA})&= \Lambda_{gc} \otimes \mathds{1}_A (\sum_{ijkl} a_{ijkl} \ket{ij} \bra{kl}_{SA})\\
&= \sum_{ijkl} a_{ijkl} \delta_{ik} \ket{0j} \bra{0l}_{SA}\\
&= \ket{0} \bra{0}_S \otimes (\sum_{ijl} a_{ijil} \ket{j} \bra{l}_A)\\
&= \ket{0} \bra{0}_S \otimes \Tr_S(\rho_{SA}) \geq 0.
\end{align}

As we will see in Section~\ref{sec:unibound}, due to the finite dimensionality constaint of the machine, this map cannot be implemented within the coherent scenario.\\

The proof that any CPTP map can be implemented when one allows infinite dimensional machines in the coherent scenario relies on the Stinespring's dilation theorem. To be able to apply the theorem we essentially need to have access to a big enough pure subspace of the machine. This is where the infinite dimensionality of the machine is crucial. The thermal state of any finite dimensional machine with finite energy gaps is of full rank and therefore does not have any pure subspace. As we will show now, this limitation can be circumvented with infinite dimensional machines. Of course, not all infinite dimensional machines allow cooling to the ground state. Nevertheless, choosing the machine wisely, it is possible to cool the system to the ground state with machines for which each constituent has a maximal energy gap no greater than a given fixed value. For the formal proof we first need the result of a Lemma.

Let \(\sigma_S = \sum_{i,j=0}^{d-1} []\sigma_S]_{i,j} \ket{i}\bra{j}_S\) be an arbitrary state of our system and let the machine be given in its usual thermal state \(\rho_M= \sum_{k=0}^{n-1} (\rho_M)_k \ket{k}\bra{k}_M\), i.e., $(\rho_M)_k= \frac{e^{-\beta_R \mathcal{E}_k}}{\Tr (e^{-\beta_R H_M})}$. Let \(K=\text{span}(\ket{i_0}_M, \dots, \ket{i_{c-1}}_M)\), with \(0 \leq i_0, \dots, i_{c-1} \leq n-1\) pairwise different, be a $c \leq n$ dimensional subspace of our machine. The (unnormalized) state of our machine on \(K\) is then given by \(\rho_K = \sum_{k=0}^{c-1} (\rho_M)_{i_k} \ket{i_k} \bra{i_k}_M\) and has trace \(N= \sum_{k=0}^{c-1} (\rho_M)_{i_k}\). Then we have that,

\begin{lemma}\label{lemma:unitary}
By acting unitarily on \(\sigma_S \otimes \rho_M \in \mathcal{S}(\mathcal{H}_S \otimes \mathcal{H}_M)\), one can induce with accuracy N any state on the system that one can induce by acting unitarily on \(\sigma_S \otimes \frac{\rho_K}{N} \in \mathcal{S}(\mathcal{H}_S \otimes K)\). That is
\(\forall \, U : \mathcal{H}_S \otimes K \rightarrow \mathcal{H}_S \otimes K  \; \exists \, \tilde{U}: \mathcal{H}_S \otimes \mathcal{H}_M \rightarrow \mathcal{H}_S \otimes \mathcal{H}_M \text{ such that }\)
\begin{equation} \label{eq:approxapply}
\Tr_M ( \tilde{U} \sigma_S \otimes \rho_M \tilde{U}^{\dagger})= N \left[ \Tr_K \left( U (\sigma_S \otimes \frac{\rho_K}{N}) U^{\dagger}\right) \right] + (1-N) \sigma_S.
\end{equation}

\end{lemma} 

\begin{proof}
Let \(U\) be a unitary on \(\mathcal{H}_S \otimes K\). We define the unitary \(\tilde{U}\) acting on \(\mathcal{H}_S \otimes \mathcal{H}_M\) as the trivial extension of \(U\) on \(\mathcal{H}_S \otimes \mathcal{H}_M\), i.e.,
\begin{equation}
\tilde{U} = U \oplus \left. \mathds{1}\right|_{\mathcal{H}_S \otimes (\mathcal{H}_M \ominus K)}.
\end{equation}

One then simply calculates

\begin{equation}
\begin{aligned}
\Tr_M ( \tilde{U} \sigma_S \otimes \rho_M \tilde{U}^{\dagger}) &= \Tr_M \left[ \tilde{U} \sigma_S \otimes \left( \rho_K \oplus (\rho_M -\rho_K) \right) \tilde{U}^{\dagger} \right] \\
&= \Tr_M \left[ \tilde{U} \left( \sigma_S \otimes \rho_K \right) \oplus \left( \sigma_S \otimes (\rho_M -\rho_K) \right) \tilde{U}^{\dagger} \right] \\
&= \Tr_M \left[ \left( U \sigma_S \otimes \rho_K U^{\dagger} \right) \oplus \left( \sigma_S \otimes (\rho_M -\rho_K) \right) \right] \\
&= N \left[\Tr_K \left( U \sigma_S \otimes \frac{\rho_K}{N} U^{\dagger} \right) \right] + \underbrace{\Tr \left(\rho_M -\rho_K \right)}_{=1-N} \sigma_S.
\end{aligned}
\end{equation}

\end{proof}

With this result we can now prove our assertion

\begin{proposition} \label{prop:cohlimit}
Given the ability to pick an arbitrary machine of possibly (countably) infinite Hilbert space dimension (but of local constituents restricted to having bounded energy gaps), one is able to implement any CPTP map on the system within a single application of the coherent scenario.
\end{proposition}

\begin{proof}
To prove the assertion we only need to find a specific machine that allows implementing any CPTP map. In the following we construct such a machine. Using Stinespring's dilation theorem \cite[Chapter~8.2.3]{Nielsen-2010}, Lemma~\ref{lemma:unitary} implies that as soon as one has a pure subspace of dimension $d^2$ in the machine, one can implement (approximately) any CPTP map on the system. For example, consider $n \geq d^2$ qubits each of energy gap \(E < \infty\) and let \(\ket{i_0}_M = \ket{0 \dots 0}_{M_1 \cdots M_n}, \ket{i_1}_M = \ket{0 1 \dots 1}_{M_1 \cdots M_n} , \ket{i_2}_M= \ket{101\dots 1}_{M_1 \cdots M_n} , \dots , \ket{i_{d^2-1}}_M = \ket{1\dots 101\dots1}_{M_1 \cdots M_n}\) (for \(j=1,\dots,d^2-1\), \(\ket{i_j}_M \) has the \((j-1)^{\text{th}}\) qubit in the ground state and all other qubits in the excited state). In this case, the state of our subspace \(K=\text{span}(\ket{i_0}_M, \dots, \ket{i_{c-1}}_M)\) of Lemma~\ref{lemma:unitary}, if each qubit is at room temperature, is pure in the limit of n going to infinity. By the above, one can therefore implement any CPTP map with accuracy \(N=\sum_{k=0}^{d^2-1}(\rho_M)_{i_k}\) in this limit. Taking m copies of this machine and performing with each copy the same unitary, one gets, by repeatedly using Eq.~\ref{eq:approxapply},

\begin{align}
\sigma_S^{(m)} - \Tr_K \left(U \sigma_S \otimes \frac{\rho_K}{N} U^{\dagger}\right) &= (1-N) \left[ \sigma_S^{(m-1)}- \Tr_K \left(U \sigma_S \otimes \frac{\rho_K}{N} U^{\dagger}\right)\right]\\
&= (1-N)^2 \left[ \sigma_S^{(m-2)}- \Tr_K \left(U \sigma_S \otimes \frac{\rho_K}{N} U^{\dagger}\right)\right]\\
&= \dots \\
&=(1-N)^m \left[ \sigma_S- \Tr_K \left(U \sigma_S \otimes \frac{\rho_K}{N} U^{\dagger}\right)\right],
\end{align}
where for $i=1,\dots,m$, $\sigma_S^{(i)}= \Tr_M( \tilde{U}_i \dots \tilde{U}_1 \sigma_S \otimes \rho_M \tilde{U}_1^{\dagger} \dots \tilde{U}_i^{\dagger})$, with $\tilde{U}_i$ being the unitary between the $i^{\text{th}}$ copy of the machine and the system. With the number of copies m going to infinity one gets
\begin{equation}
\sigma_S^{(\infty)}=\Tr_K \left(U \sigma_S \otimes \frac{\rho_K}{N} U^{\dagger}\right).
\end{equation}
Note that since countable unions of countable sets are countable, one only needs countably many qubits of energy gap \(E < \infty\) to implement any CPTP map perfectly, i.e., one copy of the original machine viewed differently. Also since concatenations of unitaries are unitary, performing all these unitaries on the separate copies amounts to performing one big unitary on the joint machine, i.e., the one copy machine viewed differently. And so, one can, with a room temperature machine of n qubits each of energy gap \(E < \infty\) with n going to infinity, implement any CPTP map on our system in a single application of the coherent scenario.
\end{proof}

Now that we have talked about the limiting case of the coherent scenario, we can investigate the limiting case of the incoherent scenario. Of special interest here is the case where $M_H$, the part of the machine connected to the hot bath, is left unchanged and $M_R$, the part of the machine connected to the room temperature bath, is allowed to be arbitrary. We remind ourselves that the state of the machine after having been in contact with the hot bath is given by $\rho_M^{R,H}= \tau_{M_H}^H \otimes \tau_{M_R}^R$, where $\tau_{M_H}^H$ denotes the thermal state of $M_H$ at $T_H$ and $\tau_{M_R}^R$ denotes the thermal state of $M_R$ at $T_R$. This limiting case of the incoherent scenario amounts to looking at the state transformations of the system $\tilde{S}=SM_H$ with initial state $\rho_{\tilde{S}}= \tau_S \otimes \tau_{M_H}^H$ within the framework of thermal operations \cite{Lostaglio-2019,Gour-2015,Brandao-2013}. One is then interested in the state transformations such that $\Tr_{M_H}[\text{TO}(\rho_{\tilde{S}})]$ is colder than $\tau_S$. Table~\ref{tab:infinitemachines} gives an overview of the limiting cases of the coherent and incoherent scenarios. 

\begin{table} [h!]
\begin{center} 
\begin{tabular}{ |c|c| } 
 \hline
 \textbf{Finite Machine} & \textbf{Infinite Machine}  \\
 \hline 
 Coherent $\Lambda_{\text{coh}}$ & CPTP map  \\ 
 \hline
  Incoherent $\Lambda_{\text{inc}}$ & Thermal Operations  \\ 
 \hline
\end{tabular}
\caption{The fact that the machine in both our scenarios is finite is emphasized. Infinite machine sizes allow implementing either any CPTP on our system in the coherent scenario or corresponds to a specific state transition within the framework of thermal operations in the incoherent scenario.}
\label{tab:infinitemachines}
\end{center}
\end{table}

Not only the limiting cases of the corresponding scenarios relate to studied paradigms. The coherent scenario is very closely related to heat bath algorithmic cooling (HBAC) \cite{Schulman-1999, Boykin-2002, Rodriguez-Briones-2017, Raeisi-2015, Rodriguez-Briones-2017, Rodriguez-Briones-2016}. In this scheme, one is usually interested in cooling a qubit system aided by N qubits. The N qubits that allow cooling are divided in two categories: reset qubits and compression qubits. Reset qubits are sometimes also denoted scratch qubits. All the N qubits as well as the system qubit start at room temperature $T_R$. A global unitary is then applied to cool the system S. The reset qubits are then rethermalized to $T_R$ while the compression qubits are left untouched. The procedure of global unitary cooling and rethermalization of the reset qubits is then allowed to be repeated as many times as wished. The limit of infinite repetitions is typically of interest. Our coherent scenario has no compression qubit and is in that sense a special case of HBAC. The system of interest S as well as the machine M, embodying the reset system, is, however, not limited to being qubit systems. The coherent scenario is therefore a generalization of HBAC without compression qubits. It is also worth mentioning that there exists generalizations of HBAC. The compression qubits can for example be promoted to arbitrary finite dimensional quantum systems \cite{Park-2016}. One can also generalize the thermalization step \cite{Alhambra-2019}. In particular, the authors in \cite{Alhambra-2019} take a similar stance than us in generalizing HBAC in that they develop a generalization of HBAC in the special case of HBAC with no compression qubit. Finally, by extending the coherent scenario, one can fully generalize HBAC including compression systems, see~\cite{Taranto-2020}.\\

The coherent scenario furthermore includes every implementation of a quantum Otto engine. This is due to the fact that a given quantum Otto engine implements a specific Hamiltonian, which results in a specific unitary, in the strokes, where the system is disconnected to the thermal baths. Furthermore, as we will see in detail in Section~\ref{sec:unibound}, we can absorb the contact with the hot thermal bath into the unitary action step --- see Lemma~\ref{lemma:singlecohpower} and Lemma~\ref{lemma:kcohpower}.\\

Finally, the incoherent scenario is linked to autonomous cooling in the sense that in the limit of infinite repetition of incoherent cooling, we recover the steady state achieved in autonomous cooling, when the corresponding interacting Hamiltonian is turned on in autonomous cooling.

\chapter{General Remarks} \label{chap:remarks}

\section{Remarks on the Coherent Scenario} \label{sec:remarkscoh}

The fact that we have a preferred basis in mind, namely the energy eigenbasis, has some major consequences for the coherent scenario. Indeed, it implies that we can harness the tools laid by the theory of majorization to help us derive most of our results. We will here not adventure ourselves in the delicate task of giving a survey of majorization theory and will assume that the reader is familiar with it. We nevertheless state the definition of majorization for completeness before highlighting how the theory relates to our problem. For the rest we refer to the excellent book of Marshal and Olkin \cite{Marshall-2011}.\\

We remind the reader that given an integer $d \in \mathds{N}$ and a vector $v=\{v_0,\dots,v_{d-1}\}$ we denote by 
\begin{equation}
v_0^{\downarrow} \geq v_1^{\downarrow} \geq \dots \geq v_{d-1}^{\downarrow}
\end{equation}
 the components of v arranged in decreasing order and 
 \begin{equation}
 v^{\downarrow}=(v_0^{\downarrow},\dots,v_{d-1}^{\downarrow}).
 \end{equation}
 
 Analogously
 
 \begin{equation}
v_0^{\uparrow} \leq v_1^{\uparrow} \leq \dots \leq v_{d-1}^{\uparrow}
\end{equation}
 denote the components of v arranged in increasing order and 
 \begin{equation}
 v^{\uparrow}=(v_0^{\uparrow},\dots,v_{d-1}^{\uparrow}).
 \end{equation}
 
 Then we have
 
 \begin{definition}[Majorization]
 Let $d \in \mathds{N}$ and $v,q \in \mathds{R}^d$. We say that $v$ {\bf majorizes} $w$, and write $v \succ w$, if
 
 \begin{align}
 \sum_{i=0}^k v_i^{\downarrow} &\geq \sum_{i=0}^k w_i^{\downarrow}, \quad \forall i=0,\dots, d-2\\
 \sum_{i=0}^{d-1} v_i^{\downarrow} &= \sum_{i=0}^{d-1} w_i^{\downarrow}.
 \end{align}
 We say that $v$ is {\bf  majorized } by $w$, and write $v \prec w$, if $w$ majorizes $v$.
 \end{definition}
 Two results of majorization theory combined set the link between our problem and this theory. The first result was obtained by Schur in 1923 \cite[Chapter~9.B]{Marshall-2011} and reads

\begin{theorem}[Schur]\label{thm:Schur}
Let A be a $d \times d$ hermitian matrix. Let $\mathfrak{D}(A)$ be the vector of diagonal entries of A and $\lambda(A)$ the vector of eigenvalues of A. Then
\begin{equation}
\mathfrak{D}(A) \prec \lambda(A).
\end{equation}
\end{theorem}

There is a small Corollary of Schur's Theorem that is of particular interest to us.

\begin{corollary} \label{cor:Uincluded}
Let $\sigma_{SM}$ be a state diagonal in $(\ket{i}_{SM})_{i=0}^{d_{SM}-1}$. Let $\lambda(\sigma_{SM})$ be the vector of eigenvalues of $\sigma_{SM}$. Let U be a unitary. Then
\begin{equation}
\mathfrak{D}(U \sigma_{SM} U^{\dagger}) \prec \lambda(\sigma_{SM}),
\end{equation} 
where $\mathfrak{D}( U \sigma_{SM} U^{\dagger})$ is the diagonal of $U \sigma_{SM} U^{\dagger}$ in $(\ket{i}_{SM})_{i=0}^{d_{SM}-1}$.
\end{corollary}

What Corollary~\ref{cor:Uincluded} proves is that given a state $\sigma_{SM}$ diagonal in $(\ket{i}_{SM})_{i=0}^{d_{SM}-1}$
\begin{equation}
\left\{ x \in \mathds{R}^{d_{SM}} \mid x=\mathfrak{D}(U \sigma_{SM} U^{\dagger}) \text{ for some unitary } U \right\} \subset \left\{ x \in \mathds{R}^{d_{SM}} \mid x \prec \lambda(\sigma_{SM}) \right\}.
\end{equation}

The natural next question to ask is if the reverse is also true, that is if

\begin{equation}
\left\{ x \in \mathds{R}^{d_{SM}} \mid x \prec \lambda(\sigma_{SM}) \right\} \subset \left\{ x \in \mathds{R}^{d_{SM}} \mid x=\mathfrak{D}(U \sigma_{SM} U^{\dagger}) \text{ for some unitary } U \right\} .
\end{equation}

This is what Horn in 1954 showed by proving a result stronger than the converse of Corollary~\ref{cor:Uincluded} \cite[Chapter~9.B]{Marshall-2011}.

\begin{theorem}[Horn] \label{thm:Horn}
Let $x,y \in \mathds{R}^d$. Suppose $x \prec y$. Then
\begin{equation}
x_j = \sum_k \lvert u_{jk} \rvert^2 y_k,
\end{equation}
for some real unitary $U=(u_{ij})$.
\end{theorem}

The result is stronger since a real unitary suffices, i.e., a unitary with real entries. To see that Horn's Theorem provides the desired converse, let $\sigma_{SM}$ be diagonal in the energy eigenbasis. Let $y=\lambda(\sigma_{SM})$ and let $x \prec y$. What Horn's Theorem tells us is that there exists a (real) unitary $U$ such that $x = \mathfrak{D}(U \sigma_{SM} U^{\dagger})$. Indeed, for all $j=0,\dots,d_{SM}-1$

\begin{align}
\mathfrak{D}(U \sigma_{SM} U^{\dagger})_j &= \sum_{kl} u_{jk} [\sigma_{SM}]_{kl} \bar{u}_{jl}\\
&= \sum_k \lvert u_{jk} \rvert^2 [\sigma_{SM}]_{kk}\\
&=\sum_k \lvert u_{jk} \rvert^2 y_k = x_j.
\end{align}

This therefore proves the following result

\begin{theorem}[Schur-Horn]\label{thm:Schur-Horn}
Given a state $\sigma_{SM}$ diagonal in $(\ket{i}_{SM})_{i=0}^{d_{SM}-1}$
\begin{equation}
\left\{ x \in \mathds{R}^{d_{SM}} \mid x=\mathfrak{D}(U \sigma_{SM} U^{\dagger}) \text{ for some unitary } U \right\} = \left\{ x \in \mathds{R}^{d_{SM}} \mid x \prec \lambda(\sigma_{SM}) \right\}.
\end{equation}
\end{theorem}

We will refer to Theorem~\ref{thm:Schur-Horn} as the Schur-Horn Theorem. It should be noted, however, that some people prefer to refer to it as Horn's Theorem. We will reserve that label to refer to Theorem \ref{thm:Horn} to avoid confusion.\\

Now that we have laid out the connection that the coherent scenario has with majorization theory, we can exploit this to see how it can help us to formally formulate our problem of interest. We will then work out an interesting instance of the problem in order to give more intuition  as well as to prepare us to the bottom-up analysis that will follow in the subsequent sections.

Given a $d_S$ dimensional system S and a $d_M$ dimensional machine M, what we are in general interested in is to find the unitary that cools to target system at the minimum work cost $\Delta F_{\text{coh}}$ to a given temperature $T_{\text{coh}} \in [T_{\text{coh}}^*, T_R]$. From our definition of temperature, this means that we would like to cool the target to a given ground state population $ r_{\text{coh}} \in [r_S,r_{\text{coh}}^*]$. Given some $c \in [r_S,r_{\text{coh}}^*]$, we are therefore interested in solving

\begin{equation}
\min_{U} \Delta F_{\text{coh}}, \quad \text{s.t. } r_{\text{coh}}=c.
\end{equation}

Now we have that 

\begin{align}
r_{\text{coh}} &= \sum_{i=0}^{d_M-1} \bra{0i} U \rho_{SM} U^{\dagger} \ket{0i}_{SM}\\
&= \sum_{i=0}^{d_M-1} [\mathfrak{D}(U \rho_{SM} U^{\dagger}) ]_i.
\end{align}

Also, given any state $\sigma_{SM}$,

\begin{align}
\Tr(\sigma_{SM} H_{SM}) &= \sum_{i,j=0}^{d_{SM}-1} [\sigma_{SM}]_{ij} [H_{SM}]_{ji}\\
&= \sum_{i=0}^{d_{SM}-1} [\sigma_{SM}]_{ii} [H_{SM}]_{ii}\\
&= \mathfrak{D}(\sigma_{SM}) \cdot \mathfrak{D}(H_{SM}),
\end{align}

so that

\begin{align}
\Delta F_{\text{coh}}&=\Tr \left( U \rho_{SM} U^{\dagger} H_{SM}\right)-\Tr \left(\rho_{SM} H_{SM}\right)\\
&= \mathfrak{D}(U \rho_{SM} U^{\dagger}) \cdot \mathfrak{D}(H_{SM})-\underbrace{\mathfrak{D}(\rho_{SM} ) \cdot \mathfrak{D}(H_{SM})}_{=\text{const.}.}.
\end{align}

We can therefore reformulate our problem as

\begin{align}
& \min_{U} \left( \mathfrak{D}(U \rho_{SM} U^{\dagger}) \cdot \mathfrak{D}(H_{SM}) \right), \quad \text{s.t. } \sum_{i=0}^{d_{M}-1} \left[ \mathfrak{D}(U \rho_{SM} U^{\dagger}) \right]_i\\
=&\min_{v \in A} v \cdot \mathfrak{D}(H_{SM}), \quad \text{s.t. } \sum_{i=0}^{d_{M}-1} v_i=c,
\end{align}

with
\begin{equation}
A= \left\{v \in \mathds{R}^{d_{SM}} \mid v=\mathfrak{D}(U \rho_{SM} U^{\dagger}) \text{ for some unitary } U \right\}.
\end{equation}

What the Schur-Horn Theorem, Theorem \ref{thm:Schur-Horn}, tells us is that

\begin{equation}
A= \left\{v \in \mathds{R}^{d_{SM}} \mid v \prec \mathfrak{D}(\rho_{SM}) \right\}.
\end{equation}

The problem of finding the unitary that cools the target system to a given ground state population $c$ at a minimal work cost therefore reduces to

\begin{equation} \label{eq:deltaFoptimization}
\min_{v \prec \mathfrak{D}(\rho_{SM})} v \cdot \mathfrak{D}(H_{SM}), \quad \text{s.t. } \sum_{i=0}^{d_{M}-1} v_i=c.
\end{equation}

With that reformulation of the problem established it is now easy to explicitly calculate $r_{\text{coh}}^*$ from the entries of $ \mathfrak{D}(\rho_{SM})$.

\begin{lemma}
$r_{\text{coh}}^*=\sum_{i=0}^{d_M-1} \left[\mathfrak{D}^{\downarrow}(\rho_{SM})\right]_i$.
\end{lemma}

\begin{proof}
The proof is a straightforward application of the Schur-Horn theorem and the basic definition of majorization. Indeed, as we saw,

\begin{align}
\max_U r_{\text{coh}} &= \max_U \sum_{i=0}^{d_M-1} \left[ \mathfrak{D}(U \rho_{SM} U^{\dagger}) \right]_i\\
&= \max_{v \in A} \sum_{i=0}^{d_M-1} v_i\\
&= \max_{v \prec \mathfrak{D}(\rho_{SM})} \sum_{i=0}^{d_M-1} v_i.
\end{align}

However, for all $v \prec \mathfrak{D}(\rho_{SM})$, by definition

\begin{equation}
\sum_{i=0}^{d_M-1} v_i \leq \sum_{i=0}^{d_M-1} v_i^{\downarrow} \leq \sum_{i=0}^{d_M-1} \left[ \mathfrak{D}^{\downarrow} (\rho_{SM}) \right]_i.
\end{equation}

Furthermore, $\sum_{i=0}^{d_M-1} v_i=\sum_{i=0}^{d_M-1} \left[ \mathfrak{D}^{\downarrow} (\rho_{SM}) \right]_i$ by choosing the unitary $U$ acting on SM that reorders the $d_M$ biggest entries of $\mathfrak{D}(\rho_{SM})$ in the $\ket{00}_{SM}, \dots, \ket{0, d_M-1}_{SM}$ eigenstates.
\end{proof}

We next want to investigate the case $c=r_{\text{coh}}^*$ more closely. That is, we are interested in

\begin{equation}
\min_{v \prec \mathfrak{D}(\rho_{SM})} v \cdot \mathfrak{D}(H_{SM}), \quad \text{s.t. } \sum_{i=0}^{d_{M}-1} v_i=r_{\text{coh}}^*.
\end{equation}

First of all, one may wonder if the example of the unitary achieving $r_{\text{coh}}^*$ is just an example or if it is in some sense unique. To this end we formulate the following Lemma.

\begin{lemma} \label{lemma:biggest}
For any vector $v \prec \mathfrak{D}(\rho_{SM})$ such that $ \sum_{i=0}^{d_M-1} v_i = r_{\text{coh}}^*$, its first $d_M$ entries are its biggest entries.
\end{lemma}

\begin{proof}
Let $v \in \mathds{R}^{d_{SM}-1}$ with $\sum_{i=0}^{d_M-1} v_i = r_{\text{coh}}^*$. Suppose the statement is wrong. That is, suppose there exists an index $i \leq d_M$ and an index $j > d_M$ such that $v_i < v_j$. Then, denoting by $P_{ij}$ the permutation matrix permuting index $i$ and $j$ only, we have
\begin{equation}
v'= P_{ij} v \prec v \prec \mathfrak{D}(\rho_{SM}).
\end{equation} 
And so $\sum_{k=0}^{d_M-1} [v']_k \leq r_{\text{coh}}^*$. But

\begin{equation}
\sum_{k=0}^{d_M-1} [v']_k=\sum_{\substack{k=0 \\ k \neq i}}^{d_M-1} [v']_k + v_j > \sum_{\substack{k=0 \\ k \neq i}}^{d_M-1} v_k + v_i = r_{\text{coh}}^* \quad \lightning. 
\end{equation}
\end{proof}

The result of Lemma \ref{lemma:biggest} is by itself quite uninteresting. It is appealing to us, however, since it allows us to prove the next, much more interesting result. For this next result we need a bit of notation. Given a vector $w \in \mathds{R}^{d_{SM}}$ we define
\begin{align}
a_w &= (w_0, \dots, w_{d_M-1})\\
b_w &= (w_{d_M}, \dots, w_{d_{SM}-1}).
\end{align}

Note that $w = a_w \oplus b_w$, i.e., $a_w$ and $b_w$ split $w$ in two. With this we can state our next result.

\begin{lemma}\label{lemma:split}
Let $v \in \mathds{R}^{d_{SM}}$ such that $\sum_{i=0}^{d_M-1} v_i = r_{\text{coh}}^*$. Then
\begin{equation}
v \prec \mathfrak{D}(\rho_{SM}) \Leftrightarrow a_v \prec a_{\mathfrak{D}^{\downarrow}(\rho_{SM})} \text{ and } b_v \prec b_{\mathfrak{D}^{\downarrow}(\rho_{SM})}.
\end{equation}
\end{lemma}

\begin{proof}
We first prove the forward implication \enquote{$\Rightarrow$}. What Lemma~\ref{lemma:biggest} tells us is that given a $v \in \mathbb{R}^{d_{SM}}$
\begin{align}
a_v^{\downarrow}&= \left( [v^{\downarrow}]_0, \dots, [v^{\downarrow} ]_{d_M-1} \right),\\
b_v^{\downarrow}&= \left( [v^{\downarrow} ]_{d_M}, \dots, [v^{\downarrow} ]_{d_{SM}-1} \right).
\end{align}

Using $v \prec \mathfrak{D}(\rho_{SM})$, we have that for all $l=0, \dots, d_M-1$,
\begin{equation}
\sum_{i=0}^l [a_v^{\downarrow}]_i = \sum_{i=0}^l [v^{\downarrow}]_i \leq \sum_{i=0}^l [ \mathfrak{D}^{\downarrow} (\rho_{SM})]_i
\end{equation}
with equality for $l=d_M-1$ as
\begin{equation}
\sum_{i=0}^{d_M-1} [a_v^{\downarrow}]_i= \sum_{i=0}^{d_M-1} v_i = r_{\text{coh}}^*= \sum_{i=0}^{d_M-1} [\mathfrak{D}^{\downarrow}(\rho_{SM})]_i.
\end{equation}

This means that $a_v \prec a_{\mathfrak{D}^{\downarrow} (\rho_{SM})}$. Analogously, for all $l=0,\dots, d_{SM}-d_M-1$
\begin{align}
\sum_{i=0}^l [b_v^{\downarrow}]_i= \sum_{i=0}^l [v^{\downarrow}]_{d_M+i} &= \sum_{i=0}^{l+d_M} [v^{\downarrow}]_i-r_{\text{coh}}^*\\
&\leq \sum_{i=0}^{l+d_M} [\mathfrak{D}^{\downarrow}(\rho_{SM})]_i-r_{\text{coh}}^*\\
&= \sum_{i=0}^{l} [\mathfrak{D}^{\downarrow}(\rho_{SM})]_{d_M+i}
\end{align}
with equality for $l=d_{SM}-d_M-1$ as
\begin{equation}
\sum_{i=0}^{d_{SM}-d_M-1} [b_v^{\downarrow}]_i= \sum_{i=0}^{d_{SM}-d_M-1} v_{d_M+i}=1-r_{\text{coh}}^* = \sum_{i=0}^{d_{SM}-d_M-1} [\mathfrak{D}^{\downarrow} (\rho_{SM})]_{d_M+i}.
\end{equation}

So $b_v \prec b_{\mathfrak{D}^{\downarrow} (\rho_{SM})}$. This proves the forward implication. The backward implication \enquote{$\Leftarrow$} follows from
\begin{equation}
v_1 \prec v_2 \text{ and } w_1 \prec w_2 \Rightarrow v_1 \oplus w_1 \prec v_2 \oplus w_2,
\end{equation}
see \cite{Bondar-2003}. With this we get
\begin{equation}
v=a_v \oplus b_v \prec a_{\mathfrak{D}^{\downarrow} (\rho_{SM})} \oplus b_{\mathfrak{D}^{\downarrow} (\rho_{SM})} = \mathfrak{D}^{\downarrow} (\rho_{SM}).
\end{equation}
And so, as desired $v \prec \mathfrak{D}(\rho_{SM})$.
\end{proof}

Lemma \ref{lemma:split} allows us to divide our constrained problem into two unconstrained ones which can each individually easily be solved. Indeed

\begin{align}
& \min_{v \prec \mathfrak{D} (\rho_{SM})} v \cdot \mathfrak{D}(\rho_{SM}), \quad \text{s.t. } \sum_{i=0}^{d_M-1} v_i = r_{\text{coh}}^*\\
=& \min_{a_v \prec a_{\mathfrak{D}^{\downarrow} (\rho_{SM})}} \min_{b_v \prec b_{\mathfrak{D}^{\downarrow} (\rho_{SM})}} (a_v \oplus b_v) \cdot (a_{\mathfrak{D}(\rho_{SM})} \oplus b_{\mathfrak{D} (\rho_{SM})})\\
=& \min_{a_v \prec a_{\mathfrak{D}^{\downarrow} (\rho_{SM})}}  a_v \cdot a_{\mathfrak{D} (\rho_{SM})} \oplus \min_{b_v \prec b_{\mathfrak{D}^{\downarrow} (\rho_{SM})}}  b_v \cdot  b_{\mathfrak{D} (\rho_{SM})}.
\end{align}

The constraint in the second line disappears because for any $a_v \prec a_{\mathfrak{D}^{\downarrow} (\rho_{SM})}$, automatically
\begin{equation}
\sum_{i=0}^{d_M-1} v_i=\sum_{i=0}^{d_M-1} [a_v]_i= \sum_{i=0}^{d_M-1} [a_{\mathfrak{D}^{\downarrow} (\rho_{SM})}]_i= r_{\text{coh}}^*.
\end{equation}

We have therefore reformulated our problem in two instances of the well-known passivity problem~\cite{Pusz-1978, Lenard-1978}. The solution of the latter is given by $\sum_{i=0}^{d_M-1} [a_v]_i^{\uparrow} [a_{\mathfrak{D} (\rho_{SM})}]_i^{\downarrow}$ and $\sum_{i=0}^{d_M-1} [b_v]_i^{\uparrow} [b_{\mathfrak{D} (\rho_{SM})}]_i^{\downarrow}$ respectively. That corresponds to inversely ordering the entries of $a_v$ and $b_v$ with respect to those of $a_{\mathfrak{D} (\rho_{SM})}$ and $b_{\mathfrak{D} (\rho_{SM})}$.

We have therefore solved the endpoint of our optimization problem for any machine M and system S. We will see, however, that solving the problem for instances where one is not interested in maximal cooling, is much more involved. This is of particular interest if one has a restricted amount of free energy, or in other words resource, at hand. Solving the optimization problem in that case is in fact already surprisingly challenging for the simplest instances of the problem, namely the one and two qubit machine case for a single qubit system. We will explore this in Section \ref{sec:onequbitmachine} and Section \ref{sec:twoqubitmachine} respectively.

\section{Remarks on the Incoherent Scenario} \label{sec:remarksincoh}

We would like here to analyze more closely what the energy conserving constraint on the unitary operation implies in this scenario. This is essentially a straightforward application of basic linear algebra and is probably immediately clear to the reader familiar with it. Nevertheless, we state and prove the results to remind ourselves of these facts as well as for precise and concise reference in the subsequent sections.\\

\begin{lemma}\label{lemma:directsum}
$U_{\text{inc}} = \oplus_{i=0}^{k_{SM}-1} U_{\tilde{\epsilon}_i}$, where $U_{\tilde{\epsilon}_i} : \text{Eig}_{H_{SM}} (\tilde{\epsilon}_i) \rightarrow \text{Eig}_{H_{SM}} (\tilde{\epsilon}_i)$.
\end{lemma}

\begin{proof}
This is a direct consequence of $[U_{\text{in}}, H_{SM}]=0$. Indeed, let $\text{Eig}_{H_{SM}}(\tilde{\epsilon})$ be the eigenspace of $H_{SM}$ with eigenvalue $\tilde{\epsilon}$. Let $ \ket{i}_{SM} \in \text{Eig}_{H_{SM}}(\tilde{\epsilon})$. Per definition $H_{SM} \ket{i}_{SM} = \tilde{\epsilon} \ket{i}_{SM}$.
Furthermore, as $[U_{\text{inc}}, H_{SM}] = 0$, we have
\begin{align}
U_{\text{inc}} H_{SM} \ket{i}_{SM} &= H_{SM} U_{\text{inc}} \ket{i}_{SM}\\
\Leftrightarrow \tilde{\epsilon} (U_{\text{inc}} \ket{i}_{SM}) &= H_{SM} (U_{\text{inc}} \ket{i}_{SM}),
\end{align} 
which shows that $ U_{\text{inc}} \ket{i}_{SM} \in \text{Eig}_{H_{SM}}(\tilde{\epsilon})$. This means that every energy eigenspace is invariant under $U_{\text{inc}}$ and so as the whole vector space can be decomposed as a direct sum of $\text{Eig}_{H_{SM}}(\tilde{\epsilon})$, $U_{\text{inc}} = \oplus_{\tilde{\epsilon}} U_{\tilde{\epsilon}}$ as desired.
\end{proof}

\begin{lemma} \label{lemma:staydiag}
If $\sigma_S$ is diagonal in $(\ket{i}_S)_{i=0}^{d_S-1}$ up to local energy conserving unitaries, then so is $\Lambda_{\text{inc}}(\sigma_S)$.
\end{lemma}

\begin{proof}
This is a straightforward calculation. For it, it is convenient to introduce a notation for the energy eigenbasis that distinguishes if an energy eigenvector has the same eigenvalue as another one or not. For $ i \in \{0, \dots, k_S-1\}$, let $n_i = \text{dim}\left( \text{Eig}_{H_S} (\tilde{E}_i) \right)$. For $i=0,\dots, k_S-1$ and $j_i=0, \dots, n_i-1$ let $\ket{v_{\tilde{E}_i}^{j_i}}_S$ be such that
\begin{equation}
H_S \ket{v_{\tilde{E}_i}^{j_i}}_S = E_i \ket{v_{\tilde{E}_i}^{j_i}}_S.
\end{equation}

That $\sigma_S$ is diagonal in the energy eigenbasis up to local energy conserving unitaries means that there exists a $U_S$, with $[U_S,H_S]=0$, such that $U_S \sigma_S U_S^{\dagger}$ is diagonal in the energy eigenbasis. Lemma \ref{lemma:directsum} implies that $\sigma_S$ can only have off diagonal entries within each eigenspace of $H_S$, i.e.,

\begin{equation}
\sigma_S= \sum_{i=0}^{k_S-1} \sum_{j_i,m_i=0}^{n_i-1}  (\sigma_S)_{i}^{j_i, m_i} \ket{v_{\tilde{E}_i}^{j_i}} \bra{v_{\tilde{E}_i}^{m_i}}_S,
\end{equation}
with $(\sigma_S)_{i}^{j_i, m_i} \in \mathbb{C}$ for  $i=0,\dots, k_S-1$ and $j_i, m_i=0, \dots, n_i-1$. Then

\begin{align}
\Lambda_{\text{inc}}(\sigma_S)&= \Tr_M ( U_{\text{inc}} \sigma_S \otimes \tau_M U_{\text{inc}}^{\dagger})\\
&= \Tr_M ( U_{\text{inc}} \sum_{i=0}^{k_S-1} \sum_{j_i,m_i=0}^{n_i-1}(\sigma_S)_{i}^{j_i, m_i} \ket{v_{\tilde{E}_i}^{j_i}} \bra{v_{\tilde{E}_i}^{m_i}}_S \otimes \sum_{k=0}^{d_M-1} (\tau_M)_k \ket{k} \bra{k}_M U_{\text{inc}}^{\dagger})\\
&=\Tr_M ( \sum_{i=0}^{k_S-1} \sum_{j_i,m_i=0}^{n_i-1} \sum_{k=0}^{d_M-1} (\sigma_S)_{i}^{j_i, m_i}  (\tau_M)_k U_{\text{inc}} \ket{v_{\tilde{E}_i}^{j_i} k} \bra{v_{\tilde{E}_i}^{m_i} k}_{SM} U_{\text{inc}}^{\dagger})\\
\end{align}

From Lemma \ref{lemma:directsum} we have that 

\begin{align}
U_{\text{inc}} \ket{v_{\tilde{E}_i}^{j_i} k}_{SM}&= \sum_{\tilde{E}_p+\mathcal{E}_q=\tilde{E}_i+\mathcal{E}_k} \sum_{l_p=0}^{n_p-1} a_{pq}^{l_p}(i,j_i,k) \ket{v_{\tilde{E}_p}^{l_p} q}_{SM}\\
 \bra{v_{\tilde{E}_i}^{m_i} k}_{SM} U_{\text{inc}}^{\dagger}&= \sum_{\tilde{E}_r+\mathcal{E}_s=\tilde{E}_i+\mathcal{E}_k} \sum_{t_r=0}^{n_r-1} \bar{a}_{rs}^{t_r}(i,m_i,k) \bra{v_{\tilde{E}_r}^{t_r} s}_{SM},
\end{align}
with $a_{pq}^{l_p}(i,j_i,k) \in \mathbb{C}$ for $i,p=0,\dots, k_S-1$, $j_i, l_p=0, \dots, n_i-1$, and $ k,q=0, \dots, d_M-1 $. The sum $ \tilde{E}_p+\mathcal{E}_q=\tilde{E}_i+\mathcal{E}_k$ is to be understood as summing over all p and q such that $ \tilde{E}_p+\mathcal{E}_q=\tilde{E}_i+\mathcal{E}_k$. With this we have

\begin{align}
\Lambda_{\text{inc}}(\sigma_S)&= \Tr_M \left( \sum_{i=0}^{k_S-1} \sum_{j_i,m_i=0}^{n_i-1}  \sum_{k=0}^{d_M-1} \sum_{\tilde{E}_p+\mathcal{E}_q=\tilde{E}_i+\mathcal{E}_k} \sum_{l_p=0}^{n_p-1} \sum_{\tilde{E}_r+\mathcal{E}_s=\tilde{E}_i+\mathcal{E}_k} \sum_{t_r=0}^{n_r-1} \right.\\
& \quad \left. (\sigma_S)_{i}^{j_i, m_i}  (\tau_M)_k  a_{pq}^{l_p}(i,j_i,k) \bar{a}_{rs}^{t_r}(i,m_i,k) \ket{v_{\tilde{E}_p}^{l_p} q} \bra{v_{\tilde{E}_r}^{t_r} s}_{SM}\right)\\
&=  \sum_{i=0}^{k_S-1} \sum_{j_i,m_i=0}^{n_i-1}  \sum_{k=0}^{d_M-1} \sum_{\tilde{E}_p+\mathcal{E}_q=\tilde{E}_i+\mathcal{E}_k} \sum_{l_p=0}^{n_p-1} \sum_{\tilde{E}_r+\mathcal{E}_s=\tilde{E}_i+\mathcal{E}_k} \sum_{t_r=0}^{n_r-1}\\
& \quad (\sigma_S)_{i}^{j_i, m_i}  (\tau_M)_k  a_{pq}^{l_p}(i,j_i,k) \bar{a}_{rs}^{t_r}(i,m_i,k) \ket{v_{\tilde{E}_p}^{l_p}} \bra{v_{\tilde{E}_r}^{t_r}}_S \delta_{qs}\\
&=  \sum_{i=0}^{k_S-1} \sum_{j_i,m_i=0}^{n_i-1}  \sum_{k=0}^{d_M-1} \sum_{\tilde{E}_p=\tilde{E}_i+\mathcal{E}_k-\mathcal{E}_q} \sum_{l_p=0}^{n_p-1} \sum_{t_p=0}^{n_p-1}\\
& \quad (\sigma_S)_{i}^{j_i, m_i}  (\tau_M)_k  a_{pq}^{l_p}(i,j_i,k) \bar{a}_{pq}^{t_p}(i,m_i,k) \ket{v_{\tilde{E}_p}^{l_p}} \bra{v_{\tilde{E}_p}^{t_p}}_S,
\end{align}
where in the last step we used that $q=s$ implies
\begin{equation}
\tilde{E}_p=\tilde{E}_i+\mathcal{E}_k-\mathcal{E}_q=\tilde{E}_i+\mathcal{E}_k-\mathcal{E}_s=\tilde{E}_r,
\end{equation}
meaning $p=r$.
So $\Lambda_{\text{inc}}(\sigma_S)$ is block diagonal with each block in $ \text{Eig}_{H_S}(\tilde{E}_i)$ for some $i=0,\dots, k_S$. As $\Lambda_{\text{inc}}(\sigma_S)$ is positive semi-definite, so is each one of its block. This implies that each of its block is unitarily diagonalizable. Each of this unitary is of course local. They also are energy conserving since within each block $H_S \propto \mathds{1}$. This ends the proof.
\end{proof}

The result of Lemma \ref{lemma:staydiag} greatly simplifies our view of the reachable states in the incoherent paradigms. Indeed, our initial state of the system, $\rho_S$, is diagonal in the energy eigenbasis. Lemma \ref{lemma:staydiag} implies that $\Lambda_{\text{inc}}(\rho_S)$ is diagonal up to local energy conserving unitaries. Applying Lemma \ref{lemma:staydiag} to $\Lambda_{\text{inc}}(\rho_S)$, we find that the same holds for $\Lambda^{(2)}_{\text{inc}}(\rho_S)$, and therefore for any $\Lambda^{(k)}_{\text{inc}}(\rho_S)$, $k \in \mathbb{N} \cup \{\infty\}$. Furthermore, since local energy conserving unitaries are allowed within the incoherent paradigm and assumed to be for free, we can w.l.o.g. assume $\Lambda^{(k)}_{\text{inc}}(\rho_S)$ to be diagonal in the energy eigenbasis for any $k \in \mathbb{N} \cup \{\infty\}$. \\

This result finally also formally elucidates the name of the scenario. The scenario is indeed called incoherent since given that its initial state is diagonal, i.e., incoherent, in the energy eigenbasis, it will stay so after any application of the scenario. We are now ready to move to our next remark.

\begin{lemma} \label{lemma:degenerate}
$\Lambda_{\text{inc}}$ can only cool the target if $U_{\text{inc}}$ acts in the degenerate subspaces of $H_{SM}$.
\end{lemma}

\begin{proof}
This proof is again straightforward. From Lemma \ref{lemma:directsum} we know that $U_{\text{inc}}$ only acts within the energy eigenspaces of $H_{SM}$, i.e., $U_{\text{inc}} = \Pi_{\epsilon} U_{\text{inc}}^{\epsilon}$, with

\begin{equation} \label{eq:oplusunitary}
U_{\text{inc}}^{\epsilon} = U_{\text{inc}}\big|_{\text{Eig}_{H_{SM}}(\epsilon)} \oplus \mathds{1} \big|_{ \text{Eig}^c_{H_{SM}}(\epsilon)}.
\end{equation}

We therefore have only left to show that if $\text{dim}(\text{Eig}_{H_{SM}}(\epsilon))=1$, then $U_{\text{inc}}^{\epsilon}$ does not affect the temperature of the target system. Let $\epsilon$ be an eigenvalue of $H_{SM}$ with $\text{dim}(\text{Eig}_{H_{SM}}(\epsilon))=1$. Let $\ket{v}_{SM} \in \text{Eig}_{H_{SM}}(\epsilon)$. As $\text{dim}(\text{Eig}_{H_{SM}}(\epsilon))=1$, there exists a $\lambda \in \mathbb{C}$ such that 
\begin{equation}
U_{\text{inc}}^{\epsilon} \ket{v}_{SM} = \lambda \ket{v}_{SM},
\end{equation}
meaning that $\ket{v}_{SM}$ is an eigenvector of $U_{\text{inc}}^{\epsilon}$ with eigenvalue $\lambda$. Since $U_{\text{inc}}^{\epsilon}$ is unitary, $\lambda = e^{i \theta}$ for some $ \theta \in [0,2 \pi)$ and so
\begin{equation}
U_{\text{inc}}^{\epsilon} \ket{v} \bra{v}_{SM} (U_{\text{inc}}^{\epsilon})^{\dagger}= \ket{v} \bra{v}_{SM}.
\end{equation}
Trivially
\begin{equation}
U_{\text{inc}}^{\epsilon} \ket{w} \bra{w}_{SM} (U_{\text{inc}}^{\epsilon})^{\dagger}= \ket{w} \bra{w}_{SM}, \quad \forall \ket{w}_{SM} \notin \text{Eig}_{H_{SM}}(\epsilon).
\end{equation}

 From Lemma \ref{lemma:staydiag} we can w.l.o.g. assume $ \Lambda_{\text{inc}}^{(k)}(\rho_S)$ to be diagonal in $(\ket{i}_S)_{i=0}^{d_S-1}$ for any $k \in \mathbb{N}$. Therefore $\Lambda_{\text{inc}}^{(k)}(\rho_S) \otimes \tau_M $ is diagonal in $ (\ket{i}_{SM})_{i=0}^{d_{SM}}$ and
 \begin{equation}
 U_{\text{inc}}^{\epsilon} \Lambda_{\text{inc}}^{(k)}(\rho_S) \otimes \tau_M (U_{\text{inc}}^{\epsilon})^{\dagger} = \Lambda_{\text{inc}}^{(k)}(\rho_S) \otimes \tau_M,
 \end{equation}
 showing that $U_{\text{inc}}^{\epsilon}$ cannot cool the target system, as desired.
\end{proof}

It is tempting to claim that: 

\say{Lemma \ref{lemma:degenerate} conceptually solves the incoherent scenario, or at least reduces it to its coherent counterpart. Indeed, all what one needs to do is identify the degenerate subspaces of $H_{SM}$ and then apply the coherent scenario within each of these subspaces.}

 While this is true in essence, it oversees two subtleties. The first one is that one does not exactly apply the coherent scenario within each subspace. The subspaces in which the unitaries are applied are indeed in general not of the form of a tensor product of the target system with some part of the machine. Moreover, the state of the system and machine within these subspaces is not necessarily thermal at $T_R$, part of the machine is indeed heated up to $T_H$ before the unitary is applied. This heating up of part of the machine is actually crucial. Without it the state of SM is proportional to the identity on those subspaces and hence no cooling can be performed. Indeed
 
 \begin{lemma} \label{lemma:roommachine}
 Let $\rho_{SM}^{R,H}=\tau_S \otimes \tau_M$. Then for every energy conserving unitary $U_{\text{inc}}$,
 \begin{equation}
 U_{\text{inc}} \; \rho_{SM}^{R,H} \; U_{\text{inc}}^{\dagger} = \rho_{SM}^{R,H}.
 \end{equation}
 \end{lemma}
 
 \begin{proof}
 The proof is a direct consequence of $[U_{\text{inc}},H_{SM}]=0$. Indeed as
 \begin{equation}
 \tau_S \otimes \tau_M = \frac{e^{-\beta_R H_{SM}}}{\Tr(e^{-\beta_R H_{SM}})},
 \end{equation}
 we immediately see that
 
\begin{align}
U_{\text{inc}} \tau_S \otimes \tau_M U_{\text{inc}}^{\dagger} &= U_{\text{inc}} \frac{e^{-\beta_R H_{SM}}}{\Tr(e^{-\beta_R H_{SM}})} U_{\text{inc}}^{\dagger}\\
&=\frac{e^{-\beta_R H_{SM}}}{\Tr(e^{-\beta_R H_{SM}})}\\
&= \tau_S \otimes \tau_M.
\end{align}
 \end{proof}
 
 Since our initial state is thermal, there is another use of the hot thermal bath that is generically useless for cooling, namely using it to thermalize the entire machine at $T_H$.
 
  \begin{lemma} \label{lemma:hotmachine}
 Let $\rho_{SM}^{R,H}=\tau_S \otimes \tau_M^H$. Then for every energy conserving unitary $U_{\text{inc}}$,
 \begin{equation}
 \Tr_M \left( U_{\text{inc}} \; \rho_{SM}^{R,H} \; U_{\text{inc}}^{\dagger}\right)
 \end{equation}
  cannot be colder than $ \tau_S $ according to our notion of temperature, i.e. according to our more ambitious view on cooling of Eq.~\ref{equ:ambitiouscool}.
 \end{lemma}
 
 \begin{proof}
 From Lemma \ref{lemma:directsum}, we only need to look at the action of $U_{\text{inc}}$ on $\text{Eig}_{H_{SM}}(\tilde{\epsilon}_i)$, $i=0,\dots,k_{SM}$. For $\text{dim} \left(\text{Eig}_{H_{SM}} (\tilde{\epsilon}_i) \right)=1$ we already know the statement to be true. Let therefore $l \in \{0, \dots, k_{SM}-1\}$ such that 
 
 \begin{equation}
 \text{dim} \left( \text{Eig}_{H_{SM}}(\tilde{\epsilon}_l) \right)=n_l>1.
 \end{equation}
 
 That is
 
 \begin{equation}
 \text{Eig}_{H_{SM}} (\tilde{\epsilon}_l) = \text{span} \left\{ \ket{p_0 q_0}_{SM}, \dots, \ket{p_{n_l-1} q_{n_l-1}}_{SM}  \right\}
 \end{equation}
 for some $p_0,p_1,\dots,p_{n_l-1} \in \{0, \dots, d_S-1\}$ and $q_0,q_1,\dots,q_{n_l-1} \in \{0, \dots, d_M-1\}$. Note that 
 
 \begin{equation} \label{eq:degen}
 \tilde{\epsilon}_l=E_{p_0}+\mathcal{E}_{q_0}=E_{p_1}+\mathcal{E}_{q_1}= \dots= E_{p_{n_l-1}}+\mathcal{E}_{q_{n_l-1}}.
 \end{equation}
 
 So with
 
 \begin{equation}
 \rho_{SM}^{R,H} \big|_{\text{Eig}_{H_{SM}} (\tilde{\epsilon}_l)}= \sum_{r=0}^{n_l-1} \frac{e^{- \beta_R E_{p_r}}}{Z_S} \frac{e^{- \beta_H \mathcal{E}_{q_r}}}{Z_M^H} \ket{p_r q_r} \bra{p_r q_r}_{SM},
 \end{equation}
 
 where $Z_S = \Tr(e^{-\beta_R H_S}), \, Z_M^H = \Tr(e^{-\beta_H H_M})$, we have that
 \begin{specialalign}
 &\bra{p_r q_r} \rho_{SM}^{R,H} \ket{p_r q_r}_{SM}  > \bra{p_s q_s} \rho_{SM}^{R,H} \ket{p_s q_s}_{SM}\\
 \Leftrightarrow  &\; e^{- \beta_R E_{p_r}} e^{- \beta_H \mathcal{E}_{q_r}} > e^{- \beta_R E_{p_s}} e^{- \beta_H \mathcal{E}_{q_s}}\\
\Leftrightarrow  &\; e^{- \beta_R (E_{p_r}-E_{p_s})} >  e^{- \beta_H (\mathcal{E}_{q_s}-\mathcal{E}_{q_r})}\\
\stackrel{\text{Eq.}~\ref{eq:degen}}{\Leftrightarrow}  &\; e^{- \beta_R (E_{p_r}-E_{p_s})} >  e^{- \beta_H (E_{p_r}-E_{p_s})}\\
 \Leftrightarrow & \; \underbrace{(\beta_H - \beta_R)}_{<0} (E_{p_r}-E_{p_s}) >0\\
 \Leftrightarrow  &\; E_{p_r} < E_{p_s}.
 \end{specialalign}
 
 This implies that
 \begin{equation} \label{eq:restrictlyordered}
 \mathfrak{D}^{\downarrow} \left( \rho_{SM}^{R,H} \big|_{\text{Eig}_{H_{SM}} (\tilde{\epsilon}_l)} \right) =  \mathfrak{D} \left( \rho_{SM}^{R,H} \big|_{\text{Eig}_{H_{SM}} (\tilde{\epsilon}_l)} \right).
 \end{equation}
 As $\rho_{SM}^{R,H} \big|_{\text{Eig}_{H_{SM}} (\tilde{\epsilon}_l)}$ is diagonal in the energy eigenbasis, also
 
 \begin{equation}
 \mathfrak{D} \left( \rho_{SM}^{R,H} \big|_{\text{Eig}_{H_{SM}} (\tilde{\epsilon}_l)} \right) = \lambda \left( \rho_{SM}^{R,H} \big|_{\text{Eig}_{H_{SM}} (\tilde{\epsilon}_l)}\right).
 \end{equation}
 
 Furthermore, from Corollary \ref{cor:Uincluded} of Schur's Theorem, for every unitary $U^{\tilde{\epsilon}_l}_{\text{inc}}$ acting on $\text{Eig}_{H_{SM}} (\tilde{\epsilon}_l)$ we have that
 
 \begin{equation} \label{eq:diagmaj}
 \mathfrak{D} \left( U^{\tilde{\epsilon}_l}_{\text{inc}} \, \rho_{SM}^{R,H} \big|_{\text{Eig}_{H_{SM}} (\tilde{\epsilon}_l)} \left(U^{\tilde{\epsilon}_l}_{\text{inc}}\right)^{\dagger} \right) \prec  \mathfrak{D} \left( \rho_{SM}^{R,H} \big|_{\text{Eig}_{H_{SM}} (\tilde{\epsilon}_l)} \right).
 \end{equation}
 
 So for every $k \in \{0, \dots, d_S-1\}$,
 
 \begin{align}
 \sum_{i=0}^{k} \left[ \Tr_M \left( U^{\tilde{\epsilon}_l}_{\text{inc}} \, \rho_{SM}^{R,H} \left(U^{\tilde{\epsilon}_l}_{\text{inc}}\right)^{\dagger} \right) \right]_{ii} &= \sum_{i=0}^{k} \left[ \Tr_M \left( U^{\tilde{\epsilon}_l}_{\text{inc}} \, \rho_{SM}^{R,H}\big|_{\text{Eig}_{H_{SM}} (\tilde{\epsilon}_l)} \left(U^{\tilde{\epsilon}_l}_{\text{inc}}\right)^{\dagger} \right. \right. \\
 & \hspace{2cm} \left. \left. \oplus \rho_{SM}^{R,H}\big|_{\text{Eig}^c_{H_{SM}} (\tilde{\epsilon}_l)} \right) \right]_{ii} \\
 &= \sum_{i=0}^{k} \underbrace{\left[ \Tr_M \left( U^{\tilde{\epsilon}_l}_{\text{inc}} \, \rho_{SM}^{R,H}\big|_{\text{Eig}_{H_{SM}} (\tilde{\epsilon}_l)} \left(U^{\tilde{\epsilon}_l}_{\text{inc}} \right)^{\dagger} \right) \right]_{ii}}_{(*)} \\
 & \hspace{2cm} + \left[ \Tr_M \left( \rho_{SM}^{R,H}\big|_{\text{Eig}^c_{H_{SM}} (\tilde{\epsilon}_l)} \right) \right]_{ii}.
 \end{align}
 
 For the first equality we used that only the diagonal blocks of $U_{\text{inc}}^{\tilde{\epsilon}_l} \rho_{SM}^{R,H} \left( U_{\text{inc}}^{\tilde{\epsilon}_l} \right)^{\dagger}$, according to the $\text{Eig}_{H_{SM}}(\tilde{\epsilon}_l) \oplus \text{Eig}^c_{H_{SM}}(\tilde{\epsilon}_l)$ division, contribute to the diagonal elements of the reduced state on $S$. Note that $(*)$ is different from zero only if $i=p_k$ for some $k=0, \dots, n_l-1$. Also
 
 \begin{align}
(*) &= \sum_{l=0}^{n_l-1} \left[  U^{\tilde{\epsilon}_l}_{\text{inc}} \, \rho_{SM}^{R,H}\big|_{\text{Eig}_{H_{SM}} (\tilde{\epsilon}_l)} \left(U^{\tilde{\epsilon}_l}_{\text{inc}} \right)^{\dagger} \right]_{iq_l,iq_l} \\
 &= \sum_{l=0}^{n_l-1} \left[ \mathfrak{D} \left( U^{\tilde{\epsilon}_l}_{\text{inc}} \, \rho_{SM}^{R,H}\big|_{\text{Eig}_{H_{SM}} (\tilde{\epsilon}_l)} \left(U^{\tilde{\epsilon}_l}_{\text{inc}} \right)^{\dagger} \right) \right]_{iq_l} \\
 &\stackrel{Eq.~\ref{eq:diagmaj}}{\leq} \sum_{l=0}^{n_l-1} \left[ \mathfrak{D}^{\downarrow} \left( \, \rho_{SM}^{R,H}\big|_{\text{Eig}_{H_{SM}} (\tilde{\epsilon}_l)} \right) \right]_{i q_l} \\
  &\stackrel{Eq.~\ref{eq:restrictlyordered}}{=}  \left[ \Tr_M \left( \, \rho_{SM}^{R,H}\big|_{\text{Eig}_{H_{SM}} (\tilde{\epsilon}_l)} \right) \right]_{ii}.
 \end{align}
 
 So we have that for every $k \in \{0, \dots, d_S-1\}$,
 
\begin{equation}
\sum_{i=0}^{k} \left[ \Tr_M \left( U^{\tilde{\epsilon}_l}_{\text{inc}} \, \rho_{SM}^{R,H} \left(U^{\tilde{\epsilon}_l}_{\text{inc}} \right)^{\dagger} \right) \right]_{ii} \leq \sum_{i=0}^{k} \left[ \Tr_M \left( \rho_{SM}^{R,H} \right) \right]_{ii},
\end{equation}

proving that the unitary $U^{\tilde{\epsilon}_l}_{\text{inc}}$ acting on  $\text{Eig}_{H_{SM}} (\tilde{\epsilon}_l)$ cannot cool the target according to our notion of cooling introduced in Section~\ref{sec:temperature}.
 \end{proof}
 
%
%
%

Given some degenerate subspaces, Lemma~\ref{lemma:roommachine} and Lemma~\ref{lemma:hotmachine} reduces the number of ways one can usefully make use of the hot thermal bath. For a machine with a lot of components this nevertheless leaves a lot or room for thermalizing different parts of it. The next natural question that comes in mind is what kind of degenerate subspaces are possible and interesting. This brings us to the second subtlety that Lemma~\ref{lemma:degenerate} leaves open. Indeed, what one is really interested in, is to figure out the interesting machines, that is the machines able to perform some cooling on the target system, for a given machine dimension $d_M$. Unfortunately the amount of possible degeneracies for $H_{SM}$ scales very unfavorably with $d_M$ such that the problem of identifying the relevant degeneracies for a given $d_M$ quickly becomes intractable. Rather than having to scan through all the possible degeneracies, it would therefore be very handy to have a way of identifying useful degeneracies. The following two results will allow us to systematically discard the two kinds of intuitively trivially useless degeneracies, namely those arising solely from $H_M$ or solely from $H_S$.

\begin{lemma} \label{lemma:machinedeg}
$\Lambda_{\text{inc}}$ cannot cool the target if $U_{\text{inc}}$ acts in degenerate subspaces consisting of degeneracies of the machine only.
\end{lemma}

\begin{proof}
 Let $k \in \{0,\dots, k_{SM}-1\}$ be such that $\text{Eig}_{H_{SM}} (\tilde{\epsilon}_k)$ is a degenerate subspace consisting of degeneracies of the machine only. That is

\begin{equation}
\text{Eig}_{H_{SM}} (\tilde{\epsilon}_k)= \text{span} \left\{ \ket{m n_0}_{SM}, \cdots, \ket{m n_{l-1}}_{SM} \right\}
\end{equation}
for some $m \in \{0, \dots,d_S-1\}$, $l \in \mathds{N}, l \leq d_M$, and some pairwise different $n_0, \cdots, n_{l-1} \in \{0, \dots, d_M-1\}$. Let $U_{\text{inc}}^{\tilde{\epsilon}_k}$ be the part of $U_{\text{inc}}$ acting solely on $\text{Eig}_{H_{SM}} (\tilde{\epsilon}_k)$, i.e., 

\begin{equation}
U_{\text{inc}}^{\tilde{\epsilon}_k}= U \oplus \mathds{1} \big|_{\text{Eig}^c_{H_{SM}} (\tilde{\epsilon}_k)},
\end{equation}

with $U=\sum_{p,q=0}^{l-1} u_{pq} \ket{m n_p} \bra{m n_q}_{SM}$ being some unitary on $\text{Eig}_{H_{SM}} (\tilde{\epsilon}_k)$. We will keep the state of our machine and system general, i.e.,

\begin{equation}
\sigma_{SM} = \sum_{j,l'=0}^{d_M-1} \sum_{i,k=0}^{d_S-1} [\sigma]_{ij, kl'} \ket{i j} \bra{k l'}_{SM}.
\end{equation}

Then

\begin{align}
\Tr_M \left( U_{\text{inc}}^{\tilde{\epsilon}_k} \sigma_{SM} \left(U_{\text{inc}}^{\tilde{\epsilon}_k}\right)^{\dagger} \right) &= \Tr_M \left( U \sigma_{SM} \big|_{\text{Eig}_{H_{SM}} (\tilde{\epsilon}_k)} U^{\dagger} \oplus \sigma_{SM} \big|_{\text{Eig}^c_{H_{SM}} (\tilde{\epsilon}_k)}\right)\\
&= \Tr_M \left( U \sigma_{SM} \big|_{\text{Eig}_{H_{SM}} (\tilde{\epsilon}_k)} U^{\dagger} \right) + \Tr_M \left( \sigma_{SM} \big|_{\text{Eig}^c_{H_{SM}} (\tilde{\epsilon}_k)}\right)
\end{align}

and 

\begin{align}
\Tr_M \left( U \sigma_{SM} \big|_{\text{Eig}_{H_{SM}} (\tilde{\epsilon}_k)} U^{\dagger} \right) &= \Tr_M \left( \sum_{p,q}^{l-1} \sum_{r,s=0}^{l-1} u_{pr} [\sigma_{SM}]_{m n_r, m n_s} \bar{u}_{qs} \ket{m n_p} \bra{m n_q}_{SM}\right)\\
&=   \sum_{r,s=0}^{l-1} \underbrace{\left( \sum_{p}^{l-1} u_{pr} \bar{u}_{ps} \right)}_{=\delta_{rs}}  [\sigma_{SM}]_{m n_r, m n_s}  \ket{m} \bra{m}_S\\\
&= \Tr_M \left( \sigma_{SM} \big|_{\text{Eig}_{H_{SM}} (\tilde{\epsilon}_k)}\right).
\end{align}

So, as desired, 

\begin{equation}
\Tr_M \left( U_{\text{inc}}^{\tilde{\epsilon}_k} \sigma_{SM} \left(U_{\text{inc}}^{\tilde{\epsilon}_k}\right)^{\dagger} \right)= \Tr_M \left( \sigma_{SM}\right).
\end{equation}

\end{proof}

\begin{lemma} \label{lemma:systdeg}
$\Lambda_{\text{inc}}$ cannot cool the target if $U_{\text{inc}}$ acts in degenerate subspaces consisting of degeneracies of the target system only.
\end{lemma}

\begin{proof}
This is almost a tautology. Indeed, let $k \in \{0,\dots, k_{SM}-1\}$ be such that $\text{Eig}_{H_{SM}} (\tilde{\epsilon}_k)$ is a degenerate subspace consisting of degeneracies of the system only. That is,

\begin{equation}
\text{Eig}_{H_{SM}} (\tilde{\epsilon}_k)= \text{span} \left\{ \ket{n_0 m}_{SM}, \cdots, \ket{n_{l-1} m}_{SM} \right\}
\end{equation}
for some $m \in \{0, \dots,d_M-1\}$, $l \in \mathds{N}, l \leq d_S$, and some pairwise different $n_0, \cdots, n_{l-1} \in \{0, \dots, d_S-1\}$. Per definition all the eigenstates $ \ket{n_0}_S, \dots, \ket{n_{l-1}}_S$ have the same energy and so what impacts the temperature of our target system is the sum of the populations of these eigenstates. However, all what our unitary does by acting on $\text{Eig}_{H_{SM}} (\tilde{\epsilon}_k)$ is move populations around these eigenstates, thereby conserving their sum. Indeed, let $U_{\text{inc}}^{\tilde{\epsilon}_k}$ be the part of $U_{\text{inc}}$ acting solely on $\text{Eig}_{H_{SM}} (\tilde{\epsilon}_k)$, i.e., 

\begin{equation}
U_{\text{inc}}^{\tilde{\epsilon}_k}= U \oplus \mathds{1} \big|_{\text{Eig}^c_{H_{SM}} (\tilde{\epsilon}_k)},
\end{equation}

with $U=\sum_{p,q=0}^{l-1} u_{pq} \ket{n_p m} \bra{n_q m}_{SM}$ being some unitary on $\text{Eig}_{H_{SM}} (\tilde{\epsilon}_k)$. Then

\begin{align}
\sum_{k=0}^{l-1} \left[ U_{\text{inc}} \sigma_{SM} U_{\text{inc}}^{\dagger} \right]_{n_k m, n_k m} &= \Tr \left( U \sigma_{SM} \big|_{\text{Eig}_{H_{SM}} (\tilde{\epsilon}_k)} U^{\dagger}\right)\\
&=\sum_{k=0}^{l-1} [\sigma_{SM}]_{n_k m, n_k m}.
\end{align}

One can nevertheless order $ \ket{n_0}_S, \dots, \ket{n_{l-1}}_S$ (according to some criteria of our choice) and diagonalize $ \sigma_{SM} \big|_{\text{Eig}_{H_{SM}} (\tilde{\epsilon}_k)} $ such that its greatest eigenvalue contributes to the first eigenstate, its next greatest eigenvalue to the second eigenstate and so on. This would result in ``cooling'' the system the most.
\end{proof}

Lemma~\ref{lemma:machinedeg} and Lemma~\ref{lemma:systdeg} confirm the intuition that in order to cool, the degenerate subspace should be somehow spanned across the machine and the system. They restrict the number of possible useful machines of a given dimension $d_M$ and are as such great tools to systematically discard some machines. They however still leave room for a lot of potential interesting machines.

\chapter{Qubit System } \label{chap:qubitsyst}

We would like to start our analysis by studying the one qubit target system case in detail. This has two main motivations. We would first of all like to follow a bottom-up approach. This will enable us to gain some valuable insights into the problem by studying simpler cases. It will also provide some practical examples that we will later be able to refer to. Second of all, the definition of temperature for qubit systems is much easier and less debated. Indeed, every qubit state, as soon as diagonal in the energy eigenbasis, can be assigned a unique temperature T. This is due to the simple but crucial fact that qubits only have two real eigenvalues summing up to one, and therefore only one free (real) parameter. From Lemma~\ref{lemma:staydiag} we know that every target state is diagonal in the energy eigenbasis in the incoherent scenario. In the coherent scenario, although one can in principle generate any coherence in the target state, we will see that all the operations that we will encounter will leave the target system diagonal in the energy eigenbasis. For the sake of making a choice and for consistency throughout the thesis, we will work with the ground state population as a notion of temperature in the target qubit case.

\section{Notation} \label{sec:notationqubit}

We would here like to state some specific notation that we will use in this chapter. We would like to remind the reader that since this chapter is intended for reference, it might simply be skipped upon a linear reading.

 As already  mentioned, we will be working a lot with the ground state population. We will generically reserve the letter r to denote it. To extract the temperature of the qubit target system from its grounds state population we will use the inverse of the following formula:

\begin{equation}
r= \frac{1}{1+e^{-\frac{1}{T} E_S}};
\end{equation}

that is,

\begin{equation}
T(r)= \frac{E_S}{\text{ln} \left(\frac{r}{1-r} \right)}.
\end{equation}

The initial ground state population on the target system will be denoted by $r_S$ and that of part $M_i$ of the machine by $r_{M_i}$.  $r_{M_i}^H$ will denote the ground state population of $M_i$ heated up to $T_H$. We will denote the ground state population of the target after an application of the coherent scenario by $r_{\text{coh}}$ and its associated temperature by $T_{\text{coh}}$, i.e.,

\begin{align}
r_{\text{coh}}&= \bra{0} \Lambda_{\text{coh}} (\rho_S) \ket{0}_S,\\
T_{\text{coh}}&= \frac{E_S}{\text{ln} \left(\frac{r_{\text{coh}}}{1-r_{\text{coh}}} \right)}. \label{eq:Tcoh}
\end{align}

In the incoherent scenario, the ground state of the target after a single cycle depends on the temperature of the hot bath $T_H$. To stress this we will denote it by $r_{\text{inc}} (T_H)$, that is

\begin{align}
r_{\text{inc}} (T_H) &= \bra{0} \Lambda_{\text{inc}}(\rho_S) \ket{0}_S\\
&= \bra{0} \Tr_M (U_{\text{inc}} \rho_S \otimes \tau_M^H U_{\text{inc}}^{\dagger}) \ket{0}_S.
\end{align}

Its associated temperature will be denoted by $T_{\text{inc}}(T_H)$, i.e.,
\begin{equation}
T_{\text{inc}}(T_H)= \frac{E_S}{\text{ln} \left(\frac{r_{\text{inc}}(T_H)}{1-r_{\text{inc}}(T_H)} \right)}.
\end{equation}
 We remind the reader that we denote the free energy change associated to a single application of the incoherent scenario by $\Delta F_{\text{inc}} (T_H)$ and that associated to a single application of the coherent scenario by $\Delta F_{\text{coh}}$. Furthermore,

\begin{align}
\Delta F_{\text{inc}} (T_H)& = Q^H (1-\frac{T_R}{T_H}),\\
\Delta F_{\text{coh}} &= \Tr \left( ( U \rho_{SM} U^{\dagger} - \rho_{SM}) H_{SM} \right),
\end{align}

where $Q^H$ denotes the heat drawn from the hot bath and

\begin{equation}
Q^H= \Tr \left( (\tau_{M_H}^H-\tau_{M_H}) H_{M_H}\right).
\end{equation}

In each scenario we can tune some parameters upon application of the corresponding maps that will alter the cooling of the target system. In the incoherent scenario, the tunable parameter is $T_H$, whereas in the coherent scenario it is the unitary $U$ that one chooses to apply. The quantities associated to these optimal application of each scenario set bounds on achievable performances and are as such of special interest. To distinguish them, we will tag them with a * subscript. An optimal application of the incoherent scenario corresponds to $T_H \rightarrow \infty$, and so,

\begin{align}
r_{\text{inc}}^* &= \lim_{T_H \rightarrow + \infty} r_{\text{inc}}(T_H),\\
T_{\text{inc}}^* &= \lim_{T_H \rightarrow + \infty} T_{\text{inc}}(T_H),\\
\Delta F_{\text{inc}}^* &= \lim_{T_H \rightarrow + \infty} \Delta F_{\text{inc}}(T_H).
\end{align}

Note in particular that since $U_{\text{inc}}$ is implemented at no cost, we will always w.l.o.g. choose a $U_{\text{inc}}$ that cools the target system maximally.

An optimal application of the coherent scenario corresponds to applying an optimal unitary $U^*$. We stress here that $U^*$ solves two optimization problems at once. It cools the system  maximally and does so at a minimal work cost. We therefore have

\begin{align}
r_{\text{coh}}^* &= \bra{0} \Tr_M \left( U^* \rho_S \otimes \tau_M (U^*)^{\dagger} \right) \ket{0}_S,\\
T_{\text{coh}}^* &= \frac{E_0}{\text{ln}\left( \frac{r_{\text{coh}}^*}{1-r_{\text{coht}}^*}\right)},\\
\Delta F_{\text{coh}}^* &= \Tr \left( (U^* \rho_{SM} (U^*)^{\dagger}-\rho_{SM}) H_{SM} \right).
\end{align}

  We will also consider repeated applications of each scenario. For the incoherent scenario we have after $n \in \mathds{N} \cup \{\infty\}$ cycles
  
  \begin{align}
  r_{\text{inc},n}(T_H) &= \bra{0} \Lambda_{\text{inc}}^n(\rho_S) \ket{0}_S,\\
  T_{\text{inc},n}(T_H) &= \frac{E_S}{\text{ln} \left( \frac{r_{\text{inc},n}(T_H)}{1-r_{\text{inc},n}} \right)},
  \Delta F_{\text{inc},n}(T_H) = Q_n^H \left( 1-\frac{T_R}{T_H} \right),
  \end{align}

where $Q_n^H$ is the total heat drawn from the hot bath after n steps, i.e.,
\begin{equation}
Q_n^H= Q^H+ \sum_{i=2}^n \Tr \left( (\tau_{M_H}^H-\sigma_{M_H}^i) H_{M_H} \right),
\end{equation}

with
\begin{equation}
\sigma_{M_H}^i= \Tr_{S M_H} \left( U_{\text{inc}} \Lambda_{\text{inc}}^{i-2}(\rho_S) \otimes \rho_{M}^{R,H} (U_{\text{inc}})^{\dagger} \right).
\end{equation}

$\sigma_{M_H}^i$ is the state of the part of the machine $M_H$ at the end of step $i-1$, i.e., at the beginning of step $i$ before it is brought in contact with the hot bath to be rethermalized to $T_H$. Note that if we had made use of the room temperature bath to thermalize the entire machine and then brought $M_H$ back to $\tau_{M_H}^H$, assuming that $M_H$ stays the same at each step, we would have gotten $Q_n^{H} = n Q^H$, which diverges for $n \rightarrow + \infty$. 

For the coherent scenario we have

\begin{align}
r_{\text{coh},n} &=\bra{0} \Lambda_{\text{coh}}^n(\rho_S) \ket{0}_S,\\
T_{\text{coh},n} &= \frac{E_S}{\text{ln} \left( \frac{r_{\text{coh},n}}{1-r_{\text{coh},n}}\right)},\\
\Delta F_{\text{coh},n}&= \Delta F_{\text{coh}} + \sum_{i=2}^n \Tr \left( (U \Lambda_{\text{coh}}^{i-1}(\rho_S)\otimes \rho_{SM} U^{\dagger} -\Lambda_{\text{coh}}^{i-1}(\rho_S)\otimes \rho_{SM}) H_{SM}\right).
\end{align}

As for the single application of each scenario, a * superscript mean $T_H \rightarrow \infty$ in the incoherent scenario and a choice of maximally cooling unitary $U^*$ at each step in the coherent scenario.

\section{One Qubit Machine} \label{sec:onequbitmachine}

Following our bottom-up approach, we naturally start by investigating the one qubit machine in both scenarios. We start by proving an impossibility-to-cool result in the incoherent scenario with this type of machine before moving on to analyzing what the coherent scenario allows for.

\subsection{Incoherent one Qubit Machine}

In the case of the single qubit machine, one has two options of making use of the hot thermal bath at $T_H$. One can either not make use of the hot bath at all and leave the machine untouched or thermalize the entire machine to $T_H$. From Lemma~\ref{lemma:roommachine} and Lemma~\ref{lemma:hotmachine} we know that both options do not allow cooling and directly have the desired result. We can here, however, also quite simply directly prove the desired result. Indeed, if one does not make use of the hot thermal bath, the SM state commutes with the unitary and nothing happens. Thermalizing the entire machine, we have the state 

\begin{equation}
\rho_{SM}^H=\rho_S \otimes \tau_M^H
\end{equation}
at hand. The question is then, which machine can cool the target system. Answering this amounts to choosing $\mathcal{E}_M$, the gap of the machine qubit, adequately. For 

\begin{equation}
H_{SM} = E_S \ket{1} \bra{1}_S \otimes \mathds{1}_M + \mathds{1}_S \otimes \mathcal{E}_M \ket{1} \bra{1}_M
\end{equation}

to have some degeneracy, one of the two following conditions must be satisfied:

\begin{enumerate}[i)]
\item $E_S=0$ or $\mathcal{E}_M=0$,
\item $\mathcal{E}_M=E_S$.
\end{enumerate}

In i) the state $\rho_{SM}^{R,H}$ is proportional to the identity within the degenerate subspace and so no cooling is possible. In ii) the degenerate subspace is $\text{span}\{\ket{0 1}_{SM}, \ket{10}_{SM}\}$ but as

\begin{equation}
\bra{0 1} \rho_{SM}^{R,H} \ket{0 1}_{SM} > \bra{1 0} \rho_{SM}^H \ket{1 0}_{SM},
\end{equation}
the unitaries acting on this subspace can only heat up the target system. We have therefore just proven the desired result.

\begin{lemma}
It is impossible to cool a target qubit with a single qubit machine in the incoherent scenario.
\end{lemma}

\subsection{Coherent one Qubit Machine} \label{subsec:coherent1qubitmachine}

We are here interested in finding the state that cools to a given temperature at a minimal free energy change for the case of the one qubit machine $ H_M = \mathcal{E}_M \ket{1} \bra{1}_M$, and one qubit target system, $ H_S= E_S \ket{1} \bra{1}_S$. We are hence interested in solving the following instance of the general optimization problem of Eq.~\ref{eq:deltaFoptimization} presented in Section~\ref{sec:remarkscoh}.

\begin{equation} \label{eq:deltaFoptimization1-1}
\min_{v \prec \mathfrak{D}(\rho_{SM})} v \cdot \mathfrak{D}(H_{SM}), \quad \text{s.t. } v_0+v_1=c,
\end{equation}

with

\begin{align}
v&=(v_0,v_1,v_2,v_3) \in \mathds{R}^4,\\
\mathfrak{D}(H_{SM}) &=(0, \mathcal{E}_M, E_S, E_S + \mathcal{E}_M),\\
\mathfrak{D}(\rho_{SM}) &= \frac{1}{\Tr(e^{-\beta_R H_{SM}}} \left(1, e^{- \beta_R \mathcal{E}_M}, e^{- \beta_R E_S}, e^{- \beta_R (E_S+\mathcal{E}_M)} \right).
\end{align}

Furthermore, $c \in [r, r_{\text{coh}}^*]$, with
\begin{equation}
r= \left[ \mathfrak{D} (\rho_{SM}) \right]_0 + \left[ \mathfrak{D} (\rho_{SM})\right]_1 = \frac{1}{1+ e^{- \beta_R E_S}}
\end{equation}

and
\begin{equation}
r_{\text{coh}}^*= \left[ \mathfrak{D}^{\downarrow} (\rho_{SM}) \right]_0 + \left[ \mathfrak{D}^{\downarrow} (\rho_{SM})\right]_1.
\end{equation}

$r_{\text{coh}}^*$ is therefore the sum of the two largest entries of $\mathfrak{D}(\rho_{SM})$. One can hence cool the target if and only if $\left[ \mathfrak{D} (\rho_{SM}) \right]_0$ and $\left[ \mathfrak{D} (\rho_{SM}) \right]_1$ are not already the biggest entries of $\mathfrak{D}(\rho_{SM})$. One readily sees that $\left[ \mathfrak{D} (\rho_{SM}) \right]_0= \frac{1}{\Tr(e^{-\beta_R H_{SM}})}$ is the biggest entry of $\mathfrak{D}(\rho_{SM})$ and that $\left[ \mathfrak{D} (\rho_{SM}) \right]_3= \frac{ e^{- \beta_R (E_S+\mathcal{E}_M)}}{\Tr(e^{-\beta_R H_{SM}}}$ is the smallest. We furthermore have

\begin{equation}
\frac{\left[\mathfrak{D} (\rho_{SM}) \right]_1}{\left[\mathfrak{D} (\rho_{SM}) \right]_2}= e^{\beta_R (E_S- \mathcal{E}_M)} <1
\end{equation}
if and only if $E_S < \mathcal{E}_M$. This means that we can cool the target exactly when $E_S < \mathcal{E}_M$. In that regime

\begin{equation}
r_{\text{coh}}^* = \left[ \mathfrak{D}^{\downarrow} (\rho_{SM}) \right]_0 + \left[ \mathfrak{D}^{\downarrow} (\rho_{SM})\right]_1= r_M
\end{equation}
and the temperature associated to $r_{\text{coh}}^*$ is

\begin{equation} \label{eq:Tstarcoh1-1}
T_{\text{coh}}^*= \frac{E_S}{\ln \left(\frac{r_{\text{coh}}^*}{1-r_{\text{coh}}^*} \right)} = \frac{E_S}{E_M} T_R.
\end{equation}

A transformation achieving this cooling is the unitary $U^*$ swapping the energy eigenstates $\ket{0 1}_{SM}$ and $\ket{1 0}_{SM}$. This unitary operates at the minimal work cost since 

\begin{equation}
a_{\mathfrak{D}(U^* \rho_{SM} (U^*)^{\dagger})} = \left( [\mathfrak{D} (\rho_{SM})]_0, [\mathfrak{D} (\rho_{SM})]_2 \right)
\end{equation}

is inversely ordered, i.e., passive~\cite{Pusz-1978, Lenard-1978}, with respect to $ a_{\mathfrak{D}(H_{SM})} = \left( 0, \mathcal{E}_M \right)$ and 

\begin{equation}
b_{\mathfrak{D}(U^* \rho_{SM} (U^*)^{\dagger})} = \left( [\mathfrak{D} (\rho_{SM})]_1, [\mathfrak{D} (\rho_{SM})]_3 \right)
\end{equation}

is inversely ordered with respect to $ b_{\mathfrak{D}(H_{SM})} = \left( E_S, E_S+\mathcal{E}_M \right)$. The work cost associated to $U^*$ is 

\begin{equation}
\Delta F_{\text{coh}}^* = (r_{\text{coh}}^*-r_S) (\mathcal{E}_M - E_S).
\end{equation}

This fully solves the endpoint problem and is a straightforward exemplification of the general solution we exposed in Section~\ref{sec:remarkscoh}. We next turn our attention to the case $r_{\text{coh}}=c < r_{\text{coh}}^*$ of the optimization problem of Eq.~\ref{eq:deltaFoptimization1-1}. That is, we are interested in finding out how much one can cool the target with finite resources. In that case we find the following.

\begin{theorem} \label{thm:optsolution1-1}
Let $r_{\text{coh}} \in [r_S, r_{\text{coh}}^*]$. Let $\mu = \frac{r_{\text{coh}}-r_S}{r_M-r_S}$. Let $t = \arcsin(\sqrt{\mu})$. Let
\begin{equation}
\mathcal{L}= i \ket{0 1} \bra{1 0}_{SM} - i \ket{10} \bra{0 1}_{SM}.
\end{equation}
Then
\begin{equation}
v= \mathfrak{D} (U \rho_{SM} U^{\dagger}), \quad \text{with } U=e^{- i t \mathcal{L}}
\end{equation}
minimizes the optimization problem of Eq.~\ref{eq:deltaFoptimization1-1} for $c=r_{\text{coh}}$ and has an associated work cost of
\begin{equation}
\Delta F_{\text{coh}} = (r_{\text{coh}}-r_S) (\mathcal{E}_M-E_S).
\end{equation}
\end{theorem}

\begin{proof}
We want to solve 
\begin{equation}
\min_{v \prec \mathfrak{D}(\rho_{SM})} v \cdot \mathfrak{D}(H_{SM}), \quad \text{s.t. } v_0+v_1=c
\end{equation}
for $v \in \mathbb{R}^4$, $\mathfrak{D}(\rho_{SM})=(0, \mathcal{E}_M, E_S, \mathcal{E}_M+E_S)$ and $c=r_{\text{coh}}$. We here derive the core idea of the proof. The idea is based on rewriting
\begin{equation}
v \cdot \mathfrak{D}(H_{SM})
\end{equation}
such that the majorization conditions as well as the constraint can naturally be expressed. This rewriting is the following.
\begin{align}
v \cdot \mathfrak{D}(H_{SM}) &= v_1 \mathcal{E}_{M} + v_2 E_S + v_3 (E_S + \mathcal{E}_M)\\
&= \left[ (v_0+v_1) - v_0 + v_3 \right] \mathcal{E}_M + \left[ 1-(v_0+v_1) \right] E_S.
\end{align}

Now using $\min (A+B) \geq \min A + \min B$  in our rewriting of $ v \cdot \mathfrak{D}(H_{SM})$, we get
\begin{align}
\min_{\stackrel{v \prec \mathfrak{D}(\rho_{SM}}{\text{s.t. } v_0 + v_1=c}}  v \cdot \mathfrak{D}(H_{SM}) &\geq \min_{\stackrel{v \prec \mathfrak{D}(\rho_{SM}}{\text{  s.t. } v_0 + v_1=c}} [(v_0+v_1)-v_0+v_3] \mathcal{E}_M   \\
&+ \min_{\stackrel{v \prec \mathfrak{D}(\rho_{SM}}{\text{  s.t. } v_0 + v_1=c }} [1-(v_0+v_1)] E_S  \\
&\geq \left(c- [\mathfrak{D}(\rho_{SM})]_0 + [\mathfrak{D}(\rho_{SM})]_3 \right) \mathcal{E}_M + (1-c) E_S\\
&= v(c) \cdot \mathfrak{D}(H_{SM}),
\end{align}
where
\begin{align}
v(c)=& \left( [\mathfrak{D}(\rho_{SM})]_0, (1-\mu(c)) [\mathfrak{D}(\rho_{SM})]_1+ \mu(c) [\mathfrak{D}(\rho_{SM})]_2, \right.\\
 &\left. (1-\mu(c)) [\mathfrak{D}(\rho_{SM})]_2 + \mu(c) [\mathfrak{D}(\rho_{SM})]_1, [\mathfrak{D}(\rho_{SM})]_3 \right),
\end{align}
with
\begin{equation}
\mu(c)= \frac{c-[\mathfrak{D}(\rho_{SM})]_0 - [\mathfrak{D}(\rho_{SM})]_1}{[\mathfrak{D}(\rho_{SM})]_2 - [\mathfrak{D}(\rho_{SM})]_1}=\frac{c-r_S}{r_M-r_S}.
\end{equation}
The first minimization follows from $v_0+v_1=c$ and 
\begin{align}
v_3 &\geq v_3^{\downarrow} \geq [\mathfrak{D}^{\downarrow}(\rho_{SM})]_3=[\mathfrak{D}(\rho_{SM})]_3,\\
v_0 &\leq v_0^{\downarrow} \leq [\mathfrak{D}^{\downarrow}(\rho_{SM})]_0=[\mathfrak{D}(\rho_{SM})]_0.
\end{align}
The second minimization just uses $v_0+v_1=c$. As $v(c)$ satisfies the constraint $[v(c)]_0+[v(c)]_1=c$, $v(r_{\text{coh}})$ is the solution of our problem. For the rewriting of $v(r_{\text{coh}})$ as in the statement as well as the expression of $\Delta F_{\text{coh}}$ we refer to Appendix C of \cite{Clivaz-2019bis}.
\end{proof}

Note that $\mathcal{L}$ is the Hamiltonian generating the swapping between the energy levels $\ket{0 1}_{SM}$ and $\ket{1 0}_{SM}$. Also note that once one has cooled the target to $r_M$, further applications of the scenario do not enhance cooling here. Indeed, one can only cool further if

\begin{equation}
\left[ \mathfrak{D} (\Lambda_{\text{coh}}^* (\rho_S) \otimes \tau_M)\right]_1 < \left[ \mathfrak{D} (\Lambda_{\text{coh}}^* (\rho_S) \otimes \tau_M)\right]_2.
\end{equation}
But we have that

\begin{align}
\left[ \mathfrak{D} (\Lambda_{\text{coh}}^* (\rho_S) \otimes \tau_M)\right]_1 &= \frac{ e^{-\beta_{\text{coh}}^* E_S}}{\left(1+e^{-\beta_{\text{coh}}^* E_S} \right) \left(1+e^{-\beta_{R} \mathcal{E}_M} \right)}\\
&=  \frac{ e^{-\beta_{R} \mathcal{E}_M}}{\left(1+e^{-\beta_{\text{coh}}^* E_S} \right) \left(1+e^{-\beta_{R} \mathcal{E}_M} \right)}\\
&= \left[ \mathfrak{D} (\Lambda_{\text{coh}}^* (\rho_S) \otimes \tau_M)\right]_2,
\end{align} 

where we used that from Eq.~\ref{eq:Tstarcoh1-1}

\begin{equation}
\beta_{\text{coh}}^* E_S = \beta_R \mathcal{E}_M.
\end{equation}

With this we have fully solved the qubit machine qubit target case of the coherent scenario and are ready to move to the second most simple instance of our scenario, namely the 2 qubit machine 1 qubit target one.

\section{Two Qubit Machine} \label{sec:twoqubitmachine}

We now move our attention to the two qubit machine case. We will denote the qubits of the machine by $M_1$ and $M_2$ and their associated energy gaps by $\mathcal{E}_{M_1}$ and $\mathcal{E}_{M_2}$ respectively. The Hamiltonian of the machine is therefore given by
\begin{equation}
H_M = \mathcal{E}_{M_1} \ket{1} \bra{1}_{M_1} \otimes \mathds{1}_{M_2} + \mathds{1}_{M_1} \otimes \mathcal{E}_{M_2} \ket{1} \bra{1}_{M_2}.
\end{equation}
Our target system is as before a qubit of energy gap $E_S$ and Hamiltonian $H_S= E_S \ket{1} \bra{1}_S$. The joint initial state is therefore given by

\begin{equation}
\rho_{SM} = \tau_S \otimes \tau_{M_1} \otimes \tau_{M_2}.
\end{equation}

We will start by investigating the single cycle regime in each scenario, i.e., a single application of $\Lambda_{\text{inc}}$ and $\Lambda_{\text{coh}}$ respectively. We will then move on to studying the repeated cycle regime as well as the asymptotic regime, i.e., when the number of repetitions $n \rightarrow \infty$.

\subsection{Single Cycle} \label{subsubsec:single}

\subsubsection{Incoherent Scenario}

Given a machine dimension $d_M$ and a system $S$, the main challenge of the incoherent scenario is to identify which specific machines allow cooling system $S$. From Section~\ref{sec:remarksincoh} we know that the machine has to have a tensor product structure to be able to cool $S$. Indeed, if not, the only way it can use the hot thermal bath is to either stay thermal at $T_R$ or to entirely thermalize to $T_H$. But we know from Lemma~\ref{lemma:hotmachine} and Lemma~\ref{lemma:roommachine} that this doesn't lead to any cooling of $S$, whatever the degeneracies of $H_{SM}$. For machines of dimension 4, i.e., $d_M=4$, this means that the useful machines have to consist of 2 qubits. Furthermore, the only way of making good use of the hot thermal bath is to leave one qubit at $T_R$ and thermalize the other to $T_H$. W.l.o.g. let $M_1$ remain at $T_R$ and $M_2$ be the one thermalizing to $T_H$. We therefore have

\begin{equation}
\rho_{SM}^{R_,H}= \tau_S \otimes \tau_{M_1} \otimes \tau_{M_2}^H.
\end{equation}

The next step is to screen through the potential degeneracies of $H_{SM}$ to see which ones enable to cool $S$. From Lemma~\ref{lemma:machinedeg} we know that degenerate subspaces consisting of degeneracies of the machine only are of no use. This gets rid of the degeneracies $\mathcal{E}_{M_1}=0, \mathcal{E}_{M_2}=0, \text{and } \mathcal{E}_{M_1}+\mathcal{E}_{M_1}=0$ as potential candidates. Similarly, Lemma~\ref{lemma:systdeg} tells us that $E_S=0$ is also not a useful degeneracy to have in $H_{SM}$. This leaves us with 5 potential degeneracies

\begin{enumerate}
\item $E_S= \mathcal{E}_{M_1}$
\item $E_S = \mathcal{E}_{M_2}$
\item $E_S = \mathcal{E}_{M_1} + \mathcal{E}_{M_2}$
\item $ \mathcal{E}_{M_1} = E_S + \mathcal{E}_{M_2}$ \label{enum:coolingdeg}
\item $\mathcal{E}_{M_2} = E_S + \mathcal{E}_{M_1}$.
\end{enumerate} 

One can show explicitly, see Appendix A of \cite{Clivaz-2019bis}, that none of these degeneracies lead to cooling except number \ref{enum:coolingdeg}
\begin{equation}
\mathcal{E}_{M_1} = E_S + \mathcal{E}_{M_2}.
\end{equation}

This seems to fully take care of the degeneracy question for $d_M=4$. There is one small subtlety, however. Degenerate subspaces can be defined by more than just one degeneracy condition. To fully conclude the proof that one cannot cool incoherently with machines of dimension 4 unless
\begin{equation}
\mathcal{E}_{M_1} = E_S + \mathcal{E}_{M_2},
\end{equation}
we therefore have to make sure that combined subspaces cannot get activated, i.e., cool, although they individually were useless. For the 2 qubit machine this can be done explicitly. For the details we refer to Appendix A of \cite{Clivaz-2019bis}. We therefore have the following result

\begin{theorem}
Given a qubit target system and a machine of dimension 4, the incoherent scenario can cool the target if and only if the following holds
\begin{enumerate}
\item The machine comprises two qubits $M_1$ and $M_2$
\item $\mathcal{E}_{M_1} = E_S + \mathcal{E}_{M_2}$
\item Only $M_2$ is brought into contact with the hot thermal bath.
\end{enumerate}
\end{theorem}

Once the machine, as well as which part to heat up, has been identified, the cooling part of the incoherent scenario is fairly straightforward. Indeed, given the hot temperature $T_H$, one heats up $M_2$ to $T_H$ and then performs the unitary cooling maximally within the degenerate subspace. Note that the unitary operation is in particular performed at no cost. In the case of the 2 qubit machine this unitary is given by

\begin{equation}
U = \ket{010} \bra{101}_{SM} + \ket{101} \bra{010}_{SM} + \mathds{1}_{\text{span}^c\{\ket{010}, \ket{101}\}},
\end{equation}
where $\text{span}^c\{\ket{010}, \ket{101}\}$ denotes the complement of the set $\text{span}\{\ket{010}, \ket{101}\}$. Doing so we get
\begin{equation}
r_{\text{inc}}(T_H)= r_S r_{M_1} + \left[ (1-r_S) r_{M_1} + r_S (1-r_{M_1}) \right] (1-r_{M_2}^H).
\end{equation}

The associated temperature $T_{\text{inc}}(T_H)$ is given by the usual formula,
\begin{equation} \label{eq:Tinc2qubits}
T_{\text{inc}}(T_H)=\frac{E_S}{\text{ln} \left( \frac{r_{\text{inc}}(T_H)}{1-r_{\text{inc}}(T_H)}\right)}.
\end{equation}

 The work cost is calculated from the hot bath. According to Eq.~\ref{eq:generaldeltaFinc} this amounts to calculating the heat drawn from the bath, $Q^H$, which is equal to the change of energy of $M_2$ upon heating it. We therefore get
 \begin{equation}
 Q^H= \mathcal{E}_{M_2} (r_{M_2} - r_{M_2}^H)
 \end{equation}
 
 and
 \begin{equation} \label{eq:Finc2qubits}
 \Delta F_{\text{inc}} (T_H) = \mathcal{E}_{M_2} (r_{M_2} - r_{M_2}^H) \left( 1- \frac{T_R}{T_H} \right).
 \end{equation}
 
 From this we get the * quantities by taking the limit $T_H \rightarrow \infty$, i.e.,
 
 \begin{align}
 r_{\text{inc}}^* &= \frac{1}{2} (r+r_{M_1})\\
 T_{\text{inc}}^* &= \frac{E_S}{\ln \left(\frac{r_S+ r_{M_1}}{2-(r_S+r_{M_2})} \right)}\\
 \Delta F_{\text{inc}}^* &= \mathcal{E}_{M_2} (r_{M_2} - \frac{1}{2} ).
 \end{align}

\subsubsection{Coherent Scenario}

We now turn our attention to the coherent scenario. Since we are mostly interested in seeing how it compares to the incoherent scenario, we will impose here the restriction

\begin{equation} \label{eq:restrictioncoh2qubit}
\mathcal{E}_{M_1}= E_S + \mathcal{E}_{M_2},
\end{equation}
although this is not needed here for the machine to cool. This will also simplify the analysis, as will become clearer in the following. As in Section~\ref{subsec:coherent1qubitmachine}, we are interested in solving a special instance of the general problem of Eq.~\ref{eq:deltaFoptimization}. More precisely, we are interested in solving

\begin{equation} \label{eq:cohproblem2qubit}
\min_{v \prec \mathfrak{D}(\rho_{SM})} v \cdot \mathfrak{D}(H_{SM}), \quad \text{s.t. } \sum_{i=0}^3 v_i = c,
\end{equation}

with

\begin{align}
v &= (v_0, \dots, v_7) \in \mathbb{R}^8\\
\mathfrak{D} (H_{SM}) &= (0, \mathcal{E}_{M_2}, \mathcal{E}_{M_1}, \mathcal{E}_{M_1} + \mathcal{E}_{M_2}, E_S, E_S + \mathcal{E}_{M_2}, E_S+ \mathcal{E}_{M_1}, E_S + \mathcal{E}_{M_1} + \mathcal{E}_{M_2})\\
\mathfrak{D}(\rho_{SM})&= \frac{1}{\Tr \left( e^{-\beta_R H_{SM}} \right)} \left(1, e^{- \beta_R \mathcal{E}_{M_2}}, e^{- \beta_R \mathcal{E}_{M_1}},  e^{- \beta_R (\mathcal{E}_{M_1}+\mathcal{E}_{M_2})}, e^{- \beta_R E_S},\right. \\
&\quad   \left. e^{- \beta_R (E_S+\mathcal{E}_{M_2})}, e^{- \beta_R (E_S+\mathcal{E}_{M_1})}, e^{- \beta_R (E_S+\mathcal{E}_{M_1}+\mathcal{E}_{M_1})} \right).
\end{align}

Furthermore, $c \in [r, r_{\text{coh}}^*]$ with 
\begin{equation}
r= \sum_{i=0}^3 \left[ \mathfrak{D}(\rho_{SM}) \right]_i= \frac{1}{1+ e^{-\beta_R E_S}}
\end{equation}

and

\begin{equation}
r_{\text{coh}}^*= \sum_{i=0}^3 \left[ \mathfrak{D}^{\downarrow} (\rho_{SM}) \right]_i.
\end{equation}

We first would like to determine $r_{\text{coh}}^*$ to know what range $c$ is allowed to evolve in, and also to know what the maximal cooling on the target is. For this we need to determine what the 4 biggest entries of $\mathfrak{D}(\rho_{SM})$ are. Looking at this more closely, and taking our restriction of Eq.~\ref{eq:restrictioncoh2qubit} into account, we find the following ordering

\begin{align}
&\left[ \mathfrak{D} (\rho_{SM}) \right]_0 > \left\{\left[ \mathfrak{D} (\rho_{SM}) \right]_1 , \left[ \mathfrak{D} (\rho_{SM}) \right]_4 \right\} > \left[ \mathfrak{D} (\rho_{SM}) \right]_2 = \left[ \mathfrak{D} (\rho_{SM}) \right]_5 > \\
& >\left\{ \left[ \mathfrak{D} (\rho_{SM}) \right]_3, \left[ \mathfrak{D} (\rho_{SM}) \right]_6 \right\} > \left[ \mathfrak{D} (\rho_{SM}) \right]_7,
\end{align}

where $\{ , \}$ means that the ordering depends on whether $\mathcal{E}_{M_2} < E_S$ or $\mathcal{E}_{M_2} > E_S$. More precisely

\begin{equation}
\{x,y\} = \begin{cases}
x < y &, \text{if } \mathcal{E}_{M_2} < E_S\\
x > y & , \text{if } \mathcal{E}_{M_2} > E_S\\
x=y & , \text{if } \mathcal{E}_{M_2} =E_S.
\end{cases}
\end{equation}

In both cases

\begin{align}
r_{\text{coh}}^* &= \left[ \mathfrak{D} (\rho_{SM}) \right]_0 + \left[ \mathfrak{D} (\rho_{SM}) \right]_1 + \left[ \mathfrak{D} (\rho_{SM}) \right]_4 +\left[ \mathfrak{D} (\rho_{SM}) \right]_5\\
&= r_{M_1},
\end{align}

which gives an associated temperature of

\begin{equation}
T_{\text{coh}}^* = \frac{E_S}{\mathcal{E}_{M_1}} T_R.
\end{equation}

A transformation achieving this cooling is the unitary $U$ swapping the energy eigenstates $\ket{011}_{SM_1 M_2}$ and $\ket{010}_{SM_1 M_2}$, i.e.,
\begin{equation}
U= \ket{011} \bra{010}_{SM_1M_2} + \ket{010} \bra{011}_{SM_1M_2} + \mathds{1}_{\text{span}^c \{ \ket{011}_{SM_1M_S}, \ket{010}_{SM_1M_2}\}}.
\end{equation}

This is however not the unitary operating at the minimal work cost since the state that one gets after performing $U$ is not passive in the respective subspaces of Lemma~\ref{lemma:split}. We will now work out the energetically most efficient transformation, which depends on the ordering of $ \mathfrak{D}(\rho_{SM})$. As the ordering depends on whether $\mathcal{E}_{M_2} \leq E_S$ or $\mathcal{E}_{M_2} > E_S$, we treat both cases separately.

\paragraph{{\bf a.} $\mathcal{E}_{M_2} \leq E_S$}

In this case we have 
\begin{equation} \label{eq:orderingSmallEs}
\begin{aligned}
&\left[ \mathfrak{D} (\rho_{SM}) \right]_0 > \left[ \mathfrak{D} (\rho_{SM}) \right]_1 , \left[ \mathfrak{D} (\rho_{SM}) \right]_4  > \left[ \mathfrak{D} (\rho_{SM}) \right]_5 = \left[ \mathfrak{D} (\rho_{SM}) \right]_2 > \\
& > \left[ \mathfrak{D} (\rho_{SM}) \right]_3, \left[ \mathfrak{D} (\rho_{SM}) \right]_6  > \left[ \mathfrak{D} (\rho_{SM}) \right]_7,
\end{aligned}
\end{equation}

and so the vector $v_{\leq} \in \mathbb{R}^8$ minimizing Eq.~\ref{eq:cohproblem2qubit} for $c=r_{\text{coh}}^*$ is given by

\begin{align}
v&=\left(\left[ \mathfrak{D} (\rho_{SM}) \right]_0 , \left[ \mathfrak{D} (\rho_{SM}) \right]_1 , \left[ \mathfrak{D} (\rho_{SM}) \right]_4  , \left[ \mathfrak{D} (\rho_{SM}) \right]_5 , \left[ \mathfrak{D} (\rho_{SM}) \right]_2 , \right.\\
& , \left. \left[ \mathfrak{D} (\rho_{SM}) \right]_3, \left[ \mathfrak{D} (\rho_{SM}) \right]_6  , \left[ \mathfrak{D} (\rho_{SM}) \right]_7 \right).
\end{align}

One checks that, indeed, the first half of $v_{\leq}$,
 \begin{equation}
 a_{v_{\leq}}= \left( \left[ \mathfrak{D} (\rho_{SM}) \right]_0 , \left[ \mathfrak{D} (\rho_{SM}) \right]_1 , \left[ \mathfrak{D} (\rho_{SM}) \right]_4  , \left[ \mathfrak{D} (\rho_{SM}) \right]_5  \right),
 \end{equation}
  is inversely ordered with respect the first half of $\mathfrak{D}(H_{SM})$,

\begin{equation}
a_{\mathfrak{D}(H_{SM})}= \left(0, \mathcal{E}_{M_2}, \mathcal{E}_{M_1}, \mathcal{E}_{M_1}+ \mathcal{E}_{M_2} \right),
\end{equation}

and that the second half of $v_{\leq}$, 

 \begin{equation}
 b_{v_{\leq}}= \left( \left[ \mathfrak{D} (\rho_{SM}) \right]_2 , \left[ \mathfrak{D} (\rho_{SM}) \right]_3 , \left[ \mathfrak{D} (\rho_{SM}) \right]_6  , \left[ \mathfrak{D} (\rho_{SM}) \right]_7  \right),
 \end{equation}
  is also inversely ordered with respect the second half of $\mathfrak{D}(H_{SM})$,

\begin{equation}
b_{\mathfrak{D}(H_{SM})}= \left(E_S, E_S+\mathcal{E}_{M_2}, E_S+\mathcal{E}_{M_1}, E_S+\mathcal{E}_{M_1}+ \mathcal{E}_{M_2} \right).
\end{equation}

A unitary $U_{\leq}^*$ that achieves $v_{\leq}$ is given by

\begin{equation}
U_{\leq}^* = U_{24} U_{35},
\end{equation}

with
\begin{equation}
U_{ij} = \ket{i_2} \bra{j_2} + \ket{j_2} \bra{i_2} + \mathds{1}_{\text{span}^c\{ \ket{i_2}, \ket{j_2} \}},
\end{equation}
where $i_2$ is the 3 digit display of $i= 0,\dots,7$ in base 2, e.g., $0_2=000$. This corresponds to swapping the energy eigenstates $\ket{010}_{SM}$ with $\ket{100}_{SM}$ as well as $\ket{011}_{SM}$ with $\ket{101}_{SM}$. Physically it corresponds to swapping the population of $S$ with $M_1$ and $U_{\leq}^*$ can be written compactly as

\begin{equation}
U_{\leq}^*= e^{-i \frac{\pi}{2} \mathcal{L}_{SM_1}},
\end{equation}
where
\begin{equation}
\mathcal{L}_{SM_k}= i \ket{01} \bra{10}_{S M_k} - i \ket{10} \bra{01}_{S M_k}.
\end{equation}

The work cost associated to $U_{\leq}^*$ is
\begin{align}
\Delta F_{\text{coh}}^* &= (\mathcal{E}_{M_1}- E_S) (r_{M_1} - r_S)\\
&= \mathcal{E}_{M_2} (r_{M_1}- r_S).
\end{align}

\paragraph{{\bf b.} $\mathcal{E}_{M_2} > E_S$}

In this case we have 

\begin{align}
&\left[ \mathfrak{D} (\rho_{SM}) \right]_0 > \left[ \mathfrak{D} (\rho_{SM}) \right]_4 , \left[ \mathfrak{D} (\rho_{SM}) \right]_1  > \left[ \mathfrak{D} (\rho_{SM}) \right]_5 = \left[ \mathfrak{D} (\rho_{SM}) \right]_2 > \\
& > \left[ \mathfrak{D} (\rho_{SM}) \right]_6, \left[ \mathfrak{D} (\rho_{SM}) \right]_3  > \left[ \mathfrak{D} (\rho_{SM}) \right]_7,
\end{align}

and so the vector $v_{>} \in \mathbb{R}^8$ minimizing Eq.~\ref{eq:cohproblem2qubit} for $c=r_{\text{coh}}^*$ is given by

\begin{align}
v_{>}&=\left(\left[ \mathfrak{D} (\rho_{SM}) \right]_0 , \left[ \mathfrak{D} (\rho_{SM}) \right]_4 , \left[ \mathfrak{D} (\rho_{SM}) \right]_1  , \left[ \mathfrak{D} (\rho_{SM}) \right]_5 , \left[ \mathfrak{D} (\rho_{SM}) \right]_2 , \right.\\
& , \left. \left[ \mathfrak{D} (\rho_{SM}) \right]_6, \left[ \mathfrak{D} (\rho_{SM}) \right]_3  , \left[ \mathfrak{D} (\rho_{SM}) \right]_7 \right).
\end{align}

As before, $a_{v_{>}}$ and $b_{v_{>}}$ are indeed inversely ordered w.r.t. $a_{\mathfrak{D}(H_{SM})}$ and $b_{\mathfrak{D}(H_{SM})}$ respectively. A unitary $U_{\geq}^*$ that achieves $v_{>}$ is here given by

\begin{equation}
U_{>}^* = U_{24} U_{35} U_{14} U_{36}.
\end{equation}

 Physically, $U_{14}$ and $U_{36}$ together correspond to swapping the populations of $S$ and $M_2$. Once this is done, $U_{24} U_{35}$ is applied, which corresponds to swapping the population of $S$ with $M_1$. One can write $U_{>}^*$ compactly as

\begin{equation}
U_{>}^*= e^{-i \frac{\pi}{2} \mathcal{L}_{SM_1}} e^{-i \frac{\pi}{2} \mathcal{L}_{SM_2}}.
\end{equation}

The work cost associated to $U_{>}^*$ is
\begin{align}
\Delta F_{\text{coh}, >}^* &= (\mathcal{E}_{M_2}- E_S) (r_{M_2} - r_S)+ \mathcal{E}_{M_2} (r_{M_1} - r_{M_2}).
\end{align}

All the analysis of this section until now is a straight exemplification of Lemma~\ref{lemma:split}. It allowed us to nevertheless gather some important intuition about the general case of $r_{\text{coh}} \in [r_S, r_{\text{coh}}^*]$. We also saw that in the case $\mathcal{E}_{M_2} > E_S$, one could adopt a better strategy than swapping the populations of $S$ and $M_1$ directly by swapping the populations of $S$ with that of $M_2$ first. The achieved temperature on the target is the same but the work cost differs between both strategies. \\

Next we turn our attention to the case, where $r_{\text{coh}} \in [r_S, r_{\text{coh}}^*]$. Here we again distinguish between the case $\mathcal{E}_{M_2} \leq E_S$ and $\mathcal{E}_{M_2} > E_S$. For $\mathcal{E}_{M_2} \leq E_S$ we find a result that is a complete analog of the 1 qubit machine result of Section~\ref{subsec:coherent1qubitmachine}.

\begin{theorem} \label{thm:optsolution1-2small}
Let $\mathcal{E}_{M_2} \leq E_S$. Let $r_{\text{coh}} \in [r_S, r_{\text{coh}}^*]$. Let $\mu = \frac{r_{\text{coh}}-r_S}{r_{M_1}-r_S}$. Let $t = \arcsin(\sqrt{\mu})$. 
Then
\begin{equation}
v= \mathfrak{D} (U \rho_{SM} U^{\dagger}), \quad \text{with } U=e^{- i t \mathcal{L}_{S M_1}}
\end{equation}
minimizes the optimization problem of Eq.~\ref{eq:cohproblem2qubit} for $c=r_{\text{coh}}$ and has an associated work cost of
\begin{equation} \label{eq:Fcoh2qubitssmall}
\Delta F_{\text{coh}} = (r_{\text{coh}}-r_S) (\mathcal{E}_{M_1}-E_S).
\end{equation}
\end{theorem}

\begin{proof}[Proof idea]

The idea of the proof is exactly the same as that of Theorem~\ref{thm:optsolution1-1}, namely to rewrite
\begin{equation}
v \cdot \mathfrak{D}(H_{SM})
\end{equation}
such that the majorization conditions as well as the constraint can naturally be expressed. The practical rewriting depends on the ordering of $\mathfrak{D}(\rho_{SM})$ as well as how $E_S, \mathcal{E}_{M_1}, \text{ and } \mathcal{E}_{M_2}$ relate to one another. In the case $\mathcal{E}_{M_1}=E_S + \mathcal{E}_{M_2}$ and $\mathcal{E}_{M_2} \leq E_S$ the ordering of $\mathfrak{D}(\rho_{SM})$ is fixed, and given by Eq.~\ref{eq:orderingSmallEs}, and the useful rewriting is the following
\begin{equation} \label{eq:usefulvHSmall}
\begin{aligned}
v \cdot \mathfrak{D}(H_{SM}) &= \left( \sum_{i=0}^3 v_i-v_0 \right) \mathcal{E}_{M_2} +[1-(v_0+v_1+v_4+v_2)] \mathcal{E}_{M_2}+(1-\sum_{i=0}^3 v_i) E_S\\
& + [\sum_{i=0}^3 v_i-(v_0+v_1)] (\mathcal{E}_{M_1}-\mathcal{E}_{M_2})+ v_7 \mathcal{E}_{M_2} + (v_6+v_7) (\mathcal{E}_{M_1}-\mathcal{E}_{M_2}).
\end{aligned}
\end{equation}

Using again that the minimum of the sum is greater than the sum of the minima we get
\begin{align}
\min_{\stackrel{v \prec \mathfrak{D}(\rho_{SM})}{\sum_{i=0}^3 v_i =c}} v \cdot \mathfrak{D}(H_{SM}) &\geq v(c) \cdot \mathfrak{D}(\rho_{SM}),
\end{align}
where
\begin{equation}
v(c)= T_{24}(\mu(c)) T_{35}(\mu(c)) \mathfrak{D}(\rho_{SM}),
\end{equation}
with
\begin{align}
T_{ij}(\lambda) &= (1-\lambda) \mathds{1} + \lambda Q_{ij}\\
Q_{ij}&: \text{ permutation matrix exchanging coordinates $i$ and $j$ only},\\
\mu(c)&= \frac{c-r_S}{r_{M_1}-r_S}.
\end{align}

As $\sum_{i=0}^3 [v(c)]_i =c$, $v(r_{\text{coh}})$ is the solution of our problem. For further details, the rewriting of $v(r_{\text{coh}})$ as in the statement, as well as the expression of $\Delta F_{\text{coh}}$ we refer to Appendix C of \cite{Clivaz-2019bis}.
\end{proof}

The solution of Theorem~\ref{thm:optsolution1-2small} corresponds to partially performing the swapping of $S$ and $M_1$ instead of performing the full swap. In the case of $\mathcal{E}_{M_2} > E_S$ we also find a very intuitive result, namely

\begin{theorem} \label{thm:optsolution1-2big}
Let $\mathcal{E}_{M_2} > E_S$. Let $r_{\text{coh}} \in [r_S, r_{\text{coh}}^*]$. Let 

\begin{equation}
\mu = \begin{cases}
 \frac{r_{\text{coh}}-r_S}{2( r_{M_2}-r_S)} &, \text{if } r_S \leq r_{\text{coh}} \leq r_{M_2},\\
 \frac{r_{\text{coh}}-r_{M_2}}{2( r_{M_1}-r_{M_2})}+\frac{1}{2} &, \text{if } r_{M_2} < r_{\text{coh}} \leq r_{M_1}.
 \end{cases}
\end{equation}

 Let $f(\mu) = \arcsin(\sqrt{\min\{2 \mu,1\}})$. Let $g(\mu) = \arcsin(\sqrt{\max\{2 \mu-1,0\}})$.
Then
\begin{equation}
v= \mathfrak{D} (U \rho_{SM} U^{\dagger}), \quad \text{with } U=e^{- i g(\mu) \mathcal{L}_{S M_1}} e^{- i f(\mu) \mathcal{L}_{S M_2}},
\end{equation}
minimizes the optimization problem of Eq.~\ref{eq:cohproblem2qubit} for $c=r_{\text{coh}}$ and has an associated work cost of
\begin{equation}
\Delta F_{\text{coh}} =\begin{cases} 
(r_{\text{coh}}-r_S) (\mathcal{E}_{M_2}-E_S) &, \text{if } r \leq r_{\text{coh}} < r_{M_2}\\
(r_{M_2}-r_S) (\mathcal{E}_{M_2}-E_S)+(r_{\text{coh}}-r_{M_2}) (\mathcal{E}_{M_1}-E_S) &, \text{if } r_{M_2} \leq r_{\text{coh}} < r_{M_1}.
\end{cases}
\end{equation}
\end{theorem}

\begin{proof}[Proof idea]

 In this case we have two practical rewritings of $v \cdot \mathfrak{D}(H_{SM})$ depending on the value of $r_{\text{coh}}$. The first one is practical for $r_{\text{coh}} \in [r_S, r_{M_2}]$ and is the following
 \begin{equation} \label{eq:usefulbeginBig}
\begin{aligned}
v \cdot \mathfrak{D}(H_{SM}) =& -v_0 \mathcal{E}_{M_2} + (-v_0-v_1-v_4) (\mathcal{E}_{M_1}-\mathcal{E}_{M_2}) + v_7 \mathcal{E}_{M_2}\\
&+ (v_6+v_3+v_7) (\mathcal{E}_{M_1}-\mathcal{E}_{M_2}) + (v_5+v_6+v_3+v_7) (2 \mathcal{E}_{M_2}-\mathcal{E}_{M_1}) \\
& + \mathcal{E}_{M_1}+ E_S -\mathcal{E}_{M_2}+ \sum_{i=0}^3 v_i (\mathcal{E}_{M_2}-E_S).
\end{aligned}
\end{equation}

This leads to
\begin{align}
\min_{\stackrel{v \prec \mathfrak{D}(\rho_{SM})}{\sum_{i=0}^3 v_i =c}} v \cdot \mathfrak{D}(H_{SM}) &\geq v_1(c) \cdot \mathfrak{D}(\rho_{SM}),
\end{align}
where
\begin{equation}
v_1(c)= T_{14}(\mu_1(c)) T_{36}(\mu_1(c)) \mathfrak{D}(\rho_{SM}),
\end{equation}
with
\begin{align}
\mu_1(c)= \frac{c-r_S}{r_{M_2}-r_S}.
\end{align}

As $\sum_{i=0}^3 [v_1(c)]_i =c$ for $c \in [r_S , r_{M_2}]$, $v_1(r_{\text{coh}})$ is the solution of our problem for $r_{\text{coh}} \in [r_S, r_{M_2}]$. For $r_{\text{coh}} \in [r_{M_2}, r_{M_1}]$ we use the following rewriting
\begin{align}
v \cdot \mathfrak{D}(H_{SM}) =& -v_0 \mathcal{E}_{M_2} + (-v_0-v_1) (\mathcal{E}_{M_1}-\mathcal{E}_{M_2}) + (v_6+v_7) (\mathcal{E}_{M_1}-\mathcal{E}_{M_2})  \\
&+ v_7 \mathcal{E}_{M_2}+ (v_3+v_5+v_6+v_7) \mathcal{E}_{M_2}  +  E_S + \sum_{i=0}^3 v_i (\mathcal{E}_{M_1}-E_S).
\end{align}

This leads to
\begin{align}
\min_{\stackrel{v \prec \mathfrak{D}(\rho_{SM})}{\sum_{i=0}^3 v_i =c}} v \cdot \mathfrak{D}(H_{SM}) &\geq v_2(c) \cdot \mathfrak{D}(\rho_{SM}),
\end{align}
where
\begin{equation}
v_2(c)= T_{24}(\mu_2(c)) T_{35}(\mu_2(c)) T_{14}(1) T_{36}(1) \mathfrak{D}(\rho_{SM}),
\end{equation}
with
\begin{align}
\mu_2(c)= \frac{c-r_{M_2}}{r_{M_1}-r_{M_2}}.
\end{align}

As $\sum_{i=0}^3 [v_2(c)]_i =c$ for $c \in [r_{M_2} , r_{M_1}]$, $v_2(r_{\text{coh}})$ is the solution of our problem for $r_{\text{coh}} \in [r_{M_2}, r_{M_1}]$. For further details we refer to Appendix C of \cite{Clivaz-2019bis}.
\end{proof}

What Theorem~\ref{thm:optsolution1-2big} says is that for $r_{\text{coh}} \in [r_S, r_{M_2}]$, the best strategy is to partially swap the populations of $S$ and $M_2$. Indeed, in that case $\mu \in [0,\frac{1}{2}]$ and so $g(\mu) = 0$ and $f(\mu)= \arcsin(\sqrt{2\mu})$, so that
\begin{equation}
U= e^{-i \arcsin(\sqrt{2 \mu}) \mathcal{L}_{SM_2}}.
\end{equation}

If $r_{\text{coh}} \in (r_{M_2}, r_{M_1}]$, Theorem~\ref{thm:optsolution1-2big} tells us that the best strategy is to fully swap the populations of $S$ and $M_2$ and then subsequently partially swap the populations of $S$ and $M_1$. Indeed, in that case $\mu \in (\frac{1}{2}, 1]$ so that $f(\mu)= \frac{\pi}{2}$ and $g(\mu)= \arcsin (\sqrt{2 \mu -1})$. We therefore have
\begin{equation}
U= e^{-i \arcsin(\sqrt{2 \mu-1}) \mathcal{L}_{S M_1}} e^{- i\frac{\pi}{2} \mathcal{L}_{SM_2}}.
\end{equation}

 Note that it is a priori not evident at all, at least to us, that out of all the $ 8 \times 8 $ unitaries that one could potentially choose, that the ones performing the partial swaps of the end result are the most efficient in terms of energy expenditure. The proof has to deal with all these possibilities at first and makes sure that no other transformation performs better. To do so it uses a lot of the fine-tuned information about the problem, i.e., the specific ordering of $\mathfrak{D}(\rho_{SM})$ as well as the relationship of the different energy levels of the joint system $SM$. Given the end results and the fact that they are so intuitive suggests, however, that this fine-tuned information might not, after all, be that crucial. It would be interesting to find a proof that is based on a more general principle. This might allow generalizing the result to more complicated machines more easily as well.

\subsubsection{Coherent vs. Incoherent}

Now that we have analyzed both scenarios in detail for the two qubit machine cooling a qubit target, we would like to see how they compare both in terms of cooling performance and in terms of cooling for a given amount of injected work cost. This comparison is summarized in Figure~\ref{fig:compasingle}, where the amount of cooling vs. the associated work cost is mapped out. The incoherent curve is generated by plotting $T_{\text{inc}}(T_H)$, Eq.~\ref{eq:Tinc2qubits}, and $\Delta F_{\text{inc}}(T_H)$, Eq.~\ref{eq:Finc2qubits}, parametrically in the hot bath temperature $T_H \in [T_R, +\infty]$. The coherent curve is generated by plotting $T_{\text{coh}}(r_{\text{coh}})$, Eq.~\ref{eq:Tcoh}, and $\Delta F_{\text{coh}}(r_{\text{coh}})$, Eq.~\ref{eq:Fcoh2qubitssmall}, parametrically in the target ground state population $r_{\text{coh}} \in [r_S,r_{\text{coh}}^*]$.

For the Figure we selected $\mathcal{E}_{M_2} < E_S$. But apart from the coherent curve having a discontinuity in the first derivative at $r_{\text{coh}}=r_{M_2}$, the behavior of the curves for $\mathcal{E}_{M_2} \geq E_S$ are qualitatively the same.

\begin{figure}[t]
	\centering
	\includegraphics[width=8.5cm]{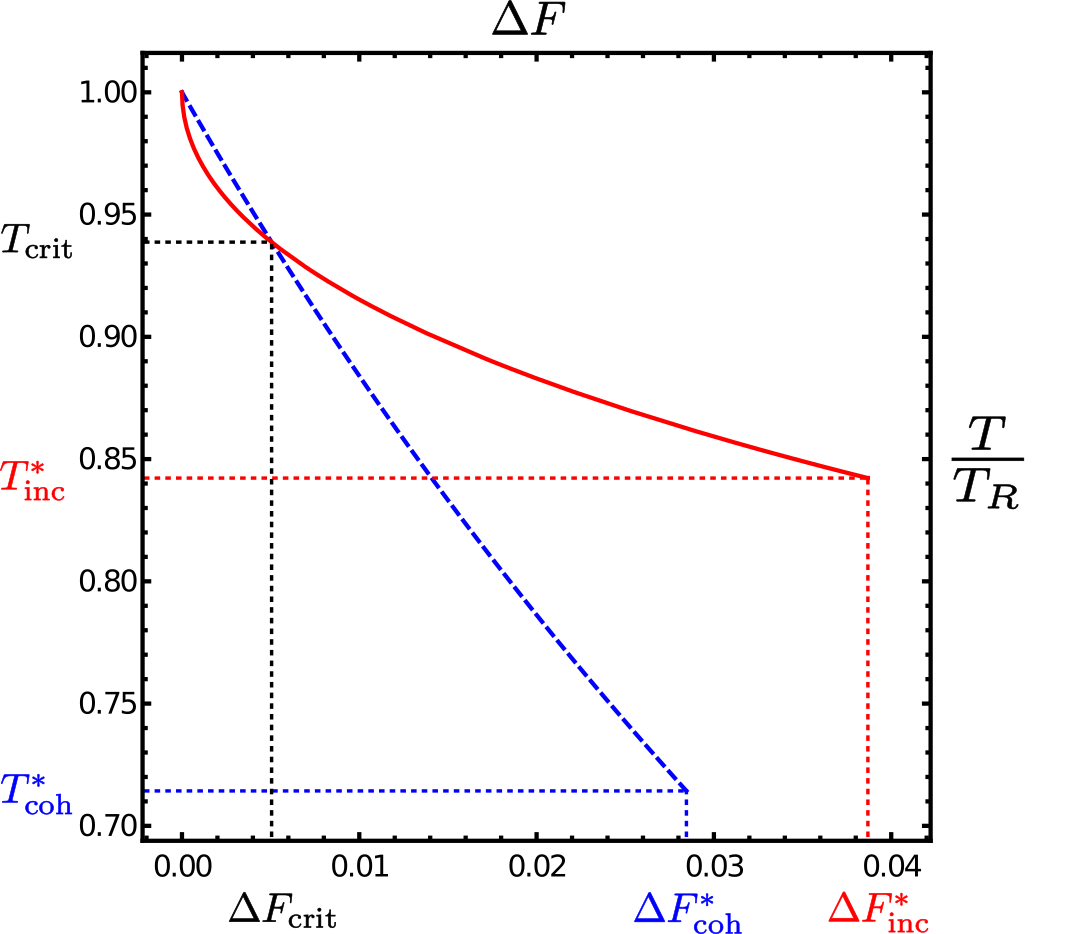}
\caption{Parametric plot of the relative temperature of the target qubit \(\frac{T}{T_R}\) as a function of its work cost $\Delta F$ for $\mathcal{E}_{M_2}=0.4$ and \(T_R=1\). The red solid curve corresponds to the incoherent scenario, the blue dashed, to the coherent scenario. When the cooling is maximal (i.e., the work cost is unrestricted), the coherent scenario always outperforms the incoherent one, $T_{\text{coh}}^{*} < T_{\text{inc}}^{*}$ and $\Delta F_{\text{coh}}^* < \Delta F_{\text{inc}}^*$. However, below a critical work cost $\Delta F_\text{crit}$, the incoherent scenario always outperforms the coherent one.\label{fig:compasingle}.}
\end{figure}

There are a few interesting observations that can be made from this comparison. First of all, looking at the end points of the curve, we see that one can coherently cool more than incoherently and that one does so at a lower work cost. This is true in general as one can easily analytically prove that
\begin{align}
T_{\text{coh}}^* &< T_{\text{inc}}^*,\\
\Delta F_{\text{coh}}^* &< \Delta F_{\text{inc}}^*,
\end{align}
always holds, see~\cite{Clivaz-2019bis}. The coherent scenario therefore always performs better than the incoherent scenario for maximal cooling in the single cycle regime. However, the coherent scenario is not universally superior to the incoherent one. Indeed, for sufficiently low work cost, the incoherent scenario always outperforms the coherent one. This is suggested by Figure~\ref{fig:compasingle} and can be easily proven in general by computing the initial slope of each curve. The incoherent scenario starts with an infinite slope while the coherent one with a finite one. The incoherent curve therefore always lies below the coherent one at the beginning. The incoherent and coherent curves must eventually cross since the endpoint of the coherent curve lies below that of the incoherent one. There therefore exists a critical work cost $\Delta F_{\text{crit}}$ such that if one only has at its disposal some work cost $\Delta F$ to invest that is smaller than $\Delta F_{\text{crit}}$, one can cool more incoherently than coherently.

\subsection{Repeated \& Asymptotic Cycles} \label{subsec:rep}

We here would like to go beyond the single cycle regime of the 2 qubit machine 1 qubit target joint system. Studying the repeated applications for each scenario is in this case still very much tractable due to the simplicity of the systems at hand. The main analytical challenges have indeed already been dealt within the previous sections. The choice of the machine that cools in the incoherent scenario has been made. In the coherent scenario, we will see that after the first cycle of cooling is completed, the optimization problem is easily solved due to the restricted subspace that is still available for cooling.

\subsubsection{Repeated Incoherent Scenario}

After a cycle is completed, applying another $\Lambda_{\text{inc}}$ to the target amounts to rethermalizing  $M_1$ to $T_R$ and $M_2$ to $T_H$ before repeating the energy conserving unitary operation
\begin{equation}
U_{\text{inc}}= \ket{010} \bra{101}_{SM} + \ket{101} \bra{010}_{SM} + \mathds{1}_{\text{span}^c\{\ket{010}_{SM}, \ket{101}_{SM}\}}.
\end{equation}

As before, the effect of $U_{\text{inc}}$ to the ground state population of the target is to remove the population of $\ket{010}_{SM}$ and add that of $\ket{101}_{SM}$. If the ground state population of the target before the application of $U_{\text{inc}}$ is $r$, then after applying $U_{\text{inc}}$ it becomes
\begin{equation} \label{eq:degpopchange}
r'=r-r (1-r_{M_1}) r_{M_2}^H+ (1-r) r_{M_1} (1-r_{M_2}^H).
\end{equation}

One can massage the right-hand side of Eq.~\ref{eq:degpopchange} and rewrite it as
\begin{equation}
r'=r_V + (1-N_V) (r-r_V),
\end{equation}
where $N_V$ is the sum of the populations of $\ket{01}_M$ and $\ket{10}_M$, i.e.,
\begin{equation}
N_V= r_{M_1} (1-r_{M_2}^H) + (1-r_{M_1}) r_{M_2}^H,
\end{equation}
and $r_V$ is the population of $\ket{01}_M$ renormalized by $N_V$, i.e.,
\begin{equation}
r_V= \frac{r_{M_1} (1-r_{M_2}^H)}{N_V}=\frac{1}{1+ e^{-\beta_R \mathcal{E}_{M_1}} e^{\beta_H \mathcal{E}_{M_2}}}.
\end{equation}

The rewriting may at first seem arbitrary but it proves to be very convenient to calculate the ground state population of the target after $n$ applications of the incoherent scenario. Indeed,

\begin{align}
r_{\text{inc},n}-r_V&=(1-N_V) (r_{\text{inc},n-1}-r_V)\\
r_{\text{inc},n-1}-r_V&=(1-N_V) (r_{\text{inc},n-2}-r_V)\\
&\vdots\\
r_{\text{inc}}-r_V&=(1-N_V) (r_{S}-r_V),
\end{align}
such that
\begin{equation}
r_{\text{inc},n}=r_V+(1-N_V)^n (r_S-r_V).
\end{equation}

And as $0 < N_V \leq 1$,
\begin{equation}
r_{\text{inc},\infty}=\lim_{n \rightarrow \infty} r_{\text{inc},n}=r_V.
\end{equation}

From the expression of $r_V$ and $r_{\text{inc}, \infty}=\frac{1}{1+e^{-\beta_{\text{inc}, \infty} E_S}}$ we get
\begin{equation}
-\beta_{\text{inc}, \infty} E_S= -\beta_R \mathcal{E}_{M_1} + \beta_H \mathcal{E}_{M_2},
\end{equation}
which yields

\begin{equation}
T_{\text{inc},\infty}=\frac{E_S}{\frac{\mathcal{E}_{M_1}}{T_R} - \frac{\mathcal{E}_{M_2}}{T_H}}.
\end{equation}

Besides being convenient to calculate, $r_V$ and $N_V$ can also be given a physical interpretation. If one pretends that the two level system of the machine $\ket{01}_M$, $\ket{10}_M$ is a qubit of energy gap $\mathcal{E}_V= \mathcal{E}_{M_1}-\mathcal{E}_{M_2}$, then its norm is given by $N_V$ and its normalized ground state population by $r_V$. As a reminder that this two level system is not exactly a real qubit but only a two level system, the name \enquote{virtual qubit} was coined to denote this two level system \cite{Brunner-2012}. See also Appendix G of \cite{Clivaz-2019bis}. With this we have calculated the attainable ground state populations after $n$ applications of the incoherent scenario for any $n \in \mathbb{N}$. We next turn our attention to the work cost. Calculating the work cost amounts to calculating $q_i^H$, the heat drawn from the hot bath to rethermalize $M_2$ to $T_H$ at the beginning of step $i=1, \dots, n$. Remember that by energy conservation $q_i^H$ equals the population change effected in $M_2$ by the rethermalization times $\mathcal{E}_{M_2}$, i.e.,
\begin{equation}
q_i^H= \mathcal{E}_{M_2} (\tau_{M_2}^H-\sigma_{M_2}^i),
\end{equation}
where $\sigma_{M_2}^i$ is the state of $M_2$ at the end of step $i-1$. The total heat drawn after $n$ applications of $\Lambda_{\text{inc}}$ is therefore given by
\begin{align}
Q_n^H &= \mathcal{E}_{M_2} (r_{M_2} - r_{M_2}^H) + \sum_{i=2}^n  q_i^H\\
&=\mathcal{E}_{M_2} (r_{M_2} - r_{M_2}^H) + \sum_{i=2}^n \mathcal{E}_{M_2} (r_{\text{inc}, i-1} - r_{\text{inc},i-2})\\
&= \mathcal{E}_{M_2} (r_{M_2} - r_{M_2}^H) + \mathcal{E}_{M_2} (r_{\text{inc}, n-1} - r_{S}),
\end{align}

where in the second step we used that from the form of $U_{\text{inc}}$, the difference of population contributing to $q_i^H$ is
\begin{equation}
r_{\text{inc},i-1}-r_{\text{inc},i-2},
\end{equation}
the difference in  ground state population that $U_{\text{inc}}$ effects on S at step $i-1$. In summary we therefore have

\begin{align}
\Delta F_{\text{inc},n}(T_H)&= \mathcal{E}_{M_2} \left( r_{M_2} - r_{M_2}^H + r_{\text{inc}, n-1} - r_{S} \right) \left( 1- \frac{T_R}{T_H} \right),\\
r_{\text{inc},n}(T_H)&=r_{\text{inc}, \infty} + (1-N_V)^n (r_S-r_{\text{inc}, \infty}),
\end{align}
where
\begin{align}
r_{\text{inc}, \infty} &= \frac{1}{1+e^{- \frac{E_S}{T_{\text{inc}, \infty}}}} , \quad T_{\text{inc}, \infty} = \frac{E_S}{\frac{\mathcal{E}_{M_1}}{T_R}-\frac{\mathcal{E}_{M_2}}{T_H}},\\
N_V &= r_{M_1} (1-r_{M_2}^H)+(1-r_{M_1}) r_{M_2}^H.
\end{align}

Note in particular that
\begin{equation}
\lim_{T_H \rightarrow \infty} r_{\text{inc}, \infty} = r_{M_1},
\end{equation}
which is the temperature achieved by a one qubit machine of energy gap $\mathcal{E}_{M_1}$ for the coherent scenario. This is an exemplification of a more general connection between the incoherent and coherent scenario that we will explore in Section~\ref{chap:quditsystem}.

\subsubsection{Repeated Coherent Scenario}

Having studied a single application of $\Lambda_{\text{coh}}$ for the special instance of the 2 qubit machine 1 qubit target such that $\mathcal{E}_{M_1}=\mathcal{E}_{M_2}+E_S$, we now turn our attention to the repeated application of $\Lambda^n_{\text{coh}}$. As done for the single application case, out of all the possible $\Lambda_{\text{coh}}$ maps, we would like to focus here on the maps that minimize the work cost. For temperatures of the target $r_{\text{coh}} \in [r_S, r_{\text{coh}}^*]$ we already know which maps $\Lambda^n_{\text{coh}}= \Lambda_{\text{coh}}^{(n)} \circ \Lambda_{\text{coh}}^{(n-1)} \circ \dots \circ \Lambda_{\text{coh}}^{(1)}$ minimize the work cost. Indeed, they are the maps such that each $\Lambda^{(i)}$ is as prescribed by Theorem~\ref{thm:optsolution1-2small} and Theorem~\ref{thm:optsolution1-2big}. This is not straightforward to see from the statement of both Theorems but upon looking at their proofs more closely, one realizes that they extend to the desired case. Indeed, given $\mathcal{E}_{M_1} = \mathcal{E}_{M_2} + E_S$, the practical rewriting used in the proof of Theorem~\ref{thm:optsolution1-2small}, Eq.~\ref{eq:usefulvHSmall}, remains practical if one has an initial system state $\tilde{\rho}_S$ at hand such that its ground state population $\tilde{r}_S$ fulfills $r_{M_2} \leq \tilde{r}_S \leq r_{M_1}$. This is because in this case the ordering of $\mathfrak{D}(\tilde{\rho}_S \otimes \rho_M)$ is the same as that of $\mathfrak{D}(\rho_{SM})$, namely
\begin{align}
&\left[ \mathfrak{D}(\tilde{\rho}_S \otimes \rho_M)\right]_0 \geq \left[ \mathfrak{D}(\tilde{\rho}_S \otimes \rho_M)\right]_1 \geq \left[ \mathfrak{D}(\tilde{\rho}_S \otimes \rho_M)\right]_4 \geq \left[ \mathfrak{D}(\tilde{\rho}_S \otimes \rho_M)\right]_5 \\
&\geq \left[ \mathfrak{D}(\tilde{\rho}_S \otimes \rho_M)\right]_2 \geq \left[ \mathfrak{D}(\tilde{\rho}_S \otimes \rho_M)\right]_3 \geq \left[ \mathfrak{D}(\tilde{\rho}_S \otimes \rho_M)\right]_6 \geq \left[ \mathfrak{D}(\tilde{\rho}_S \otimes \rho_M)\right]_7.
\end{align}

We therefore have the following result.

\begin{theorem} \label{thm:optrepsolution1-2small}
Let $\mathcal{E}_{M_1}=\mathcal{E}_{M_2} +E_S$. Let $\tilde{\rho}_S$ be a system state such that $\tilde{r}_S \in [ r_{M_2}, r_{M_1}]$. Let $r_{\text{coh}} \in [\tilde{r}_S, r_{\text{coh}}^*]$. Let $\mu = \frac{r_{\text{coh}}-\tilde{r}_S}{r_{M_1}-\tilde{r}_S}$. Let $t = \arcsin(\sqrt{\mu})$. 
Then
\begin{equation}
v= \mathfrak{D} (U \tilde{\rho_S} \otimes \rho_M U^{\dagger}), \quad \text{with } U=e^{- i t \mathcal{L}_{S M_1}}
\end{equation}
minimizes the optimization problem
\begin{equation}
\min_{v \prec \mathfrak{D}(\tilde{\rho}_S \otimes \rho_M)} v \cdot \mathfrak{D}(H_{SM}), \quad \text{s.t. } \sum_{i=0}^3 v_i = r_{\text{coh}}.
\end{equation}
 and has an associated work cost of
\begin{equation} \label{eq:Fcoh2qubitssmall}
\Delta F_{\text{coh}} = (r_{\text{coh}}-\tilde{r}_S) (\mathcal{E}_{M_1}-E_S).
\end{equation}
\end{theorem}

Note that Theorem~\ref{thm:optrepsolution1-2small} holds for $\mathcal{E}_{M_2} \leq E_S$ as well as for $\mathcal{E}_{M_2} > E_S$. This is because the rewriting of $v \cdot \mathfrak{D}(H_{SM})$ of Eq.~\ref{eq:usefulvHSmall} does not depend on how $\mathcal{E}_{M_2}$ and $E_S$ relate to one another. For the case $\mathcal{E}_{M_2} \leq E_S$, as $r_{M_2} \leq r_{S}$, the result of Theorem~\ref{thm:optrepsolution1-2small} can straightforwardly be applied. It tells us that in this case it is of no advantage, from a work cost perspective, to rethermalize our machine before the target has been cooled to $r_{M_1}$. Indeed, if we start with $\rho_S$,  any $\Lambda_{\text{coh}}$ that cools the target to less than $r_{M_1}$, does so at a work cost of at least $\mathcal{E}_{M_1}-E_S$ per population change of the target. This is the content of Theorem~\ref{thm:optsolution1-2small}. Furthermore, the next application of the coherent scenario will just continue to cool at a work cost of at least $\mathcal{E}_{M_1} - E_S$ per population change of the target and so on as long as $r_{\text{coh}} \in [r_S, r_{M_1}]$. One may therefore as well not rethermalize the machine and directly cool the target to $r_{M_1}$ at a work cost of $ (r_{M_1}-r_S) (\mathcal{E}_{M_2} - E_S)$ within a single application of the scenario, as prescribed by Theorem~\ref{thm:optsolution1-2small}.

In the case $\mathcal{E}_{M_2} > E_S$, the ordering of $\mathfrak{D}(\rho_{SM})$ is different. And so, one cannot directly make use of Theorem~\ref{thm:optrepsolution1-2small}. But a similar reasoning applies here as well. If our initial state $\tilde{\rho}_S$ is such that $\tilde{r}_S \in [r_S, r_{M_2}]$, then the ordering of $\mathfrak{D}(\tilde{\rho}_S \otimes \rho_M)$ is the same as that of $\mathfrak{D}(\rho_{SM})$ and we can use the rewriting of $v \cdot \mathfrak{D}(H_{SM})$ of Eq.~\ref{eq:usefulbeginBig} to get the following result.

\begin{theorem} \label{thm:optrepsolution1-2big}
Let $\mathcal{E}_{M_1}=\mathcal{E}_{M_2} +E_S$. Let $\tilde{\rho}_S$ be a system state such that $\tilde{r}_S \leq  r_{M_2}$. Let $r_{\text{coh}} \in [\tilde{r}_S, r_{M_2}]$. Let $\mu = \frac{r_{\text{coh}}-\tilde{r}_S}{r_{M_2}-\tilde{r}_S}$. Let $t = \arcsin(\sqrt{\mu})$. 
Then
\begin{equation}
v= \mathfrak{D} (U \tilde{\rho_S} \otimes \rho_M U^{\dagger}), \quad \text{with } U=e^{- i t \mathcal{L}_{S M_2}}
\end{equation}
minimizes the optimization problem
\begin{equation}
\min_{v \prec \mathfrak{D}(\tilde{\rho}_S \otimes \rho_M)} v \cdot \mathfrak{D}(H_{SM}), \quad \text{s.t. } \sum_{i=0}^3 v_i = r_{\text{coh}}.
\end{equation}
 and has an associated work cost of
\begin{equation} \label{eq:Fcoh2qubitssmall}
\Delta F_{\text{coh}} = (r_{\text{coh}}-\tilde{r}_S) (\mathcal{E}_{M_2}-E_S).
\end{equation}
\end{theorem}

Note that Theorem~\ref{thm:optrepsolution1-2big} also holds for both $\mathcal{E}_{M_2} \leq E_S$ and $\mathcal{E}_{M_2} > E_S$. In the case $\mathcal{E}_{M_2} \leq E_S$ one may be worried that $2 \mathcal{E}_{M_2} - \mathcal{E}_{M_1} <0$ might happen for some machines and that in that case one cannot make use of Eq.~\ref{eq:usefulbeginBig} to derive the desired result anymore. But as
\begin{equation}
(v_5+v_6+v_5+v_7) (2 \mathcal{E}_{M_2} - \mathcal{E}_{M_1}) = (1-v_0-v_1-v_2-v_4) (\mathcal{E}_{M_1} - 2 \mathcal{E}_{M_2}),
\end{equation}
the statement of Theorem~\ref{thm:optrepsolution1-2big} also holds for these machines. In any case, Theorem~\ref{thm:optrepsolution1-2big} tells us that also in the case of $\mathcal{E}_{M_2} > E_S$, it is of no advantage to rethermalize the machine before one has cooled the target to $r_{M_1}$.

Once one has cooled the target to $r_{M_1}$, the only way to further cool the target is to rethermalize the machine. Once this is done, further applications of the coherent scenario are only able to cool by acting on the $\ket{011}_{SM}$ and $\ket{100}_{SM}$ subspace of the joint system $SM$. Performing a partial swap, gradually moving population form $\ket{100}_{SM}$ to $\ket{011}_{SM}$, is then the optimal operation to further cool. At the end of this partial swap one reaches a ground sate population of
\begin{equation}\label{eq:cohrepeatedchange}
r'=r+ (1-r) r_{M_1} r_{M_2} - r (1-r_{M_1}) (1-r_{M_2}),
\end{equation}
if the ground state population of the target before the swap was $r$. Massaging the right-hand side of Eq.~\ref{eq:cohrepeatedchange} as done for Eq.~\ref{eq:degpopchange} we get
\begin{equation}
r'=r_{\tilde{V}}+(1-N_{\tilde{V}}) (r-r_{\tilde{V}}),
\end{equation}
where $N_{\tilde{V}}$ is the sum of population of $\ket{00}_M$ and $\ket{11}_M$, i.e.,
\begin{equation}
N_{\tilde{V}}=r_{M_1} r_{M_2} + (1-r_{M_1}) (1-r_{M_2}),
\end{equation}
and $r_{\tilde{V}}$ is the population of $\ket{00}_M$ renormalized by $N_{\tilde{V}}$, i.e.,
\begin{equation}
r_{\tilde{V}}= \frac{r_{M_1} r_{M_2}}{N_{\tilde{V}}} = \frac{1}{1+e^{-\beta_R (\mathcal{E}_{M_2} + \mathcal{E}_{M_1})}}.
\end{equation}
One therefore gets
\begin{equation}
r_{\text{coh},n}^*=r_{\tilde{V}} + (1-N_{\tilde{V}})^{n-1} (r_B-r_{\tilde{V}}).
\end{equation}
And as again $ 0 < N_{\tilde{V}} \leq 1$,
\begin{equation}
r_{\text{coh}, \infty}^*= \lim_{n \rightarrow \infty} r_{\text{coh},n}^*=r_{\tilde{V}}.
\end{equation}
From the expression of $r_{\tilde{V}}$ and $r_{\text{coh},\infty}^*= \frac{1}{1+ e^{-\beta_{\text{coh},\infty}^* E_S}}$ we get
\begin{equation}
\beta_{\text{coh},\infty}^*=\beta_R \frac{\mathcal{E}_{M_2} + \mathcal{E}_{M_1}}{E_S},
\end{equation}
which yields
\begin{equation}
T_{\text{coh},\infty}^*= \frac{E_S}{\mathcal{E}_{M_1}+\mathcal{E}_{M_2}} T_R.
\end{equation}
The work cost associated to it is, for $n \in \mathds{N} \cup \{+\infty\}$,
\begin{align}
\Delta F_{\text{coh},n}^* &= \Delta F_{\text{coh}}^* +2  \mathcal{E}_{M_2} \sum_{i=2}^n (r_{\text{coh},i}^*-r_{\text{coh}, i-i}^*)\\
&= \Delta F_{\text{coh}}^* + 2 \mathcal{E}_{M_2} (r_{\text{coh},n}-r_{M_1}).
\end{align}
\section{Two Open Problems} \label{sec:twoopenproblems}

By gathering intuition on each of the scenarios and systematically working out how they express themselves in the simplest settings, we have stumbled across two interesting, as well as challenging, problems.\\

In the incoherent scenario, we found that working out the degeneracies that allow cooling, turned out to be more challenging than expected. Indeed, explicitly treating all the potential degenerate subspaces of the 1 qubit system 2 qubit machine is already quite a cumbersome task. There we found that results such as those of Lemma~\ref{lemma:roommachine} and Lemma~\ref{lemma:hotmachine} render the quest for the useful machine a more tractable one. We believe that we only scratched the surface of what is possible to say in that direction and that more statements of the sort would allow getting a better picture at what is possible to achieve within the incoherent  scenario. In particular, it would allow us to better grasp which machines are expected to perform at all.\\

In the coherent scenario, we found that the main challenge was to choose among all the potential unitaries, the one operating at the lowest work cost. This made us formulate a general optimization problem. While we could solve the case of maximal cooling, we found it much more challenging to tackle the finite resource regime. There, we found that solving the 1 qubit and 2 qubit machine case for 1 qubit target stem was already surprisingly challenging. In particular, the proof techniques per se, used a lot of the fine-tuning information about the problem itself. As such, they provide little hope for generalizing to more complex system straightforwardly. However, working on them gave us a lot of intuition about the problem at hand. In particular, looking back on them, one notices that all the involved optimal transformations are T-transforms, i.e., doubly stochastic matrices $T(\lambda)$ of the form
\begin{equation}
T(\lambda)= (1-\lambda) \mathds{1} + \lambda Q,
\end{equation}
where $Q$ is a permutation matrix exchanging only two coordinates. That every vector $x$ majorized by another vector $y$ is reachable via the application of finitely many successive T-transforms is a known fact of majorization theory, see for example \cite[Chapter 2.B.1]{Marshall-2011}. But this fact alone is not quite enough to prove what we would like to have. In particular, it would be desirable that once one picks a T-transform, that one can carry it fully, i.e., fully swaps, before one chooses to carry another T-transform. This is not ensured by the standard result of majorization theory. The question that interests us here is if the solutions of Eq.~\ref{eq:deltaFoptimization} can be continuously parameterized by T-transforms in the following way.

\begin{question}
Given a $v \prec \mathfrak{D}(\rho_{SM})$ that solves the optimization problem of Eq.~\ref{eq:deltaFoptimization} for some $c \in [ r_S, r_{\text{coh}}^*]$, does there exist T-transforms $T_0(\lambda), \dots, T_{k-1}(\lambda)$, where
\begin{equation}
T_i(\lambda) = (1-\lambda) \mathds{1} + \lambda Q_i,
\end{equation}
for some permutation matrix $Q_i$ exchanging only two coordinates, such that
\begin{enumerate}
\item $v= T_k(\lambda) T_{k-1}(1) \dots T_1(1) \mathfrak{D}(\rho_{SM})$ for some $\lambda \in [0,1]$,
\item for all $\mu \in [0,k]$, $w(\mu) = T_{k-1} (t_{k-1}(\mu)) T_{k-2} (t_{k-2}(\mu)) \dots T_0(t_0(\mu))$ is a solution of Eq.~\ref{eq:deltaFoptimization} for some $\tilde{c} \in [r_S, c]$, where
\end{enumerate}

\begin{equation}
t_i = \begin{cases} 0, & \text{if } \mu < i,\\
\mu-i, & \text{if } \mu \in [i, i+1],\\
1, & \text{if } \mu > i+1.
\end{cases}
\end{equation}
\end{question}

If this were to be true, then it would mean that one could continuously move along the curve of solutions of Eq.~\ref{eq:deltaFoptimization} only by continuously exchange two coordinates of our vector at a time. All the solutions of our problem encountered until now satisfy this property. This is why we believe the answer of the above question to be affirmative. The cases treated rigorously thus far consisting of relatively elementary systems, it might, however, well be that the problem unfolds in unexpected ways for more complex systems.\\

From a mathematical perspective, an affirmative answer to the question would also drastically simplify the optimization problem. Indeed, instead of having to look through all doubly stochastic matrices to decide which one to pick, one would only need to choose the best T-transform, which boils down to choosing the pair of coordinates that can be exchanged at the cheapest in terms of energy expenditure.\\

Physically, performing a T-transform corresponds to performing a unitary in a two-dimensional subspace of the $SM$ joint system, i.e.,

\begin{equation}
U = U \big|_{\text{span} \{ \ket{ij}_{SM}, \ket{kl}_{SM} \}} \oplus \mathds{1}_{\text{span}^c \{ \ket{ij}_{SM}, \ket{kl}_{SM} \}}.
\end{equation}
 If the system is a qubit and $\ket{i}_S \neq \ket{k}_S$, this can be viewed as swapping population between $S$ and a virtual qubit of the machine $\ket{j}_M$, $\ket{l}_M$. If the answer to this question is affirmative, then all what one would need to do in order to solve the optimization problem is select the best virtual qubit to cool and then proceed to cool with this virtual qubit until the qubit is exhausted. Once the best virtual qubit is exhausted, one would then select the second best virtual qubit and cool with it until no cooling is possible with it anymore, and so on. This would elevate the concept of a virtual qubit to the fundamental operation that allows continuously cooling a qubit target system at the best energy expenditure. It would in particular mean that once the best virtual qubit has been identified, one can with confidence cool the system using that virtual qubit until the system has reached the temperature of the virtual qubit without having to continuously check while cooling with a particular virtual qubit if another operation could perform better than continuing to cool using that virtual qubit.

\chapter{Qudit System} \label{chap:quditsystem} 

We have seen in Chapter~\ref{chap:qubitsyst} that keeping track of the exact energy expenditure in the coherent scenario quickly becomes very  challenging. Similarly, precisely selecting the useful machines in the incoherent scenario has also proved itself to be a hard task, already for low dimensional machines. In order to bypass these challenges, we will in the following not keep track of the work cost of the process anymore as well as not worry in detail about what degeneracies the joint system $SM$ may or may not have. We will instead focus on the state transformations that each scenario allows performing on the target system S. Not worrying about energy expenditure and the degeneracies of our joint system will allow us to make statements about a much larger class of machines and target systems. Indeed, in the following we will allow the target system as well as the machine to be general finite dimensional quantum systems. Since our target system is not restricted to be a qubit anymore, our notion of temperature will play a bigger role in the following. We would therefore like to expand a bit more on that notion before we move on to the main results of this Chapter.

\section{Sumtemperature} \label{sec:sumtemperature}

In Section~\ref{sec:temperature} we admitted in being interested in a bit more that the ground state population of $S$ only. Indeed, we stated that we were not only interested in maximizing the ground state population of $S$ but rather in maximizing

\begin{equation}
\sum_{k=0}^l \bra{k} \sigma_S \ket{k}_S, \quad \forall l=0,\dots, d_S-1.
\end{equation}

We would here like to expand a bit more on this. First of all, we would like to define our notion of temperature more precisely.

\begin{definition}[Sumcolder] \label{def:sumcolder}
Given two system states $\sigma_1$ and $\sigma_2$, we say that $\sigma_1$ is {\bf sumcolder} than $\sigma_2$ if
\begin{equation}
\sum_{k=0}^l \bra{k} \sigma_1 \ket{k}_S \geq \sum_{k=0}^l \bra{k} \sigma_2 \ket{k}_S, \quad \forall l \in \{0, \dots, d_S-2\}.
\end{equation}
\end{definition}

\begin{definition}[Sumhotter] \label{def:sumhotter}
Given two system states $\sigma_1$ and $\sigma_2$, we say that $\sigma_1$ is {\bf  sumhotter } than $\sigma_2$ if  $\sigma_2$ is sumcolder than $\sigma_1$.
\end{definition}

\begin{definition}[Sumtemperature] \label{def:sumtemperature}
Given two system states $\sigma_1$ and $\sigma_2$, we say that $\sigma_1$ and $\sigma_2$ have the same {\bf  sumtemperature} if  $\sigma_1$ is both sumhotter and sumcolder than $\sigma_2$.
\end{definition}

Note that the sumcolder relation is (only) a preorder on the set of states (it is reflexive and transitive). In particular, the sumcolder relation is not a total order since there are states \(\sigma_1\) and \(\sigma_2\) for which \(\sigma_1\) is neither sumcolder nor sumhotter than \(\sigma_2\). In other words, not all states are comparable according to this notion of temperature. For comparable states however, our definition is a generalization of the ground state notion of temperature and so if $\sigma_1$ is colder than $\sigma_2$ according to our notion of temperature, it is also according to the ground state notion of temperature. Formally

\begin{lemma} \label{lemma:sumcolderrgen}
Let  $\sigma_1$ and $\sigma_2$ be two system states such that $\sigma_1$ is sumcolder than $\sigma_2$. Then
\begin{equation}
\bra{0} \sigma_1 \ket{0}_S \geq \bra{0} \sigma_2 \ket{0}_S.
\end{equation}
\end{lemma}

There are another few straightforward things that can be said about this new notion of temperature.

\begin{lemma}\label{lemma:sumtempcons}
$\sigma_1$ has the same sumtemperature as $\sigma_2$ if and only if
\begin{equation}
\mathfrak{D}(\sigma_1)= \mathfrak{D}(\sigma_2).
\end{equation}
\end{lemma}

\begin{lemma}\label{lemma:sumtempdiagcons}
If $\sigma_1$ and $\sigma_2$ are diagonal in the energy eigenbasis, then $\sigma_1$ has the same sumtemperature as $\sigma_2$ if and only if
\begin{equation}
\sigma_1=\sigma_2.
\end{equation}
\end{lemma}

The results of Lemma~\ref{lemma:sumtempcons} and Lemma~\ref{lemma:sumtempdiagcons} follow directly from the definitions of sumcolder and sumhotter, Definition~\ref{def:sumcolder} and Definition~\ref{def:sumhotter}, and the fact that states are normalized, i.e.,

\begin{equation}
\Tr(\sigma_1) = \sum_{k=0}^{d_S-1} [\sigma_1]_k =1=\sum_{k=0}^{d_S-1} [\sigma_2]_k=  \Tr(\sigma_2).
\end{equation}

Lemma~\ref{lemma:sumtempdiagcons} in particular means that on the set of states diagonal in the energy eigenbasis, the sumcolder relation is a partial order. The next result relates our notion of temperature with that based on the average energy $\Tr(\sigma H_S)$, see Section~\ref{sec:temperature} for more details on its definition. It states that if a state $\sigma_1$ is colder than another state $\sigma_2$ according to our notion of temperature, then it is also colder according to the average energy notion of temperature. The formal result reads as follows.

\begin{lemma} \label{lemma:sumcolderEgen}
Let $\sigma_1$ and $\sigma_2$ be two system states such that $\sigma_1$ is sumcolder than $\sigma_2$. Then
\begin{equation}
\Tr (\sigma_1 H_S) \leq \Tr(\sigma_2 H_S).
\end{equation}
\end{lemma}

\begin{proof}
We first note that for any $i \in \{1, \dots, d_S-1\}$
\begin{equation}
E_i= E_0+ \sum_{l=1}^i E_l-E_{l-1}. 
\end{equation}

So that for any system state $\sigma$
\begin{align}
\Tr(\sigma H_S) =\sum_{i=0}^{d_S-1} [\sigma]_i E_i &=  E_0+\sum_{i=1}^{d_S-1} [\sigma]_i \left( \sum_{l=1}^i E_l-E_{l-1} \right) \\
&= E_0+ \sum_{l=1}^{d_S-1}  \left( \sum_{i=l}^{d_S-1} [\sigma]_i \right) (E_l-E_{l-1}),
\end{align}
where in the last step we used $\sum_{i=1}^{d_S-1} \sum_{l=1}^i =\sum_{l=1}^{d_S-1} \sum_{i=l}^{d_S-1}$. As from $\sigma_1$ sumcolder than $\sigma_2$ follows
\begin{equation}
\sum_{i=l}^{d_S-1} [\sigma_1]_i \leq  \sum_{i=l}^{d_S-1} [\sigma_2]_i, \quad \forall l \in \{1,\dots,d_S-1\},
\end{equation}
and that 

\begin{equation}
E_l-E_{l-1} \geq 0, \quad \forall l \in \{1,\dots,d_S-1\},
\end{equation}
our result is proven.
\end{proof}

To see how our notion of temperature compares to the other notions of temperature we talked about in Section~\ref{sec:temperature}, we need two more definitions.

\begin{definition} [Ordered]
Let $\sigma$ be a target system state. Let $(\ket{v_i})_{i=0}^{d_S-1}=(\ket{v_0}, \dots, \ket{v_{d_S-1}})$  be a basis of the system. We say that $\sigma$ is { \bf  ordered } in $(\ket{v_i})_{i=0}^{d_S-1}$ if
\begin{equation}
\bra{v_i} \sigma \ket{v_i} \geq \bra{v_{i+1}} \sigma \ket{v_{i+1}}, \quad \forall i \in \{0, \dots, d_S-2\}.
\end{equation}
\end{definition}

\begin{definition}[Diagonally Ordered]
Let $\sigma$ be a target system state. Let $(\ket{v_i})_{i=0}^{d_S-1}=(\ket{v_0}, \dots, \ket{v_{d_S-1}})$  be a basis of the system. We say that $\sigma$ is {\bf  diagonally ordered } in $(\ket{v_i})_{i=0}^{d_S-1}$ if
\begin{enumerate}
\item $\sigma$ is diagonal in $(\ket{v_i})_{i=0}^{d_S-1}$,
\item $\sigma$ is ordered in $(\ket{v_i})_{i=0}^{d_S-1}$.
\end{enumerate}
\end{definition}

Note that since for $i=0, \dots, d_S-2 $, $E_i \leq E_{i+1}$, the passive states are precisely the states diagonally ordered in the energy eigenbasis~\cite{Pusz-1978, Lenard-1978}. We will use the two terms interchangeably. With these two definitions, we can now easily relate our notion of temperature to majorization. This is the content of the following two results.

\begin{lemma}\label{lemma:orderedsumcolder}
Let $\sigma_1$, $\sigma_2$ be system states ordered in $(\ket{i})_{i=0}^{d_S-1}$. Then $\sigma_1$ is sumcolder than $\sigma_2$ if and only if
\begin{equation}
\mathfrak{D}(\sigma_1) \succ \mathfrak{D}(\sigma_2).
\end{equation}
\end{lemma}

The result of Lemma~\ref{lemma:orderedsumcolder} follows directly from the definition of Majorization and of sumcolder. We remind the reader that for two hermitian matrices $\sigma_1$ and $\sigma_2$,  $\sigma_1 \succ \sigma_2$ means $\lambda(\sigma_1) \succ \lambda(\sigma_2)$. With that we can state our next result.

\begin{lemma}\label{lemma:diagorderedsumcolder}
Let $\sigma_1$, $\sigma_2$ be system states diagonally ordered in $(\ket{i})_{i=0}^{d_S-1}$. Then $\sigma_1$ is sumcolder than $\sigma_2$ if and only if
\begin{equation}
\sigma_1 \succ \sigma_2.
\end{equation}
\end{lemma}

The proof of Lemma~\ref{lemma:diagorderedsumcolder} is also direct. With this nomenclature fixed and its direct consequences spelled out, we are now ready to state the main result of this section.

\begin{theorem} \label{thm:sumtemperaturegenerality}
Let $\sigma_1$ and $\sigma_2$ be two system states diagonally ordered in $(\ket{i}_S)_{i=0}^{d_S-1}$ such that $\sigma_1$ is sumcolder than $\sigma_2$. Then

\begin{enumerate}
\item $\bra{0} \sigma_1 \ket{0} \geq \bra{0} \sigma_2 \ket{0}$,
\item $\Tr (\sigma_1 H_S) \leq \Tr(\sigma_2 H_S)$,
\item $S(\sigma_1) \leq S(\sigma_2)$,
\item $\mathcal{P}(\sigma_1) \geq \mathcal{P}(\sigma_2).$
\end{enumerate} 
\end{theorem}

\begin{proof}
1. and 2. are special cases of Lemma~\ref{lemma:sumcolderrgen} and Lemma~\ref{lemma:sumcolderEgen}. 3. and 4. follow from the fact that our sumtemperature notion of temperature is equivalent to majorization for diagonally ordered states and that the von Neumann entropy $S$ is a Schur concave function and that the purity $\mathcal{P}$ is a Schur convex function when seen as functions from the vector of eigenvalues of $\sigma$ to the reals.
\end{proof}

This means that for states diagonally ordered in the energy eigenbasis, our notion of temperature reunites all notions of temperatures. The next question to ask is if the states that one encounters in the coherent and incoherent scenario are diagonally ordered in the energy eigenbasis.\\

In the coherent scenario our target system state does not have to be diagonally ordered in $(\ket{i})_{i=0}^{d_S-1}$, since one can apply any local unitary to it. It therefore doesn't even have to be diagonal, let alone ordered, in $(\ket{i})_{i=0}^{d_S-1}$. However, since we can apply any local unitary to it and since we do not keep track of the work cost, we may w.l.o.g. assume that our system state is diagonally ordered in $(\ket{i})_{i=0}^{d_S-1}$. This will also ensure that we always pick the coldest state in the local unitary orbit of our state, where cold is according to our sumtemperature notion of temperature. This last statement follows from the following.

\begin{lemma} \label{lemma:diagorderedlowerbound}
Let $\sigma$ be a system state. Let 
\begin{equation}
\tilde{\sigma}= \sum_{i=0}^{d_S-1} \left[ \lambda^{\downarrow} (\sigma) \right]_i \ket{i} \bra{i}_S.
\end{equation}

Then $\tilde{\sigma}$ is sumcolder than $\sigma$.
\end{lemma}

\begin{proof}
From Schur's Theorem, Theorem~\ref{thm:Schur}, we know that
\begin{equation}
\mathfrak{D}(\sigma) \prec \lambda(\sigma).
\end{equation}
So
\begin{align}
\sum_{i=0}^{l} \bra{i} \sigma \ket{i} = \sum_{i=0}^l \left[ \mathfrak{D}(\sigma) \right]_i \leq \sum_{i=0}^l \left[ \mathfrak{D}^{\downarrow}(\sigma) \right]_i \leq \sum_{i=0}^l \left[ \lambda^{\downarrow}(\sigma) \right]_i= \sum_{i=0}^{l} \bra{i} \tilde{\sigma} \ket{i}.
\end{align}
\end{proof}

In the coherent scenario we will therefore always assume that $\sigma$ is diagonally ordered.\\

In the incoherent scenario, we know from Lemma~\ref{lemma:staydiag} that we may always assume our system state to be diagonal in $(\ket{i})_{i=0}^{d_S-1}$. However, we are not allowed to perform any local unitary within this scenario as these are typically energy non-conserving. The question then arises of when the states in the incoherent scenario are diagonally ordered and when not. In general, one has to be careful when comparing results obtained with the sumcolder temperature notion to the purity and entropy notion of temperature. One may nevertheless infer lower bounds for these notions of cooling from sumtemperature results by looking at the passive state~\cite{Pusz-1978, Lenard-1978} corresponding to the obtained state and making use of Lemma~\ref{lemma:diagorderedlowerbound}. For the results presented in this thesis, the explicit final states obtained in the incoherent scenario will always be diagonally ordered and as such comparable to all the other notions of cooling as prescribed by Theorem~\ref{thm:sumtemperaturegenerality}.

\section{Universal Bound on Cooling} \label{sec:unibound}

Now that we are properly equipped with our sumtemperature notion, we are ready to use it within the context of the coherent and incoherent scenario. We remind the reader that unless otherwise stated, we will in the following assume that we have a $d_S < \infty$ dimensional system $S$ at hand as well as a $d_M< \infty$ dimensional machine $M$ with maximal energy gap $\mathcal{E}_{\text{max}} = \mathcal{E}_{d_M-1} < \infty$. As usual, $S$ and $M$ will be assumed to have an initial state thermal at $T_R$ and the underlying Hamiltonian to be non-interacting. Given such a system and machine, we will be in the following interested in investigating what is the lowest temperature, in terms of sumtemperature, that can be achieved on the target system state when one is allowed to apply the coherent and incoherent scenario an unbounded number of times. In short, we will be interested in $\Lambda^{\infty}_{\text{coh}} (\rho_S)$ and $\Lambda^{\infty}_{\text{inc}} (\rho_S)$.\\

Our first result in this direction will consist of a bound on the achievable temperature of $\Lambda^{\infty}_{\text{coh}} (\rho_S)$ and $\Lambda^{\infty}_{\text{inc}} (\rho_S)$, i.e., a bound that holds for both scenarios. Before we turn to the bound itself, we would like to state and prove a preliminary result that will allow simplifying our analysis as well as clarifying the interplay between both scenarios in terms of state attainability. This result states that given a system state $\sigma_S$ diagonal in the energy eigenbasis, one can reach any sumtemperature via a single application of the coherent scenario that one can reach via a single application of the incoherent scenario. We remind the reader that this relation is a priori not clear. Indeed, while the incoherent scenario restricts to energy conserving unitaries only -- as opposed to arbitrary unitaries for the coherent scenario -- the state of the machine in the incoherent scenario is allowed to have a temperature gradient, and as such has a richer structure than in the coherent scenario. The formal result reads as follows.

\begin{lemma} \label{lemma:singlecohpower}
Let our system and machine, $S$ and $M$, be given. Let $\sigma_S$ be a state of $S$ diagonal in the energy eigenbasis. Then, for all $v \in \mathbb{R}^{d_S}$ such that there exists a $\Lambda_{\text{inc}}$ with $v = \mathfrak{D} (\Lambda_{\text{inc}}(\sigma_S))$, there also exists a $\Lambda_{\text{coh}}$ such that $v=\mathfrak{D}(\Lambda_{\text{coh}}(\sigma_S))$.
\end{lemma}

\begin{proof}
The proof is not hard. It relies mostly on the fact that a hotter thermal state is majorized by a colder one, i.e., for us
\begin{equation}
\tau^H_{M_2} \prec \tau_{M_2},
\end{equation}

and that majorization is stable under tensor product, see Corollary 1.2. of \cite{Bondar-2003}. For us that means
\begin{equation}
\rho_{M}^{R,M}=\tau_{M_1}^R \otimes \tau_{M_2}^H \prec \tau_{M_1}^R \otimes \tau_{M_2}^R= \rho_M.
\end{equation}

From this we get, using the stability under tensor product again, that
\begin{equation}
\sigma_S \otimes \rho_M^{R,H} \prec \sigma_S \otimes \rho_M,
\end{equation}

which means
\begin{equation}
\lambda \left( \sigma_S \otimes \rho_M^{R,H} \right) \prec \lambda \left( \sigma_S \otimes \rho_M \right).
\end{equation}

Now let $\tilde{\Lambda}_{\text{inc}}$ be a specific application of the incoherent scenario. Let $\tilde{U}_{\text{inc}}$ be the energy conserving unitary such that $\tilde{\Lambda}_{\text{inc}}(\cdot)=\Tr_M(\tilde{U}_{\text{inc}} \cdot \otimes \rho_M^{R,H} \tilde{U}_{\text{inc}}^{\dagger})$. Then by Schur's Theorem, Theorem \ref{thm:Schur}, we have that

\begin{equation}
\mathfrak{D} \left( \tilde{U}_{\text{inc}} \sigma_S \otimes \rho_M^{R,H} \tilde{U}_{\text{inc}}^{\dagger} \right) \prec \lambda \left(\sigma_S \otimes \rho_M^{R,H} \right) \prec \lambda \left( \sigma_S \otimes \rho_M \right).
\end{equation}

Now using Horn's Theorem, Theorem~\ref{thm:Horn}, there exists a unitary $\tilde{U}$ such that 

\begin{equation}
\mathfrak{D} \left( \tilde{U}_{\text{inc}} \sigma_S \otimes \rho_M^{R,H} \tilde{U}_{\text{inc}}^{\dagger} \right) = \mathfrak{D} \left( \tilde{U} \sigma_S \otimes \rho_M \tilde{U}^{\dagger} \right).
\end{equation}

And so
\begin{align}
\mathfrak{D} \left( \tilde{\Lambda}_{\text{inc}} (\sigma_S) \right) &= \mathfrak{D} \left[ \Tr_M \left( \tilde{U}_{inc} \sigma_S \otimes \rho_M^{R,H} \tilde{U}_{\text{inc}}^{\dagger} \right) \right]\\
&= \mathfrak{D} \left[ \Tr_M \left( \tilde{U} \sigma_S \otimes \rho_M \tilde{U}^{\dagger} \right) \right]\\
&= \mathfrak{D} \left( \tilde{\Lambda}_{\text{coh}} (\sigma_S) \right),
\end{align}
where we have chosen $\tilde{\Lambda}_{\text{coh}} (\cdot) = \Tr_M \left( \tilde{U} \cdot \otimes \rho_M \tilde{U}^{\dagger} \right)$. 
\end{proof}

What Lemma~\ref{lemma:singlecohpower} means is that, for a given machine $M$, we cannot cool a given system $S$ more with a single application of the incoherent scenario than with one application of the coherent scenario. There are two things one should point out at this stage. First of all, this result does not tell us that we can cool more within the coherent scenario than within the incoherent one. It might indeed well be that for a given machine, both scenarios, within their single application regime, allow reaching the same temperature. Whether a gap exists between both scenarios therefore remains an open question. Second of all, this result really only makes a statement about a single application of each scenario. One could imagine that the incoherent scenario is slow to start off, but that it catches up on the coherent one upon repeated applications and maybe even eventually allows for more cooling. Our next result gets rid of this possibility. It shows that for a given number $k \in \mathbb{N}$ of applications of each scenario, any temperature that $\Lambda^k_{\text{inc}}(\sigma_S)$ can reach can also be reached by $\Lambda^k_{\text{coh}}(\sigma_S)$. The statement reads as follows.

\begin{lemma} \label{lemma:kcohpower}
Let $S$ and $M$ be given. Let $k \in \mathbb{N}$. Let $\sigma_S$ be a state of $S$ diagonal in the energy eigenbasis. Then for all $v \in \mathbb{R}^{d_S}$ for which there exists a $\Lambda_{\text{inc}}^k$ such that $v= \mathfrak{D}(\Lambda_{\text{inc}}^k(\sigma_S))$, there also exists a $\Lambda^k_{\text{coh}}$ such that $v= \mathfrak{D}(\Lambda_{\text{coh}}^k(\sigma_S))$.
\end{lemma}

Note that we write $\Lambda^k_{\text{coh}}$ and $\Lambda^k_{\text{inc}}$ to mean $k$ applications of each scenario but that we allow the specific maps chosen at each step to vary from one step to the other, i.e.,
\begin{equation}
\Lambda^k_{\text{coh}}= \Lambda^{(k)}_{\text{coh}} \circ\Lambda^{(k-1)}_{\text{coh}} \circ \dots \circ \Lambda^{(1)}_{\text{coh}},
\end{equation}

and $\Lambda^{(i)}_{\text{coh}}$ does not have to be equal to $\Lambda^{(j)}_{\text{coh}}$ for $i \neq j$. Similarly for $\Lambda^k_{\text{inc}}$.\\

The statement of Lemma~\ref{lemma:kcohpower} being a straightforward generalization of Lemma~\ref{lemma:singlecohpower}, one would expect the proof of it to straightforwardly harness the result of Lemma~\ref{lemma:singlecohpower}. Since the initial state $\sigma_S$ is in both cases demanded to be diagonal and that each scenario does not necessarily deliver diagonal state, one nevertheless have to deal with a few subtleties first. To that end we first state and prove the following technical Lemma.

\begin{lemma} \label{lemma:singlecohpowerTechnical}
Let $\sigma_1$ and $\sigma_2$ be diagonal system states such that $\sigma_1 \prec \sigma_2$. Then for all $v \in \mathbb{R}^{d_S}$ for which there exists a $\Lambda_{\text{inc}}$ such that $v= \mathfrak{D}( \Lambda_{\text{inc}} (\sigma_1))$, there exists a $\Lambda_{\text{coh}}$ such that $v= \mathfrak{D}(\Lambda_{\text{coh}}(\sigma_2))$.
\end{lemma}

\begin{proof}
Using the stability of majorization under tensor product, Corollary 1.2. of \cite{Bondar-2003}, we have that
\begin{equation}
\sigma_1 \otimes \rho_M \prec \sigma_2 \otimes \rho_M.
\end{equation}

Since majorization is transitive one can straightforwardly replace all the instances of $\sigma_1 \otimes \rho_M $ in the proof of Lemma~\ref{lemma:singlecohpower} by $\sigma_2 \otimes \rho_M $, which proves our result.
\end{proof}

Now we are ready to prove Lemma~\ref{lemma:kcohpower}.

\begin{proof}[Proof of Lemma~\ref{lemma:kcohpower}]

The proof is by induction over $k$. $k=1$ is the statement of Lemma~\ref{lemma:singlecohpower}. Now suppose that the statement is true for $k$, we will show that it holds for $k+1$. Let $v \in \mathbb{R}^{d_S}$ for which there exists a $\Lambda^{k+1}_{\text{inc}} = \Lambda^{(k+1)}_{\text{inc}} \circ \Lambda^{(k)}_{\text{inc}} \circ \dots \circ \Lambda^{(1)}_{\text{inc}}$, with associated unitaries $U^{(i)}_{\text{inc}}$, such that 
\begin{equation}
v= \mathfrak{D} (\Lambda_{\text{inc}}^{k+1} (\sigma_S)).
\end{equation}

We look at $\Lambda_{\text{inc}}^k= \Lambda^{(k)}_{\text{inc}} \circ \dots \circ \Lambda^{(1)}_{\text{inc}}$. $\Lambda^k_{\text{inc}}(\sigma_S)$ may not be diagonal but by Lemma~\ref{lemma:staydiag} we know that it is diagonalizable via a local energy conserving unitary, call it $\tilde{U}_{\text{inc}}$. Then
\begin{equation}
\tilde{U}_{\text{inc}} \Lambda_{\text{inc}}^k(\sigma_S) \tilde{U}_{\text{inc}}^{\dagger}
\end{equation}

is diagonal and

\begin{align}
\tilde{U}_{\text{inc}} \Lambda_{\text{inc}}^k(\sigma_S) \tilde{U}_{\text{inc}}^{\dagger} &= \tilde{U}_{\text{inc}} \Tr_M \left[ U_{\text{inc}}^{(k)}  \Lambda_{\text{inc}}^{k-1}(\sigma_S) \otimes \rho_M^{R,H} \left( U_{\text{inc}}^{(k)}\right)^{\dagger} \right] \tilde{U}_{\text{inc}}^{\dagger}\\
&=  \Tr_M \left[\left( \tilde{U}_{\text{inc}} \otimes \mathds{1}_M U_{\text{inc}}^{(k)} \right)  \Lambda_{\text{inc}}^{k-1}(\sigma_S) \otimes \rho_M^{R,H} \left(\tilde{U}_{\text{inc}} \otimes \mathds{1}_M U_{\text{inc}}^{(k)}\right)^{\dagger} \right] \\
&=: \tilde{\Lambda}_{\text{inc}} \circ \Lambda_{\text{inc}}^{k-1}(\sigma_S) =: \tilde{\Lambda}_{\text{inc}}^k(\sigma_S).
\end{align}

Using the induction hypothesis, we know that there exists a $\tilde{\Lambda}_{\text{coh}}^k = \tilde{\Lambda}^{(k)}_{\text{coh}} \circ  \dots \circ \tilde{\Lambda}^{(1)}_{\text{coh}}$ such that
\begin{equation}
\mathfrak{D} \left( \tilde{\Lambda}_{\text{coh}}^k(\sigma_S) \right) = \mathfrak{D} \left( \tilde{\Lambda}_{\text{inc}}^k(\sigma_S) \right).
\end{equation}

$\tilde{\Lambda}_{\text{coh}}^k(\sigma_S)$ may also not be diagonal. But since it is a state, there is a unitary, call it $\tilde{U}$, such that
\begin{equation}
\tilde{U} \tilde{\Lambda}_{\text{coh}}^k(\sigma_S) \tilde{U}^{\dagger}
\end{equation}
is diagonal. As before, there is a $\Lambda_{\text{coh}}^k$ such that
\begin{equation}
\Lambda_{\text{coh}}^k (\sigma_S)= \tilde{U} \tilde{\Lambda}_{\text{coh}}^k(\sigma_S) \tilde{U}^{\dagger}.
\end{equation}

Furthermore, by Schur's Theorem (Theorem~\ref{thm:Schur}), 
\begin{equation}
\mathfrak{D} \left( \tilde{\Lambda}_{\text{inc}}^k (\sigma_S) \right)=\mathfrak{D} \left( \tilde{\Lambda}_{\text{coh}}^k (\sigma_S) \right) \prec \lambda \left( \tilde{\Lambda}_{\text{coh}}^k (\sigma_S) \right)=\mathfrak{D} \left( \Lambda_{\text{coh}}^k (\sigma_S) \right).
\end{equation}

Now let $\sigma_1= \tilde{\Lambda}_{\text{inc}}^k(\sigma_S)$ and $\sigma_2= \Lambda_{\text{coh}}^k(\sigma_S)$. Then $\sigma_1$ and $\sigma_2$ are diagonal and $\sigma_1 \prec \sigma_2$. Furthermore, there exists a $\Lambda_{\text{inc}}$ such that $v=\mathfrak{D} \left(\Lambda_{\text{inc}}(\sigma_1) \right)$. Indeed, with
\begin{equation}
\Lambda_{\text{inc}} (\cdot)= \Tr_M \left[ \left( U^{(k+1)} \tilde{U}_{\text{inc}}^{\dagger} \otimes \mathds{1}_M \right) \, \cdot \otimes \rho_M^{R,H} \left(U^{(k+1)} \tilde{U}_{\text{inc}}^{\dagger} \otimes \mathds{1}_M \right)^{\dagger} \right],
\end{equation}
we have that 
\begin{align}
\Lambda_{\text{inc}}(\sigma_1)&=\Tr_M \left[ U^{(k+1)} \left( \tilde{U}_{\text{inc}}^{\dagger} 
\sigma_1 \tilde{U}_{\text{inc}} \right) \otimes \rho_M^{R,H} \left(U^{(k+1)}\right)^{\dagger} \right]\\
&=\Lambda^{(k+1)}_{\text{inc}} \left( \Lambda_{\text{inc}}^k(\sigma_S) \right)= \Lambda^{k+1}_{\text{inc}}(\sigma_S).
\end{align}

So by Lemma~\ref{lemma:singlecohpowerTechnical}, there exists a $\Lambda_{\text{coh}}$ such that
\begin{equation}
v=\mathfrak{D} \left( \Lambda_{\text{inc}}^{k+1} (\sigma_S) \right) = \mathfrak{D} \left( \Lambda_{\text{inc}} (\sigma_1) \right)= \mathfrak{D} \left( \Lambda_{\text{coh}} (\sigma_2) \right)= \mathfrak{D} \left( \Lambda_{\text{coh}}^{k+1} (\sigma_S) \right),
\end{equation}
where $\Lambda_{\text{coh}}^{k+1} = \Lambda_{\text{coh}} \circ \Lambda_{\text{coh}}^{k}$.
\end{proof}

With these preliminary results, we have gathered some good intuition about how both scenarios compare to one another in terms of state attainability and are ready to move to the main result of this section. The result states that in the limit of unbounded, as well as infinite if well-defined, applications of each scenario, all the states that one can ever reach are sumhotter than the following state

\begin{equation}
\sigma_S^*= \sum_{k=0}^{d_S-1} \frac{ \left( e^{-\beta_R \mathcal{E}_{\text{max}}} \right)^k}{\sum_{j=0}^{d_S-1} \left( e^{-\beta_R \mathcal{E}_{\text{max}}} \right)^j} \ket{k} \bra{k}_S,
\end{equation}
given that the initial state of the system, $\rho_S$, is sumhotter than $\sigma_S^*$. Since $\rho_S$ and $\sigma_S^*$ are diagonally ordered with respect to $\left( \ket{i} \right)_{i=0}^{d_S-1}$, this is the same as saying that one cannot cool more than $\sigma_S^*$ within both paradigms given that $\rho_S \prec \sigma_S^*$ holds. In formal terms, the statement reads as follows.

\begin{theorem}[Universal bound]\label{thm:universalbound}
Let $S$ be a system. Let $M$ be a machine with maximal energy gap $\mathcal{E}_{\text{max}}$, such that $\rho_S \prec \sigma_S^*$. Then for all $k \in \mathbb{N}$
\begin{align}
\mathfrak{D} \left( \Lambda_{\text{inc}}^{k}(\rho_S) \right) \prec \sigma_S^*, \quad \mathfrak{D} \left( \Lambda_{\text{coh}}^{k}(\rho_S) \right) &\prec \sigma_S^*,
\end{align}
and if the limits exist, i.e., if $k \rightarrow \infty$ makes sense, 
\begin{align}
\mathfrak{D} \left( \Lambda_{\text{inc}}^{\infty}(\rho_S) \right) \prec \sigma_S^*, \quad \mathfrak{D} \left( \Lambda_{\text{coh}}^{\infty}(\rho_S) \right) &\prec \sigma_S^*.
\end{align}

\end{theorem}

\begin{proof}[Proof Idea]
The proof idea is the following. First of all, since the incoherent scenario cannot perform better than the coherent scenario, Lemma~\ref{lemma:kcohpower}, one really only needs to prove the result for the coherent scenario. For the coherent scenario one can show, see Section A of the Supplementary material of \cite{Clivaz-2019}, that if $\rho_S \prec \sigma_S^*$ then also $\Lambda_{\text{coh}}(\rho_S) \prec \sigma_S^*$ for any choice of $\Lambda_{\text{coh}}$. To prove this, one first only considers $\Lambda^*_{\text{coh}}$, the coherent operations cooling the most, i.e., the $\Lambda_{\text{coh}}$ that has the associated $U$ reordering $\rho_{SM}$ such that the greatest eigenvalues of $\rho_{SM}$ contribute to the ground state and so on. One shows $\Lambda_{\text{coh}}^*(\rho_S) \prec \sigma_S^*$ and then uses that for any other choice of $\Lambda_{\text{coh}}$, $\Lambda_{\text{coh}}(\rho_S) \prec \Lambda_{\text{coh}}^*(\rho_S)$. With this done, one can repeat the argument and one gets that
\begin{equation}
\Lambda_{\text{coh}}^k (\rho_S) \prec \sigma_S^*, \quad \text{for any } k \in \mathbb{N}.
\end{equation}

In particular $\left( \Lambda_{\text{coh}}^* \right)^k (\rho_S) \prec \sigma_S^*$ for all $k \in \mathbb{N}$. This means that the partial sums of $\left( \Lambda_{\text{coh}}^* \right)^k (\rho_S) \prec \sigma_S^*$ are each upper bounded by those of $\sigma_S^*$. Since they are monotonically increasing they converge and
\begin{equation}
\left( \Lambda_{\text{coh}}^* \right)^{\infty} (\rho_S) \prec \sigma_S^*.
\end{equation}
For any other sequence $ \Lambda_{\text{coh}}^k (\rho_S)$ for which all partial sums converge, the same argument holds and
\begin{equation}
\Lambda_{\text{coh}}^{\infty} (\rho_S) \prec \sigma_S^*.
\end{equation}
\end{proof}

For a qubit target system the bound can be expressed as a temperature and tells us that all the states that can be reached with both scenarios have a temperature higher or equal than
\begin{equation}
T^* = \frac{E_S}{\mathcal{E}_{\text{max}}} T_R,
\end{equation}
provided that the initial state has a temperature higher or equal than $T^*$, i.e., provided $T^* \leq T_R$. The qubit bound was first derived in \cite{Allahverdyan-2011} and also appears in \cite{Reeb-2014}.
To relate this result to the machines that we considered in chapter~\ref{chap:qubitsyst}, for the one qubit machine, $\mathcal{E}_{\text{max}}=\mathcal{E}_M$, and for the two qubit machine, $\mathcal{E}_{\text{max}}=\mathcal{E}_{M_1}+\mathcal{E}_{M_2}$. \\

Note also that the state $\sigma_S^*$ can be viewed as a thermal state at inverse temperature $\beta_R$ of the modified target Hamiltonian
\begin{equation}
\tilde{H}_S= \sum_{k=0}^{d_S-1} k \mathcal{E}_{\text{max}} \ket{k} \bra{k}_S.
\end{equation}

That is
\begin{equation}
\sigma_S^* = \frac{e^{-\beta_R \tilde{H}_S}}{\Tr \left( e^{- \beta_R \tilde{H}_S} \right)}.
\end{equation}

We will, however, try to avoid this notation since writing $\sigma_S^*$ as such might give the wrong impression that what we are doing is modifying the Hamiltonian of the system from $H_S$ to $\tilde{H}_S$ instead of actually cooling the system. Indeed, changing the Hamiltonian of a system is equivalent to completely changing it and might therefore be considered as cheating. As a comparison, if upon given a warm beer to cool one is given back a cold glass of water, most of us would not see this procedure as cooling a beer.\\

Finally, note that a particular feature of this bound is that it does not depend on all the intricacies of the machine or even on its dimension. It only cares about its maximal energy gap, $\mathcal{E}_{\text{max}}$. This simplifies a lot the analysis as one can very efficiently determine what the lowest achievable temperature on the target system is given some machine $M$. 

\section{Attainability of Bound}

Now that we have a bound, the question is if this bound is attainable. In particular, since the bound has such a simple form in that it only cares about a single parameter of the machine, one may wonder if it is not overseeing relevant parameters of the problem that would hinder its attainability. We will in the following show that the bound is attainable in the coherent scenario. This is Section~\ref{subsec:cohattain}. Concerning the incoherent scenario, note that nothing has been specified regarding the degeneracies of the joint system $SM$ so far. One therefore cannot expect to attain the bound in general. Indeed, given a target system $S$, if we choose a machine $M$ such that $H_{SM}$ has no degeneracies, we have that the system cannot be cooled incoherently. For all such machines the bound is therefore far from being reachable. What is more, even if degeneracies in the joint system exist, if the machine does not exhibit a tensor product structure, one can also perform no cooling at all incoherently. We then have two choices. We can choose to specify on the useful machines only by systematically disregarding the useless ones. The hope being that for the machines that are left, the bound might be attainable. For this, we would have to tackle the degeneracy problem we talked about in Section~\ref{sec:twoopenproblems} more systematically. The other avenue is to see how much we would potentially have to modify a given arbitrary machine to make it useful enough to reach the bound. We will investigate the second strategy in Section~\ref{subsec:incattain}. Finally, we will see in Section~\ref{subsec:autoattain} how to relate the incoherent scenario to another paradigm of cooling, namely autonomous machines.

\subsection{Coherent Scenario} \label{subsec:cohattain}

In the coherent scenario we find that the bound set by Theorem~\ref{thm:universalbound} is attainable. That is, that we can cool our system $S$ to $\sigma_S^*$, given that $\rho_S \prec \sigma_S^*$. To show that the bound is attainable, we need to find a protocol, i.e., a sequence of $\Lambda_{\text{coh}}^{(k)}$ maps, such that in the limit of $k \rightarrow \infty$,
\begin{equation}
\Lambda_{\text{coh}}^k (\rho_S) = \Lambda_{\text{coh}}^{(k)} \circ \dots \circ \Lambda_{\text{coh}}^{(1)} \rightarrow \sigma_S^*.
\end{equation}

The best candidate for this is to choose the map $\Lambda_{\text{coh}}^*$, i.e., the map that cools the most at each step. Formally this protocol can be defined as.

\begin{definition}[Optimal coherent protocol]

Given a joint state $\sigma_{SM}$ let $U_{\text{opt}}$ be the unitary that reorders the eigenvalues of $\sigma_{SM}$ as largest in the energy subspace $\ket{00} \bra{00}_{SM}$, second largest in $\ket{01} \bra{01}_{SM}$ and so on all the way up to $\ket{d_S-1, d_M-1} \bra{d_S-1 ,d_M-1}_{SM}$. That is 
\begin{equation}\label{eq:optcohunitary}
    U_{\text{opt}} \sigma_{SM} U_{\text{opt}}^{\dagger} = \sum_{\substack{i \in \{0, \dots, d_S-1\},\\j \in \{0, \dots, d_M-1\}}} \left[\lambda^{\downarrow} (\sigma_{SM}) \right]_{i \cdot d_M +j} \ket{ij} \bra{ij}_{SM}.
\end{equation}

The optimal coherent protocol is then defined as applying $A$ to the system state in each step, where
\begin{equation}
   \sigma_S \mapsto A(\sigma_S)= \Tr(U_{\text{opt}} \sigma_S \otimes \tau_M U_{\text{opt}}^{\dagger}).
\end{equation}

\end{definition}

Note that this protocol needs to harness all the information of the machine to be implemented, as one needs to study the order of $\sigma_{SM}$ at each step. However, the end result suggests that this information might not be needed. This motivates us to look at another protocol that makes use of only the $\mathcal{E}_{\text{max}}$ information of the machine. We call this protocol the max-swap protocol. What this protocol does is perform a swap between two consecutive levels of the system $(i-1,i)$ and the maximal energy gap of the machine. That is, it performs the following unitary

\begin{equation} \label{eq:Uidef}
U_i= \ket{i-1, d_M-1} \bra{i 0}_{SM} + \ket{i 0} \bra{i-1 ,d_M-1}_{SM} \oplus \mathds{1}_{\text{span}^c\{\ket{i-1, d_M-1}_{SM}, \ket{i 0}_{SM} \}}.
\end{equation}

In order to make sure to cool at each step, we pick $i$ such that $\Delta_i$, the population difference that the application of $U_i$ causes on energy level $\ket{i-1}$, is positive. If there is no such $i$, then the protocol does nothing. In order to be efficient, the protocol picks $i$ such that $\Delta_i$ is the greatest among all positive ones. Before and after applying $U_i$, we also reorder $\sigma_S$ to make it passive~\cite{Pusz-1978, Lenard-1978}. This makes the protocol yet a bit more efficient. The protocol can be formally defined as follows.

\begin{definition}[Coherent max-swap protocol]

Given a system state $\sigma_S$, let $\bar{k}$ be the index $i \in \{1,\dots,d_S-1\}$ for which $\Delta_i = \left[\mathfrak{D}(\sigma_S) \right]_{i} \left[\mathfrak{D}(\tau_M) \right]_{0}  - \left[\mathfrak{D}(\sigma_S) \right]_{i-1} \left[\mathfrak{D}(\tau_M) \right]_{d_M-1}$ is the greatest if there exists a positive $\Delta_i$, else let $\bar{k}=0$. That is
\begin{equation}
\bar{k}= \begin{cases}
\argmax\limits_{i=1,\dots,d_S-1} \Delta_i &, \text{if} \max_i \Delta_i > 0\\
0&, \text{else}.
    \end{cases}
\end{equation}
Let $U_0=\mathds{1}_{SM}$, and for $i=1,\dots, d_S-1$ let $U_i$ be given by Eq.~\ref{eq:Uidef}.

For a given system state $\sigma_S$, let $\mathcal{P}(\sigma_S)$ be its corresponding passive state, i.e., 
\begin{equation}
\mathcal{P}(\sigma_S) = \sum_k^{d_S-1} \left[\lambda^{\downarrow}(\sigma_S) \right]_k \ket{k} \bra{k}_S.
\end{equation}

The coherent max-swap protocol is then defined as applying $B$ to the system state in each step, where
\begin{equation}
    \sigma_S \mapsto B(\sigma_S) = \mathcal{P} \left( \Tr_M \left[ U_{\bar{k}} \mathcal{P}(\sigma_S) \otimes \tau_M U_{\bar{k}}^{\dagger} \right]\right).
\end{equation}

Note that the transformation $\rho \rightarrow \mathcal{P}(\rho)$ is a unitary, so the above map can be expressed via a single joint unitary on the system and machine, which is an allowed coherent operation.

\end{definition}

With both of these protocols defined, we can state the main result of this section, which is that both protocols converge to $\sigma_S^*$ in the limit of $k \rightarrow \infty$ applications, given that $\rho_S \prec \sigma_S^*$ holds.

\begin{theorem}[Coherent Attainability] \label{thm:cohattain}
Let $S$ be a system. Let $M$ be a machine such that $\rho_S \prec \sigma_S^*$. Then
\begin{equation}
A^{\infty} (\rho_S) = B^{\infty} (\rho_S) = \sigma_S^*.
\end{equation}
\end{theorem}

\begin{proof}[Proof idea]
The fully detailed proof is given in the Supplementary Material B of \cite{Clivaz-2019}. We here give its road map. First of all, both protocols converge since the partial sums of $A^{k}(\rho_S)$ and $B^k(\rho_S)$ are monotonically increasing and bounded by 1. We then show that the point of convergence of the max-swap protocol majorizes $\sigma_S^*$, this is Lemma 5 of the supplementary material B of \cite{Clivaz-2019}. Using the result of Theorem~\ref{thm:universalbound}, we have that $\sigma_S^* \prec B^{\infty}(\rho_S) \prec \sigma_S^*$, from which we get $B^{\infty} (\rho_S)=\sigma_S^*$.

One can prove that the optimal protocol converges to $\sigma_S^*$ in the same way. We, however, choose a different route by showing that for all $k \in \mathbb{N}$, $A^k(\sigma_S) \succ \Lambda_{\text{coh}}^k(\sigma_S)$ for any choice of $\Lambda_{\text{coh}}^k= \Lambda_{\text{coh}}^{(k)} \circ \dots \circ \Lambda_{\text{coh}}^{(1)}$ and any system state $\sigma_S$. This also formally motivates the name of the protocol by making sure that it is not only optimal at each step but also over all $k$ steps. With this result we have that $A^{\infty} (\rho_S) \succ B^{\infty}(\rho_S) \succ \sigma_S^*$. Again from Theorem~\ref{thm:universalbound}, $A^{\infty}(\rho_S) \prec \sigma_S^*$, which ends the proof. 
\end{proof}

\subsection{Incoherent Scenario} \label{subsec:incattain}

While the bound is always attainable in the coherent scenario, we know  that this in not the case for the incoherent scenario. We can attain the bound incoherently, however, if we slightly modify the machine. Indeed, the coherent max-swap protocol only requires to perform swaps between the various levels of the system and the maximal energy gap of the machine. While these swaps are certainly not energy conserving, we can make them energy conserving by adding for each swap a qubit of the right energy gap to the machine. The original machine $M$ together with these newly added qubits constitute a new machine $\tilde{M}$ that we call the extended machine. In precise terms the extended machine is defined as follows.

\begin{definition}[Extended Machine]
Given the machine M with Hamiltonian $H_M= \sum_{i=0}^{d_M-1} \mathcal{E}_i \ket{\mathcal{E}_i} \bra{\mathcal{E}_i}$ and the system S with Hamiltonian $H_S= \sum_{i=0}^{d_S-1} E_i \ket{E_i} \bra{E_i}$ we define the extended machine $\tilde{M}$ by appending the qubits $Q_i, \, i=1, \dots, d_S-1$, of energy gap
\begin{equation}
    \mathcal{E}_{Q_i}=\mathcal{E}_{\text{max}}-(E_{i}-E_{i-1}), \quad i=1, \dots, d_S-1
\end{equation}
to the machine $M$. The Hamiltonian of the extended machine is therefore given by
\begin{equation}
    H_{\tilde{M}}= H_M \otimes \mathds{1}_{\tilde{M}} + \sum_{i=1}^{d_S-1} H_{Q_i} \otimes \mathds{1}_{\{Q_i\}^c},
\end{equation}
with $H_{Q_i} = \mathcal{E}_{Q_i} \ket{1} \bra{1}_{Q_i}$.
\end{definition}

Adding these qubits in essence bridges the relevant energy levels of the original joint system $SM$ and allows passing population between them via the newly created degenerate subspace of $S \tilde{M}$. Of course, for some population to be passed around meaningfully, one needs to heat up part of the extended machine. We opt for heating up the added qubits while leaving the entirety of the original machine at room temperature. At each step, the state of the extended machine before applying the energy conserving unitary is therefore given by
\begin{equation}
\rho_{ \tilde{M}}^{R,H}= \rho_M \otimes \tau_{Q_1}^H \otimes \cdots \otimes \tau_{Q_{d_S-1}}^H.
\end{equation}

 To show that this new machine allows us to cool the system at least to the bound of Theorem~\ref{thm:universalbound}, we design a protocol that converges does so. The protocol is based on the coherent max-swap protocol. Indeed, it performs a swap between two consecutive levels of the system $(i-1,i)$ and the maximal energy gap of the machine $M$. Only, to ensure that the swap is energy preserving, it does so simultaneously aided with a swap on the added qubit $Q_i$. The protocol is formally defined as follows.
 
 \begin{definition}[Incoherent max-swap protocol]

Given a system state $\sigma_S$ and the extended machine $\tilde{M}$, let $\bar{k}$ be the index $i \in \{1,\dots,d_S-1\}$ for which $\tilde{\Delta}_i = \left[\mathfrak{D}(\sigma_S) \right]_{i} \left[\mathfrak{D}(\tau_M) \right]_{0}  \left[\mathfrak{D}(\tau_{Q_i}) \right]_{1} - \left[\mathfrak{D}(\sigma_S) \right]_{i-1} \left[\mathfrak{D}(\tau_M) \right]_{d_M-1} \left[\mathfrak{D}(\tau_{Q_i}) \right]_{0}$ is the greatest if there exists a positive $\tilde{\Delta}_i$, else let $\bar{k}=0$. That is
\begin{equation}
\bar{k}= \begin{cases}
\argmax\limits_{i=1,\dots,d_S-1} \tilde{\Delta}_i &, \text{if } \max_i \tilde{\Delta}_i > 0\\
0&, \text{else}.
    \end{cases}
\end{equation}
Let $\tilde{U}_0=\mathds{1}_{S\tilde{M}}$, and for $i=1,\dots, d_S-1$ let $\tilde{U}_i$ be defined as follows.
\begin{equation}
\begin{aligned}
    \tilde{U}_i = \mathds{1}_{S\tilde{M}} &- \ket{i-1, d_M-1, 0 } \bra{i-1, d_M-1, 0}_{SMQ_i}\otimes \mathds{1}_{\{SMQ_i\}^c} \\
    &- \ket{i 0 1} \bra{i01}_{SMQ_i} \otimes \mathds{1}_{\{SMQ_i\}^c}\\
    &  + \ket{i-1, d_M-1, 0} \bra{i01}_{SMQ_i} \otimes \mathds{1}_{\{SMQ_i\}^c} \\
    &+ \ket{i01} \bra{i-1, d_M-1, 0}_{SMQ_i} \otimes \mathds{1}_{\{SMQ_i\}^c}.
    \end{aligned}
\end{equation}

The incoherent max-swap protocol is then defined as applying $\tilde{B}$ to the system state in each step, where
\begin{equation}
    \sigma_S \mapsto \tilde{B}(\rho_S) =  \Tr_{\tilde{M}} [ \tilde{U}_{\bar{k}} \sigma_S \otimes \rho_{ \tilde{M}}^{R,H} \tilde{U}_{\bar{k}}^{\dagger}].
\end{equation}

The above unitary corresponds to swapping the pair of levels $\{\bar{k}-1,\bar{k}\}$ of the target with the maximum energy gap of the machine and the particular qubit subspace in the extension that makes the unitary an energy preserving swap between degenerate states. Indeed
\begin{equation}
    E_{\bar{k}-1} + \mathcal{E}_{\text{max}}= E_{\bar{k}} + \mathcal{E}_{Q_{\bar{k}}}.
\end{equation}

\end{definition}

Note that in the incoherent version of the max-swap protocol, we do not render the state passive before and after the application of $\tilde{U}_{\bar{k}}$. This is because doing so is typically not doable via energy conserving unitaries and thereby not allowed within the incoherent scenario. This has as a consequence that we cannot ensure that the state that we get after application of the incoherent max-swap protocol is passive. Our sumtemperature notion of temperature is therefore not equivalent to the majorization relation, and we are to make cautious use of how they relate to one another to derive our result. That being said, one can show in a very similar manner than done for the coherent max-swap protocol that this protocol can cool the system state to at least the sumtemperature of $\sigma_S^*$, provided that $\rho_S \prec \sigma_S^*$. The precise statement reads as follows.

\begin{theorem}
Let $S$ be a system. Let $M$ be a machine such that $\rho_S \prec \sigma_S^*$. Then in the limit $T_H \rightarrow \infty$, $\tilde{B}^{\infty}(\rho_S)$ is sumcolder than $\sigma_S^*$.
\end{theorem}

\begin{proof}[Proof idea]
Since for any $l=0,\dots, d_S-1$
\begin{equation}
\sum_{i=0}^{l} \left[ \mathfrak{D} \left( \tilde{B}^k(\rho_S) \right)\right]_i
\end{equation}
is monotonically increasing and bounded by 1 as a sequence of $k \in \mathbb{N}$, it converges. This implies that the protocol converges. One then shows completely analogously to how the result is proven for the coherent max-swap protocol, that the converging point is a fixed point of the protocol. This is Lemma 9 of the supplementary material C of \cite{Clivaz-2019}. Since the converging point is a fixed point, every $\tilde{\Delta}_i$ is smaller equal zero. In the limit $T_H \rightarrow \infty$, this means

\begin{equation}
\frac{\left[ \mathfrak{D} \left( \tilde{B}^{\infty}(\rho_S) \right)\right]_i} {\left[ \mathfrak{D} \left( \tilde{B}^{\infty}(\rho_S) \right)\right]_{i-1}} \leq \frac{ \left[ \mathfrak{D} \left( \tau_M \right)\right]_{d_M-1}} { \left[ \mathfrak{D} \left( \tau_M \right)\right]_{0}}= \frac{ \left[ \mathfrak{D} \left( \sigma_S^* \right)\right]_{i}} { \left[ \mathfrak{D} \left( \sigma_S^* \right)\right]_{i-1}}.
\end{equation}

From this, following the steps of the proof of Lemma 5 of the supplementary material B of \cite{Clivaz-2019} we get that for every $k= 1, \dots d_S$,

\begin{equation}
\sum_{i=0}^{k-1} \left[ \mathfrak{D} \left( \tilde{B}^{\infty}(\rho_S) \right)\right]_i \geq \sum_{i=0}^{k-1} \left[ \mathfrak{D} \left( \sigma_S^* \right)\right]_i,
\end{equation}
which proves that $\tilde{B}^{\infty}(\rho_S)$ is sumcolder than $\sigma_S^*$ as desired.
\end{proof}

\subsection{Autonomous Cooling} \label{subsec:autoattain}

The attainability of the bound for both our scenarios having been discussed, we would like to make a small detour in our analysis of the performance of both scenarios by drawing a parallel between the incoherent scenario and another paradigm of cooling, namely autonomous cooling. Within both of our scenarios we deal with stroke type machines in the sense that we clearly separate the rethermalization step and unitary operations in discrete time steps. However, in autonomous cooling, both processes of thermalization and unitary evolution happen simultaneously and continuously. In this paradigm, one typically turns on interaction terms between the various components of the machine the system and the environment and then looks at the steady state of the system that is typically achieved if one waits long enough. The machine is autonomous in the sense that once the right interactions have been turned on, which can be seen as setting up the machine, it cools the system by itself, without any external help. In particular, no external source of work or precise timing is required. As such, this paradigm of cooling requires the least amount of control out of all the paradigms discussed so far.

In the regime where the interactions are weak, the dynamics of the system is well approximated by a linear master equation and is analytically solvable for the simple case of a two qubit machine cooling a qubit system \cite{Linden-2010, Skrzypczyk-2011}. The interaction term between the two qubit machine and the qubit system is in that case generated by the following interaction Hamiltonian

\begin{equation}
H_{\text{int}}= g \left( \ket{010} \bra{101}_{SM_1 M_2} + \ket{101} \bra{010}_{S M_1 M_2} \right),
\end{equation}
where the strength of the interaction $g$ is much smaller than the energy gaps of $S$, $M_1$ and $M_2$. Note that this Hamiltonian generates the energy conserving unitary of the incoherent two qubit machine. The analogy between both paradigm goes further. Indeed, in the ideal case where the system is completely isolated from the room temperature environment, i.e., solely interacting with the machine, the steady state achieved on the system by the autonomous machine is the same as that achieved after infinite repetitions of the incoherent paradigm namely the thermal state at temperature

\begin{equation}
T^*= \frac{E_S}{\mathcal{E}_{\text{max}}} T_R.
\end{equation}

As the master equation governing the dynamics in the autonomous case is linear, the analogy between both paradigms also straight forwardly carries to any machines in the following sense. Given a qubit system and an arbitrary machine of maximal energy gap $\mathcal{E}_{\text{max}}$, one needs only add a single qubit of the right energy gap to the machine to autonomously cool the system to the bound of Theorem~\ref{thm:universalbound}. That is, we have the following.

\begin{theorem}
Let $S$ be a qubit at $T_R$ isolated from its environment. Let $M$ be a machine with $\mathcal{E}_{\text{max}} \geq E_S$. Then, one can cool the target to

\begin{equation}
T^*= \frac{E_S}{\mathcal{E}_{\text{max}}} T_R
\end{equation}
autonomously by adding a single qubit $Q$ of energy gap $\mathcal{E}_{\text{max}}-E_S$ to $M$.
\end{theorem}

\begin{proof}
One needs only couple the machine $M$ to the room temperature bath and couple the added qubit to the bath at temperature $T_H \rightarrow \infty$. We then engineer the following interaction between the modified machine and the target system

\begin{equation}
H_{\text{int}}= g (\ket{0 ,d_M-1, 0} \bra{101}_{SMQ} + \ket{101} \bra{0,d_M-1, 0}_{SMQ})
\end{equation}

and harness the result of \cite{Linden-2010} using the linearity of the master equation governing the dynamics.
\end{proof}

\chapter{Conclusion and Outlook} \label{ref:concl}

We defined two cooling scenarios that we called coherent and incoherent in this Part. After having gained a lot of insight in their inner workings we could derive an attainable bound valid in both our scenarios that only depends on the maximal energy gap of the machine and is in particular independent of the target systems energy gaps. The bound is valid for any finite dimensional target system and a large class of machines. When valid, the bound is furthermore always attainable in the coherent scenario, and we exhibited two cooling protocols that converge to it. In the incoherent scenario, while the bound is not always attainable, a minimal modification of the machine allows reaching it. Interestingly, for a qubit target system, the bound is reachable already with the simplest machines consisting of one qubit of energy gap $\mathcal{E}_{\text{max}}$ in the coherent scenario and a two qubit machine of energy gap $\mathcal{E}_{\text{max}}$ and $\mathcal{E}_{\text{max}}-E_S$. The minimal incoherent machine also allows cooling the target to the desired bound autonomously, showcasing the fact that the results go beyond a particular approach and unify different operational approaches to quantum thermodynamics.\\

For a qubit target we furthermore made a thorough study of the one and two qubit machines. There we showed that it is impossible to incoherently cool a qubit with a single qubit machine and that given a target qubit, the only two qubit incoherent machine that allows cooling is that satisfying the degeneracy constraint $\mathcal{E}_{M_1}= \mathcal{E}_{M_2} + E_S$. This lead us to formulate the main remaining open problem of the incoherent scenario, namely that of characterizing the degenerate subspaces in which cooling can be performed. 

For the coherent scenario we could analytically find the best cooling strategy in terms of work cost expenditure for the one and two qubit machines as well as solve the end-point cooling of a single cycle for arbitrary machines. This lead us to formulate the main open problem of the coherent scenario, namely that of calculating the best cooling strategy for arbitrary machines in terms of work cost expenditure.\\

As we have seen, while we have gained significant insight in the workings of both of our scenarios, there are a number of open problems that remain. More generally, while the concept of a virtual qubit is suspected to play a significant role in unraveling the best coherent cooling strategy for qubit targets, it would be interesting to see if this is also the case for target qudits or if the right concept of a virtual qudit is more appropriate in that context. For qudit target system, the right optimization problem has, however, yet to be defined.

While our universal and attainable bound embodies one of the central conceptual pillars of statistical physics, namely that despite the potential complexity, thermodynamics tasks can be characterized by a few relevant parameters, the protocols attaining the bound are of course highly idealized and are not expected to perform perfectly in realistic many-body quantum systems scenarios. That makes the attainability by few qubit machines all the more interesting as they make the bound potentially attainable with state-of-the-art quantum technology. To render the protocols for general machines more realistic, investigating the reachability of the bound with limited unitary control or complexity would be of great interest. As already mentioned, the scenarios are inherently limited in the way they are defined by for example being markovian. Extending them to the non-markovian regime as done in~\cite{Taranto-2020} or more generally extending them beyond their original limitations is also a route for further investigations. One could for example extend the scenarios to other initial states such as non-thermal states or correlated states. Finally, while we focused solely on the task of refrigeration, other tasks such as that of work extraction are of interest in quantum thermodynamics.

\part{Creation of Correlations} 

\label{part:corre} 


\emph{This part is based on the following paper:}  \\

\begin{itemize}
\item F. Bakhshinezhad, \underline{F. Clivaz}, G. Vitagliano, P. Erker, A. Rezakhani, M. Huber, and N. Friis, "Thermodynamically optimal creation of correlations," \href{https://doi.org/10.1088/1751-8121/ab3932}{Journal of Physics A: Mathematical and Theoretical {\bf 52}, 465303 (2019)}, \href{https://arxiv.org/abs/1904.07942}{arXiv:1904.07942}.

\end{itemize}

\chapter{Introduction} \label{sec:corrintro}

While Part~\ref{part:refri} of this thesis concerned itself with a purely thermodynamic task, we will here turn our attention to studying the interplay of Thermodynamics with another fascinating field of research, namely that of Information theory. Both fields have already been influencing one another for quite some time, the starting point of which might arguably be Maxwell's demon \cite{Maxwell-1871, Bennett-1982, Leff-2003} in 1871, well before information theory had a name for itself. The treatment of information in a rigorous and abstract setting by Shannon~\cite{Shannon-1948} in 1948 allowed for many more later connections, one of which is Jaynes', who viewed the foundation of thermodynamics from a different perspective and formulated a principle of maximal entropy~\cite{Jaynes-1957}. In the meantime both theories have been extended to the quantum regime \cite{Nielsen-2010, Wilde-2013, Binder-2018} and the quantum information theory community has contributed to the renewed research interest in the now rapidly evolving field of quantum thermodynamics~\cite{Goold-2016, Vinjanampathy-2016, Millen-2016}.\\

In studying the interplay between both fields, it is of particular interest to understand how their respective resources are interchangeable. Energy is a well-established resource in thermodynamics. From a physical perspective, correlations can be seen as the resource in information theory, in that in order to acquire information from a system, one needs to correlate it with another system, sometimes called a pointer~\cite{Guryanova-2020}. As long as both systems, or system and pointer, are not already interacting, by means of some interaction Hamiltonian term turned on for example, establishing these correlations necessitate some investment of energy~\cite{Friis-2016}. Conversely, energy can be extracted from any kind of correlations~\cite{PerarnauLlobet-2015}. This settles the question of resource inter-convertibility from a qualitative point of view and highlights the importance of correlations in quantum thermodynamics \cite{Friis-2016, PerarnauLlobet-2015, Huber-2015, Bruschi-2015, Binder-2015, Alipour-2016, Bera-2017, Mueller-2018, Bera-2019, Sapienza-2019}.\\

Quantitative statements are however much harder to establish in that regard and for the most part remain elusive beyond existing bounds~\cite{Vitagliano-2018, Huber-2015, Bruschi-2015, Jevtic-2012,  Jevtic-2012b}. We will here partially fill this gap in addressing the question of how much correlations can be created for a given amount of energy. The precise problem we will be looking at has first been formulated in~\cite{Huber-2015}, see also~\cite{Vitagliano-2018}. The rest of this Part is structured as follows. In Chapter~\ref{chap:corrnotation} we will set some notation before formulating the general problem we are interested in Chapter~\ref{chap:genquest}. In Chapter~\ref{chap:pure} we will present the solution of our problem for initial pure states. We will then make some considerations in Chapter~\ref{chap:cons} that will allow us to appreciate the complexity of the general problem. In Chapter~\ref{chap:symmetric} we will explore in more details the case, where both systems are copies of one another. For this case we will construct a framework in Sec.~\ref{sec:framework} and Sec.~\ref{sec:remarks}, which we will make use of in 3 different ways in Sec.~\ref{sec:marginalapproach}, Sec.~\ref{sec:passingnormapproach}, and in Sec.~\ref{sec:geometric}. This will allow us to find the operations that create correlations at an optimal energy expenditure for all 3 dimensional as well as 4 dimensional symmetric systems.

\chapter{Notation} \label{chap:corrnotation}

We would here like to set some further notation that will be used throughout this Part. We remind the reader that since this chapter is intended for reference, it might simply be skipped upon a linear reading.

 We are here interested in creating correlations between two systems $A$ and $B$. Their Hamiltonians are denoted by
\begin{align}
H_A&=\sum_{i=0}^{d_A-1} E_i^A \ket{i} \bra{i}_A, \quad \text{with } E_i^A \leq E_{i+1}^A \forall i=0,\dots,d_A-1,\\
H_B&=\sum_{i=0}^{d_B-1} E_i^B \ket{i} \bra{i}_B, \quad \text{with } E_i^B \leq E_{i+1}^B \forall i=0,\dots,d_B-1.
\end{align}

Given the initial state $\rho_{AB}=\tau_A \otimes \tau_B$, we denote the set of all the unitaries that leave the marginals of $\rho_{AB}$ diagonal in the energy eigenbasis by $\mathcal{A}$, i.e.,
\begin{equation}
\begin{aligned}
\mathcal{A}= &\left\{ U: \mathcal{H}_{AB} \rightarrow \mathcal{H}_{AB} \mid \left[ \Tr_A \left(U \rho_{AB} U^{\dagger} \right) \right]_{ij}= \left[ \Tr_B \left(U \rho_{AB} U^{\dagger} \right)\right]_{ij}=0, \right. \\ 
& \quad \left. \forall i \neq j, \right\}.
\end{aligned}
\end{equation}

We denote by $\mathcal{B} \subset \mathcal{A}$ the set of all unitaries that do not alter the entries $\left[ \rho_{AB} \right]_{ik, jk}$ and $\left[ \rho_{AB} \right]_{ki, kj}$ if $i \neq j$, i.e.,
\begin{align}
\mathcal{B}= & \left\{ U: \mathcal{H}_{AB} \rightarrow \mathcal{H}_{AB} \mid \left[ U \rho_{AB} U^{\dagger} \right]_{ik, jk}=\left[ U \rho_{AB} U^{\dagger} \right]_{ki, kj}=0 \; \forall k \text{ and } i \neq j \right\}.
\end{align}

Also, for $H_A=H_B$ we have the following further notation

\begin{align}
\mathcal{H}_i &= \text{span}\{\ket{j, j+i} \}_{j=0}^{d-1}, \quad i=0,\dots,d-1\\
r_i &= \mathfrak{D} \left( \rho_{AB} \big|_{\mathcal{H}_i} \right), \\
U_i &: \mathcal{H}_i \rightarrow \mathcal{H}_i , \text{ unitary},\\
\tilde{r}_i &= \mathfrak{D} ( U_i \rho_{AB} \big|_{\mathcal{H}_i} U_i^{\dagger} ),\\
\mathbf{p}_A &=\mathfrak{D}(\rho_A),\\
\mathbf{p}_B &=\mathfrak{D}(\rho_B),\\
\tilde{\mathbf{p}}_A &= \mathfrak{D} \left( \Tr_B (U \rho_{AB} U^{\dagger}) \right),\\
\tilde{\mathbf{p}}_B &= \mathfrak{D} \left( \Tr_A (U \rho_{AB} U^{\dagger}) \right),\\
p_{ij}&= \bra{ij} \rho_{AB} \ket{ij}\\
\tilde{p}_{ij} &= \bra{ij} U \rho_{AB} U^{\dagger} \ket{ij}\\
p_{ij}(\beta')&= \bra{i} \tau(\beta') \ket{j} \bra{i} \tau(\beta') \ket{j}, \quad \forall \beta' \leq \beta_R\\
b_i &= (p_{0i}(\beta'), p_{1 (i+1)} (\beta'), \dots, p_{(d-1) (d-1+i)} (\beta')),\\
\delta_{i} &= E_{i+1}-E_i, \quad \forall i=0,\dots,d-2.
\end{align}

\chapter{The General Problem} \label{chap:genquest}

We are interested here in studying how well one can correlate two initially uncorrelated systems $A$ and $B$ with respective Hamiltonians $H_A= \sum_{i=0}^{d_A-1} E_i^A \ket{i} \bra{i}_A$ and $H_B= \sum_{i=0}^{d_B-1} E_i^B \ket{i} \bra{i}_B$. The general interest behind this stems from the fact that correlations are found to be useful, especially to perform information processing tasks. We might, however, not always have them at hand, from which the need to create them emerges.\\

Since we assume that correlations do not generate themselves or are already present, we will assume that the joint Hamiltonian is non-interacting, i.e.,
\begin{equation}
H_{AB} = H_A \otimes \mathds{1}_B + \mathds{1}_A \otimes H_B,
\end{equation}
and that the initial state of the system is uncorrelated, i.e.,
\begin{equation}
\rho_{AB}= \rho_A \otimes \rho_B.
\end{equation}

We also want to take a thermodynamic perspective on the problem, meaning that we will assume a background temperature $T_R$. We will allow to unitarily control the joint system $AB$ to create the desired correlation. This unitary control allows us to engineer any interaction within the $AB$ system, while keeping a closed system perspective by explicitly treating all the involved parties. It therefore appears to us as being the good middle ground of allowing enough control to be able to achieve interesting states on $AB$. And at the same time of not allowing too much and risk to completely loose track of what is done and implicitly allow for ``cheating'' by allowing an ancillary system to tacitly provide us with the desired state.\\

Correlations being a resource, we expect to have to invest some other resource to create it. Here we will be interested in how much energy, a ubiquitous resource in thermodynamics, will be needed to be invested to create the desired state. To make sure that we are not hiding any energy in the initial state, we will assume the initial state to be passive with respect to $H_{AB}$~\cite{Pusz-1978, Lenard-1978}. In fact, we will assume the initial state of $AB$ to be a special passive state, namely the thermal state of $H_{AB}$ at the background temperature $T_R$. This is, after all, the state that comes for free in a thermal background at temperature $T_R$.\\

We will be interested in all kinds of correlations and as such choose to measure correlations according to the mutual information of the system that we denote as $\mathcal{I}(\rho_{AB})$. As a reminder, for some state $\sigma_{AB}$ on $AB$, its mutual information is given by

\begin{equation}
\mathcal{I}(\sigma_{AB}) = S(\sigma_A) + S(\sigma_B) - S(\sigma_{AB}),
\end{equation}
where $S(\sigma)= -\Tr(\sigma \ln(\sigma))$ is the von Neumann entropy of $\sigma$.

We will measure the energy invested in the system as the average energy change of the system, where the average energy of a state $\sigma_{AB}$ is given by

\begin{equation}
\Tr(\sigma_{AB} H_{AB}).
\end{equation}

All in all, given some amount of energy $c$, we are interested in knowing the maximal amount of correlation that can be unitarily created in a non-interacting joint system $AB$ initially thermal at $T_R$. That is, we are interested in the following.

\begin{problem} [Physical Problem] \label{prob:optiproblem}
	Given $\rho_{AB}=\tau_A \otimes \tau_B$, solve
\begin{equation} 
 \max_U \mathcal{I}( U \rho_{AB} U^{\dagger}), \quad \text{s.t. } \Tr (U \rho_{AB} U^{\dagger} H_{AB}) \leq c.
 \end{equation}
\end{problem}

One might be surprised that we are considering the mutual information and average energy of the final state only, rather than their difference with that of the initial state. Both problems are however equivalent. Indeed, the initial value of the mutual information being a constant, does not affect the choice of $U$. The initial value of the average energy can be absorbed into $c$, since the latter can be of any numerical value. And so, solving one problem for any $c$ automatically solves the other for any $c$ as well.\\

Actually, looking at $\mathcal{I}(U \rho_{AB} U^{\dagger})$ more closely, we see that for every unitary
\begin{equation}
S(U \rho_{AB} U^{\dagger})= S(\rho_{AB}).
\end{equation}
The global entropy part of the mutual information therefore also does not affect the choice of $U$. Furthermore, as
\begin{equation}
H_{AB}= H_A \otimes \mathds{1}_B + \mathds{1}_A \otimes H_B,
\end{equation}

We have that for any joint state $\sigma_{AB}$
\begin{equation}
\Tr (\sigma_{AB} H_{AB}) = \Tr (\sigma_A H_A) + \Tr (\sigma_B H_B),
\end{equation}
where $\sigma_A = \Tr_B(\sigma_{AB})$ and $\sigma_B = \Tr_A (\sigma_{AB})$. This means that we can rewrite Problem.~\ref{prob:optiproblem} in terms of local quantities only as

\begin{problem} [Technical Problem] \label{prob:optilocalproblem} Given $\rho_{AB}= \tau_A \otimes \tau_B$, solve
\begin{equation}
\max_U S(\tilde{\rho}_A)+ S(\tilde{\rho}_B), \quad \text{s.t. } \Tr (\tilde{\rho}_A H_A) + \Tr (\tilde{\rho}_B H_B) \leq c,
\end{equation}
where $\tilde{\rho}_A= \Tr_B (U \rho_{AB} U^{\dagger})$ and $\tilde{\rho}_B= \Tr_A (U \rho_{AB} U^{\dagger})$.
\end{problem}

\chapter{Pure State Solution} \label{chap:pure}

The problem having been exposed, we would like to solve a special instance of it, namely the case where $T_R=0$. This corresponds to when the initial state of the joint system $AB$ is pure. Surprisingly, we will be able to solve the problem in full generality when $T_R=0$, that is for any local Hamiltonian $H_A$ and $H_B$. We will jump right into stating the result before explaining its derivation. The fully detailed proof of the result can be found in Appendix A.II of \cite{Bakhshinezhad-2019}.\\

Let $d= \min(d_A,d_B)$ be the minimum of both local dimensions, and let
\begin{align}
\tilde{H}_A &= \sum_{i=0}^{d-1} (E_i^A+E_i^B) \ket{i} \bra{i}_A\\
\tilde{H}_B &= \sum_{i=0}^{d-1} (E_i^A+E_i^B) \ket{i} \bra{i}_B.
\end{align}

Then the following holds.
\begin{theorem}\label{thm:purestate}
If $\rho_{AB}=\ket{00}\bra{00}_{AB}$, i.e., if $T_R=0$, we have that for all $c>0$ there exists a unique $\beta(c) < \beta_R$ such that
\begin{align}
\tilde{\rho}_{\text{opt},A}(c) &= \frac{  e^{- \beta(c) \tilde{H}_A} \Pi_A}{\Tr \left(e^{- \beta(c) \tilde{H}_A} \Pi_A \right)}\\
\tilde{\rho}_{\text{opt},B} (c)&= \frac{  e^{- \beta(c) \tilde{H}_B} \Pi_B}{\Tr \left(e^{- \beta(c) \tilde{H}_B} \Pi_B \right)},\label{eq:rhooptpure}
\end{align}
with $\Pi_A= \sum_{i=0}^{d-1} \ket{i} \bra{i}_A$ and $\Pi_B= \sum_{i=0}^{d-1} \ket{i} \bra{i}_B$, are solutions of Problem~\ref{prob:optilocalproblem}. Furthermore, $\beta(c)$ is uniquely determined by the equation
\begin{equation}
\Tr \left(\tilde{\rho}_{\text{opt},A}(c) \tilde{H}_A \right) =c
\end{equation}
when $c < \frac{1}{d} \sum_{i=0}^{d-1} E_i^A+E_i^B$. When $c \geq \frac{1}{d} \sum_{i=0}^{d-1} E_i^A+E_i^B$, $\beta(c)= 0$.
\end{theorem}

For the following, we will w.l.o.g. assume $d=d_A$. There are two main ingredients that build the proof of Theorem~\ref{thm:purestate}. The first one allows us to have a better grasp at the unitarily achievable local spectra and is due to the well-known Schmidt decomposition for pure states \cite{Nielsen-2010}. This drastically reduces the complexity of the analysis because thanks to it we know that both marginal entropies are equal, i.e.,
\begin{equation}
S(\tilde{\rho}_A)= S(\tilde{\rho}_B).
\end{equation}

We therefore only have to keep track of one spectrum rather than two. The Schmidt decomposition brings in a second advantage as well, it clarifies which marginal spectra are reachable in the unitary orbit, namely all spectra. For our problem, this means that the only constraint on the achievable spectra comes from the amount of achievable energy c, and not from the unitary evolution itself. Carrying the above argument through, the problem can be reformulated as

\begin{equation} \label{eq:maxSrhoV}
\max_{\rho, V} S(\rho), \quad \text{s.t. } \Tr \left[ \rho \left( H_A+V^{\dagger} H_B V \right) \right] \leq c,
\end{equation}
where $\rho$ is a state on $A$ and
\begin{equation}
V=\sum_{i=0}^{d-1} \ket{\phi_i^B} \bra{\phi_i^A},
\end{equation}
with $(\ket{\phi_i^A})_i$ and $(\ket{\phi_i^B})_i$ being the Schmidt bases of $U \rho_{AB} U^{\dagger}$. The second main step of the proof is to realize that all the difficulties have now been shifted to the constraints and to consider the converse problem of Eq.\ref{eq:maxSrhoV} instead. That is, one instead looks at

\begin{equation}\label{eq:converseproblem}
\min_{\rho,V} \left( \underbrace{\Tr(\rho H_A)}_{(*)} + \underbrace{\Tr (\rho V^{\dagger} H_B V)}_{(\#)} \right), \quad \text{s.t. } S(\rho)=\kappa.
\end{equation}

Fixing the spectrum of $\rho$, which temporarily gets rid of the entropy constraint, one can then separately solve $(*)$ and $(\#)$. To solve $(*)$ note that $V: \mathcal{H}_A \rightarrow \mathcal{H}_B$ is not unitary, since 
\begin{equation}
V V^{\dagger}= \sum_{i=0}^{d-1} \ket{\phi_i^B} \bra{\phi_i^B} \neq \mathds{1}_{\mathcal{H}_B},
\end{equation}
but is easily unitarily extendable  by defining vectors $\ket{\phi_d^A}, \dots, \ket{\phi_{d_B-1}^A}$ such that
\begin{equation}
\braket{\phi_i^A}{\phi_j^A} = \delta_{ij}, \quad \forall i,j=0,\dots, d_B-1.
\end{equation}

One can then extend $\mathcal{H}_A$ as $\tilde{\mathcal{H}}_A = \text{span} \{ \ket{\phi_0^A}, \dots, \ket{\phi_{d_B-1}^A} \}$ and define
\begin{equation}
\tilde{V}: \tilde{\mathcal{H}}_A \rightarrow \mathcal{H}_B, \quad \tilde{V}= \sum_{i=0}^{d_B-1} \ket{\phi_i^B} \bra{\phi_i^A}.
\end{equation}

$\tilde{V}$ is by construction unitary and $\tilde{V} \big|_{\mathcal{H}_A} = V$. Similarly one extends $\rho$ to $\tilde{\mathcal{H}}_A$ as
\begin{equation}
\rho \oplus {\bf 0} \big|_{\tilde{\mathcal{H}}_A - \mathcal{H}_A}.
\end{equation}

With this , both $(*)$ and $(\#)$ become instances of the famous passivity problem~\cite{Pusz-1978, Lenard-1978} of which we know the solution. This allows to fully solve the problem of Eq.~\ref{eq:converseproblem}, see Proposition~1 of Appendix A.II of \cite{Bakhshinezhad-2019}. Using the fact that the solution is strictly monotonically parameterized by $\kappa$, see Proposition~2 of Appendix A.II of \cite{Bakhshinezhad-2019}, one then shows that it is also a solution of the original problem, see Proposition~3 of Appendix A.II of \cite{Bakhshinezhad-2019}.

\chapter{General Considerations} \label{chap:cons}

The fact that our problem could be solved in full generality for $T_R=0$ gives hope that it might also be the case for when $T_R >0$. However, for non-zero temperatures, the problem is much more complicated. Indeed, the initial state of $AB$ is for finite background temperatures not pure anymore. This has the dramatic consequence that the Schmidt decomposition technique, which is only valid for pure states, does not help to simplify the problem anymore. Instead, one has to truly deal with the potentially very different local spectra. What is more, even if one drops the energy constraint, figuring out the allowed local spectra within the global unitary orbit is already a challenging problem in itself \cite{Christandl-2006}. With no symmetry at hand to help us break the problem into a simpler one, it therefore seems that a frontal take on the problem might not be our best option. Instead, trying to set bounds on the marginal entropies and finding unitaries that achieve these bounds might be a more fruitful avenue. There is a bound of particular interest that is easily derivable as follows.\\

First of all, note that as 
\begin{equation}
S(\rho_1 \otimes \rho_2)= S(\rho_1)+S(\rho_2),
\end{equation}

and $\Tr(\rho_1 \otimes \rho_2 H_{AB})= \Tr (\rho_1 H_A) + \Tr (\rho_2 H_B)$, by defining 
\begin{equation}
\tilde{\sigma} = \tilde{\rho}_A \otimes \tilde{\rho}_B,
\end{equation}
where we remind the reader that $\tilde{\rho}_A= \Tr_B (U \rho_{AB} U^{\dagger})$ and $\tilde{\rho}_B= \Tr_A (U \rho_{AB} U^{\dagger})$, we can rewrite Problem~\ref{prob:optilocalproblem} as

\begin{align}
&\max_U S(\tilde{\sigma}), \quad \text{s.t. } \Tr (\tilde{\sigma} H_{AB}) \leq c \\
\leq & \max_{\rho} S(\rho), \quad \text{s.t. } \Tr (\rho H_{AB}) \leq c \\
=& S \left(\tau_A (\beta(c)) \right) + S \left(\tau_B(\beta(c)) \right), \label{eq:entropybound}
\end{align}

where for the first inequality we dropped the intricate dependence of $\tilde{\sigma}$ on $U$ and maximized over all allowed states $\rho$ on $AB$ instead. The last equality is Jaynes principle. What this tells us is that the sum of the local entropies of our problem is upper bounded by the sum of the entropies of two local thermal states at the same inverse temperature $\beta(c)$. Since the joint initial state is the thermal state at the inverse background temperature $\beta_R$, and since the entropy of the thermal state is a strictly increasing function of its temperature, it directly follows that $\beta(c) \leq \beta_R$ must hold. Furthermore, $\beta(c)$ is uniquely determined by
\begin{equation}
\Tr(\tau_A (\beta(c)) H_A) + \Tr(\tau_B(\beta(c)) H_B ) =c,
\end{equation} 
such that scanning through all the $c > 0$ is equivalent to scanning through all the inverse temperatures $0 \leq \beta(c) < \beta_R$.
The next natural question to ask is if this bound is unitarily attainable from $\rho_{AB} = \tau_A \otimes \tau_B$. An excellent candidate for saturating the bound is a state with marginals thermal at $\beta(c)$. The bad news is that in full generality such a state does not lie in the unitary orbit of $\rho_{AB} = \tau_A \otimes \tau_B$, meaning that there exist local Hamiltonians $H_A$ and $H_B$, background temperatures $\beta_R$, and $c \geq 0$ for which the marginals of $U \rho_{AB} U^{\dagger}$ cannot be thermal at the inverse temperature $\beta(c)$. Our solution for the case of $T_R=0$ of Chapter~\ref{chap:pure} already suggests that this might be the case. Indeed, if $d_A < d_B$ then $d=d_A$ and $\tilde{\rho}_{\text{opt},B}(c)$ of Eq.~\ref{eq:rhooptpure} is not full rank. It therefore cannot be thermal. However, note that, since we have not proven that $\tilde{\rho}_{\text{opt},A}(c)$ and $\tilde{\rho}_{\text{opt},B}(c)$ are the unique solutions of Problem~\ref{prob:optilocalproblem} for a given $c$ at $T_R=0$, this does not form a counter example. There might indeed be another $\beta < \beta_R$ fulfilling
\begin{equation}
S(\tau_A(\beta)) + S(\tau_B(\beta)= S(\tilde{\rho}_{\text{opt},A}(c)) + S(\tilde{\rho}_{\text{opt},B}(c)),
\end{equation}
and 
\begin{equation}
\Tr(\tau_A(\beta) H_A)+ \Tr(\tau_B(\beta) H_B) \leq c.
\end{equation}

To construct a proper counter example we will make use of (the left hand side of) the triangle inequality for the von Neumann entropy, which for an arbitrary state $\sigma_{AB}$ on $AB$ can be written as 
\begin{equation}
 \lvert S(\sigma_A) - S(\sigma_B) \rvert \leq S(\sigma_{AB}).
\end{equation}

Assuming that $\sigma_{AB}= U \rho_{AB} U^{\dagger}$ has thermal marginals at $\beta(c)$, i.e., that 
\begin{align}
\sigma_A &= \tau_A (\beta(c)),\\
\sigma_B &= \tau_B (\beta(c)),
\end{align}
and using that 
\begin{equation}
S(\sigma_{AB})=S(\rho_{AB})=S(\tau_A)+S(\tau_B),
\end{equation}

we have that

\begin{equation}\label{eq:diffentropy}
\lvert S(\tau_A(\beta(c))) - S(\tau_B(\beta(c))) \rvert \leq S(\tau_A) + S(\tau_B)
\end{equation}
must hold true.
But it is well possible to keep the right-hand side of Eq.~\ref{eq:diffentropy} small, by for example choosing a big $\beta_R$, i.e. a small room temperature, while forcing its left-hand side to be big no matter $\beta(c) \leq \beta_R$, by for example having very different local Hamiltonians, maybe even of different dimensions. See \cite[Section III]{Vitagliano-2018} for an explicit example.\\

This shows that given a $c \geq 0$ and some local Hamiltonians $H_A$ and $H_B$, one cannot hope to always be able to find a unitary $U$ such that $U \rho_{AB} U^{\dagger}$ has thermal marginal at $\beta(c)$. Note, however, that this does not imply that the bound on the sum of the local entropies of Eq.~\ref{eq:entropybound} is not reachable in general. It only shows that if reached, the state doing so does not have to have marginals thermal at the same temperature. \\

Introducing some symmetry back into the problem might nevertheless allow us to prove the attainability of the bound for a large class of systems. Inspired by this, we ask what happens if we restrict $H_A$ and $H_B$ such that $H_A =H_B$. In that case Eq.~\ref{eq:diffentropy} is trivially fulfilled since its left-hand side vanishes. One may therefore hope to solve the problem for that restricted class of Hamiltonians by looking for global states fulfilling the constraint that have equal thermal marginals. This brings us to the following conjecture.

\begin{conjecture} [General Conjecture] \label{conj:general}
Given a $c >0$ and local Hamiltonians $H_A=H_B$, the solution of Problem~\ref{prob:optilocalproblem} is given by
\begin{equation}
2 S(\tau_A(\beta(c)).
\end{equation}
\end{conjecture}

Conjecture~\ref{conj:general} is already known to hold true for equally gapped Hamiltonians $H=H_A=H_B$ \cite{Huber-2015, Jevtic-2012}. This in particular solves the qubit case, i.e., when $d$, the local dimension, is $2$. We will here prove the conjecture for $d=3$ and $d=4$. We will also set some conditions for $d>4$ that if satisfied are sufficient to prove the conjecture for all $H=H_A=H_B$. All the above mentioned proofs are based on the same idea, namely proving the existence of unitaries $U$ that transform $\rho_{AB}$ into a state with equal thermal marginals. We will refer to these unitaries as symmetrically thermalizing unitaries (STU). Formally we define the following.

\begin{definition}[STU($\beta, \beta'$)]
Given inverse temperature $\beta, \beta' \in [-\infty, +\infty]$, a  unitary $U$ on $AB$ is called a symmetrically thermalizing unitary from $\beta$ to $\beta'$, written STU($\beta$, $\beta'$), if the following holds
\begin{align}
\Tr_B (U \tau_{AB} (\beta) U^{\dagger})&=\tau_A(\beta')\\
\Tr_A (U \tau_{AB} (\beta) U^{\dagger})&=\tau_B(\beta').
\end{align}
\end{definition} 

On a more technical level, what we really conjecture is that STU($\beta$,$\beta'$) exist between all pairs of inverse temperatures $\beta$ and $\beta'$ such that $\beta' \leq \beta$. That is,

\begin{conjecture}[Technical Conjecture] \label{conj:technical}
Given a pair of local Hamiltonians $H_A=H_B$ and an inverse background temperature $\beta_R$, there exists a STU($\beta_R$, $\beta'$) for all $\beta' \leq \beta_R$.
\end{conjecture}

Note that as scanning through all the $\beta' \leq \beta_R$ of Conjecture~\ref{conj:technical} effectively scans through all the $c \geq 0$ of Conjecture~\ref{conj:general}, proving Conjecture~\ref{conj:technical} proves Conjecture~\ref{conj:general}. However, disproving Conjecture~\ref{conj:technical} does not disprove Conjecture~\ref{conj:general} since it might well be that states with different marginals nevertheless achieve the bound of Conjecture~\ref{conj:general}.

\chapter{Symmetric Mixed State Framework} \label{chap:symmetric}

\section{General Framework} \label{sec:framework}

Now that we have set our technical problem, we would like to define a framework valid in all dimensions that will allow us to tackle it. In Sections~\ref{sec:marginalapproach}, Section~\ref{sec:passingnormapproach} and Section~\ref{sec:geometric} we will use this framework to prove the existence of STUs as in Conjecture~\ref{conj:technical} for all local Hamiltonians $H_A=H_B$ of dimension $d=3$ and $d=4$. The framework will consist in characterizing a class of unitary transformations that have the two following  properties.

\begin{enumerate}
\item Leave the marginals diagonal in the energy eigenbasis, \label{enum:diagprop}
\item Transform the marginals equally. \label{enum:equaltrafo}
\end{enumerate}

Since our initial state is diagonal, its marginals also are. We furthermore aim at producing thermal marginals, which per definition are diagonal in the energy eigenbasis. This explains why property~\ref{enum:diagprop} is desirable. Working with a class of unitaries that have this property will furthermore allow us to solely focus on the diagonal entries of the marginals without having to worry about the off-diagonal entries. This reduces the complexity of the problem and renders it much more tractable. Note, however, that given our initial state $\rho_{AB}=\tau_A \otimes \tau_B$, the entire set of unitaries fulfilling property~\ref{enum:diagprop}, i.e.,

\begin{equation}
\mathcal{A}= \left\{ U: \mathcal{H}_{AB} \rightarrow \mathcal{H}_{AB} \mid \left[ \Tr_A \left(U \rho_{AB} U^{\dagger} \right) \right]_{ij}= \left[ \Tr_B \left(U \rho_{AB} U^{\dagger} \right)\right]_{ij}=0, \quad \forall i \neq j\right\},
\end{equation}

is in general quite complex to characterize. The reason is that given a state $\sigma_{AB}$, 

\begin{equation}
\left[ \Tr_B \left( \sigma_{AB} \right)\right]_{ij}= \sum_{k} \left[ \sigma_{AB} \right]_{ik, jk}=0
\end{equation}
does not imply that each complex entry $\left[ \sigma_{AB} \right]_{ik, jk}$ vanishes. And keeping track of the unitarily attainable entries $\left[ \sigma_{AB} \right]_{ik, jk}$ such that $\sum_{k} \left[ \sigma_{AB} \right]_{ik, jk}=0$ is no easy task. The same of course also holds for $\Tr_A$. This is why we decide to trade generality for tractability and restrict ourselves on a subset of $\mathcal{A}$, call it $\mathcal{B}$, that keeps every $\left[ \rho_{AB} \right]_{ik, jk}=\left[ \rho_{AB} \right]_{ki, kj}=0$ if $i \neq j$. Once property~\ref{enum:diagprop} is ensured, property~\ref{enum:equaltrafo} will make sure that both marginals are equal. This will then give us a set of allowed marginal transformations and our goal will be to prove that $\tau(\beta)$ is within this set for all $\beta \leq \beta_R$.\\

Let us now define our set more precisely. Let $d$ be the dimension of our local Hilbert space, i.e., $d=\dim(\mathcal{H}_A) = \dim(\mathcal{H}_B)$. For $i \in \{ 0, \dots, d-1\}$ let

\begin{equation} \label{eq:Hspacepartition}
\mathcal{H}_i = \text{span} \{ \ket{j, j+i}\}_{j=0}^{d-1}.
\end{equation}

The unitaries that we will be interested in will be unitaries $U_{\mathcal{B}}$ of the form

\begin{equation} \label{eq:oplusunitary}
U_{\mathcal{B}} = \oplus_{i=0}^{d-1} U_i, \quad \text{s.t. } U_i: \mathcal{H}_i \rightarrow \mathcal{H}_i.
\end{equation}

Note that theses unitaries do not create off diagonal entries in the marginals when applied to a diagonal state. Indeed

\begin{lemma} \label{lemma:diagmarginals}
If $\sigma_{AB}$ is diagonal in the energy eigenbasis, then so are
\begin{equation}
\Tr_{B/A} \left(U_{\mathcal{B}} \sigma_{AB} U_{\mathcal{B}}^{\dagger} \right)
\end{equation}
for all $U_{\mathcal{B}}$ as in Eq.~\ref{eq:oplusunitary}.
\end{lemma}

\begin{proof}
Note that $U_{\mathcal{B}}$ only has non-zero elements of the form of $\ket{j, j+i} \bra{k, k+i}$ for $k,j,i=0,\dots,d-1$. $\sigma_{AB}$ being diagonal has non-zero element of the form $\ket{mn} \bra{mn}$, $m,n= 0, \dots, d-1$ only. And so $U_{\mathcal{B}} \sigma_{AB} U_{\mathcal{B}}^{\dagger}$ has elements of the form

\begin{align}
&\ket{ j ,j+i} \bra{k, k+i} \ket{mn} \bra{mn} \ket{p, p+l} \bra{q, q+l} \\
&= \ket{j ,j+i} \delta_{km} \bra{q, q+l} \delta_{k+i, n} \delta_{mp} \delta_{n, p+l} \\
&= \ket{j, j+i}\bra{q ,q+l} \delta_{mk} \delta_{n, k+i} \delta_{pk} \underbrace{\delta_{k+i ,k+l}}_{=\delta_{il}} \\
&=\ket{j ,j+i}\bra{q ,q+i} \delta_{mk} \delta_{n ,k+i} \delta_{pk}.
\end{align}
And so as desired

\begin{align}
\Tr_B (\ket{j ,j+i} \bra{q ,q+i})&= \ket{j} \bra{j}\\
\Tr_A (\ket{j ,j+i} \bra{q ,q+i})&= \ket{j+i} \bra{j+i},
\end{align}
meaning that the marginals of $U_{\mathcal{B}} \sigma_{AB} U_{\mathcal{B}}^{\dagger}$ only have diagonal non-vanishing entries.
\end{proof}

For visual purposes, it is practical to represent $\rho_{AB}$ in the reordered energy eigenbasis

\begin{align}
\mathfrak{B}=  ( &\ket{00}, \ket{11}, \dots, \ket{d-1 ,d-1}, \\
&   \ket{01} ,\ket{12}, \dots, \ket{d-1, 0},\\
&  \dots  \\
&  \ket{0, d-1} ,\ket{10}, \dots, \ket{d-1, d-2} ).
\end{align}

Doing so, we get for $d=3$

\begin{equation}
\left[\rho_{AB} \right]_{\mathfrak{B}}=\begin{pmatrix}
  \begin{matrix}
  p_{00} & &\\
   & p_{11}&\\
  & & p_{22}
  \end{matrix}
  & \rvline & & & \\
\cline{1-4}
  & \rvline & \begin{matrix}
  p_{01} & &\\
   & p_{12}&\\
  & & p_{20}
  \end{matrix} & \rvline & \\
\cline{2-5}
  & &  & \rvline &\begin{matrix}
  p_{02} & &\\
   & p_{10}&\\
  & & p_{21}
  \end{matrix} \\
  \end{pmatrix},
\end{equation}

where $p_{ij} = \bra{ij} \rho_{AB} \ket{ij}$. With this representation we directly see the action of $U_{\mathcal{B}} = U_0 \oplus U_1 \oplus U_3$ on $\rho_{AB}$. Indeed, $U_0$ acts on the top left block, $U_1$ on the middle block and $U_2$ on the bottom right one. In general, as $\rho_{AB}$ is diagonal in the energy eigenbasis, 
\begin{equation}
\rho_{AB}= \oplus_{i=0}^{d-1} \rho_{AB} \big|_{\mathcal{H}_i},
\end{equation}
 and
\begin{equation}
U \rho_{AB} U^{\dagger} = \oplus_{i=0}^{d-1} U_i \rho_{AB} \big|_{\mathcal{H}_i} U_i^{\dagger}.
\end{equation}

We now focus our attention on the diagonal of the $i^{\text{th}}$ block

\begin{equation}
r_i = \mathfrak{D} (\rho_{AB} \big|_{\mathcal{H}_i}) = \begin{pmatrix}
p_{0i}\\
p_{1 i+1}\\
\vdots\\
p_{d-1 i+d-1}
\end{pmatrix},
\end{equation}

where the indices are to be understood as modulo $d$. Doing so, we see that $\tilde{r}_i$, the transformed diagonal under the action of $U_i$, can be written as 

\begin{equation} \label{eq:rtrafo}
\tilde{r}_i = \mathfrak{D} ( U_i \rho_{AB} \big|_{\mathcal{H}_i} U_i^{\dagger} ) = M_i r_i,
\end{equation}
   
with $M_i= \left( [U_i]_{kl}\right)_{k,l=0}^{d-1}$. Remember that the off diagonal elements of each block do not contribute to the reduced states and that as such the $\tilde{r}_i$ contain all the information about our transformation. More precisely, with this notation we can very concisely keep track of the transformations that the unitaries of $\mathcal{B}$ induce on the reduced states. Indeed, from Lemma~\ref{lemma:diagmarginals}, the transformed marginal of system $A$ is diagonal and is therefore fully described by

\begin{align}
\tilde{\mathbf{p}}_A &= \begin{pmatrix}
\tilde{p}_{00} + \tilde{p}_{01} + \dots + \tilde{p}_{0 d-1}\\
\tilde{p}_{11} + \tilde{p}_{12} + \dots + \tilde{p}_{10}\\
\vdots\\
\tilde{p}_{d-1 d-1} + \tilde{p}_{d-1 0} + \dots + \tilde{p}_{d-1 d-2}
\end{pmatrix}\\
&= \begin{pmatrix} \tilde{p}_{00}\\
\tilde{p}_{11}\\
\vdots\\
\tilde{p}_{d-1 d-1}
\end{pmatrix}
+
\begin{pmatrix} \tilde{p}_{01}\\
\tilde{p}_{12}\\
\vdots\\
\tilde{p}_{d-1 0}
\end{pmatrix}
+ \dots+
\begin{pmatrix} \tilde{p}_{0d-1}\\
\tilde{p}_{10}\\
\vdots\\
\tilde{p}_{d-1 d-2}
\end{pmatrix} = \sum_{i=0}^{d-1} \tilde{r}_i,
\end{align}
where $\tilde{p}_{ij} = \bra{ij} U_{\mathcal{B}} \rho_{AB} U^{\dagger}_{\mathcal{B}} \ket{ij}$. Similarly

\begin{align}
\tilde{\mathbf{p}}_B &= \begin{pmatrix}
\tilde{p}_{00} + \tilde{p}_{d-1 0 } + \dots + \tilde{p}_{1 0}\\
\tilde{p}_{11} + \tilde{p}_{01} + \dots + \tilde{p}_{21}\\
\vdots\\
\tilde{p}_{d-1 d-1} + \tilde{p}_{d-2 d-1 } + \dots + \tilde{p}_{0 d-1}
\end{pmatrix}\\
&= \begin{pmatrix} \tilde{p}_{00}\\
\tilde{p}_{11}\\
\vdots\\
\tilde{p}_{d-1 d-1}
\end{pmatrix}
+
\begin{pmatrix} \tilde{p}_{d-10}\\
\tilde{p}_{01}\\
\vdots\\
\tilde{p}_{d-2 d-1}
\end{pmatrix}
+ \dots+
\begin{pmatrix} \tilde{p}_{1 0}\\
\tilde{p}_{21}\\
\vdots\\
\tilde{p}_{0 d-1}
\end{pmatrix} = \sum_{i=0}^{d-1} \Pi^i \tilde{r}_i,
\end{align}

where $\Pi$ is the cyclic permutation matrix that ``pushes down'' every element of a vector once. That is,
\begin{equation}
\Pi = (\Pi_{ij}), \quad \Pi_{ij} = \delta_{i j+1 \, \text{mod}\, d}.
\end{equation}

And so with Eq.~\ref{eq:rtrafo} we have

\begin{align}
\tilde{\mathbf{p}}_A &= \sum_{i=0}^{d-1} M_i r_i,\\
\tilde{\mathbf{p}}_B &= \sum_{i=0}^{d-1} \Pi^i M_i r_i.
\end{align}

Now that we have ensured that our marginals transform in a way that keeps them diagonal with respect to the energy eigenbasis, we turn our attention to the second property we would like our transformation to have, property~\ref{enum:equaltrafo}. That is, we will see how to make sure that $\tilde{\mathbf{p}}_A= \tilde{\mathbf{p}}_B$. For this, it is convenient to separate the first term of $\sum_{i=0}^{d-1}$ and split the rest in roughly two equal parts, depending on if $d$ is even of odd. Let $k= \frac{d-1}{2}$ if $d$ is odd and $k=\frac{d}{2}$ if $d$ is even. Then,

\begin{align}
    \tilde{\mathbf{p}}_A  &=
    \begin{cases}
        M_{0} r_0 + \sum\limits_{i=1}^{k} (M_{i} r_{i} + M_{d-i} r_{d-i})
        & d\ \text{odd}\\[2mm]
        M_0 r_0 + \sum\limits_{i=1}^{k-1} (M_i r_{i} + M_{d-i} r_{d-i})+M_{d/2} r_{d/2}
        & d\ \text{even}
    \end{cases}.
    \label{gen evol marginal A}
\end{align}

This can be written compactly as:

\begin{equation}
\tilde{\mathbf{p}}_A  = M_{0} r_0 + \sum_{i=1}^{k} ( \lfloor \tfrac{2 i}{d} \rfloor +1)^{-1} (M_{i} r_{i} + M_{d-i} r_{d-i}),
\end{equation}
where $\lfloor x \rfloor$ denotes the floor function of $x$ and the prefactor $ ( \lfloor \tfrac{2 i}{d} \rfloor +1)^{-1}$ is equal to 1 unless $d$ is even and $i=k$, in which case it is $\frac{1}{2}$. Similarly, using $\Pi^{d-i} = \Pi^{-i}$, we get
\begin{equation}
\tilde{\mathbf{p}}_B  = M_{0} r_0 + \sum_{i=1}^{k} ( \lfloor \tfrac{2 i}{d} \rfloor +1)^{-1} ( \Pi^i M_{i} r_{i} + \Pi^{-i} M_{d-i} r_{d-i}).
\end{equation}

Now using the symmetry $p_{ij} = p_{ji}$ of the initial state, we have

\begin{align}
r_{d-i}= \begin{pmatrix}
p_{0 d-i}\\
p_{1 d-i+1}\\
\vdots\\
p_{d-1 d-i+1}
\end{pmatrix}
=\begin{pmatrix}
p_{d-i 0}\\
p_{d-i+1 1}\\
\vdots\\
p_{d-i+1 d-1}
\end{pmatrix}
=
\begin{pmatrix}
p_{-i 0}\\
p_{-i+1 1}\\
\vdots\\
p_{-i+1 d-1}
\end{pmatrix} = \Pi^i \begin{pmatrix}
p_{0 i}\\
p_{1 i+1}\\
\vdots\\
p_{d-1 d-1+i}
\end{pmatrix} = \Pi^i r_i.
\end{align}

This gives us

\begin{align}
\tilde{\mathbf{p}}_A  &= M_{0} r_0 + \sum_{i=1}^{k} ( \lfloor \tfrac{2 i}{d} \rfloor +1)^{-1} (M_{i} + M_{d-i} \Pi^i) r_{i},\\
\tilde{\mathbf{p}}_B  &= M_{0} r_0 + \sum_{i=1}^{k} ( \lfloor \tfrac{2 i}{d} \rfloor +1)^{-1} ( \Pi^i M_{i} + \Pi^{-i} M_{d-i} \Pi^i) r_{i}.
\end{align}

So one way to ensure that $\tilde{\mathbf{p}}_A = \tilde{\mathbf{p}}_B$ is to demand
\begin{equation}
M_{d-i} = \Pi^i M_i \Pi^{-i}, \quad \text{for } i=1,\dots,k.
\end{equation}
Indeed, that way
\begin{align}
\Pi^i M_i &= M_{d-i} \Pi^i,\\
\Pi^{-i} M_{d-1} \Pi^i &=M_i.
\end{align}
This fixes the matrices $M_{k+1}, \dots, M_{d-1}$, while keeping $M_0, \dots, M_k$ as free variables. The end result is
\begin{equation} \label{eq:trafo}
\tilde{\mathbf{p}}=\tilde{\mathbf{p}}_A=\tilde{\mathbf{p}}_B  = M_{0} r_0 + \sum_{i=1}^{k} ( \lfloor \tfrac{2 i}{d} \rfloor +1)^{-1} ( \mathds{1} + \Pi^i) M_i r_{i}.
\end{equation}

Note that since $\tilde{\mathbf{p}}_A = \tilde{\mathbf{p}}_B$, we will for convenience drop the subscripts and from now on talk about the marginal transformation $\tilde{\mathbf{p}}$. To prove Conjecture~\ref{conj:technical}, we therefore must show that for any $\beta' < \beta_R$, there exist unistochastic matrices $M_0, \dots, M_k$ such that $\tilde{\mathbf{p}}= \mathfrak{D}(\tau(\beta'))$. We have therefore reduced the problem of Conjecture~\ref{conj:technical} to the following

\begin{problem} \label{prob:findbetaprime}
Given $\tau(\beta_R)$, the thermal state of some Hamiltonian $H= \sum_{i=0}^{d-1} E_i \ket{i} \bra{i}$ at inverse background temperature $\beta_R$, does there exist for all $\beta' < \beta_R$ unistochastic matrices $M_0,\dots,M_k$ such that
\begin{equation}\label{eq:findbetaprime}
\mathfrak{D}(\tau(\beta'))=M_{0} r_0 + \sum_{i=1}^{k} ( \lfloor \tfrac{2 i}{d} \rfloor +1)^{-1} ( \mathds{1} + \Pi^i) M_i r_{i},
\end{equation}
where $r_i=(p_{0i}, \dots, p_{(d-1)(d-1+i)})$, with $p_{ij}=\bra{ij} \tau(\beta_R) \otimes \tau(\beta_R) \ket{ij}$, and where $k= \frac{d-1}{2}$ if $d$ is odd and $k=\frac{d}{2}$ if $d$ is even.
\end{problem}

\section{Remarks on the General Framework} \label{sec:remarks}

\subsection{Properties of the $M_i$'s}
The $M_i$ matrices of Eq.~\ref{eq:trafo}, that dictate the transformation of our marginals, are members of a special class of matrices called unistochastic. A unistochastic matrix is a matrix that can be written in the form $ \left( \lvert u_{ij} \rvert^2 \right)$ for some unitary matrix $\left( u_{ij} \right)$. That is,
\begin{definition}[Unistochastic matrix]
A matrix $M = (m_{ij}) \in \mathbb{R}^{d \times d}$ is called unistochastic if there exists a unitary matrix $U= (u_{ij}) \in \mathbb{C}^{d \times d}$ such that
\begin{equation}
m_{ij}= \lvert u_{ij} \rvert^2, \quad \forall i,j=0,\dots,d-1.
\end{equation}
\end{definition}

Unistochastic matrices are a subclass of another special kind of matrices called doubly stochastic. Doubly stochastic matrices are matrices with positive real entries such that each row and column sums to 1. That is,

\begin{definition}[Doubly stochastic matrix]
A matrix $M = (m_{ij}) \in \mathbb{R}^{d \times d}$ is called doubly stochastic if
\begin{enumerate}
\item $m_{ij} \geq 0, \quad \forall i,j =0,\dots,d-1$,
\item $\sum_{i=0}^{d-1} m_{ij} =1, \quad \forall j=0, \dots, d-1$,
\item $\sum_{j=0}^{d-1} m_{ij} =1, \quad \forall i=0, \dots, d-1$.
\end{enumerate}
\end{definition}

One immediately sees that a unistochastic matrix is doubly stochastic. Indeed, each of its entry is positive and
\begin{align}
\sum_{i=0}^{d-1} \lvert u_{ij} \rvert^2 &= \left( U U^{\dagger} \right)_{ii} =1, \quad \forall j=0, \dots, d-1,\\
\sum_{j=0}^{d-1} \lvert u_{ij} \rvert^2 &= \left( U^{\dagger} U \right)_{jj} =1, \quad \forall i=0, \dots, d-1.
\end{align}

However, not every doubly stochastic matrix is unistochastic, see \cite[Chapter 2.G]{Marshall-2011} and \cite{Bengtsson-2005}. Quite peculiarly though, the action of every doubly stochastic matrix on a vector is the same as the action of every unistochastic matrix on a vector. This is due to a combination of two results of majorization theory. The first was obtained by Hardy, Littlewood and P\'olya in 1929 \cite[Chapter~2.B]{Marshall-2011} and reads

\begin{theorem} [Hardy, Littlewood and P\'olya] \label{thm:Hardy}
Let $x,y \in \mathbb{R}^d$. Then $x \prec y$ iff there exists a doubly stochastic matrix $M$ such that
\begin{equation}
x = M y.
\end{equation}
\end{theorem}

The second result is Horn's Theorem, Theorem~\ref{thm:Horn}, that was presented in Section~\ref{sec:remarkscoh}. Combining both results we directly get

\begin{corollary}
Let $y \in \mathbb{R}^d$ and let $M \in \mathbb{R}^{d \times d}$ be doubly stochastic. Then there exists a unistochastic matrix $M_y$ such that
\begin{equation}
M_y y = M y.
\end{equation}
\end{corollary}

\begin{proof}
Let $x= M y$. Then from Theorem~\ref{thm:Hardy} $x \prec y$. From Horn's Theorem, Theorem~\ref{thm:Horn}, there exists a unitary matrix $U=(u_{ij})$ such that
\begin{equation}
x_i = \sum_{j=0}^{d-1} \lvert u_{ij} \rvert^2 y_j.
\end{equation}
Let $M_y = \left( \lvert u_{ij} \rvert^2 \right)_{ij}$. Then $x= M_y y$ as desired.
\end{proof}

For us this means that we can relax the condition on our $M_i$ matrices by demanding that they are doubly stochastic matrices only. That is,
\begin{lemma}
It suffices to find doubly stochastic matrices $M_0, \dots, M_k$ fulfilling Eq.~\ref{eq:findbetaprime} to solve Problem~\ref{prob:findbetaprime}.
\end{lemma}
Lastly, we would like to point out a result derived by Birkhoff in 1946 \cite[Chapter~2.A]{Marshall-2011} that sheds some light on the geometry of doubly stochastic matrices. The result reads as follows.

\begin{theorem}[Birkhoff] \label{thm:Birkhoff}
The permutation matrices are the extremal points of the set of doubly stochastic matrices.
Moreover, the set of doubly stochastic matrices is the convex hull of
the permutation matrices.
\end{theorem}

\subsection{Choice of $\mathcal{H}_i$}

We would here like to briefly comment on our choice of partitioning our Hilbert space $\mathcal{H}_{AB}$ as in Eq.~\ref{eq:Hspacepartition}, i.e., as
\begin{equation}
\mathcal{H}_i = \text{span} \{ \ket{j ,j+i}\}_{j=0}^{d-1}.
\end{equation}

What makes this decomposition special is that it ensures that the off diagonal elements of the marginals are untouched, provided that our state is block diagonal in $\mathfrak{B}$ to start with. If one view the indices $i$ and $j$ in a square matrix as

\begin{equation}
\begin{pmatrix}
00 & 01 & \cdots & 0 (d-1)\\
10 & 11 & \cdots & 1 (d-1)\\
\vdots & \vdots & & \vdots\\
(d-1) 0 & (d-1) 1 & \cdots & (d-1) (d-1)
\end{pmatrix},
\end{equation}

the choice of $\mathcal{H}_i$ that we made corresponds to picking the diagonal elements for $\mathcal{H}_0$, the superdiagonal elements for $\mathcal{H}_1$, that is the elements right above the diagonal, and so on. Note that this is done is a cyclic manner so that $(d-1) 0$ is for us also part of the superdiagonal. What makes this picking of indices particular is that each $\mathcal{H}_i$ picks exactly one index per row and column. Such a picking of index corresponds to a choice of Latin square \cite{Wallis-2017}. The correspondence is seen by writing $i$ in place of the elements of $\mathcal{H}_i$ in the above matrix representation. For dimension 3 doing this delivers
\begin{equation}
\begin{pmatrix}
0 & 1 & 2\\
2 & 0 & 1\\
1 & 2  & 0
\end{pmatrix}.
\end{equation}

In dimension 3 there are a total of 12 Latin squares. But since relabeling corresponds for us to the same choice, there really is one other possible choice in dimension 3, namely

\begin{equation}
\begin{pmatrix}
0 & 1 & 2\\
1 & 0 & 2\\
2 & 0  & 1
\end{pmatrix},
\end{equation}

corresponding to choosing the following separation of our Hilbert space
\begin{align}
\tilde{\mathcal{H}}_0 &= \text{span} \{ \ket{00}, \ket{12}, \ket{21} \},\\
\tilde{\mathcal{H}}_1 &= \text{span} \{ \ket{01}, \ket{10}, \ket{22} \},\\
\tilde{\mathcal{H}}_2 &= \text{span} \{ \ket{02}, \ket{11}, \ket{20} \}.
\end{align}

The two choices correspond to either picking the diagonal or anti-diagonal in the above matrix representation as being part of one space. For bigger dimensions the number of Latin squares is drastically bigger. We built our framework on a particular Latin square but one could adapt it to any Latin square straightforwardly since the essence of our framework relies on the fact that our picking is of a Latin square type, rather than on the specific pick itself.
\section{Majorized Marginal Approach} \label{sec:marginalapproach}
We want here to use the framework exposed in Section~\ref{sec:framework} to prove Conjecture~\ref{conj:technical} for $d=3$. This approach is based on the following observation. Looking at the transformation of our marginal as dictated by Eq.~\ref{eq:trafo}, if one somehow manages to invert the order of $(\mathds{1}+\Pi^i)$ and $M_i$, then by choosing $M_i=M$ for all $i=0,\dots, k$, we can reach any $\tilde{\mathbf{p}}$ such that

\begin{equation} \label{eq:Mreach}
\tilde{\mathbf{p}}= M \left( r_0 + \sum_{i=1}^{k} ( \lfloor \tfrac{2 i}{d} \rfloor +1)^{-1} ( \mathds{1} + \Pi^i) r_{i} \right) = M \mathbf{p}.
\end{equation}
Then using Theorem~\ref{thm:Hardy}, the above means that we can reach any $\tilde{\mathbf{p}}$ such that
\begin{equation}
\tilde{\mathbf{p}} \prec \mathbf{p} = \mathfrak{D}(\tau(\beta_R)).
\end{equation}
In particular, $\mathfrak{D}(\tau(\beta'))$ is reachable for all $\beta' \leq \beta_R$ since 
\begin{equation}
\mathfrak{D}(\tau(\beta')) \prec \mathfrak{D}(\tau(\beta_R)), \quad  \text{for all } \beta' \leq \beta_R.
\end{equation}

The most straightforward way of inverting $(\mathds{1}+\Pi^i)$ and $M$, is to assume that they commute. This is however a big restriction that not all doubly stochastic matrices $M$ fulfill. And so taking this avenue, one cannot make use of Theorem~\ref{thm:Hardy} in Eq.~\ref{eq:Mreach} and the above argument breaks down. In fact, only circulant doubly stochastic matrices fulfill the commutation restriction and those do not suffice to reach any $\tilde{\mathbf{p}}=\mathfrak{D}(\tau(\beta'))$, as was noted in the original formulation of the problem already \cite{Huber-2015}. One can nevertheless salvage the original idea in dimension 3 by relaxing the commuting condition in a way that still allows for the desired inversion. This is the content of the following.

\begin{lemma} \label{lemma:commute}
Let $M \in \mathbb{R}^{3 \times 3}$ be a doubly stochastic matrix. Then there exists a doubly stochastic matrix $\tilde{M} \in \mathbb{R}^{3 \times 3}$ such that
\begin{equation}
M (\mathds{1}+\Pi) = (\mathds{1} + \Pi) \tilde{M}.
\end{equation}
\end{lemma}

\begin{proof}
Using Birkhoff's Theorem, Theorem~\ref{thm:Birkhoff},
\begin{equation}
M = \sum_{i=1}^{6} \alpha_i P_i,
\end{equation}

where $P_1, \dots, P_6$ denote the six permutation matrices in dimension 3 and $\alpha_i \in [0,1]$. With the above it suffices to prove that the statement is true for $M$ being a permutation matrix. For the 3 circulant permutation matrices the statement is direct. For the 3 remainder ones, this is done by direct calculation, see Appendix~A.IV. of~\cite{Bakhshinezhad-2019}.
\end{proof}

As in dimension 3
\begin{equation}
\tilde{\mathbf{p}}= M_0 r_0 + (\mathds{1}+\Pi) M_1 r_1,
\end{equation}

by choosing $M=M_0$, Lemma~\ref{lemma:commute} tells us that there exists a doubly stochastic matrix $M_1$ such that
\begin{equation}
M_0 (\mathds{1}+\Pi) = (\mathds{1} + \Pi) M_1.
\end{equation}

Therefore
\begin{equation}
\tilde{\mathbf{p}}= M_0 \mathbf{p},
\end{equation}
and as now $M_0$ may be any doubly stochastic matrix, one can make use of Theorem~\ref{thm:Hardy} and the original argument carries through, proving Conjecture~\ref{conj:technical} for $d=3$. We hence proved the following.

\begin{theorem}
For dimension 3 the \enquote{majorized marginal} approach allows to reach any marginal $\tilde{\mathbf{p}}$ thermal at $\beta' < \beta_R$ and therefore proves Conjecture~\ref{conj:technical} for $d=3$.
\end{theorem}
 Since nothing in the argument crucially used that the original state was thermal and since local unitaries are allowed, we can actually prove a more general statement, see~\cite{Bakhshinezhad-2019} for a detailed proof.

\begin{theorem}
Let $\mathcal{H}$ be a 3 dimensional Hilbert space and let $\sigma_1$ and $\sigma_2$ be two states on that Hilbert space such that $\sigma_1 \prec \sigma_2$. Then there exists a unitary $U$ on $\mathcal{H} \otimes \mathcal{H}$ such that
\begin{equation}
\sigma_1= \Tr_B(U \sigma_2 \otimes \sigma_2 U^{\dagger})=\Tr_A (U \sigma_2 \otimes \sigma_2 U^{\dagger}).
\end{equation}
\end{theorem}

The result of Lemma~\ref{lemma:commute} breaks down in dimensions greater than 3 as the following counter example shows. 
\begin{counterexample} \label{counterexample:commute}
For the doubly stochastic matrix
\begin{equation}
M=\begin{pmatrix}
1&0&0&0\\
0&0&0&1\\
0&1&0&0\\
0&0&1&0
\end{pmatrix},
\end{equation}
there exists no doubly stochastic matrix $\tilde{M} \in \mathbb{R}^{4 \times 4}$ such that
\begin{equation}
M (\mathds{1}+\Pi)=(\mathds{1}+\Pi) \tilde{M}.
\end{equation}
\end{counterexample}
The proof of the statement of Counterexample~\ref{counterexample:commute} is by contradiction and can be found in Appendix~V of \cite{Bakhshinezhad-2019}. The counterexample straightforwardly generalizes to higher dimensions showing that Lemma~\ref{lemma:commute} is tight to dimension 3. To assess Conjecture~\ref{conj:technical} for higher dimensions, we therefore turn to other ways of making use of our framework of Section~\ref{sec:framework}.

\section{Passing on the norm approach} \label{sec:passingnormapproach}

Here we would like to look at a different way of exploring our framework of Section~\ref{sec:framework}. This will enable us to find an alternative proof of Conjecture~\ref{conj:technical} for dimension 3. This proof will have the advantage to partially generalize to dimension 4. This approach is based on the following observation.

Since we are only interested in reaching marginals $\tilde{\mathbf{p}}$ that are thermal at inverse temperature $\beta' < \beta_R$, what we really are interested in is to reach marginals of the following form
\begin{equation}
\tilde{\mathbf{p}}= b_0 + \sum_{i=1}^k ( \lfloor \tfrac{2 i}{d} \rfloor +1)^{-1} (\mathds{1}+\Pi^i) b_i,
\end{equation}
where $b_i= (p_{0i}(\beta'), p_{1 (i+1)}(\beta'), \dots, p_{(d-1) (d-1+i)}(\beta'))$ with
\begin{equation}
p_{ij} (\beta')= \bra{i} \tau(\beta') \ket{i} \bra{j} \tau(\beta') \ket{j}.
\end{equation}

One way to transform $\mathbf{p}$ into $\tilde{\mathbf{p}}$ is to transform each $r_i$ into its corresponding $b_i$. However, since doubly stochastic transformations conserve the norm and since in general $ \lVert r_i \rVert \neq \lVert b_i \rVert$, where $\lVert \cdot \lVert$ denotes the one norm, i.e.,
\begin{equation}
\lVert v \rVert = \sum_{i=0}^{d-1} \lvert v_i \rvert,
\end{equation}

one cannot perform this transformation directly with $M_i$. What one can do with $M_i$ instead is to transform the normalized $r_i$ into the normalized $b_i$, as indeed the following holds.

\begin{lemma} \label{lemma:ribi}
For any dimension $d$ and Hamiltonian

\begin{equation} 
\frac{r_i}{\lVert r_i \rVert} \succ \frac{b_i}{\lVert b_i \rVert}, \quad \forall i=0,\dots k.
\end{equation}
\end{lemma}

\begin{proof}
The proof is quite simple. We have
\begin{equation}
(r_i)_j= p_{j (j+i)} = \frac{e^{-\beta_R (E_j+E_{j+i})}}{\left(\sum_{k=0}^{d-1} e^{-\beta_R E_k} \right)^2},
\end{equation}
and so

\begin{equation}
\left( \frac{r_i}{\lVert r_i \rVert} \right)_j= \frac{e^{-\beta_R (E_j+E_{j+i})}}{\left(\sum_{k=0}^{d-1} e^{-\beta_R ( E_k+E_{k+i})} \right)^2}.
\end{equation}
This means that $\frac{r_i}{\lVert r_i \rVert} = \mathfrak{D}(\tau_{\tilde{H}_i}(\beta))$, where $\tau_{\tilde{H}_i}(\beta)$ is the thermal state at inverse temperature $\beta$ of the Hamiltonian
\begin{equation}
\tilde{H}_i = \sum_{j=0}^{d-1} (E_j+E_{j+i}) \ket{j} \bra{j}.
\end{equation}
Finally, as
\begin{equation}
 \mathfrak{D}(\tau_{\tilde{H}_i}(\beta)) \succ \mathfrak{D}(\tau_{\tilde{H}_i}(\beta'))= \frac{b_i}{\lVert b_i \rVert},
\end{equation}
our result is proven.
\end{proof}

What Lemma~\ref{lemma:ribi} tells us is that we are able to transform the $r_i$'s into the $b_i$'s, just not in the right amount. When $\lVert r_i \rVert > \lVert b_i \rVert$, we will end up with too much $b_i$ if one fully transforms $r_i$ into $b_i$ as prescribed by Lemma~\ref{lemma:ribi}. And so, one can instead use the excessive $r_i$ to try and create some $b_j$ for which Lemma~\ref{lemma:ribi} does not allow to create enough of. One succeeds if the norm can adequately be passed across the different subspaces as such.\\

One $r_i$ that always has excessive norm is $r_0$ as indeed

\begin{lemma} \label{lemma:bigr0}
For all Hamiltonians, inverse temperature $\beta$ and dimension d
\begin{equation}
\lVert r_0 \rVert \geq \lVert b_0 \rVert, \quad \forall \beta' \leq \beta.
\end{equation}
\end{lemma}

\begin{proof}[Proof idea]
The proof consists in verifying that $ \partial_{\beta} \lVert r_0 \rVert \geq 0$. See Appendix VII of \cite{Bakhshinezhad-2019} for more details.
\end{proof}

In dimension 3, this excessive norm can also adequately be passed to the other $b_i$ as the following holds.

\begin{lemma} \label{lemma:r0b1}
For $d=3$, any Hamiltonian and inverse background temperature $\beta$.
\begin{equation}
\frac{r_0}{\lVert r_0 \rVert} \succ \frac{(\mathds{1}+\Pi) b_1}{2 \lVert b_1 \rVert}, \quad \forall \beta' < \beta.
\end{equation}
\end{lemma}

\begin{proof}[Proof idea]
The proof consists of two steps. First, one proves that 
\begin{equation}
\frac{r_0}{\lVert r_0 \rVert} \succ \frac{(\mathds{1}+\Pi) r_1}{2 \lVert r_1 \rVert}.
\end{equation}
One then proceeds to prove that
\begin{equation} \label{eq:pir1pib1}
 \frac{(\mathds{1}+\Pi) r_1}{2 \lVert r_1 \rVert} \succ  \frac{(\mathds{1}+\Pi) b_1}{2 \lVert b_1 \rVert}.
\end{equation}

Using the transitivity of majorization, we get the desired result. The detailed proof of each inequality is found in Appendix IX of \cite{Bakhshinezhad-2019}. Note in particular that \ref{eq:pir1pib1} does not merely follow from
\begin{equation}
\frac{r_1}{\lVert r_1 \rVert} \succ \frac{b_1}{\lVert b_1 \rVert}.
\end{equation}
\end{proof}

Lemma~\ref{lemma:ribi}, Lemma~\ref{lemma:bigr0} and Lemma~\ref{lemma:r0b1} together prove Conjecture~\ref{conj:technical} for dimension 3. Indeed, according to Lemma~\ref{lemma:ribi} and Lemma~\ref{lemma:r0b1}, there exist doubly stochastic matrices $M_{r_0 \rightarrow b_0}$, $M_{r_1 \rightarrow b_1}$, and $M_{r_0 \rightarrow b_1}$ such that
\begin{align}
M_{r_0 \rightarrow b_0} r_0 &= \frac{\lVert r_0 \rVert}{\lVert b_0 \rVert} b_0,\\
M_{r_1 \rightarrow b_1} r_1 &= \frac{\lVert r_1 \rVert}{\lVert b_1 \rVert} b_1,\\
M_{r_0 \rightarrow b_1} r_0 &= \frac{\lVert r_0 \rVert}{2 \lVert b_1 \rVert} (\mathds{1}+\Pi) b_1.
\end{align}

Now let
\begin{align}
M_0 &= \frac{\lVert  b_0 \rVert}{\lVert r_0 \rVert} M_{r_0 \rightarrow b_0} + \left( 1- \frac{\lVert b_0 \rVert}{\lVert r_0 \rVert} \right) M_{r_0 \rightarrow b_1},\\
M_1 &= M_{r_1 \rightarrow b_1}.
\end{align}
$M_0$ is doubly stochastic since according to Lemma~\ref{lemma:bigr0} $ 0 \leq \frac{\lVert b_0 \rVert}{\lVert r_0 \rVert} \leq 1$. Then using that $1=\lVert r_0 \rVert + 2 \lVert r_1 \rVert = \lVert b_0 \rVert + 2 \lVert b_1 \rVert$ we get

\begin{align}
\tilde{\mathbf{p}}&= M_0 r_0 + (\mathds{1} + \Pi)  M_1 r_1\\
&=\frac{\lVert  b_0 \rVert}{\lVert r_0 \rVert} M_{r_0 \rightarrow b_0} r_0 + \left( 1- \frac{\lVert b_0 \rVert}{\lVert r_0 \rVert} \right) M_{r_0 \rightarrow b_1} r_0 + (\mathds{1} + \Pi) M_{r_1 \rightarrow b_1} r_1\\
&= b_0 + \frac{\lVert r_0 \rVert- \lVert b_0 \rVert + 2 \lVert r_1 \rVert}{ 2 \lVert b_1 \rVert} (\mathds{1}+\Pi) b_1\\
&= b_0 + (\mathds{1} + \Pi) b_1,
\end{align}
as desired. We have therefore proven the following.

\begin{theorem}
For dimension 3 the \enquote{passing on the norm} approach allows to reach any marginal $\tilde{\mathbf{p}}$ thermal at $\beta' < \beta_R$ and therefore proves Conjecture~\ref{conj:technical} for $d=3$.
\end{theorem}

There are various ways that this approach can unfold to higher dimensions. The fact that Lemma~\ref{lemma:ribi} as well as Lemma~\ref{lemma:bigr0} hold for any dimension already gives a good basis to start from. However, the result of Lemma~\ref{lemma:r0b1} proves more difficult to generalize. To this end, note that we choose to shift the excessive norm of $r_0$ to $(\mathds{1}+\Pi) b_1$ and not to $b_1$ and $\Pi b_1$ separately. This is simply because shifting the norm separately is not always possible to do. Indeed, already in dimension 3
\begin{equation}
\frac{r_0}{\lVert r_0 \rVert} \succ \frac{b_1}{\lVert b_1 \rVert}
\end{equation}
does not always hold. For example for $E_2 \rightarrow \infty, E_0=E_1=0$, the greatest component of $r_0 / \lVert r_0 \rVert$ is $\frac{1}{2}$ while that of $b_1 / \lVert b_1 \rVert$ is $1$. This tells us that shifting the norm to a group of terms is advantageous. The fact that for any dimension
\begin{equation} \label{equ:r0majs}
\frac{r_0}{\lVert r_0 \rVert} \succ \frac{s}{\lVert s \rVert} , \quad \text{where } s= \sum_{i=1}^{d-1} r_i,
\end{equation}
holds, see the proof below, suggests that a fruitful generalization of Lemma~\ref{lemma:r0b1} for higher dimension might be to group as many terms as possible together, namely all the $r_i \neq r_0$. However, shifting the norm from $r_0$ to all the other $r_i$ uniformly oversees the fact that different $b_i$ need to have their norm compensated by a different amount. With this given, there are 3 potential avenues we see to generalize the approach to higher dimensions.
\begin{enumerate}
\item Investigate the possibility of shifting some norm between the different $r_i, i \neq 0$ as well as from some $r_i, i \neq 0$ to $r_0$ such that a subsequent uniform shift from $r_0$ to all the $r_i$'s delivers the desired state. In doing so, one has to keep in mind that $(\mathds{1} + \Pi^i)$ mixes $r_i$ before its norm can be shifted, which limits the shifting possibilities.

\item Shift norm from $r_0$ to all the $r_i$ but not uniformly, such that the shifting results in the desired state. This amounts to checking whether
\begin{equation}
r_0 \succ b_0 +\sum_{i=1}^{k} ( \lfloor \tfrac{2 i}{d} \rfloor +1)^{-1} \left( 1- \frac{\lVert r_i \rVert}{\lVert b_i \rVert} \right) (\mathds{1}+\Pi^i) b_i,
\end{equation}
holds.

\item Shift norm from $r_0$ to the simplest non-trivial grouping of $r_i$, namely 
\begin{equation}
(\mathds{1}+\Pi^i) r_i.
\end{equation}
\end{enumerate}

We found option 3 the most tractable alternative and therefore opted for that avenue. Doing so, we were able to prove the following.

\begin{lemma} \label{lemma:dim4passing}
In dimension 4, for any Hamiltonian $H= \sum_{i=0}^3 E_i \ket{i} \bra{i}$ such that $\delta_{i+1} \leq \delta_i$, where $\delta_i = E_{i_1} - E_i$, the following holds for $i=1,2$.
\begin{align}
\frac{r_0}{\lVert r_0 \rVert} &\succ \frac{(\mathds{1}+\Pi^i) b_i}{2 \lVert b_i \rVert},\\
\lVert r_i \rVert &\leq \lVert b_i \rVert.
\end{align}
\end{lemma}

\begin{proof}[Proof idea]
The idea of the proof is the same as that of Lemma~\ref{lemma:r0b1} and Lemma~\ref{lemma:bigr0}. We prove the majorization relation in two steps, namely
\begin{equation}
\frac{r_0}{\lVert r_0 \rVert} \succ \frac{(\mathds{1}+\Pi^i) r_i}{2 \lVert r_i \rVert}
\end{equation}
and
\begin{equation}
\frac{(\mathds{1}+\Pi^i) r_i}{2 \lVert r_i \rVert} \succ \frac{(\mathds{1}+\Pi^i) r_i}{2 \lVert r_i \rVert}.
\end{equation}
The inequality is proven by showing that $\partial_{\beta} \lVert r_i \rVert \leq 0$. For the details of the proof we refer to Appendix IX of \cite{Bakhshinezhad-2019}.
\end{proof}

Lemma~\ref{lemma:dim4passing} fails to be true in the regime $\delta_{i+1} < \delta_i$ as one may find a counter example to the relation
\begin{equation}
\frac{r_0}{\lVert r_0 \rVert} \succ \frac{(\mathds{1} + \Pi) r_1}{2 \lVert r_1 \rVert}
\end{equation}
in that regime. See Appendix IX of~\cite{Bakhshinezhad-2019} for more details. With the result of Lemma~\ref{lemma:dim4passing} we are able to prove the following for dimension 4.

\begin{theorem}
For dimension 4 and Hamiltonians $H= \sum_{i=0}^3 E_i \ket{i} \bra{i}$ such that $\delta_{i+1} \leq \delta_i$,  the \enquote{passing on the norm} approach allows to reach any marginal $\tilde{\mathbf{p}}$ thermal at $\beta' < \beta_R$ and therefore proves Conjecture~\ref{conj:technical} in the restricted case of $\delta_{i+1} \leq \delta_i$ for $d=4$.
\end{theorem}

\begin{proof}
The results of Lemma~\ref{lemma:ribi} and Lemma~\ref{lemma:dim4passing} ensure the existence of doubly stochastic matrices $M_{r_0 \rightarrow b_0}$, $M_1$, $M_2$, $M_{r_0 \rightarrow b_1}$ and $M_{r_0 \rightarrow b_2}$ such that
\begin{align}
M_{r_0 \rightarrow b_0} r_0 &= \frac{\lVert r_0 \rVert}{\lVert b_0 \rVert} b_0,\\
M_1 r_1 &= \frac{\lVert r_1 \rVert}{\lVert b_1 \rVert} b_1,\\
M_2 r_2 &= \frac{\lVert r_2 \rVert}{\lVert b_2 \rVert} b_2,\\
M_{r_0 \rightarrow b_1} r_0 &= \frac{\lVert r_0 \rVert}{2 \lVert b_1 \rVert} (\mathds{1} + \Pi) b_1,\\
M_{r_0 \rightarrow b_2} r_0 &= \frac{\lVert r_0 \rVert}{2 \lVert b_2 \rVert} (\mathds{1} + \Pi^2) b_2.
\end{align}

Now let
\begin{equation}
M_0 = \alpha_1 M_{r_0 \rightarrow b_0} + \alpha_2 M_{r_0 \rightarrow b1} + \alpha_3 M_{r_0 \rightarrow b_2},
\end{equation}
with
\begin{align}
\alpha_1 &= \frac{\lVert b_0 \rVert}{\lVert r_0 \rVert},\\
\alpha_2 &= 2\frac{\lVert b_1 \rVert-\lVert r_1 \rVert}{\lVert r_0 \rVert},\\
\alpha_3 &= \frac{\lVert b_2 \rVert-\lVert r_2 \rVert}{\lVert r_0 \rVert}.
\end{align}

Lemma~\ref{lemma:dim4passing} ensures that $\alpha_2 \geq 0 $ and $\alpha_3 \geq 0$. Furthermore, using that $\lVert b_0 \rVert + 2 \lVert b_1 \rVert+\lVert b_2 \rVert= \lVert r_0 \rVert+ 2 \lVert r_1 \rVert+\lVert r_2 \rVert = 1$ we have

\begin{equation}
\alpha_1 + \alpha_2 + \alpha_3=1.
\end{equation}
So $M_0$ is doubly stochastic as required and
\begin{align}
\tilde{\mathbf{p}} &= M_0 r_0 + (\mathds{1} + \Pi) M_1 r_1 + \frac{1}{2} ( \mathds{1} + \Pi^2) M_2 r_2\\
&= b_0 + \frac{\alpha_2 \lVert r_0 \rVert + 2 \lVert r_1 \rVert}{ 2 \lVert b_1 \rVert} (\mathds{1} + \Pi) b_1 + \frac{ \alpha_3 \lVert r_0 \rVert + \lVert r_2 \rVert}{\lVert b_2 \rVert} \frac{\mathds{1} + \Pi^2}{2} b_2\\
&= b_0 + (\mathds{1} + \Pi) b_1 + \frac{\mathds{1} + \Pi^2}{2} b_2,
\end{align}
as desired.
\end{proof}

For completeness, as well as for the sake of potential future investigations, before moving on to the next approach we prove the result claimed in Eq.~\ref{equ:r0majs}, namely that

\begin{lemma}
\begin{equation}
\frac{r_0}{\lVert r_0 \rVert} \succ \frac{s}{\lVert s \rVert} , \quad \text{where } s= \sum_{i=1}^{d-1} r_i,
\end{equation}
\end{lemma}

\begin{proof}
	Let
	\begin{align}
	x&=\frac{r_0}{\lVert r_0 \rVert},\\
	y&= \frac{s}{\lVert s \rVert}.
	\end{align}
	Note that $x_i=p_{i}^2/ \lVert r_0 \rVert$, $y_i = p_i (1-p_i) / \lVert s_0 \rVert$ and that $\lVert r_0 \rVert= \sum_{i=0}^{d-1} p_i^2$, $\lVert s_0 \rVert= 1- \lVert r_0 \rVert$. In order to prove our assertion, we need to prove $x \succ y$, i.e.,
	
	\begin{align}
	\sum_{i=0}^{k} x_i^{\downarrow} &\geq \sum_{i=0}^k y_i^{\downarrow}, \; \forall k=0, \dots, d-2\\
	\sum_{i=0}^{d-1} x_i^{\downarrow} &= \sum_{i=0}^{d-1} y_i^{\downarrow}
	\end{align}
	
	The equality condition is ensured since both $x$ and $y$ are normalised. Alternatively, it is also easily directly verifiable. To verify the inequality conditions, note that since the $p_i$'s form a probability distribution, i.e. $ 0 \leq p_i \leq 1$ and $\sum_{i=0}^{d-1}p_i =1$, and that $p_i \geq p_{i+1}$, we have $ p_i^2 \geq p_{i+1}^2$ and
	\begin{align}
	x_i &\geq x_{i+1},\\
	y_i &\geq y_{i+1}.
	\end{align}
	This implies $x_{i}^{\downarrow}=x_i$ and $y_{i}^{\downarrow}=y_i$. To conclude the proof, we calculate given $k \in {0, \dots, d-2}$,
	
	\begin{align}
	&\sum_{i=0}^{k} x_i \geq \sum_{i=0}^k y_i\\
	\Leftrightarrow & \frac{\sum_{i=0}^k p_i^2}{\sum_{j=0}^{d-1} p_j^2} \geq \frac{\sum_{i=0}^k p_i (1-p_i)}{\sum_{j=0}^{d-1} p_j (1-p_j)}\\
	\Leftrightarrow & \frac{1}{1+\frac{\sum_{j=k+1}^{d-1} p_j^2}{\sum_{i=0}^{k} p_i^2}} \geq \frac{1}{1+\frac{\sum_{j=k+1}^{d-1} p_j (1-p_j)}{\sum_{i=0}^{k} p_i (1-p_i)}}\\
	\Leftrightarrow & \frac{\sum_{j=k+1}^{d-1} p_j (1-p_j)}{\sum_{i=0}^{k} p_i (1-p_i)} \geq \frac{\sum_{j=k+1}^{d-1} p_j^2}{\sum_{i=0}^{k} p_i^2} \\
	\Leftrightarrow & \sum_{i=0}^k \sum_{j=k+1}^{d-1} [ p_i^2 p_j (1-p_j) - p_j^2 p_i (1-p_i)] \geq 0\\
	\Leftrightarrow & \sum_{i=0}^k \sum_{j=k+1}^{d-1} p_i p_j\underbrace{[ p_i (1-p_j) - p_j (1-p_i)]}_{=p_i-p_j} \geq 0.
	\end{align}
	The last relation holds true for any $k \in {0,\dots, d-2}$ since $p_i - p_j \geq 0$ for all $i=0, \dots, k$ and $j=k+1, \dots d-1$.
\end{proof}

\section{Geometric approach} \label{sec:geometric}
We now would like to turn out attention to yet another way of tackling Conjecture~\ref{conj:technical} that makes use of our framework of Section~\ref{sec:framework}. This will allow us to prove Conjecture~\ref{conj:technical} for dimension 3 and 4. This approach is based on the geometry of doubly stochastic matrices, that is on the fact that, according to Theorem~\ref{thm:Birkhoff}, they build a polytope, the vertices of which are the permutation matrices. Applying this to our transformation $\tilde{\mathbf{p}}$, we get that the reachable $\tilde{\mathbf{p}}$ are contained within the polytope that has as vertices the following points

\begin{align}
 \mathbf{\tilde{p}}^{(i_{0},i_{1},\ldots,i_{k})}
    &= P^{(i_{0})} r_0 +
    \sum_{n=1}^{k}
    (\lfloor\tfrac{2n}{d}\rfloor+1)^{-1}
    (\mathds{1}+ \Pi^{n})\,P^{(i_{n})} r_{n},
\label{eq:EqualEvoMarginal}
\end{align}

where $i_j \in \{0,\dots,d!-1\}$ for all $j = 0, \dots, k$, and where $P^{(0)}, \dots, P^{(d!-1)}$ denote the $d!$ permutation matrices in dimension $d$. To solve Conjecture~\ref{conj:technical}, one needs to show that the curve $\mathfrak{D}(\tau(\beta'))$ for $\beta' \leq \beta_R$ is within this polytope, that is we must show
\begin{equation}
\{\mathfrak{D}(\tau(\beta')) \mid \beta' \leq \beta_R\} \subset \text{conv} \{\mathbf{\tilde{p}}^{(i_{0},i_{1},\ldots,i_{k})} \mid i_0,\dots,i_k = 0 , \dots, d!-1 \},
\end{equation}

where $\text{conv} \{ a_0, \dots, a_n\}$ denotes the set of all convex combinations of the points $a_0, \dots, a_n$. Note that the number of vertices of the polytope rapidly grows with the dimension of the problem as there are $(d!)^{k+1}$ vertices in dimension $d$. For dimension 3 the polytope therefore has $(3!)^2 = 36$ vertices and for dimension 4, $(4!)^3=13'824$. To nevertheless keep the problem tractable it is therefore crucial to select the relevant vertices or points of the polytope adequately. To this end, visualizing the thermal curve as well the polytope is of great help to build intuition, see Figure~\ref{fig:polytope} for a depiction in dimension 3.

\begin{figure*}
(a)\includegraphics[width=0.5\textwidth,trim={0cm 0mm 0cm 0mm}]{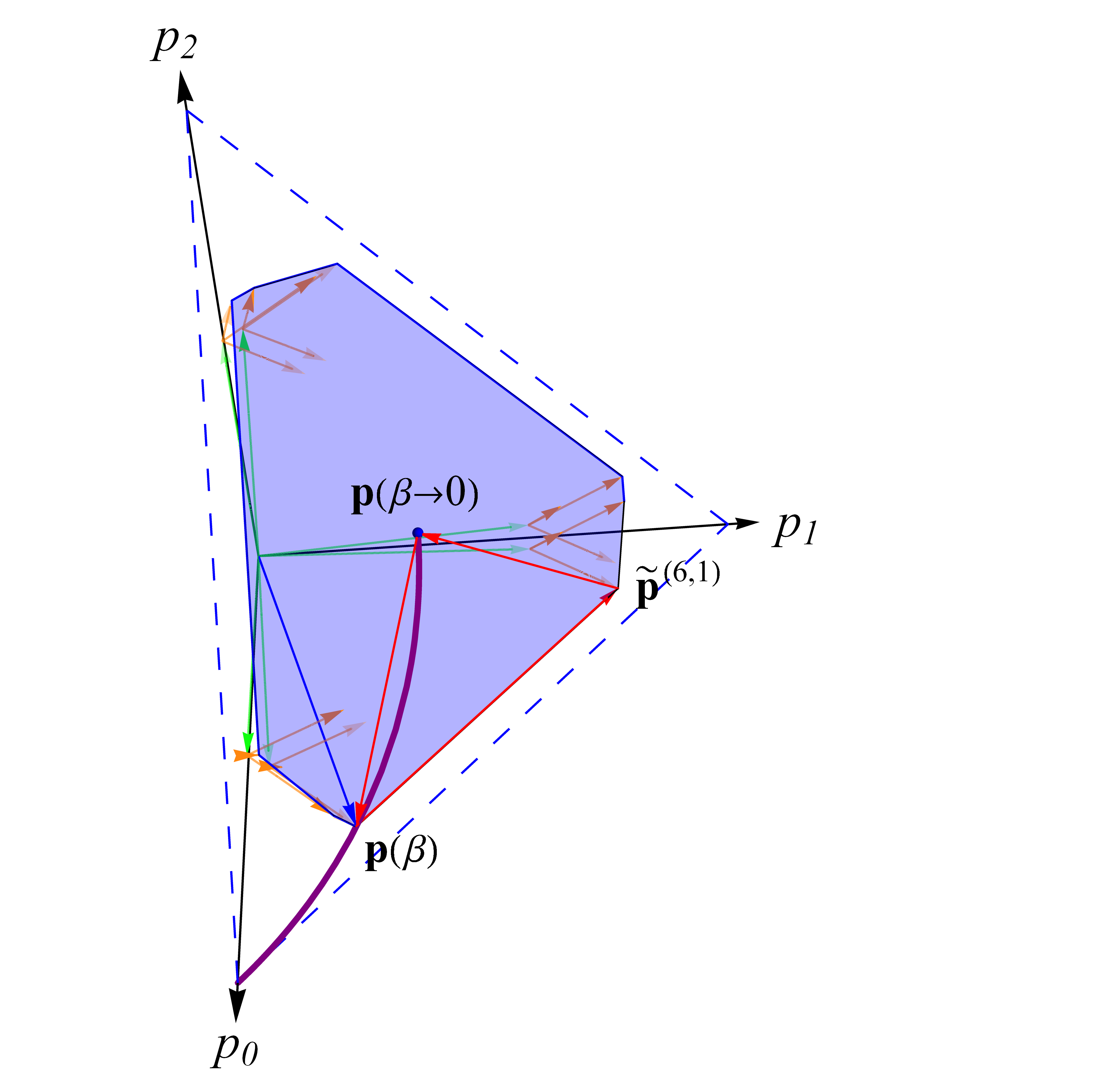}
(b)\ \includegraphics[width=0.35\textwidth,trim={0cm 0mm 0cm 0mm}]{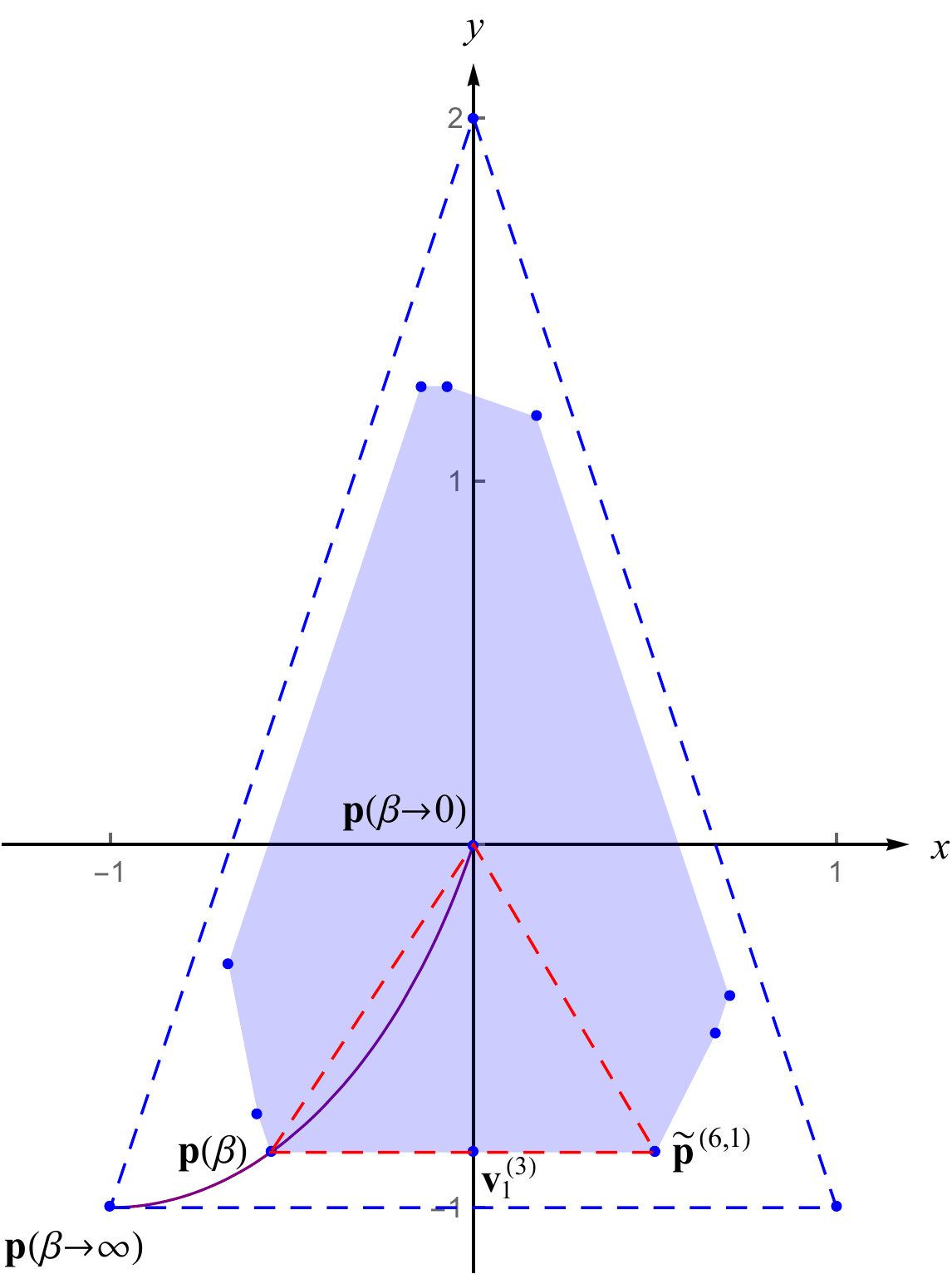}
\caption{\textbf{Polytope of reachable marginals}. (a) The axes show the components $p_{i}$ (with $i=0,1,2$) in the standard basis of $\mathds{R}^3$, $\{e_{i}\}_{i=0,1,2}$ and the standard simplex is indicated by the dashed blue triangle. The parameter values chosen for the illustration are $\beta_R=1.35 E_{1}$ and $E_{2}=2 E_{1}$. The thermal states are shown as a solid purple curve from $\mathbf{p}(\beta\rightarrow\infty)=(1,0,0)^{T}$ to $\mathbf{p}(\beta\rightarrow0)=(\tfrac{1}{3},\tfrac{1}{3},\tfrac{1}{3})^{T}$. The shaded blue area shows the polytope of reachable $\tilde{\mathbf{p}}$ from the point $\mathbf{p}(\beta)$. We have chosen to restrict to cyclic permutations on $r_1$ here to illustrate that this is enough for $d=3$, whereas this is no longer the case when $d=4$. (b) The polytope is shown in terms of the new basis $x\equiv q_{0}$ and $y\equiv q_{1}$.}
\label{fig:polytope}
\end{figure*}

To help such a visualization note that as the points
\begin{equation}
\tilde{\mathbf{p}}=\begin{pmatrix}
\tilde{p}_0\\
\vdots\\
\tilde{p}_{d-1}
\end{pmatrix}
\end{equation}
are normalized, i.e., $\sum_{i=0}^{d-1} \tilde{p}_i=1$, they geometrically lie on the standard simplex in $\mathds{R}^d$, and so does our polytope. To ease the analysis it is convenient to get rid of the normalization condition by picking coordinates that parameterize the simplex. We do this by working in the coordinates $\{x_i\}_{i=0,\dots,d-1}$ with
\begin{align}
x_i&= (i+1) \tilde{p}_{i+1} - \sum_{j=0}^i \tilde{p}_j, \quad i=0,\dots, d-2,\\
x_{d-1} &= - \sum_{i=0}^{d-1} \tilde{p}_i.
\end{align}

This way, $x_0, \dots, x_{d-2}$ parameterize the simplex and $x_{d-1}=-1$ for all the points of the simplex and can as such be disregarded. Making this change of coordinates corresponds to performing a change of basis from the standard basis of $\mathds{R}^d$, $(e_i)_{i=0}^{d-1}$ with $(e_i)_j=\delta_{ij}$, to the following basis
\begin{align}
q_i &= \frac{1}{i+2} \left( \frac{-1}{i+1} \sum_{j=0}^i e_j + e_{i+1}\right), \quad i=0,\dots,d-2,\\
q_{d-1}& = - \frac{1}{d} \sum_{j=0}^{d-1} e_j.
\end{align}

This basis is orthogonal but not orthonormal. The basis vectors are normalized such that the coordinate change only deals with integer multiples of the original coordinates. The $d-1$ first basis vectors, $q_0, \dots q_{d-2}$, constitute a basis of the simplex and $q_{d-1}$ is orthogonal to it. One furthermore verifies that indeed

\begin{equation}
\begin{pmatrix}
x_0\\
\vdots\\
x_{d-1}
\end{pmatrix} = B \begin{pmatrix} \tilde{p}_0\\
\vdots\\
\tilde{p}_{d-1} \end{pmatrix},
\end{equation}
where the matrix $B$ is such that $B^{-1}$ is the matrix of which columns are the $q_i$'s, i.e., $(B^{-1})_{ij} = (q_j)_i, \quad \forall i,j = 0,\dots,d-1$. Note that from the definition of the coordinates $(x_i)_{i=0}^{d-1}$,
\begin{equation}
B= \begin{pmatrix}
-1 &1 & \multicolumn{2}{c}{\text{\kern0.5em\smash{\raisebox{-1ex}{\Large 0}}}}\\
& \ddots & \ddots & \\
\multicolumn{2}{c}{\text{\kern0.5em\smash{\raisebox{-1ex}{\Large -1}}}} & -1 & d-1\\
 & & &-1
\end{pmatrix}.
\end{equation}

Remembering that $x_{d-1}=-1$ on the simplex, we work solely with the $d-1$ first coordinates $x_0, \dots, x_{d-2}$. In these new coordinates, the infinite temperature point is the origin, i.e.,
\begin{equation}
\lim_{\beta \rightarrow 0} \mathfrak{D}(\tau(\beta)) = (0,\dots,0),
\end{equation}

and the zero temperature point is
\begin{equation}
\lim_{\beta \rightarrow + \infty} \mathfrak{D}(\tau(\beta)) = (-1, \dots, -1).
\end{equation}

The thermal states furthermore lie on a curve connecting these points that is strictly confined to $x_i \leq 0 , \, i=0,\dots d-2$.\\

To try and reduce the complexity of the problem and not have to deal with all $(d!)^k$ vertices of the polytope, we break down the proof that the thermal curve is contained within the polytope into two steps.
\begin{itemize}
\item Identifying $d$ points $v_0, \dots, v_{d-1}$ of the simplex, the convex hull of which contains the thermal curve for all $\beta' \leq \beta_R$, i.e., identifying $v_0, \dots,v_{d-1}$ such that
\begin{equation}
\{ \mathfrak{D} ( \tau(\beta')) \mid \beta' \leq \beta_R\} \subset \text{conv} \{ v_0, \dots, v_{d-1} \}.
\end{equation} 

\item Proving that these $d$ points are not only within the simplex but also within the polytope of allowed transformations, i.e., proving
\begin{equation}
v_i \in \text{conv} \{ \tilde{\mathbf{p}}^{(i_0,\dots, i_{k})} \mid i_0, \dots, i_k = 0, \dots, d!-1 \}.
\end{equation}
\end{itemize}

We suggest to look at the following $d$ points, which in the $x_i$, $i=0,\dots,d-2$, coordinates read as
\begin{equation}
\begin{aligned}
v_0 &= (x_0(\beta_R), \dots, x_{d-2}(\beta_R))\\
v_1 &= (0,x_0(\beta_R), \dots, x_{d-2}(\beta_R))\\
&\vdots\\
v_{d-2} &= (0, \dots,0, x_{d-2}(\beta_R))\\
v_{d-1} &= (0,\dots,0)
\end{aligned},
\end{equation}
where $x_i(\beta_R)= (i+1) p_{i+1} - \sum_{j=0}^i p_j$. Note that $v_0 = \mathfrak{D}(\tau(\beta_R))$ and $v_{d-1}= \mathfrak{D}(\tau(0))$. That  for all $\beta' \leq \beta_R $, $\mathfrak{D}(\tau(\beta'))$ can be written as a convex combination of the above $v_i$'s can be proven in all dimensions ans so we have

\begin{lemma}
For all dimension $d$ and inverse background temperature $\beta_R$,
\begin{equation}
\{\mathfrak{D}(\tau(\beta')) \mid \beta' \leq \beta_R \} \subset \text{conv} \{v_0,\dots,v_{d-1}\}.
\end{equation}
\end{lemma}

The difficulty arises in trying to prove that the points $v_0, \dots , v_{d-1}$ lie inside the polytope. One can prove for any dimensions that $v_0$ and $v_{d-1}$ are within the polytope. For dimension 3 we can furthermore construct appropriate matrices $M_0$ and $M_1$ such that 
\begin{equation}
v_1 = M_0 r_0 + (\mathds{1}+\Pi) M_1 r_1.
\end{equation}

This proves Conjecture~\ref{conj:technical} for dimension 3. For dimension 4, the point $v_1$ can be reached with the equivalent matrices than that in dimension 3 and so we find $M_0, M_1$ and $M_2$ such that

\begin{equation}
v_1 = M_0  r_0 + (\mathds{1} + 'Pi) M_1 r_1 + \frac{\mathds{1}+\Pi^2}{2} M_2 r_2.
\end{equation}

To prove that $v_2$ is within the polytope is harder. One can nevertheless prove it by showing that it is the combination of at most 5 vertices of the polytope. For details we refer to~\cite{Bakhshinezhad-2019}. All in all, this delivers the following result.

\begin{theorem}
For dimension 3 and 4 the \enquote{geometric} approach allows to reach any marginal $\tilde{\mathbf{p}}$ thermal at $\beta' \leq \beta_R$ and therefore proves Conjecture~\ref{conj:technical} for $d=3$ and $d=4$.
\end{theorem}

\chapter{Conclusion and Outlook} \label{chap:corrconcl}

In this Part, we have investigated the question of how much correlations can be generated in two initially uncorrelated thermal systems for a given amount of energy. While for vanishing background temperature the question can be answered in full generality, the problem becomes much more complicated for finite temperatures and not much is known beyond the existence of a bound based on Jaynes' principle. In particular, due to the existence of counter examples, one cannot expect the marginals of the maximally correlated state to be symmetric in that regime. For systems with equal Hamiltonians, $H_A=H_B=H$, however, enough symmetry is put back into the problem such that the counter examples dissolve and our intuition seems to be restored. This lead us to conjecture the form of the maximally correlating unitaries in that regime, namely unitaries that transform thermal sates into states with equal thermal marginals at higher temperature, or in short STU($\beta$, $\beta'$) with $\beta' \leq \beta$. We then provided a framework valid for any $d \times d$ dimensional systems that enabled us to explore this question beyond previous partial results for equally gapped Hamiltonian~\cite{Huber-2015}.

To showcase the versatility of our framework as well as to provide further guidance for a proof (or disproof) of our conjecture, we have discussed 3 different approaches that made use of our framework. The \enquote{majorised marginal} approach allowed us to prove our conjecture for $d=3$ but could not be generalized to higher dimensions. The \enquote{passing on the norm} approach allowed us to prove our conjecture for $d=3$ and generalized to $d=4$ for Hamiltonians satisfying a specific constraint on their energy gaps, namely $\delta_{i+1} \leq \delta_i$. The \enquote{geometric} approach allowed us to prove our conjecture in $d=3$ and $d=4$. For both the \enquote{passing on the norm} and \enquote{geometric} approach we furthermore formulated a set of conditions to extend the proofs to arbitrary dimensions.

This Part addressed the fundamental question of optimal resource inter-convertibility in quantum thermodynamics. The problem at hand is a particular case of the quantum marginal problem~\cite{Christandl-2006} in that it asks the question of which marginals can unitarily be reached from, or are compatible with, a particular global state. The framework we put forward for the symmetric case has been showcased to be highly versatile and can in fact also be extended to cases beyond symmetric marginal transformations. As such, we believe that it might deliver relevant insights for related questions such as that of the catalytic entropy conjecture of~\cite{Boes-2019}. Finally, refining the question of the creation of correlations to that of generating entanglement remains a significant challenge in that context~\cite{Huber-2015, Bruschi-2015, Piccione-2019, Guha-2019}.


\part{Concluding remarks} 

In this thesis, we have treated two tasks of quantum thermodynamics, namely refrigeration, in Part~\ref{part:refri}, and the creation of correlations, in Part~\ref{part:corre}. For each task we were able to build a framework enabling us to derive insightful results that we review in the respective conclusions, Chapter~\ref{ref:concl} and Chapter~\ref{chap:corrconcl}. The main result of Part~\ref{part:refri} consists of a bound on cooling that is both universally valid in the paradigms considered and attainable. The main result of Part~\ref{part:corre} is the proof that, for a fixed amount of energy, it is possible to create the intuitive maximal amount of correlations between two identical systems of dimension 3 and 4. The process of getting these results naturally raised new questions. Some, such as that of the form of the energetically most efficient cooling operations, c.f. Section~\ref{sec:twoopenproblems}, were stringent enough to be formulated in precise terms. Others were convincing enough to be conjectured. This was the case of the reachability of the correlation bound for arbitrary dimension, c.f. Chapter~\ref{chap:cons}. However, most of them had to wait until the respective conclusion and outlook chapters, Chapter~\ref{ref:concl} and Chapter~\ref{chap:corrconcl}, to be formulated.\\

Beyond the naturally raised questions that are more of less close to the specifics of the respectively considered problems, working on these questions also triggered train of thoughts that are less tied to these problems than the field of quantum thermodynamics or research in general. We would next like to take the risk to elaborate on some of them.\\

What our bound on refrigeration hints at, is that pure states do not exist in nature. Indeed, we only really find thermal states around us and unless we have a machine of infinite size, our paradigm does not enable to cool a state to absolute zero. This impossibility to generate pure states is actually found to be valid across many other paradigms and can even be formulated as a general principle, dubbed the 3$^{\text{rd}}$ law, within the resource theory of thermodynamics~\cite{Masanes-2017, Wilming-2017, Scharlau-2018}. This comes in pair with the nonexistence of perfect, i.e., projective, measurements. Indeed, if projective measurements were to naturally exist, then one could simply use these measurements as machines to generate pure states, contradicting the 3$^{\text{rd}}$ law (as well as our result). This nonexistence of projective measurement can actually be proven on formal ground~\cite{Guryanova-2020}. While this does not refute the use of pure states and perfect measurements altogether, it serves to remind ourselves that they are idealizations, and that one can never truly expect to have access to them. For most theoretical investigations and applications, where one is ready to invest amounts of energies that look infinitely large from the system's perspective, i.e., amounts that are orders of magnitude greater than the system's energy, and where approximate pure states are all what is needed, this does not have big implications. However, this does issue a word of caution to be kept in mind for when one is interested in the energetics of quantum systems.\\

Throughout this thesis we have also repeatedly mentioned the fact that the relevant resources for quantum thermodynamics are not as clearly identifiable as those of classical thermodynamics.  This, for us, hints at the broader fact that the right notions treating quantum thermodynamics fully for what it is, as opposed to trying to desperately translate concepts from classical thermodynamics, have yet to be developed. The gap between theoretical approaches that allow/demand to manipulate arbitrarily complex quantum systems and the experimental reality of being able to control at most a few quantum systems, is a speaking example of the discrepancy between our current theoretical model and reality. A potential avenue in that endeavor is to explicitly include the notion of complexity in our description of quantum thermodynamic tasks and of how we bound their performances. Another fruitful avenue in that direction could be to make a better effort at considering correlated states in our protocols. While it is often easier to consider initially uncorrelated states, as is done for example throughout this thesis, it is also true that only correlated states can exhibit quantum phenomena.\\

While investigating both of the tasks considered in this thesis, we wound up having to closely look at the unitary orbit of a given state. That unitaries play such a central role in our quest is of no surprise and is simply due to the fact that time evolution in quantum mechanics is unitary. We found that a great mathematical tool that was of tremendous help to better understand the specifications of this unitary orbits was the theory of majorization. Since unitaries are so ubiquitous in quantum mechanics, we expect the theory of majorization to further be of great help to future research in the field. We also expect that further investigations in the field of quantum thermodynamics and quantum mechanics will foster developing the already well rounded theory of majorization, thereby nourishing the fruitful dialog between mathematics and physics. The number of quantum mechanical problems that majorization theory can prove itself useful to is, however, bound to be limited. Indeed, while quantum mechanics is inherently a theory of matrices or operators, majorization theory is originally a theory of vectors. As such it only cares about the diagonal elements of a matrix if a basis is fixed or of its spectrum if one allows for the basis  to vary.\\

We have also talked at length about the energy expenditure of performing unitary operations in this thesis. While it is a question of theoretical interest in itself and is found to be of practical importance to help us identify the operations that are in principle already out of reach, it is also a fact that to perfectly implement a unitary operation, a perfect measuring device, such as an autonomous clock~\cite{Erker-2017, Mitchison-2019}, is needed. Such a device has, however, an infinite thermodynamic cost associated to it. If one were to attempt to actually account for the total cost of such a machine, then, for the machine to have a chance to theoretically be implementable, one would have to consider an imperfect measuring device allowing us to only perform approximate unitaries. The fundamental question that arises is if the following both conditions can simultaneously be satisfied.

\begin{enumerate}
	\item The total energy needed to operate the machine is not obviously out of reach.
	\item The machine is able to yield a useful operation.
\end{enumerate}

In particular, it would be interesting to see if this restricts the type of operations that a quantum machine can in principle perform.

What this also hints at is the fact that time in quantum mechanics is ultimately this abstract underlying parameter. Quantum theory does not intrinsically allow for time to be measured and this special status of time is well-known to originate great tensions with general relativity. A consistently operational treatment of time in quantum mechanics could help gain more understanding in that regard.\\

Last but not least, while going through the necessary and useful ritual of projecting the future of the field and framing its big open questions, we find it instructive to bear in mind that the technology, i.e., the knowledge about techniques, we produce, influences the society we live in and its social organization~\cite[Chapter~13]{Eriksen-2010}, \cite{Pfaffenberger-1988}. This calls for more awareness of potential consequences of the research directions we explore as well as how we explore them. This for example expresses itself in terms of ecology, but not only, and ultimately calls upon our moral duty towards society when conducting research. Thinking as such also leads to very concrete questions such as: is it moral and acceptable that the military finances fundamental (independent?) research? See for example~\cite{Figl-2019} for a glimpse at the situation in Austria.

But society and its social processes also influence the production of technology. These processes impact what we study, how we study it, and ultimately, how we see things. This suggests a demystification of the objective nature so far apart from any human interaction. It forces us to admit that our understanding of nature is much more socially intertwined and dependent on our particular social context --- and in that sense subjective --- than we would like to admit. This calls for a rupture with the traditional belief that natural sciences are opposed to social sciences. We consider it rather more fruitful to think of research and critical thinking as a whole --- which on its own motivates this excurse. This way of thinking is, in our view, especially pertinent in our present time, when information --- as opposed to disinformation --- has become such a crucial and fragile resource.


\appendix 




\printbibliography[heading=bibintoc]

@article{Alhambra-2019,
   title={Heat-Bath Algorithmic Cooling with optimal thermalization strategies},
   volume={3},
   ISSN={2521-327X},
   url={https://dx.doi.org/10.22331/q-2019-09-23-188},
   journal={Quantum},
   publisher={Verein zur Forderung des Open Access Publizierens in den Quantenwissenschaften},
   author={Alhambra, Álvaro M. and Lostaglio, Matteo and Perry, Christopher},
   year={2019},
   month={Sep},
   pages={188},
  eprint={1807.07974},
    archivePrefix={arXiv},
    primaryClass={quant-ph},
}

@article {Boykin-2002,
	author = {Boykin, P. Oscar and Mor, Tal and Roychowdhury, Vwani and Vatan, Farrokh and Vrijen, Rutger},
	title = {Algorithmic cooling and scalable NMR quantum computers},
	volume = {99},
	number = {6},
	pages = {3388--3393},
	year = {2002},
	url = {https://dx.doi.org/10.1073/pnas.241641898},
	issn = {0027-8424},
	eprint={quant-ph/0106093},
    archivePrefix={arXiv},
    primaryClass={quant-ph},
	journal = {Proceedings of the National Academy of Sciences}
}

@article{Brandao-2013,
  title = {Resource Theory of Quantum States Out of Thermal Equilibrium},
author = "Brand{\~a}o, Fernando G. S. L. and Horodecki, Micha{\l} and Oppenheim, Jonathan and Renes, Joseph M. and Spekkens, Robert W.",
  journal = {Phys. Rev. Lett.},
  volume = {111},
  issue = {25},
  pages = {250404},
  numpages = {5},
  year = {2013},
  month = {Dec},
  publisher = {American Physical Society},
  url = {https://dx.doi.org/10.1103/PhysRevLett.111.250404},
  eprint={1111.3882},
    archivePrefix={arXiv},
    primaryClass={quant-ph}
}

@article{Brandao-2015,
   title={The second laws of quantum thermodynamics},
   volume={112},
   ISSN={1091-6490},
   url={https://dx.doi.org/10.1073/pnas.1411728112},
   number={11},
   journal={Proceedings of the National Academy of Sciences},
   publisher={Proceedings of the National Academy of Sciences},
   author={Brandão, Fernando and Horodecki, Michał and Ng, Nelly and Oppenheim, Jonathan and Wehner, Stephanie},
   year={2015},
   month={Feb},
   pages={3275–3279},
 eprint={1305.5278},
    archivePrefix={arXiv},
    primaryClass={quant-ph}
}

@article{Gour-2015,
title = "The resource theory of informational nonequilibrium in thermodynamics",
journal = "Physics Reports",
volume = "583",
pages = "1 - 58",
year = "2015",
issn = "0370-1573",
url = "https://dx.doi.org/10.1016/j.physrep.2015.04.003",
author = "Gilad Gour and Markus P. M{\"u}ller and Varun Narasimhachar and Robert W. Spekkens and Nicole Yunger Halpern",
    eprint={1309.6586},
    archivePrefix={arXiv},
    primaryClass={quant-ph}
}

@article{Lostaglio-2019,
	year = "2019",
	month = {oct},
	publisher = "{IOP} Publishing",
	volume = {82},
	number = {11},
	pages = {114001},
	author = {Matteo Lostaglio},
	title = {An introductory review of the resource theory approach to thermodynamics},
	journal = {Reports on Progress in Physics},
url={https://dx.doi.org/10.1088/1361-6633/ab46e5},
eprint={1807.11549},
    archivePrefix={arXiv},
    primaryClass={quant-ph}
}

@article{Horodecki-2013,
   title={Fundamental limitations for quantum and nanoscale thermodynamics},
   volume={4},
   ISSN={2041-1723},
   url={https://dx.doi.org/10.1038/ncomms3059},
   number={1},
   journal={Nature Communications},
   publisher={Springer Science and Business Media LLC},
   author={Horodecki, Michał and Oppenheim, Jonathan},
   year={2013},
   month={Jun},
    eprint={1111.3834},
    archivePrefix={arXiv},
    primaryClass={quant-ph}
}

@article{Lostaglio-2015,
   title={Quantum Coherence, Time-Translation Symmetry, and Thermodynamics},
   volume={5},
   ISSN={2160-3308},
   url={https://dx.doi.org/10.1103/PhysRevX.5.021001},
   number={2},
   journal={Physical Review X},
   publisher={American Physical Society (APS)},
   author={Lostaglio, Matteo and Korzekwa, Kamil and Jennings, David and Rudolph, Terry},
   year={2015},
   month={Apr},
    eprint={1410.4572},
    archivePrefix={arXiv},
    primaryClass={quant-ph}
}

@article{Lostaglio-2015b,
   title={Description of quantum coherence in thermodynamic processes requires constraints beyond free energy},
   volume={6},
   ISSN={2041-1723},
   url={https://dx.doi.org/10.1038/ncomms7383},
   number={1},
   journal={Nature Communications},
   publisher={Springer Science and Business Media LLC},
   author={Lostaglio, Matteo and Jennings, David and Rudolph, Terry},
   year={2015},
   month={Mar},
 eprint={1405.2188},
    archivePrefix={arXiv},
    primaryClass={quant-ph}
}

@article{Cwiklinski-2015,
   title={Limitations on the Evolution of Quantum Coherences: Towards Fully Quantum Second Laws of Thermodynamics},
   volume={115},
   ISSN={1079-7114},
   url={https://dx.doi.org/10.1103/PhysRevLett.115.210403},
   number={21},
   journal={Physical Review Letters},
   publisher={American Physical Society (APS)},
   author={Ćwikliński, Piotr and Studziński, Michał and Horodecki, Michał and Oppenheim, Jonathan},
   year={2015},
   month={Nov},
   eprint={1405.5029},
    archivePrefix={arXiv},
    primaryClass={quant-ph}
}

@article{Skrzypczyk-2014,
   title={Work extraction and thermodynamics for individual quantum systems},
   volume={5},
   ISSN={2041-1723},
   url={https://dx.doi.org/10.1038/ncomms5185},
   number={1},
   journal={Nature Communications},
   publisher={Springer Science and Business Media LLC},
   author={Skrzypczyk, Paul and Short, Anthony J. and Popescu, Sandu},
   year={2014},
   month={Jun},
  eprint={1307.1558},
    archivePrefix={arXiv},
    primaryClass={quant-ph}
}

@article{Guryanova-2016,
   title={Thermodynamics of quantum systems with multiple conserved quantities},
   volume={7},
   ISSN={2041-1723},
   url={https://dx.doi.org/10.1038/ncomms12049},
   number={1},
   journal={Nature Communications},
   publisher={Springer Science and Business Media LLC},
   author={Guryanova, Yelena and Popescu, Sandu and Short, Anthony J. and Silva, Ralph and Skrzypczyk, Paul},
   year={2016},
   month={Jul},
eprint={1512.01190},
    archivePrefix={arXiv},
    primaryClass={quant-ph}
}

@article{Masanes-2017,
   title={A general derivation and quantification of the third law of thermodynamics},
   volume={8},
   ISSN={2041-1723},
   url={https://dx.doi.org/10.1038/ncomms14538},
   number={1},
   journal={Nature Communications},
   publisher={Springer Science and Business Media LLC},
   author={Masanes, Lluís and Oppenheim, Jonathan},
   year={2017},
   month={Mar},
 eprint={1412.3828},
    archivePrefix={arXiv},
    primaryClass={quant-ph}
}

@article{Wilming-2017,
   title={Third Law of Thermodynamics as a Single Inequality},
   volume={7},
   ISSN={2160-3308},
   url={https://dx.doi.org/10.1103/PhysRevX.7.041033},
   number={4},
   journal={Physical Review X},
   publisher={American Physical Society (APS)},
   author={Wilming, Henrik and Gallego, Rodrigo},
   year={2017},
   month={Nov},
eprint={1701.07478},
    archivePrefix={arXiv},
    primaryClass={quant-ph}
}

@article{Scharlau-2018,
   title={Quantum Horn’s lemma, finite heat baths, and the third law of thermodynamics},
   volume={2},
   ISSN={2521-327X},
   url={https://dx.doi.org/10.22331/q-2018-02-22-54},
   journal={Quantum},
   publisher={Verein zur Forderung des Open Access Publizierens in den Quantenwissenschaften},
   author={Scharlau, Jakob and Mueller, Markus P.},
   year={2018},
   month={Feb},
   pages={54},
    eprint={1605.06092},
    archivePrefix={arXiv},
    primaryClass={quant-ph}
}

@book{Nielsen-2010, 
place={Cambridge}, 
title={Quantum Computation and Quantum Information: 10th Anniversary Edition}, 
url={https://dx.doi.org/10.1017/CBO9780511976667}, 
publisher={Cambridge University Press}, 
author={Nielsen, Michael A. and Chuang, Isaac L.}, 
year={2010}}

@Inbook{Park-2016,
author="Park, Daniel K.
and Rodriguez-Briones, Nayeli A.
and Feng, Guanru
and Rahimi, Robabeh
and Baugh, Jonathan
and Laflamme, Raymond",
editor="Takui, Takeji
and Berliner, Lawrence
and Hanson, Graeme",
title="Heat Bath Algorithmic Cooling with Spins: Review and Prospects",
bookTitle="Electron Spin Resonance (ESR) Based Quantum Computing",
year="2016",
publisher="Springer New York",
address="New York, NY",
pages="227--255",
isbn="978-1-4939-3658-8",
url="https://dx.doi.org/10.1007/978-1-4939-3658-8_8",
eprint={1501.00952},
    archivePrefix={arXiv},
    primaryClass={quant-ph}
}

@article{Raeisi-2015,
  title = {Asymptotic Bound for Heat-Bath Algorithmic Cooling},
  author = {Raeisi, Sadegh and Mosca, Michele},
  journal = {Phys. Rev. Lett.},
  volume = {114},
  issue = {10},
  pages = {100404},
  numpages = {5},
  year = {2015},
  month = {Mar},
  publisher = {American Physical Society},
  url = {https://dx.doi.org/10.1103/PhysRevLett.114.100404},
eprint={1407.3232},
    archivePrefix={arXiv},
    primaryClass={quant-ph}
}

@article{Rodriguez-Briones-2017,
  title = {Correlation-Enhanced Algorithmic Cooling},
  author = {Rodr\'{\i}guez-Briones, Nayeli A. and Mart\'{\i}n-Mart\'{\i}nez, Eduardo and Kempf, Achim and Laflamme, Raymond},
  journal = {Phys. Rev. Lett.},
  volume = {119},
  issue = {5},
  pages = {050502},
  numpages = {5},
  year = {2017},
  month = {Aug},
  publisher = {American Physical Society},
  url = {https://dx.doi.org/10.1103/PhysRevLett.119.050502},
 eprint={1703.03816},
    archivePrefix={arXiv},
    primaryClass={quant-ph}
}

@article{Rodriguez-Briones-2016,
  title = {Achievable Polarization for Heat-Bath Algorithmic Cooling},
  author = {Rodr\'{\i}guez-Briones, Nayeli Azucena and Laflamme, Raymond},
  journal = {Phys. Rev. Lett.},
  volume = {116},
  issue = {17},
  pages = {170501},
  numpages = {5},
  year = {2016},
  month = {Apr},
  publisher = {American Physical Society},
  url = {https://dx.doi.org/10.1103/PhysRevLett.116.170501},
    eprint={1412.6637},
    archivePrefix={arXiv},
    primaryClass={quant-ph}
}

@inproceedings{Schulman-1999,
author = {Schulman, Leonard J. and Vazirani, Umesh V.},
title = {Molecular Scale Heat Engines and Scalable Quantum Computation},
year = {1999},
isbn = {1581130678},
publisher = {Association for Computing Machinery},
address = {New York, NY, USA},
url = {https://dx.doi.org/10.1145/301250.301332},
booktitle = {Proceedings of the Thirty-First Annual ACM Symposium on Theory of Computing},
pages = {322–329},
numpages = {8},
location = {Atlanta, Georgia, USA},
series = {STOC ’99},
eprint={quant-ph/9804060},
    archivePrefix={arXiv},
    primaryClass={quant-ph}
}

@book{Marshall-2011,
  author = {Marshall, Albert W. and Olkin, Ingram and Arnold, Barry C.},
  url = {https://dx.doi.org/10.1007/978-0-387-68276-1},
  edition = {Second},
  publisher = {Springer},
  title = {Inequalities: Theory of Majorization and its Applications},
  volume = {143},
  year = {2011}
}

@article{Gallego-2016,
	url = {https://dx.doi.org/10.1088/1367-2630/18/10/103017},
	year = 2016,
	month = {oct},
	publisher = {{IOP} Publishing},
	volume = {18},
	number = {10},
	pages = {103017},
	author = {R Gallego and J Eisert and H Wilming},
	title = {Thermodynamic work from operational principles},
	journal = {New Journal of Physics},
 eprint={1504.05056},
    archivePrefix={arXiv},
    primaryClass={quant-ph}
	}

@article{Muller-2018,
   title={Correlating Thermal Machines and the Second Law at the Nanoscale},
   volume={8},
   ISSN={2160-3308},
   url={https://dx.doi.org/10.1103/PhysRevX.8.041051},
   number={4},
   journal={Physical Review X},
   publisher={American Physical Society (APS)},
   author={Müller, Markus P.},
   year={2018},
   month={Dec},
 eprint={1707.03451},
    archivePrefix={arXiv},
    primaryClass={quant-ph}
}

@article{Bakhshinezhad-2019,
   title={Thermodynamically optimal creation of correlations},
   volume={52},
   ISSN={1751-8121},
   url={https://dx.doi.org/10.1088/1751-8121/ab3932},
   number={46},
   journal={Journal of Physics A: Mathematical and Theoretical},
   publisher={IOP Publishing},
   author={Bakhshinezhad, Faraj and Clivaz, Fabien and Vitagliano, Giuseppe and Erker, Paul and Rezakhani, Ali and Huber, Marcus and Friis, Nicolai},
   year={2019},
   month={Oct},
   pages={465303},
  eprint={1904.07942},
    archivePrefix={arXiv},
    primaryClass={quant-ph}
}

@article{Clivaz-2019,
   title={Unifying Paradigms of Quantum Refrigeration: A Universal and Attainable Bound on Cooling},
   volume={123},
   ISSN={1079-7114},
   url={https://dx.doi.org/10.1103/PhysRevLett.123.170605},
   number={17},
   journal={Physical Review Letters},
   publisher={American Physical Society (APS)},
   author={Clivaz, Fabien and Silva, Ralph and Haack, Géraldine and Brask, Jonatan Bohr and Brunner, Nicolas and Huber, Marcus},
   year={2019},
   month={Oct},
   eprint={1903.04970},
    archivePrefix={arXiv},
    primaryClass={quant-ph}
}

@article{Clivaz-2019bis,
   title={Unifying paradigms of quantum refrigeration: Fundamental limits of cooling and associated work costs},
   volume={100},
   ISSN={2470-0053},
   url={https://dx.doi.org/10.1103/PhysRevE.100.042130},
   number={4},
   journal={Physical Review E},
   publisher={American Physical Society (APS)},
   author={Clivaz, Fabien and Silva, Ralph and Haack, Géraldine and Brask, Jonatan Bohr and Brunner, Nicolas and Huber, Marcus},
   year={2019},
   month={Oct},
eprint={1710.11624},
    archivePrefix={arXiv},
    primaryClass={quant-ph}
}

@article{Pusz-1978,
author = "Pusz, W. and Woronowicz, S. L.",
url="https://dx.doi.org/10.1007/BF01614224",
journal = "Comm. Math. Phys.",
number = "3",
pages = "273--290",
publisher = "Springer",
title = "Passive states and KMS states for general quantum systems",
eprint = "https://projecteuclid.org:443/euclid.cmp/1103901491",
volume = "58",
year = "1978"
}

@article{Lenard-1978,
  url = {https://dx.doi.org/10.1007/bf01011769},
  year = {1978},
  month = dec,
  publisher = {Springer Science and Business Media {LLC}},
  volume = {19},
  number = {6},
  pages = {575--586},
  author = {A. Lenard},
  title = {Thermodynamical proof of the Gibbs formula for elementary quantum systems},
  journal = {Journal of Statistical Physics}
}

@article{Alicki-2013,
   title={Entanglement boost for extractable work from ensembles of quantum batteries},
   volume={87},
   ISSN={1550-2376},
   url={https://dx.doi.org/10.1103/PhysRevE.87.042123},
   number={4},
   journal={Physical Review E},
   publisher={American Physical Society (APS)},
   author={Alicki, Robert and Fannes, Mark},
   year={2013},
   month={Apr},
 eprint={1211.1209},
    archivePrefix={arXiv},
    primaryClass={quant-ph}
}

@article{Hovhannisyan-2013,
   title={Entanglement Generation is Not Necessary for Optimal Work Extraction},
   volume={111},
   ISSN={1079-7114},
   url={https://dx.doi.org/10.1103/PhysRevLett.111.240401},
   number={24},
   journal={Physical Review Letters},
   publisher={American Physical Society (APS)},
   author={Hovhannisyan, Karen V. and Perarnau-Llobet, Martí and Huber, Marcus and Acín, Antonio},
   year={2013},
   month={Dec},
    eprint={1303.4686},
    archivePrefix={arXiv},
    primaryClass={quant-ph}
}

@article{Skrzypczyk-2015,
   title={Passivity, complete passivity, and virtual temperatures},
   volume={91},
   ISSN={1550-2376},
   url={https://dx.doi.org/10.1103/PhysRevE.91.052133},
   number={5},
   journal={Physical Review E},
   publisher={American Physical Society (APS)},
   author={Skrzypczyk, Paul and Silva, Ralph and Brunner, Nicolas},
   year={2015},
   month={May},
    eprint={1412.5485},
    archivePrefix={arXiv},
    primaryClass={quant-ph}
}

@article{PerarnauLlobet-2015b,
   title={Most energetic passive states},
   volume={92},
   ISSN={1550-2376},
   url={https://dx.doi.org/10.1103/PhysRevE.92.042147},
   number={4},
   journal={Physical Review E},
   publisher={American Physical Society (APS)},
   author={Perarnau-Llobet, Martí and Hovhannisyan, Karen V. and Huber, Marcus and Skrzypczyk, Paul and Tura, Jordi and Acín, Antonio},
   year={2015},
   month={Oct},
    eprint={1502.07311},
    archivePrefix={arXiv},
    primaryClass={quant-ph}
}

@article{Vitagliano-2018,
   title={Trade-Off Between Work and Correlations in Quantum Thermodynamics},
   ISBN={9783319990460},
   ISSN={2365-6425},
   url={https://dx.doi.org/10.1007/978-3-319-99046-0_30},
   journal={Thermodynamics in the Quantum Regime},
   publisher={Springer International Publishing},
   author={Vitagliano, Giuseppe and Klöckl, Claude and Huber, Marcus and Friis, Nicolai},
   year={2018},
   pages={731–750},
    eprint={1803.06884},
    archivePrefix={arXiv},
    primaryClass={quant-ph}
}

@misc{Christandl-2006,
    title={The Structure of Bipartite Quantum States - Insights from Group Theory and Cryptography},
    author={Matthias Christandl},
    year={2006},
    eprint={quant-ph/0604183},
    archivePrefix={arXiv},
    primaryClass={quant-ph}
}

@article{Jevtic-2012,
   title={Maximally and Minimally Correlated States Attainable within a Closed Evolving System},
   volume={108},
   ISSN={1079-7114},
   url={https://dx.doi.org/10.1103/PhysRevLett.108.110403},
   number={11},
   journal={Physical Review Letters},
   publisher={American Physical Society (APS)},
   author={Jevtic, Sania and Jennings, David and Rudolph, Terry},
   year={2012},
   month={Mar},
 eprint={1110.2371},
     archivePrefix={arXiv},
    primaryClass={quant-ph}
}

@article{aberg_2014,
   title={Catalytic Coherence},
   volume={113},
   ISSN={1079-7114},
   url={http://dx.doi.org/10.1103/PhysRevLett.113.150402},
   number={15},
   journal={Physical Review Letters},
   publisher={American Physical Society (APS)},
   author={Åberg, Johan},
   year={2014},
   month={Oct},
    eprint={1304.1060},
     archivePrefix={arXiv},
    primaryClass={quant-ph}
}

@article{Bondar-2003,
title = "Schur majorization inequalities for symmetrized sums with applications to tensor products",
journal = "Linear Algebra and its Applications",
volume = "360",
pages = "1 - 13",
year = "2003",
issn = "0024-3795",
url = "https://dx.doi.org/10.1016/S0024-3795(02)00461-5",
author = "James V. Bondar",
}

@article{Allahverdyan-2011,
   title={Thermodynamic limits of dynamic cooling},
   volume={84},
   ISSN={1550-2376},
   url={https://dx.doi.org/10.1103/PhysRevE.84.041109},
   number={4},
   journal={Physical Review E},
   publisher={American Physical Society (APS)},
   author={Allahverdyan, Armen E. and Hovhannisyan, Karen V. and Janzing, Dominik and Mahler, Guenter},
   year={2011},
   month={Oct},
eprint={1107.1044},
    archivePrefix={arXiv},
    primaryClass={cond-mat.stat-mech}
}

@article{Reeb-2014,
   title={An improved Landauer principle with finite-size corrections},
   volume={16},
   ISSN={1367-2630},
   url={https://dx.doi.org/10.1088/1367-2630/16/10/103011},
   number={10},
   journal={New Journal of Physics},
   publisher={IOP Publishing},
   author={Reeb, David and Wolf, Michael M},
   year={2014},
   month={Oct},
   pages={103011},
eprint={1306.4352},
    archivePrefix={arXiv},
    primaryClass={quant-ph}
}

@article{Linden-2010,
   title={How Small Can Thermal Machines Be? The Smallest Possible Refrigerator},
   volume={105},
   ISSN={1079-7114},
   url={https://dx.doi.org/10.1103/PhysRevLett.105.130401},
   number={13},
   journal={Physical Review Letters},
   publisher={American Physical Society (APS)},
   author={Linden, Noah and Popescu, Sandu and Skrzypczyk, Paul},
   year={2010},
   month={Sep},
    eprint={0908.2076},
    archivePrefix={arXiv},
    primaryClass={quant-ph}
}

@article{Skrzypczyk-2011,
   title={The smallest refrigerators can reach maximal efficiency},
   volume={44},
   ISSN={1751-8121},
   url={https://dx.doi.org/10.1088/1751-8113/44/49/492002},
   number={49},
   journal={Journal of Physics A: Mathematical and Theoretical},
   publisher={IOP Publishing},
   author={Skrzypczyk, Paul and Brunner, Nicolas and Linden, Noah and Popescu, Sandu},
   year={2011},
   month={Nov},
   pages={492002},
    eprint={1009.0865},
    archivePrefix={arXiv},
    primaryClass={quant-ph}
}

@article{Brunner-2012,
   title={Virtual qubits, virtual temperatures, and the foundations of thermodynamics},
   volume={85},
   ISSN={1550-2376},
   url={https://dx.doi.org/10.1103/PhysRevE.85.051117},
   number={5},
   journal={Physical Review E},
   publisher={American Physical Society (APS)},
   author={Brunner, Nicolas and Linden, Noah and Popescu, Sandu and Skrzypczyk, Paul},
   year={2012},
   month={May},
    eprint={1106.2138},
    archivePrefix={arXiv},
    primaryClass={quant-ph}
}

@article{Venturelli-2013,
   title={Minimal Self-Contained Quantum Refrigeration Machine Based on Four Quantum Dots},
   volume={110},
   ISSN={1079-7114},
   url={https://dx.doi.org/10.1103/PhysRevLett.110.256801},
   number={25},
   journal={Physical Review Letters},
   publisher={American Physical Society (APS)},
   author={Venturelli, Davide and Fazio, Rosario and Giovannetti, Vittorio},
   year={2013},
   month={Jun},
  eprint={1210.3649},
    archivePrefix={arXiv},
    primaryClass={cond-mat.mes-hall}
}

@article{Mitchison-2016,
   title={Realising a quantum absorption refrigerator with an atom-cavity system},
   volume={1},
   ISSN={2058-9565},
   url={https://dx.doi.org/10.1088/2058-9565/1/1/015001},
   number={1},
   journal={Quantum Science and Technology},
   publisher={IOP Publishing},
   author={Mitchison, Mark T and Huber, Marcus and Prior, Javier and Woods, Mischa P and Plenio, Martin B},
   year={2016},
   month={Mar},
   pages={015001},
  eprint={1603.02082},
    archivePrefix={arXiv},
    primaryClass={quant-ph}
}

@article{Hofer-2016,
   title={Quantum heat engine based on photon-assisted Cooper pair tunneling},
   volume={93},
   ISSN={2469-9969},
   url={https://dx.doi.org/10.1103/PhysRevB.93.041418},
   number={4},
   journal={Physical Review B},
   publisher={American Physical Society (APS)},
   author={Hofer, Patrick P. and Souquet, J.-R. and Clerk, A. A.},
   year={2016},
   month={Jan},
  eprint={1512.02165},
    archivePrefix={arXiv},
    primaryClass={cond-mat.mes-hall}
}

@article{Hofer-2016b,
   title={Autonomous quantum refrigerator in a circuit QED architecture based on a Josephson junction},
   volume={94},
   ISSN={2469-9969},
   url={https://dx.doi.org/10.1103/PhysRevB.94.235420},
   number={23},
   journal={Physical Review B},
   publisher={American Physical Society (APS)},
   author={Hofer, Patrick P. and Perarnau-Llobet, Martí and Brask, Jonatan Bohr and Silva, Ralph and Huber, Marcus and Brunner, Nicolas},
   year={2016},
   month={Dec},
   eprint={1607.05218},
    archivePrefix={arXiv},
    primaryClass={quant-ph}
}

@article{Maslennikov-2019,
   title={Quantum absorption refrigerator with trapped ions},
   volume={10},
   ISSN={2041-1723},
   url={https://dx.doi.org/10.1038/s41467-018-08090-0},
   number={1},
   journal={Nature Communications},
   publisher={Springer Science and Business Media LLC},
   author={Maslennikov, Gleb and Ding, Shiqian and Hablützel, Roland and Gan, Jaren and Roulet, Alexandre and Nimmrichter, Stefan and Dai, Jibo and Scarani, Valerio and Matsukevich, Dzmitry},
   year={2019},
   month={Jan},
eprint={1702.08672},
    archivePrefix={arXiv},
    primaryClass={quant-ph}
}

@article{Levy-2012,
   title={Quantum Absorption Refrigerator},
   volume={108},
   ISSN={1079-7114},
   url={https://dx.doi.org/10.1103/PhysRevLett.108.070604},
   number={7},
   journal={Physical Review Letters},
   publisher={American Physical Society (APS)},
   author={Levy, Amikam and Kosloff, Ronnie},
   year={2012},
   month={Feb},
   eprint={1109.0728},
    archivePrefix={arXiv},
    primaryClass={quant-ph}
}

@article{Bengtsson-2005,
    author = {Bengtsson, Ingemar and Ericsson, {\AA}sa and Ku{\'s}, Marek and Tadej, Wojciech and {\.Z}yczkowski, Karol},
    title = {{Birkhoff's Polytope and Unistochastic Matrices, $N = 3$ and $N = 4$}},
    journal = {Commun. Math. Phys.},
    volume = {259},
    pages = {307--324},
    year = {2005},
    url = {https://dx.doi.org/10.1007/s00220-005-1392-8},
  eprint={math/0402325},
    archivePrefix={arXiv},
    primaryClass={math.CO}
}

@book{Wallis-2017,  
title={Introduction to Combinatorics}, 
isbn={9781498777605}, 
publisher={Chapman and Hall/CRC Press}, 
author={Wallis, Walter D. and  George, John C.}, 
year={2017},
url={https://www.worldcat.org/title/introduction-to-combinatorics-second-edition/oclc/1047943507&referer=brief_results}
}

@article{Bennett-1982,
    author = {Bennett, Charles H.},
    title = {The thermodynamics of computation -- a review},
    journal = {Int. J. Theor. Phys.},
    volume = {21},
    pages = {905--940},
    year = {1982},
    url = {https://dx.doi.org/10.1007/BF02084158}
}

@book{Leff-2003,
    title = {Maxwell Demon 2: Entropy, Classical and Quantum Information, Computing},
    editor = {Leff, Harvey and Rex, Andrew F.},
    publisher = {Institute of Physics},
    address = {Bristol},
    year = {2003},
isbn={9780750307598},
url={https://www.worldcat.org/title/maxwells-demon-2-entropy-classical-and-quantum-information-computing/oclc/50940631}
}

@book{Maxwell-1871,
author = {Maxwell, J. C.},
title = "Theory of Heat",
address={London},
publisher = {Longmans, Green and Co.},
year = {1871},
edition={2d},
url={https://hdl.handle.net/2027/hvd.32044009528449}
}

@ARTICLE{Shannon-1948,
 author={C. E. {Shannon}}, 
journal={The Bell System Technical Journal}, 
title={A mathematical theory of communication}, 
year={1948}, 
volume={27}, 
number={3}, 
pages={379-423},
url={https://dx.doi.org/10.1002/j.1538-7305.1948.tb01338.x}
}

@article{Jaynes-1957,
  title = {Information Theory and Statistical Mechanics},
  author = {Jaynes, E. T.},
  journal = {Phys. Rev.},
  volume = {106},
  issue = {4},
  pages = {620--630},
  numpages = {0},
  year = {1957},
  month = {May},
  publisher = {American Physical Society},
  url = {https://dx.doi.org/10.1103/PhysRev.106.620},
}

@book{Wilde-2013,
 place={Cambridge}, 
title={Quantum Information Theory}, 
url={https://dx.doi.org/10.1017/CBO9781139525343}, 
publisher={Cambridge University Press}, 
author={Wilde, Mark M.}, 
year={2013},
    eprint={1106.1445},
    archivePrefix={arXiv},
    primaryClass={quant-ph}}

@book{Binder-2018,
  editor    = {Felix Binder and Luis A. Correa and Christian Gogolin and Janet Anders and Gerardo Adesso}, 
  title     = {Thermodynamics in the Quantum Regime},
  publisher = {Springer International Publishing},
  year      = {2018},
  volume    = {195},
  series    = {0168-1222},
  edition   = {1},
  isbn      = {978-3-319-99046-0},
url="https://dx.doi.org/10.1007/978-3-319-99046-0"
}

@article{Goold-2016,
    author = {Goold, John and Huber, Marcus and Riera, Arnau and del~Rio, Lidia and Skrzypczyk, Paul},
    title = {The role of quantum information in thermodynamics \textemdash\ a topical review},
    journal = {J. Phys. A: Math. Theor.},
    volume = {49},
    pages = {143001},
    year = {2016},
    url = {https://dx.doi.org/10.1088/1751-8113/49/14/143001},
    eprint={1505.07835},
    archivePrefix={arXiv},
    primaryClass={quant-ph}
}

@article{Vinjanampathy-2016,
    author = {Vinjanampathy, Sai and Anders, Janet},
    title = {{Quantum Thermodynamics}},
    journal = {Contemp. Phys.},
    volume = {57},
    pages = {545--579},
    year = {2016},
    url = {https://dx.doi.org/10.1080/00107514.2016.1201896},
  eprint={1508.06099},
    archivePrefix={arXiv},
    primaryClass={quant-ph}
}

@article{Millen-2016,
    author = {Millen, James and Xuereb, Andr{\'e}},
    title = {Perspective on quantum thermodynamics},
    journal = {New J. Phys.},
    volume = {18},
    pages = {011002},
    year = {2016},
    url = {https://dx.doi.org/10.1088/1367-2630/18/1/011002},
    eprint={1509.01086},
    archivePrefix={arXiv},
    primaryClass={quant-ph}
}

@article{Guryanova-2020,
   title={Ideal Projective Measurements Have Infinite Resource Costs},
   volume={4},
   ISSN={2521-327X},
   url={https://dx.doi.org/10.22331/q-2020-01-13-222},
   journal={Quantum},
   publisher={Verein zur Forderung des Open Access Publizierens in den Quantenwissenschaften},
   author={Guryanova, Yelena and Friis, Nicolai and Huber, Marcus},
   year={2020},
   month={Jan},
   pages={222},
    eprint={1805.11899},
    archivePrefix={arXiv},
    primaryClass={quant-ph}
}

@article{Friis-2016,
    author = {Friis, Nicolai and Huber, Marcus and Perarnau-Llobet, Mart{\'i}},
    title = {Energetics of correlations in interacting systems},
    journal = {Phys. Rev. E},
    volume = {93},
    pages = {042135},
    year = {2016},
    url = {https://dx.doi.org/10.1103/PhysRevE.93.042135},
    eprint={1511.08654},
    archivePrefix={arXiv},
    primaryClass={quant-ph}
}

@article{PerarnauLlobet-2015,
    author = {Perarnau-Llobet, Mart{\'i} and Hovhannisyan, Karen V. and Huber, Marcus and Skrzypczyk, Paul and Brunner, Nicolas and Ac$\acute{\i}$n, Antonio},
    title = {Extractable work from correlations},
    journal = {Phys. Rev. X},
    volume = {5},
    pages = {041011},
    year = {2015},
    url = {https://dx.doi.org/10.1103/PhysRevX.5.041011},
   eprint={1407.7765},
    archivePrefix={arXiv},
    primaryClass={quant-ph}
}

@article{Huber-2015,
    author = {Huber, Marcus and Perarnau-Llobet, Mart{\'i} and Hovhannisyan, K. V. and Skrzypczyk, Paul and Kl{\"o}ckl, Claude and Brunner, Nicolas and Ac$\acute{\i}$n, Antonio},
    title = {Thermodynamic cost of creating correlations},
    journal = {New J. Phys.},
    volume = {17},
    pages = {065008},
    year = {2015},
    url = {https://dx.doi.org/10.1088/1367-2630/17/6/065008},
  eprint={1404.2169},
    archivePrefix={arXiv},
    primaryClass={quant-ph}
}

@article{Bruschi-2015,
    author = {Bruschi, David Edward and Perarnau-Llobet, Mart{\'i} and Friis, Nicolai and Hovhannisyan, Karen V. and Huber, Marcus},
    title = {The thermodynamics of creating correlations: Limitations and optimal protocols},
    journal = {Phys. Rev. E},
    volume = {91},
    pages = {032118},
    year = {2015},
    url = {https://dx.doi.org/10.1103/PhysRevE.91.032118},
    eprint={1409.4647},
    archivePrefix={arXiv},
    primaryClass={quant-ph}
}

@article{Binder-2015,
    author = {Binder, Felix C. and Vinjanampathy, Sai and Modi, Kavan and Goold, John},
    title = {{Quantacell: Powerful charging of quantum batteries}},
    journal = {New J. Phys.},
    volume = {17},
    pages = {075015},
    year = {2015},
    url = {https://dx.doi.org/10.1088/1367-2630/17/7/075015},
    eprint={1503.07005},
    archivePrefix={arXiv},
    primaryClass={quant-ph}
}

@article{Alipour-2016,
    author = {Alipour, Sahar and Benatti, Fabio and Bakhshinezhad, Faraj and Afsary, Maryam and Marcantoni, Stefano and Rezakhani, Ali T.},
    title = {{Correlations in quantum thermodynamics: Heat, work, and entropy production}},
    journal = {Sci. Rep.},
    volume = {6},
    pages = {35568},
    year = {2016},
    url = {https://dx.doi.org/10.1038/srep35568},
    eprint={1606.08869},
    archivePrefix={arXiv},
    primaryClass={quant-ph}
}

@article{Bera-2017,
    author = {Bera, Manabendra Nath and Riera, Arnau and Lewenstein, Maciej and Winter, Andreas},
    title = {{Generalized Laws of Thermodynamics in the Presence of Correlations}},
    journal = {Nat. Commun.},
    volume = {8},
    pages = {2180},
    year = {2017},
    url = {https://dx.doi.org/10.1038/s41467-017-02370-x},
    eprint={1612.04779},
    archivePrefix={arXiv},
    primaryClass={quant-ph}
}

@article{Mueller-2018,
    author = {M{\"u}ller, Markus P.},
    title = {{Correlating thermal machines and the second law at the nanoscale}},
    journal = {Phys. Rev. X},
    volume = {8},
    pages = {041051},
    year = {2018},
    url = {https://dx.doi.org/10.1103/PhysRevX.8.041051},
  eprint={1707.03451},
    archivePrefix={arXiv},
    primaryClass={quant-ph}
}

@article{Bera-2019,
    author = {Bera, Manabendra Nath and Riera, Arnau and Lewenstein, Maciej and Baghali Khanian, Zahra and Winter, Andreas},
    title = {{Thermodynamics as a Consequence of Information Conservation}},
    journal = {Quantum},
    volume = {3},
    pages = {121},
    year = {2019},
    url = {https://dx.doi.org/10.22331/q-2019-02-14-121},
    eprint={1707.01750},
    archivePrefix={arXiv},
    primaryClass={quant-ph}
}

@article{Sapienza-2019,
    author = {Sapienza, Facundo and Cerisola, Federico and Roncaglia, Augusto J.},
    title = {Correlations as a resource in quantum thermodynamics},
    journal = {Nat. Commun.},
    volume = {10},
    pages = {2492},
    year = {2019},
    url = {https://dx.doi.org/10.1038/s41467-019-10572-8},
  eprint={1810.01215},
    archivePrefix={arXiv},
    primaryClass={quant-ph}
}

@article{Jevtic-2012b,
   title={Quantum mutual information along unitary orbits},
   volume={85},
   ISSN={1094-1622},
   url={http://dx.doi.org/10.1103/PhysRevA.85.052121},
   number={5},
   journal={Physical Review A},
   publisher={American Physical Society (APS)},
   author={Jevtic, Sania and Jennings, David and Rudolph, Terry},
   year={2012},
   month={May},
    eprint={1112.3372},
    archivePrefix={arXiv},
    primaryClass={quant-ph}
}

@article{Boes-2019,
   title={Von Neumann Entropy from Unitarity},
   volume={122},
   ISSN={1079-7114},
   url={https://dx.doi.org/10.1103/PhysRevLett.122.210402},
   number={21},
   journal={Physical Review Letters},
   publisher={American Physical Society (APS)},
   author={Boes, Paul and Eisert, Jens and Gallego, Rodrigo and Müller, Markus P. and Wilming, Henrik},
   year={2019},
   month={May},
   eprint={1807.08773},
    archivePrefix={arXiv},
    primaryClass={quant-ph}
}

@misc{Piccione-2019,
    title={Energy bounds for entangled states},
    author={Nicolò Piccione and Benedetto Militello and Anna Napoli and Bruno Bellomo},
    year={2019},
    eprint={1904.02778},
    archivePrefix={arXiv},
    primaryClass={quant-ph}
}

@article{Guha-2019,
   title={Allowed and forbidden bipartite correlations from thermal states},
   volume={100},
   ISSN={2470-0053},
   url={https://dx.doi.org/10.1103/PhysRevE.100.012147},
   number={1},
   journal={Physical Review E},
   publisher={American Physical Society (APS)},
   author={Guha, Tamal and Alimuddin, Mir and Parashar, Preeti},
   year={2019},
   month={Jul},
    eprint={1904.07643},
    archivePrefix={arXiv},
    primaryClass={quant-ph}
}

@article{Abah-2012,
   title={Single-Ion Heat Engine at Maximum Power},
   volume={109},
   ISSN={1079-7114},
   url={https://dx.doi.org/10.1103/PhysRevLett.109.203006},
   number={20},
   journal={Physical Review Letters},
   publisher={American Physical Society (APS)},
   author={Abah, O. and Roßnagel, J. and Jacob, G. and Deffner, S. and Schmidt-Kaler, F. and Singer, K. and Lutz, E.},
   year={2012},
   month={Nov},
    eprint={1205.1362},
    archivePrefix={arXiv},
    primaryClass={quant-ph}
}

@article{Rossnagel-2016,
   title={A single-atom heat engine},
   volume={352},
   ISSN={1095-9203},
   url={https://dx.doi.org/10.1126/science.aad6320},
   number={6283},
   journal={Science},
   publisher={American Association for the Advancement of Science (AAAS)},
   author={Rossnagel, J. and Dawkins, S. T. and Tolazzi, K. N. and Abah, O. and Lutz, E. and Schmidt-Kaler, F. and Singer, K.},
   year={2016},
   month={Apr},
   pages={325–329},
    eprint={1510.03681},
    archivePrefix={arXiv},
    primaryClass={cond-mat.stat-mech}
}

@article{Niedenzu-2016,
   title={On the operation of machines powered by quantum non-thermal baths},
   volume={18},
   ISSN={1367-2630},
   url={https://dx.doi.org/10.1088/1367-2630/18/8/083012},
   number={8},
   journal={New Journal of Physics},
   publisher={IOP Publishing},
   author={Niedenzu, Wolfgang and Gelbwaser-Klimovsky, David and Kofman, Abraham G and Kurizki, Gershon},
   year={2016},
   month={Aug},
   pages={083012},
    eprint={1508.06519},
    archivePrefix={arXiv},
    primaryClass={quant-ph}
}

@article{Niedenzu-2018,
   title={Quantum engine efficiency bound beyond the second law of thermodynamics},
   volume={9},
   ISSN={2041-1723},
   url={https://dx.doi.org/10.1038/s41467-017-01991-6},
   number={1},
   journal={Nature Communications},
   publisher={Springer Science and Business Media LLC},
   author={Niedenzu, Wolfgang and Mukherjee, Victor and Ghosh, Arnab and Kofman, Abraham G. and Kurizki, Gershon},
   year={2018},
   month={Jan},
    eprint={1703.02911},
    archivePrefix={arXiv},
    primaryClass={quant-ph}
}

@article{Mitchison-2015,
   title={Coherence-assisted single-shot cooling by quantum absorption refrigerators},
   volume={17},
   ISSN={1367-2630},
   url={https://dx.doi.org/10.1088/1367-2630/17/11/115013},
   number={11},
   journal={New Journal of Physics},
   publisher={IOP Publishing},
   author={Mitchison, Mark T and Woods, Mischa P and Prior, Javier and Huber, Marcus},
   year={2015},
   month={Nov},
   pages={115013},
    eprint={1504.01593},
    archivePrefix={arXiv},
    primaryClass={quant-ph}
}

@article{Brask-2015,
   title={Small quantum absorption refrigerator in the transient regime: Time scales, enhanced cooling, and entanglement},
   volume={92},
   ISSN={1550-2376},
   url={https://dx.doi.org/10.1103/PhysRevE.92.062101},
   number={6},
   journal={Physical Review E},
   publisher={American Physical Society (APS)},
   author={Brask, Jonatan Bohr and Brunner, Nicolas},
   year={2015},
   month={Dec},
    eprint={1508.02025},
    archivePrefix={arXiv},
    primaryClass={quant-ph}
}

@article{Pozas-Kerstjens-2018,
   title={A quantum Otto engine with finite heat baths: energy, correlations, and degradation},
   volume={20},
   ISSN={1367-2630},
   url={https://dx.doi.org/10.1088/1367-2630/aaba02},
   number={4},
   journal={New Journal of Physics},
   publisher={IOP Publishing},
   author={Pozas-Kerstjens, Alejandro and Brown, Eric G and Hovhannisyan, Karen V},
   year={2018},
   month={Apr},
   pages={043034},
   eprint={1708.06363},
    archivePrefix={arXiv},
    primaryClass={quant-ph}
}

@article{Torrontegui-2017,
   title={Energy consumption for shortcuts to adiabaticity},
   volume={96},
   ISSN={2469-9934},
   url={https://dx.doi.org/10.1103/PhysRevA.96.022133},
   number={2},
   journal={Physical Review A},
   publisher={American Physical Society (APS)},
   author={Torrontegui, E. and Lizuain, I. and González-Resines, S. and Tobalina, A. and Ruschhaupt, A. and Kosloff, R. and Muga, J. G.},
   year={2017},
   month={Aug},
    eprint={1704.06704},
    archivePrefix={arXiv},
    primaryClass={quant-ph}
}

@article{Chubb-2018,
   title={Beyond the thermodynamic limit: finite-size corrections to state interconversion rates},
   volume={2},
   ISSN={2521-327X},
   url={https://dx.doi.org/10.22331/q-2018-11-27-108},
   journal={Quantum},
   publisher={Verein zur Forderung des Open Access Publizierens in den Quantenwissenschaften},
   author={Chubb, Christopher T. and Tomamichel, Marco and Korzekwa, Kamil},
   year={2018},
   month={Nov},
   pages={108},
 eprint={1711.01193},
    archivePrefix={arXiv},
    primaryClass={quant-ph}
}

@article{Brown-2016,
   title={Passivity and practical work extraction using Gaussian operations},
   volume={18},
   ISSN={1367-2630},
   url={https://dx.doi.org/10.1088/1367-2630/18/11/113028},
   number={11},
   journal={New Journal of Physics},
   publisher={IOP Publishing},
   author={Brown, Eric G and Friis, Nicolai and Huber, Marcus},
   year={2016},
   month={Nov},
   pages={113028},
    eprint={1608.04977},
    archivePrefix={arXiv},
    primaryClass={quant-ph}
}

@article{Perry-2018,
   title={A Sufficient Set of Experimentally Implementable Thermal Operations for Small Systems},
   volume={8},
   ISSN={2160-3308},
   url={https://dx.doi.org/10.1103/PhysRevX.8.041049},
   number={4},
   journal={Physical Review X},
   publisher={American Physical Society (APS)},
   author={Perry, Christopher and Ćwikliński, Piotr and Anders, Janet and Horodecki, Michał and Oppenheim, Jonathan},
   year={2018},
   month={Dec},
    eprint={1511.06553},
    archivePrefix={arXiv},
    primaryClass={quant-ph}
}

@article{Sparaciari-2017,
   title={Energetic instability of passive states in thermodynamics},
   volume={8},
   ISSN={2041-1723},
   url={https://dx.doi.org/10.1038/s41467-017-01505-4},
   number={1},
   journal={Nature Communications},
   publisher={Springer Science and Business Media LLC},
   author={Sparaciari, Carlo and Jennings, David and Oppenheim, Jonathan},
   year={2017},
   month={Dec},
    eprint={1701.01703},
    archivePrefix={arXiv},
    primaryClass={quant-ph}
}

@article{Erker-2017,
   title={Autonomous Quantum Clocks: Does Thermodynamics Limit Our Ability to Measure Time?},
   volume={7},
   ISSN={2160-3308},
   url={https://dx.doi.org/10.1103/PhysRevX.7.031022},
   number={3},
   journal={Physical Review X},
   publisher={American Physical Society (APS)},
   author={Erker, Paul and Mitchison, Mark T. and Silva, Ralph and Woods, Mischa P. and Brunner, Nicolas and Huber, Marcus},
   year={2017},
   month={Aug},
eprint={1609.06704},
archivePrefix={arXiv},
primaryClass={quant-ph}
}

@misc{Taranto-2020,
    title={Exponential improvement for quantum cooling through finite memory effects},
    author={Philip Taranto and Faraj Bakhshinezhad and Philipp Schüttelkopf and Marcus Huber},
    year={2020},
    eprint={2004.00323},
    archivePrefix={arXiv},
    primaryClass={quant-ph}
}

@article{Einstein-1905,
author = {Einstein, A.},
title = "{\"U}ber einen die Erzeugung und Verwandlung des Lichtes betreffenden heuristischen Gesichtspunkt",
journal = {Annalen der Physik},
volume = {322},
number = {6},
pages = {132-148},
url = {https://dx.doi.org/10.1002/andp.19053220607},
year = {1905}
}

@article{Scovil-1959,
  title = {Three-Level Masers as Heat Engines},
  author = {Scovil, H. E. D. and Schulz-DuBois, E. O.},
  journal = {Phys. Rev. Lett.},
  volume = {2},
  issue = {6},
  pages = {262--263},
  year = {1959},
  month = {Mar},
  publisher = {American Physical Society},
  url = {https://dx.doi.org/10.1103/PhysRevLett.2.262}
}

@article{Carnot-1872,
     author = {Carnot, S.},
     title = "R{\'e}flexions sur la puissance motrice du feu et sur les machines propres {\`a} d{\'e}velopper cette puissance",
     journal = "Annales scientifiques de l'{\'E}cole Normale Sup{\'e}rieure",
     publisher = {Elsevier},
     volume = "2e s{\'e}rie, 1",
     year = {1872},
     pages = {393-457},
     url = {https://dx.doi.org/10.24033/asens.88}
}

@article{Gogolin-2016,
   title={Equilibration, thermalisation, and the emergence of statistical mechanics in closed quantum systems},
   volume={79},
   ISSN={1361-6633},
   url={http://dx.doi.org/10.1088/0034-4885/79/5/056001},
   number={5},
   journal={Reports on Progress in Physics},
   publisher={IOP Publishing},
   author={Gogolin, Christian and Eisert, Jens},
   year={2016},
   month={Apr},
   pages={056001},
eprint={1503.07538},
archivePrefix={arXiv},
primaryClass={quant-ph},
}

@Article{Neumann-1929,
author="Neumann, J. v.",
title="Beweis des Ergodensatzes und des H-Theorems in der neuen Mechanik",
journal="Zeitschrift f{\"u}r Physik",
year="1929",
month="Jan",
volume="57",
number="1",
pages="30--70",
issn="0044-3328",
url="https://dx.doi.org/10.1007/BF01339852"
}

@article{Schroedinger-1927,
author = {Schr{\"o}dinger, E.},
title = {Energieaustausch nach der Wellenmechanik},
journal = {Annalen der Physik},
volume = {388},
number = {15},
pages = {956-968},
url = {https://dx.doi.org/10.1002/andp.19273881504},
year = {1927}
}

@Article{Neumann-2010,
author="von Neumann, J.",
title="Proof of the ergodic theorem and the H-theorem in quantum mechanics",
journal="The European Physical Journal H",
year="2010",
month="Nov",
volume="35",
number="2",
pages="201--237",
issn="2102-6467",
url="https://dx.doi.org/10.1140/epjh/e2010-00008-5",
    eprint={1003.2133},
    archivePrefix={arXiv},
    primaryClass={physics.hist-ph}
}

@article{Goldstein-2010,
author = {Goldstein, Sheldon  and Lebowitz, Joel L.  and Mastrodonato, Christian  and Tumulka, Roderich  and Zangh{\`i}, Nino },
title = {Normal typicality and von Neumann's quantum ergodic theorem},
journal = {Proceedings of the Royal Society A: Mathematical, Physical and Engineering Sciences},
volume = {466},
number = {2123},
pages = {3203-3224},
year = {2010},
URL = {https://dx.doi.org/10.1098/rspa.2009.0635},
    eprint={0907.0108},
    archivePrefix={arXiv},
    primaryClass={quant-ph}
}

@article{Srednicki-1994,
  title = {Chaos and quantum thermalization},
  author = {Srednicki, Mark},
  journal = {Phys. Rev. E},
  volume = {50},
  pages = {888--901},
  numpages = {0},
  year = {1994},
  month = {Aug},
  publisher = {American Physical Society},
  url = {https://dx.doi.org/10.1103/PhysRevE.50.888},
    eprint={cond-mat/9403051},
    archivePrefix={arXiv},
    primaryClass={cond-mat}
}

@article{Deutsch-1991,
  title = {Quantum statistical mechanics in a closed system},
  author = {Deutsch, J. M.},
  journal = {Phys. Rev. A},
  volume = {43},
  pages = {2046--2049},
  numpages = {0},
  year = {1991},
  month = {Feb},
  publisher = {American Physical Society},
  url = {https://dx.doi.org/10.1103/PhysRevA.43.2046}
}

@book{Campbell-2019,
author = {Deffner, Sebastian and Campbell, Steve},
title = {Quantum Thermodynamics},
publisher = {Morgan \& Claypool Publishers},
year = {2019},
series = {2053-2571},
isbn = {978-1-64327-658-8},
url = {https://dx.doi.org/10.1088/2053-2571/ab21c6},
    eprint={1907.01596},
    archivePrefix={arXiv},
    primaryClass={quant-ph}
}

@article{Jarzynski-1997,
  title = {Nonequilibrium Equality for Free Energy Differences},
  author = {Jarzynski, C.},
  journal = {Phys. Rev. Lett.},
  volume = {78},
  issue = {14},
  pages = {2690--2693},
  numpages = {0},
  year = {1997},
  month = {Apr},
  publisher = {American Physical Society},
  url = {https://dx.doi.org/10.1103/PhysRevLett.78.2690},
    eprint={cond-mat/9610209},
    archivePrefix={arXiv},
    primaryClass={cond-mat.stat-mech}
}

@article{Crooks-1999,
  title = {Entropy production fluctuation theorem and the nonequilibrium work relation for free energy differences},
  author = {Crooks, Gavin E.},
  journal = {Phys. Rev. E},
  volume = {60},
  pages = {2721--2726},
  numpages = {0},
  year = {1999},
  month = {Sep},
  publisher = {American Physical Society},
  url = {http://dx.doi.org/10.1103/PhysRevE.60.2721},
    eprint={cond-mat/9901352},
    archivePrefix={arXiv},
    primaryClass={cond-mat.stat-mech}
}

@Article{Kosloff-2013,
AUTHOR = {Kosloff, Ronnie},
TITLE = {Quantum Thermodynamics: A Dynamical Viewpoint},
JOURNAL = {Entropy},
VOLUME = {15},
YEAR = {2013},
NUMBER = {6},
PAGES = {2100--2128},
ISSN = {1099-4300},
url = {http://dx.doi.org/10.3390/e15062100}
}

@misc{Potts-2019,
    title={Introduction to Quantum Thermodynamics (Lecture Notes)},
    author={Patrick P. Potts},
    year={2019},
    eprint={1906.07439},
    archivePrefix={arXiv},
    primaryClass={quant-ph}
}

@article{Gorini-1975,
title = "Completely positive dynamical semigroups of N-level systems",
author = "Vittorio Gorini and Andrzej Kossakowski and George Sudarshan",
year = "1975",
month = "12",
volume = "17",
pages = "821--825",
journal = "Journal of Mathematical Physics",
issn = "0022-2488",
publisher = "American Institute of Physics Publising LLC",
number = "5",
url={http://dx.doi.org/10.1063/1.522979}
}

@Article{Lindblad-1976,
author="Lindblad, G.",
title="On the generators of quantum dynamical semigroups",
journal="Communications in Mathematical Physics",
year="1976",
month="Jun",
day="01",
volume="48",
number="2",
pages="119--130",
issn="1432-0916",
url="http://dx.doi.org/10.1007/BF01608499"
}

@article{Mitchison-2019,
author = {Mark T. Mitchison},
title = {Quantum thermal absorption machines: refrigerators, engines and clocks},
journal = {Contemporary Physics},
volume = {60},
number = {2},
pages = {164-187},
year  = {2019},
publisher = {Taylor & Francis},
URL = {http://dx.doi.org/10.1080/00107514.2019.1631555},
    eprint={1902.02672},
    archivePrefix={arXiv},
    primaryClass={quant-ph}
}

@article{Pekola-2015,
author = {Pekola, Jukka P},
title = {Towards quantum thermodynamics in electronic circuits},
journal = {Nature Physics},
volume = {11},
pages = {118},
year  = {2015},
publisher = {Nature Publishing Group},
url = {http://dx.doi.org/10.1038/nphys3169},
}

@article{Klatzow-2019,
  title = {Experimental Demonstration of Quantum Effects in the Operation of Microscopic Heat Engines},
  author = {Klatzow, James and Becker, Jonas N. and Ledingham, Patrick M. and Weinzetl, Christian and Kaczmarek, Krzysztof T. and Saunders, Dylan J. and Nunn, Joshua and Walmsley, Ian A. and Uzdin, Raam and Poem, Eilon},
  journal = {Phys. Rev. Lett.},
  volume = {122},
  issue = {11},
  pages = {110601},
  numpages = {6},
  year = {2019},
  month = {Mar},
  publisher = {American Physical Society},
  url = {http://dx.doi.org/10.1103/PhysRevLett.122.110601},
    eprint={1710.08716},
    archivePrefix={arXiv},
    primaryClass={quant-ph}
}

@article {Langen-2015,
	author = {Langen, Tim and Erne, Sebastian and Geiger, Remi and Rauer, Bernhard and Schweigler, Thomas and Kuhnert, Maximilian and Rohringer, Wolfgang and Mazets, Igor E. and Gasenzer, Thomas and Schmiedmayer, J{\"o}rg},
	title = {Experimental observation of a generalized Gibbs ensemble},
	volume = {348},
	number = {6231},
	pages = {207--211},
	year = {2015},
	publisher = {American Association for the Advancement of Science},
	issn = {0036-8075},
	url = {http://dx.doi.org/10.1126/science.1257026},
	journal = {Science},
    eprint={1411.7185},
    archivePrefix={arXiv},
    primaryClass={cond-mat.quant-gas}
}

@Inbook{Schmiedmayer-2018,
author="Schmiedmayer, J{\"o}rg",
editor="Binder, Felix
and Correa, Luis A.
and Gogolin, Christian
and Anders, Janet
and Adesso, Gerardo",
title="One-Dimensional Atomic Superfluids as a Model System for Quantum Thermodynamics",
bookTitle="Thermodynamics in the Quantum Regime: Fundamental Aspects and New Directions",
year="2018",
publisher="Springer International Publishing",
address="Cham",
pages="823--851",
isbn="978-3-319-99046-0",
url="http://dx.doi.org/10.1007/978-3-319-99046-0_34",
    eprint={1805.11539},
    archivePrefix={arXiv},
    primaryClass={quant-ph}
}

@Inbook{Bar-Gill-2018,
author="Bar-Gill, Nir",
editor="Binder, Felix
and Correa, Luis A.
and Gogolin, Christian
and Anders, Janet
and Adesso, Gerardo",
title="NV Color Centers in Diamond as a Platform for Quantum Thermodynamics",
bookTitle="Thermodynamics in the Quantum Regime: Fundamental Aspects and New Directions",
year="2018",
publisher="Springer International Publishing",
address="Cham",
pages="983--998",
isbn="978-3-319-99046-0",
url="http://dx.doi.org/10.1007/978-3-319-99046-0_41"
}

@Inbook{Dawkins-2018,
author="Dawkins, Samuel T.
and Abah, Obinna
and Singer, Kilian
and Deffner, Sebastian",
editor="Binder, Felix
and Correa, Luis A.
and Gogolin, Christian
and Anders, Janet
and Adesso, Gerardo",
title="Single Atom Heat Engine in a Tapered Ion Trap",
bookTitle="Thermodynamics in the Quantum Regime: Fundamental Aspects and New Directions",
year="2018",
publisher="Springer International Publishing",
address="Cham",
pages="887--896",
isbn="978-3-319-99046-0",
url="http://dx.doi.org/10.1007/978-3-319-99046-0_36"
}

@article{An-2014,
   title={Experimental test of the quantum Jarzynski equality with a trapped-ion system},
   volume={11},
   ISSN={1745-2481},
   url={http://dx.doi.org/10.1038/nphys3197},
   number={2},
   journal={Nature Physics},
   publisher={Springer Science and Business Media LLC},
   author={An, Shuoming and Zhang, Jing-Ning and Um, Mark and Lv, Dingshun and Lu, Yao and Zhang, Junhua and Yin, Zhang-Qi and Quan, H. T. and Kim, Kihwan},
   year={2014},
   month={Dec},
   pages={193–199},
    eprint={1409.4485},
    archivePrefix={arXiv},
    primaryClass={quant-ph}
}

@Inbook{Lu-2018,
author="Lu, Yao
and An, Shuoming
and Zhang, Jing-Ning
and Kim, Kihwan",
editor="Binder, Felix
and Correa, Luis A.
and Gogolin, Christian
and Anders, Janet
and Adesso, Gerardo",
title="Probing Quantum Fluctuations of Work with a Trapped Ion",
bookTitle="Thermodynamics in the Quantum Regime: Fundamental Aspects and New Directions",
year="2018",
publisher="Springer International Publishing",
address="Cham",
pages="917--938",
isbn="978-3-319-99046-0",
url="http://dx.doi.org/10.1007/978-3-319-99046-0_38",
    eprint={1902.00206},
    archivePrefix={arXiv},
    primaryClass={quant-ph}
}

@article{Popescu-2006,
   title={Entanglement and the foundations of statistical mechanics},
   volume={2},
   ISSN={1745-2481},
   url={http://dx.doi.org/10.1038/nphys444},
   number={11},
   journal={Nature Physics},
   publisher={Springer Science and Business Media LLC},
   author={Popescu, Sandu and Short, Anthony J. and Winter, Andreas},
   year={2006},
   month={Oct},
   pages={754–758},
    eprint={quant-ph/0511225},
    archivePrefix={arXiv},
    primaryClass={quant-ph}
}

@book{Gemmer-2010,
  url = {https://dx.doi.org/10.1007/978-3-540-70510-9},
  year = {2010},
  publisher = {Springer Berlin Heidelberg},
  author = {Jochen Gemmer and M. Michel and G\"{u}nter Mahler},
  title = {Quantum Thermodynamics}
}

@book{Eriksen-2010,
publisher = {Pluto Press},
isbn = {9780745330495},
title = {Small Places, Large Issues - Third Edition : An Introduction to Social and Cultural Anthropology},
address = {345 Archway Road, London N6 5AA},
author = {Eriksen, Thomas Hylland},
year = {2010}
}

@article{Pfaffenberger-1988,
  url = {https://dx.doi.org/10.2307/2802804},
  year = {1988},
  month = jun,
  publisher = {{JSTOR}},
  volume = {23},
  number = {2},
  pages = {236},
  author = {Bryan Pfaffenberger},
  title = {Fetishised Objects and Humanised Nature: Towards an Anthropology of Technology},
  journal = {Man}
}

@article{Figl-2019,
 author  = {Figl, Bettina},
 year    = {10.12.2019},
 title   = {Forschung, finanziert vom US-Militär},
 journal = {Wiener Zeitung},
 url     = {https://www.wienerzeitung.at/nachrichten/chronik/oesterreich/2041937-Forschung-finanziert-vom-US-Militaer.html?em_no_split=1},
Note={Accessed on 16.05.2020}
}


\end{document}